\documentclass[oneside,chap]{thesis}

\usepackage{algorithm2e,algorithmic,framed}
\usepackage{amssymb,amsfonts,amsmath}
\usepackage{amsfonts,amsmath}
\usepackage{graphicx}
\usepackage{url}
\usepackage{bm}
\usepackage{amsthm,amsmath,amssymb,graphicx,epsfig,color}
\long\def\symbolfootnote[#1]#2{\begingroup%
\def\thefootnote{\fnsymbol{footnote}}\footnote[#1]{#2}\endgroup}

\newcommand{\Exp}{\mbox{}{\bf{E}}}
\newcommand{\Expect}[1]{\mbox{}{\bf{E}}\left[#1\right]}

\newcommand{\FNorm }[1]{\mbox{}\|#1\|_F  }
\newcommand{\FNormS}[1]{\mbox{}\|#1\|_F^2}

\newcommand{\TNorm }[1]{\mbox{}\|#1\|_2  }
\newcommand{\TNormS}[1]{\mbox{}\|#1\|_2^2}

\newcommand{\XNorm }[1]{\mbox{}\|#1\|_{\xi}  }
\newcommand{\XNormS}[1]{\mbox{}\|#1\|_{\xi}^2}
\newcommand{\PNorm }[1]{\mbox{}\|#1\|_{p}  }

\newtheorem{theorem}{\bf Theorem}[]
\newtheorem{lemma}[theorem]{Lemma}
\newtheorem{definition}[theorem]{Definition}
\newtheorem{proposition}[theorem]{Proposition}
\newtheorem{corollary}[theorem]{Corollary}

\newcommand{\qedsymb}{\hfill{\rule{2mm}{2mm}}}

\newcommand{\transp}{^{\textsc{T}}}
\newcommand{\trace}{\text{\rm Tr}}
\newcommand{\mat}[1]{{\ensuremath{\textsc{#1}}}}

\newcommand{\abs }[1]{\left|#1\right|}
\newcommand{\EE}[1]{\ensuremath{\mathbb{E}\left[#1\right] } }
\newcommand{\var}[1]{\text{Var}\ensuremath{\left[#1\right] } }

\def\rank{\hbox{\rm rank}}

\def\b{{\mathbf b}}
\def\e{{\mathbf e}}
\def\q{{\mathbf q}}

\def\s{{\mathbf s}}
\def\u{{\mathbf u}}
\def\v{{\mathbf v}}

\def\A{\matA}

\def\C{\matC}

\def\matA{\mat{A}}
\def\matB{\mat{B}}
\def\matC{\mat{C}}
\def\matD{\mat{D}}
\def\matE{\mat{E}}
\def\matH{\mat{H}}
\def\matI{\mat{I}}
\def\matP{\mat{P}}
\def\matQ{\mat{Q}}
\def\matR{\mat{R}}
\def\matS{\mat{S}}

\def\matU{\mat{U}}
\def\matV{\mat{V}}
\def\matW{\mat{W}}
\def\matX{\mat{X}}
\def\matY{\mat{Y}}
\def\matZ{\mat{Z}}

\def\scl{{\textsc{l}}}
\def\scu{{\textsc{u}}}
\def\phiu{{\overline{\phi}}}

\def\phil{{\underbar{\math{\phi}}}}
\DeclareMathSymbol{\Prob}{\mathbin}{AMSb}{"50}
\newcommand\remove[1]{}

\def\nnz{{ \rm nnz }}

\def\math#1{$#1$}

\def\mand#1{$$#1$$}
\def\frac#1#2{{#1\over #2}}

\def\mld#1{\begin{equation}
#1
\end{equation}}
\def\eqar#1{\begin{eqnarray}
#1
\end{eqnarray}}
\def\eqan#1{\begin{eqnarray*}
#1
\end{eqnarray*}}

\DeclareMathSymbol{\R}{\mathbin}{AMSb}{"52}

\def\qed{\hfill\rule{2mm}{2mm}}

\def\cl#1{{\cal #1}}

\def\argmin{\mathop{\hbox{argmin}}\limits}

\def\x{{\mathbf x}}
\def\y{{\mathbf y}}
\def\z{{\mathbf z}}

\def\a{{\mathbf a}}
\def\b{{\mathbf b}}

\def\norm#1{{\|#1\|}}

\def\ceil#1{{\left\lceil\,#1\,\right\rceil}}
\def\r#1{{(\ref{#1})}}

\def\dotfil{\leaders\hbox to 1.5mm{.}\hfill}
\newcounter{rmnum}
\def\RN#1{\setcounter{rmnum}{#1}\uppercase\expandafter{\romannumeral\value{rmnum}}}
\def\rn#1{\setcounter{rmnum}{#1}\expandafter{\romannumeral\value{rmnum}}}

\begin{document}

%
\thesistitle{\bf Topics in Matrix Sampling Algorithms}
\author{Christos Boutsidis}
\degree{Doctor of Philosophy}
\department{Computer Science}
\signaturelines{6}     
\thadviser{Petros Drineas}

\memberone{Kristin P. Bennett} 
\membertwo{Sanmay Das}
\memberthree{Malik Magdon-Ismail} 
\memberfour{Michael W. Mahoney}
\memberfive{Mark Tygert}

\submitdate{May 2011\\(For Graduation May 2011)}
\copyrightyear{2011}   

%
\titlepage
\abstitlepage          
\copyrightpage         
\tableofcontents

\listoftables 
 
\specialhead{ACKNOWLEDGMENT}
 
I did most of the work presented in this document  as a graduate student 
in the Computer Science Department at Rensselaer Polytechnic Institute 
during the period Sept. 2006 - May 2011 supervised by Petros Drineas. 
Part of this work was done as a visiting student in 
(i)  the IBM T.J Watson Research Lab (summer, 2008);
(ii) the Institute of Pure and Applied  Mathematics at the 
University of California, at Los Angeles (fall, 2008);
(iii) the IBM Zurich Research Lab (summer/fall, 2009); and
(v) WorldQuant, LLC (summer, 2010). I am grateful to all these
places for their hospitality. 

During these years I was benefited by discussing the topics
of this dissertation with 
Nikos Anerousis        (at IBM Watson), 
Costas Bekas           (at IBM Zurich), 
Kristin Bennet         (at Rensselaer),
Ali Civril             (st Rensselaer),
Petros Drineas         (at Rensselaer),
Efstratios Gallopoulos (at University of Patras),
Michael Mahoney        (at Stanford),
Malik Magdon Ismail    (at Rensselaer),
Efi Kokiopoulou        (at ETH Zurich), 
Michalis Raptis        (at UCLA), 
Jimeng Sun             (at IBM Watson), 
Spyros Stefanou        (at WorldQuant),
Charalabos Tsourakakis (at Carnegie Mellon University), 
Mark Tygert            (at UCLA), 
Michalis Vlachos       (at IBM Zurich), and
Anastasios Zouzias     (at University of Toronto).

Finally, many thanks go to the people that read parts
of this document and found numerous typos and mistakes:
Sanmay Das, Sotiris Gkekas, Charalabos Tsourakakis, and 
Anastasios Zouzias. 

\specialhead{ABSTRACT} 
We study three fundamental problems of Linear Algebra, 
lying in the heart of various Machine Learning applications, namely:
(i)   Low-rank  Column-based Matrix Approximation, 
(ii)  Coreset   Construction in Least-Squares Regression, and
(iii) Feature   Selection in $k$-means Clustering. 
A high level description of these problems is as follows: 
given a matrix \math{\matA} and an integer $r$, what
are the $r$ most ``important'' columns (or rows) in $\matA$? A more detailed
description is given momentarily. 

{\bf 1. Low-rank Column-based Matrix Approximation.} 
We are given a matrix $\matA$ and a target rank $k$.
The goal is to select a subset of columns of $\matA$ and, by using
only these columns, compute a rank $k$ approximation to $\matA$ that is as good as the
rank $k$ approximation that would have been obtained by using all the columns.

{\bf 2. Coreset Construction in Least-Squares Regression.} 
We are given a matrix $\matA$ and a vector $\b$. 
Consider the (over-constrained) 
least-squares problem of minimizing $\TNorm{\matA\x-\b}$, over all vectors $\x \in\cl D$.
The domain \math{\cl D} represents the constraints on the solution and can be arbitrary.
The goal is to select a subset of the rows of $\matA$ and $\b$
and, by using only these rows, find a solution vector that is as good 
as the solution vector that would have been obtained by using all the rows. 
  
{\bf 3. Feature Selection in K-means Clustering.} 
We are given a set of points described with respect to a large number of features. 
The goal is to select a subset of the features and, by using only this subset, 
obtain a $k$-partition of the points that is as good as the partition that would 
have been obtained by using all the features.

We present novel algorithms for all three problems mentioned above.
Our results can be viewed as follow-up research to a 
line of work known as ``Matrix Sampling Algorithms''. Frieze et al~\cite{FKV98}
presented the first such algorithm for the Low-rank Matrix Approximation
problem. Since then, such algorithms have been developed for several other
problems, e.g. Regression~\cite{DMM06a}, Graph Sparsification~\cite{SS08}, 
and Linear Equation Solving~\cite{Spi10}. 
Our contributions to this line of research are: 
(i) improved algorithms for Low-rank Matrix Approximation and Regression
(ii) algorithms for a new problem domain ($K$-means Clustering).            

\chapter{INTRODUCTION} \label{chap1}

We study several topics in the area of ``Matrix Sampling Algorithms''.
Chapter \ref{chap3} gives a comprehensive overview of this area; 
here, we summarize the main idea. Let the matrix $\matA \in \R^{m \times n}$ be the input to a
linear algebraic or machine learning problem.  
Assume that there is an \emph{Algorithm} that solves this problem exactly in 
$O( f(m,n) )$. Consider the following question. 
\begin{center}
{ \bf Does there exist a small subset of columns (or rows, or elements) in $\matA$ 
such that running this \emph{Algorithm} on this subset gives an approximate 
solution to the problem in $O( g(m,n) )$ time with $g(m,n) = o(f(m,n))$? } 
\end{center}
A positive answer to this question implies the following approach to solve the problem:
first, by using a matrix sampling algorithm, select a few columns from $\matA$; then,
compute an approximate solution to the problem by running the \emph{Algorithm} on the selected columns. 
The challenge is to select a subset of columns from $\matA$ such that the solution obtained by running the \emph{Algorithm}
on this subset is as good as the solution that would have been obtained by running the \emph{Algorithm} on $\matA$. 
The focus of this dissertation is exactly on developing such ``good'' matrix sampling algorithms. We do so for three
problems: 
\begin{enumerate}
\item { \bf  Low-rank Matrix Approximation. }
\item { \bf  Least-Squares Regression. }
\item { \bf  $K$-means Clustering. }
\end{enumerate}
\noindent Prior work has offered good matrix sampling algorithms for several other problems,
including Low-rank Matrix Approximation and Regression, e.g.: 
(i)   Matrix Multiplication~\cite{DK01}, 
(ii)   Graph Sparsification~\cite{ST08b}, and 
(iii)  Linear Equation Solving~\cite{ST08c}. 
The - high level - contributions of this thesis are: 
(i) improved algorithms for Low-rank Matrix Approximation and Regression
(ii) novel algorithms for a new problem domain ($k$-means Clustering). 
 
Our {\bf motivation} to study the above problems is two-fold. 
First, by selecting the most important columns from the input matrix, 
we quickly reveal the most important and meaningful information in it.
Consider, for example, a matrix describing stock prices with the rows
corresponding to dates and the columns corresponding to stocks. A small
subset of columns that reconstructs this matrix corresponds to a small subset of
stocks that ``reconstructs'' the whole portfolio. Identifying the ``dominant'' 
stocks in a portfolio is certainly a (possibly multimillion) worth doing task. 
Second, by selecting a small portion of the input data and solving a smaller problem, 
we are able to improve the computational efficiency of standard algorithms. 
For example, continuing on the theme of the above discussion, solving a regression 
problem with a small number of rows from $\matA$ and $\b$, as opposed to the
(possibly severely overconstrained) problem involving $\matA, \b$ would
make statistical arbitrage algorithms able to deliver solutions much faster,
which is important, for example, in high-frequency trading strategies.
On top of the improved computational efficiency, since the smaller problem is often
more robust to noise and outliers, these strategies would be attractive
from the risk minimization point of view. 

The highlights of our {\bf contributions} for the three problems we mentioned above are as follows:
{\bf 1)} We offer fast, accurate, deterministic 
algorithms for column-based low-rank matrix approximations. We achieve computational 
efficiency by introducing a novel framework in Section~\ref{chap22} where one is able to work with approximate SVD 
factorizations and select columns from matrices. Previous work for column-based low-rank matrix
approximations uses the exact SVD, which is expensive. We achieve near-optimal and
deterministic algorithms by generalizing a recent important result for decompositions
of the identity~\cite{BSS09} in Sections~\ref{chap317} and~~\ref{chap318}. 
Previous work offers algorithms that are not
optimal and are typically randomized. 
{\bf 2)} We present the first deterministic algorithm for coreset construction in
least-squares regression with arbitrary constraints. We achieve that by putting the 
result of~\cite{BSS09} in the linear regression setting. Previous
work on this topic offers randomized algorithms and less accurate bounds than ours. 
{\bf 3)} We present the first provably accurate feature selection algorithm for $k$-means clustering. 
We achieve that by looking at this problem from a linear algebraic point of view.

\section{Structure }\label{chap11}
Sections \ref{chap12}, \ref{chap13}, and \ref{chap14}  
of this chapter introduce the three problems that we study
in detail in the present dissertation. For each problem, we
give motivating examples, a brief review of prior work, 
and a resume of our results. 
We conclude this chapter in Section \ref{chap15} by describing the main idea
behind the algorithms of this dissertation and the proof techniques
of our results.
Chapter \ref{chap2} introduces the notation and provides the necessary background from Linear Algebra and Probability Theory. 
Chapter \ref{chap3} gives a comprehensive overview of existing work on the  topic of matrix sampling algorithms.
Our results on low-rank matrix approximation, least-squares regression,
and $k$-means clustering are given in detail in Chapters \ref{chap4}, \ref{chap5}, and \ref{chap6}, respectively. 
Finally, we discuss directions for future research in Chapter \ref{chap7}. 

\section{Low-rank Column-based Matrix Approximation }\label{chap12}

Given \math{\matA \in \R^{m \times n}} of rank $\rho$ and $k < \rho$, 
the best rank \math{k} approximation to $\matA$ is 
$$\matA_k=\sum_{i=1}^k\sigma_i\u_i\v_i\transp,$$ 
where
\math{\sigma_1\ge\sigma_2\ge\cdots\ge\sigma_k\ge 0} are the top \math{k} singular values of $\matA$, with associated left and right singular vectors
\math{\u_i \in \R^m} and \math{\v_i \in \R^n}, respectively. 
The singular values and singular vectors of $\matA$ can be computed via the Singular Value Decomposition (SVD) of
$\matA$ in deterministic \math{O( mn\min\{m,n\})} time. There is considerable interest (e.g. \cite{CH92,FKV98,DV06,DRVW06,LWMRT07,SXZF07,DR10,HMT}) in determining a minimum set of $r \ll n$ columns of $\matA$  which is approximately as good as \math{\matA_k} at reconstructing \math{\matA}.
For example, \cite{DRVW06} provides a connection between a ``good'' subset of columns of $\matA$ and the projective clustering problem; \cite{DV06} efficiently computes low-rank approximations by using a subset of ``important'' columns of $\matA$; 
\cite{SXZF07} uses carefully selected ``informative'' columns and rows of $\matA$ to interpret large scale datasets; 
\cite{CH92} shows that a small ``linearly independent'' set of columns of $\matA$ are more robust to noise in least-squares regression.

\paragraph{Problem Setup.}
Fix $\matA \in \mathbb{R}^{m \times n}$ of rank $\rho$, integer $k < \rho$, and oversampling parameter $k \leq r < n$.
The goal is to construct $\matC \in \mathbb{R}^{m \times r}$ consisting of $r$ columns of $\matA$. $\matA$ is the target 
matrix for the approximation, $k$ is the target rank, and $r$ is the number of columns to be selected.  
We are interested in the reconstruction errors: 
$$
\XNorm{\matA-\matC\matC^+\matA} \qquad \mbox{and} \qquad
\XNorm{\matA-\Pi^\xi_{\matC,k}(\matA)},
$$
for $\xi=2,F$.
The former is the reconstruction error
for \math{\matA} using the columns in \math{\matC};
the latter is the error from the
best (under the appropriate norm) rank \math{k}
reconstruction of
\math{\matA} within the column space of $\matC$. 
For fixed $\matA$, $k$, and $r$, we would like these errors to be as close
to \math{\XNorm{\matA-\matA_k}} as possible. 
Note that $\XNorm{\matA-\matC\matC^+\matA} \le \XNorm{\matA-\Pi^\xi_{\matC,k}(\matA)}$; 
so, the way we will present our results is the following:
$$ \XNorm{\matA-\Pi^\xi_{\matC,k}(\matA)} \le \alpha \XNorm{\matA-\matA_k}.$$ 
$\alpha$ is the approximation factor of the corresponding algorithm; the goal
is to design algorithms that
select $r$ columns from $\matA$ and offer ``small'' $\alpha$. 

\paragraph{Prior Work.} If $r = k$, prior work provides almost near-optimal algorithms.
``Near-optimal'' means that the factors $\alpha$ offered by the corresponding algorithms are
- asymptotically - the best possible. So, there is no room to improve on these results, modulo running time. 
We will indeed present fast algorithms for the $r = k$ case that are almost as accurate as the best 
existing ones
(see Theorem \ref{thmCSSPs} for $\xi = 2$, Theorem \ref{thmCSSPf} for $\xi = F$, and Theorem
\ref{fastcssp} for $\xi=2,F$).

For general $r > k$, not much is known in existing literature. For spectral norm, we are not familiar
with any technique addressing this problem; for Frobenius norm, existing algorithms are not optimal and work only
if r = $\Omega(k \log(k))$~\cite{DMM06d,DV06}. 
For example, ~\cite{DMM06d} describes a $O( m n \min\{m,n\} +  r \log(r) )$ time
algorithm that, with constant probability, guarantees:   
$$ \FNorm{\matA-\Pi^F_{\matC,k}(\matA)} \le 
\sqrt{  1 + O\left( \frac{ k \log(k)}{r} \right)} \FNorm{\matA-\matA_k}.$$ 

\paragraph{Results.}
We present the first polynomial-time near-optimal algorithms for arbitrary
$r>k$ for both $\xi=2,F$. To be precise, 
for $\xi=F$, we present a ``near-optimal'' algorithm 
only for $r > 10k$ (see Theorem \ref{thmFast3} in Section \ref{chap42}); 
for arbitrary $r > k$ our result is 
optimal up to an additive factor $1$ 
(see Theorem \ref{theorem:intro2} in Section \ref{chap42}). We note that, for $\xi = 2$ and $r > k$, 
the best possible value for $\alpha$ is $\hat\alpha = \sqrt{\frac{n}{r}}$; for $\xi = F$ and $r > k$, 
the best possible approximation is $\hat\alpha = \sqrt{1 + \frac{k}{2r}}$. 
We break our results into three categories.

{\bf 1.} For $r > k$ and $\xi = 2$: (i) Theorem \ref{theorem:intro1} in Section \ref{chap41}
presents a deterministic algorithm that runs in 
$O\left( mn\min\{m,n\} + rn\left(k^2+\left(\rho-k\right)^2\right)\right)$
and achieves error:
$$ \TNorm{\matA-\Pi^2_{\matC,k}(\matA)} \le 
O\left(  \sqrt{\frac{\rho}{r}}  \right)
\TNorm{\matA-\matA_k};
$$ 
(ii) Theorem \ref{thmFast1} in Section \ref{chap41} presents a considerably faster randomized algorithm that runs in
$O\left( 
mnk \log\left( k^{-1} \min\{m,n\}\right)  + nrk^2      
\right)$ and achieves, in expectation, error: 
$$ \Expect{\TNorm{\matA-\Pi^2_{\matC,k}(\matA)}} \le 
O\left( \sqrt{ \frac{n}{r} } \right)
\TNorm{\matA-\matA_k}.
$$

{\bf 2.} For $r > k$ and $\xi = F$: (i) Theorem \ref{theorem:intro2} in Section \ref{chap42} 
presents a deterministic algorithm that runs in $O\left(mn\min\{m,n\} + nrk^2\right)$
and achieves error:
$$ \FNorm{\matA-\Pi^F_{\matC,k}(\matA)} \le 
\sqrt{ 2  + O\left( k/r  \right) }
\FNorm{\matA-\matA_k};
$$ 
(ii) Theorem \ref{thmFast2} in Section \ref{chap42} presents a considerably faster randomized algorithm that runs in
$O\left(mnk+nrk^2\right)$ 
and achieves, in expectation, error:
$$ \Expect{\FNorm{\matA-\Pi^F_{\matC,k}(\matA)}} \le 
\sqrt{ 3  + O\left( k/r  \right) }
\FNorm{\matA-\matA_k}.
$$

{\bf 3.} Finally, for $r > 10 k$ and $\xi = F$, 
Theorem \ref{thmFast3} in Section \ref{chap42} presents a randomized algorithm that runs in 
$O\left(mnk + n k^3 + n\log(r) \right)$ and achieves, in expectation,
error:
$$ \Expect{\FNorm{\matA-\Pi^F_{\matC,k}(\matA)}} \le 
\sqrt{ 1  + O\left( k/r  \right) }
\FNorm{\matA-\matA_k}. 
$$

\section{Coreset Construction in Least-Squares Regression  }\label{chap13}

Linear regression is one of the beloved techniques in the statistical analysis
of data, since it provides a powerful tool for information extraction~\cite{SLS77}. 
Research in this area ranges from numerical solutions of 
regression problems~\cite{Bjo96} to robustness of the prediction 
error to noise (e.g. using feature selection methods~\cite{GE03}).
We don't address neither of these issues here; rather, we
focus on constructing coresets for constrained least-squares regression. 
A coreset is a subset of the data that contains essentially as much information 
(when viewed through the lens of the learning model) as the original data. 
For example, with support vector classification, the support 
vectors are a coreset~\cite{CS00}. A coreset contains the meaningful or important
information and provides a good summary of the data. If such a coreset can be found quickly, 
one could obtain good approximate solutions by solving a (much) smaller regression problem. 
When the constraints are complex  (e.g. non-convex constraints), 
solving a smaller  problem could be a significant saving~\cite{Gao07}.
(See section \ref{chap33} for further motivation on least-squares problems with constraints.)
\vspace{-.19in}
\paragraph{Problem Setup.} 
Assume \math{m} data points
\math{(\z_1,y_1),\ldots,(\z_m,y_m)}; \math{\z_i \in \R^n}
are features and \math{y_i \in \R} are targets (responses). Typically $m \gg n$. 
The linear regression problem asks to determine a vector \math{\x_{opt}
\in\cl D\subseteq\R^n}
that minimizes
\vspace{-.1in}
\mand{
\cl E(\x)=\sum_{i=1}^n(\z_i\transp\cdot \x-y_i)^2,
}
over \math{\x\in\cl D};
so, \math{\cl E(\x_{opt})\le \cl E(\x)}, \math{\forall\x\in\cl D}.
The domain \math{\cl D} represents the constraints on the solution
and can be arbitrary. A coreset of size \math{r < m} is a subset of the data,
\math{(\z_{i_1},y_{i_1}),\ldots,(\z_{i_r},y_{i_r})}.
The coreset regression problem considers the squared error, 
\vspace{-.1in}
\mand{
\tilde{\cl E}(\x)=\sum_{j=1}^r(\z_{i_j}\transp\cdot \x-y_{i_j})^2.
}
Suppose that \math{\tilde{\cl E}} is minimized at
\math{\tilde\x_{opt}}, so 
\math{\tilde{\cl E}(\tilde\x_{opt})\le \tilde{\cl E}(\x)}, 
\math{\forall\x\in\cl D}; the goal is to construct a coreset 
\math{(\z_{i_1},y_{i_1}),\ldots,(\z_{i_r},y_{i_r})} such that
\math{\tilde\x_{opt}} is nearly as good as \math{\x_{opt}}.
For fixed $\matA, \b$, \math{\epsilon>0}, and $r$ being as small as possible, the goal is to
find \math{\tilde\x_{opt}} with
$$
\cl E(\tilde\x_{opt})\le (1+\epsilon) \cl E(\x_{opt}).
$$ 

Below, we switch to a more convenient matrix formulation of the problem. 
Let \math{\matA\in\R^{m\times n}} ($m \gg n$) be the
\emph{data matrix} whose rows are the data points 
 \math{\z_i\transp}, so
\math{\matA_{ij}=\z_i[j]}; and,
\math{\b\in\R^m} is the  target vector, so
\math{b_i=y_i}. Also,  
\math{\cl E(\x)=\TNormS{\matA\x-\b}} and
$$\x_{opt}=\argmin_{\x\in\cl D}\TNormS{\matA\x-\b}.$$
A coreset of size \math{r < m}
is a subset \math{\matC \in \R^{r \times n}} of the rows of \math{\matA} 
and the corresponding
elements \math{\b_c \in \R^r}
of \math{\b} (possibly rescaled); so, \math{
\tilde{\cl E}(\x)=\TNormS{\matC\x-\b_c}} and 
$$
\tilde\x_{opt}=\argmin_{\x\in\cl D}\norm{\matC\x-\b_c}_2^2.
$$
A coreset \math{\matC, \b_c} is \math{(1+\epsilon)}-approximate if the corresponding vector $\tilde\x_{opt}$ satisfies
$$
\TNormS{\matA\tilde\x_{opt}-\b}\le
(1+\epsilon)\TNormS{\matA\x_{opt}-\b}
.$$

\paragraph{Prior Work.} The best (in terms of coreset size)
$(1+\epsilon)$-coreset construction algorithm is in~\cite{DMM06a}.
It finds a coreset of size \math{r= O(n \log(n) \epsilon^{-2})}, works only
for unconstrained regression, runs in \math{O( mn^2 +  n \log(n) \epsilon^{-2} \log(n \log(n) \epsilon^{-1}))}, and
fails with some constant probability.
There are also several random-projection type algorithms that construct ``coresets''
for unconstrained regression~\cite{Sar06, DMMS07, RT08, Zou10}.
Here, the rows in $\matC$ are linear combinations of the rows of $\matA$ (similarly
for $\b_c$ and $\b$). The best of these methods (in terms of ``coreset'' size) is in~\cite{Zou10} and constructs a 
$(1+\epsilon)$-``coreset'' with $r = O(n \epsilon^{-1})$ .
(Table~\ref{table:34} in Section~\ref{chap33} gives a summary of these results.)

\paragraph{Results.} Our main result is Theorem~\ref{lem:regression} in Section \ref{chap51}, which 
describes a deterministic $O\left( mn^2 + n^3 \epsilon^{-2}  \right)$ algorithm that  
constructs a \math{(1+\epsilon)}-approximate coreset of size \math{r = O(n \epsilon^{-2})}. 
Our result improves upon~\cite{DMM06a} on three aspects: (i) it is deterministic,
(ii) the coreset size is $O(\log(n))$ smaller, and (iii) handles arbitrary constraints.  
Also, Theorem \ref{lem:regression2} in Section \ref{chap52} presents a randomized algorithm that, with constant 
probability, constructs a \math{(1+\epsilon)}-approximate ``coreset'' of size \math{r = O(n \ln(n) \log(n m) \epsilon^{-2})}.
Our result improves upon~\cite{Sar06, DMMS07, RT08, Zou10} 
by means of providing a ``coreset'' for regression of arbitrary constraints. 

\section{Feature Selection in $k$-means Clustering  } \label{chap14}
Clustering is ubiquitous in science and engineering, with numerous
and diverse application domains, ranging from bioinformatics and
medicine to the social sciences and the web~\cite{Har75}. Perhaps
the most well-known clustering algorithm is the so-called
``$k$-means'' algorithm or Lloyd's method \cite{Llo82}, an
iterative expectation-maximization type approach, which attempts
to address the following objective: given a set of points in a
Euclidean space and the number of
clusters $k$, split the points into $k$ clusters so that the total
sum of the (squared Euclidean) distances of each point to its
nearest cluster center is minimized. The good behavior of the Lloyd's method
(\cite{Llo82,ORSS06}), have made $k$-means enormously popular in
applications~\cite{Wu07}.

In recent years, the high dimensionality of the modern massive
datasets has provided a considerable challenge to $k$-means
clustering approaches. First, the curse of dimensionality makes
algorithms for $k$-means clustering very slow, and, second, the
existence of many irrelevant features may not allow the
identification of the relevant underlying structure in the data
\cite{GGBD05}. Practitioners addressed such obstacles by
introducing feature selection and feature extraction techniques.
It is worth noting that feature selection selects a small subset
of actual features from the data and then runs the clustering
algorithm only on the selected features, whereas feature
extraction constructs a small set of artificial features
and then runs the clustering algorithm on the constructed
features. 

\paragraph{Problem Setup.}
Consider $m$ points $\mathcal{P} = \{ p_1, p_2, ..., p_m \} \in \R^n$, and integer
$k$ denoting the number of clusters. The objective of 
$k$-means is to find a $k$-partition of 
$\mathcal{P}$ such that points that are ``close'' to each other belong to the same cluster and points 
that are ``far'' from each other belong to different clusters. 
A $k$-partition of $\mathcal{P}$ is a collection 
$\cl S=\{\mathcal{S}_1, \mathcal{S}_2, ..., \mathcal{S}_k\}$
of \math{k} non-empty pairwise disjoint
sets which covers \math{\cl P}.
Let $s_j=|\mathcal{S}_j|$ be the size of $\mathcal{S}_j$. 
For each set \math{S_j}, let \math{\bm\mu_j\in\R^n} be its centroid 
(the mean point):
\math{\bm\mu_j=(\sum_{p_i\in S_j}p_i)/s_j}.
The $k$-means objective function is 
$$
\mathcal{F}(\mathcal{P}, \cl S) = 
\sum_{i=1}^m\norm{p_i-\bm\mu(p_i)}_2^2;
$$
\math{\bm\mu(p_i) \in \R^n} is the centroid of the cluster to which \math{p_i} belongs.
The goal of $k$-means is to find the partition 
$$\cl S_{opt}= \arg \min_{\cl S} \cl F (\mathcal{P}, \cl S).$$
The goal of feature selection is to construct $r$-dimensional points 
$\mathcal{\hat{P}} = \{ \hat{p}_1, \hat{p}_2, ..., \hat{p}_m \} \in \R^r$ 
($r \ll n$, and each $\hat{p}_i$ contains a subset of elements of the corresponding $p_i$), such that,
for $\hat{\cl S}_{opt} = \arg \min_{\cl S} \cl F (\mathcal{\hat{P}}, \cl S)$:
$$ \cl F (\mathcal{P}, \hat{\cl S}_{opt}) \leq (\beta + \epsilon) \cl F (\mathcal{P}, \cl S_{opt}) .$$
Here $\beta$ is a small constant, e.g., $\beta = 1,2,3$. The parameter $\epsilon > 0$
is given as input and one minimizes $r$ to achieve the desired accuracy $\beta + \epsilon$.

\paragraph{Prior Work.}
Despite the significance of the problem, as well as the
wealth of heuristic methods addressing it~\cite{GE03}, 
there exist no provably accurate feature selection methods for $k$-means clustering.
On the other hand, there are two provably accurate feature extraction methods.
First, a folklore result~\cite{JL84} indicates that one can 
construct $r = O( \log(m)\epsilon^{-2})$ artificial features with Random Projections and,
with constant probability, get a $(1+\epsilon)$-approximate clustering. 
Second, the work in \cite{DFKVV99} argues that one can construct $r = k$ artificial features with the SVD,
in $O( mn \min\{ m,n \} )$ time and get a $2$-approximation on the clustering quality. 

\paragraph{Results.} 
We present the first provably accurate feature selection algorithm 
for $k$-means: Theorem \ref{fastkmeans} in Section \ref{chap61} presents
a $O( mnk\epsilon^{-1} + k \log(k) \epsilon^{-2} \log(k \log(k) \epsilon^{-1}))$ 
randomized algorithm that, with constant probability, achieves a
$(3+\epsilon)$-error with $r = O( k \log(k) \epsilon^{-2} ) $ features. 
We also describe a random-projection-type feature
extraction algorithm: Theorem \ref{thm:second_result} in Section \ref{chap62} presents a 
$O(m n \lceil \epsilon^{-2} k / \log(n) \rceil )$
algorithm that, with constant probability, 
achieves a $(2+\epsilon)$-error with $r = O( k \epsilon^{-2})$ ``features''. 
We improve the above folklore result by showing that a
smaller number of dimensions are enough for obtaining 
an approximate clustering. 
Finally, Theorem \ref{thm:first_result} in Section~\ref{chap63}
describes a feature extraction algorithm
that uses an approximate SVD to construct $r = k$ ``features''
in $O(m n k \epsilon^{-1})$ time such that, with constant probability, the 
error is at most a $2+\epsilon$ factor from the optimal. 

\section{Algorithms and General Methodology}\label{chap15}

Algorithmically, to select $r$ columns from $\matA = [\a_1, \a_2,..., \a_n] \in \R^{m \times n}$, 
we  compute the matrix $\matV_k \in \R^{n \times k}$ of the right singular vectors 
of $\matA$ and consider the matrix $\matV_k\transp = [\v_1, \v_2,..., \v_n] \in \R^{k \times n}$. Here 
$k \ll m$ is part of the input of the corresponding problem and denotes, for example,
the target rank for the approximation (Low-rank Matrix Approximation)
or the number of clusters ($k$-means Clustering). 
Notice that $\matV_k\transp$ has
the same number of columns with $\matA$. So, in some sense,
selecting columns from $\matA$ and $\matV_k\transp$ is equivalent,
since there is a one-to-one correspondence between the columns of these matrices.
Actually, the columns of $\matV_k\transp$ is a ``compact'' representation 
of the columns of $\matA$ in some $k$-dimensional space, which is exactly
the subspace of interest for the corresponding problem. 
To construct
(rescaled) columns $\matC = \matA \Omega \matS$, it suffices to construct
sampling and rescaling matrices $\Omega, \matS$. To do so, we select
columns from $\matV_k\transp$ such that the submatrix 
$\matV_k\transp \Omega \matS$ has columns that are as ``linearly independent'' 
as possible. Intuitively, this means that we select the columns of $\matV_k\transp$ that
capture essentially all the information in $\matV_k\transp$, and, since
$\matV_k\transp$ is a compact representation of $\matA$, this corresponds to
selecting 
columns of $\matA$ that capture most of the information in $\matA$.  
Notice that $\matV_k\transp$ has rank at most $k$, so there are exactly 
$k$ columns that are linearly independent, but there might be many such 
$k$-subsets of which we do not know what is the best. The challenge is to 
find the ``best'' such subset of $r \ge k$ columns  and to do so in 
low-order polynomial time. Evaluating the
``linearly independence'' of a set of columns can be done through different 
quantities, for example, the determinant, the volume of the parallelepiped 
spanned by those columns, or the singular values of the matrix formed by
those columns. We will mostly do so by measuring the smallest singular
value of $\matV_k\transp \Omega \matS$. Notice that 
$\sigma_k( \matV_k\transp \Omega \matS ) = 0$ means that the selected
submatrix has rank less than $k$; on the other hand, 
$\sigma_k( \matV_k\transp \Omega \matS ) \gg 0$ implies that the selected
columns are almost as linearly independent as possible. To summarize, 
our main goal is to construct matrices
$\Omega, \matS$ and guarantee that $\sigma_k( \matV_k\transp \Omega \matS )$
is as large as possible. The later goal can be viewed as the main task behind almost
all the column sampling techniques of Section \ref{chap31}. We should note here
that Section \ref{chap31} presents several methods for selecting columns from
short-fat matrices of orthonormal rows; then, in Chapters \ref{chap4}, \ref{chap5},
and \ref{chap6} we show how to use these elementary methods for three different problems. 
An important issue is that sometimes we want to ensure that both $\sigma_k( \matV_k\transp \Omega \matS )$
is large and $\sigma_1( \matV_{\rho-k}\transp \Omega \matS )$ is small, i.e. we want
to construct $\Omega, \matS$ that simultaneously select columns from two different matrices.
Existing techniques do not provide such advantages. To do so, we developed novel sampling techniques
for selecting columns from two different matrices simultaneously in Sections \ref{chap317} and
\ref{chap318}. 
Another important observation is that we developed a novel theory (see Lemma \ref{lem:genericNoSVD})
in Section \ref{chap22} that shows that there is no need to compute exactly the matrix $\matV_k$, approximations
suffice. This theory allows us to design fast algorithms without the SVD. 

{ \bf Proof Techniques.}  
Our proofs rely on matrix perturbation theory and the
aforementioned spectral properties of the submatrix 
$\matV_k\transp \Omega \matS$. As an example of the
proof techniques that we employ here, consider 
Lemma~\ref{lem:generic} in Section \ref{chap22}
with $\matW = \Omega \matS$ and $\xi=2$ (after some standard properties of matrix norms):
$$
\TNormS{\matA - \Pi_{\matC,k}^{2}(\matA)}
\leq \TNormS{\matA-\matA_k} + \TNormS{(\matA-\matA_k)} \TNormS{ \matV_{\rho-k}\transp \Omega \matS} \TNormS{ (\matV_k\transp \Omega \matS)^+}.
$$
To prove this general bound, we combined some results from matrix perturbation theory
with the matrix analog of the Pythagorean theorem. This generic bound is quite useful:
it indicates that for any sampling and rescaling matrices $\Omega, \matS$ the approximation
error is bounded from above by certain terms, so, it tells us that we should focus on 
algorithms that control exactly these terms.   
Consider, for example, constructing a sampling and a rescaling matrix 
with the technique of Section~\ref{chap317}:
\mand{
\TNormS {(\matV_k\transp \Omega \matS)^+} \le  \left( 1 - \sqrt{\frac{k}{r}} \right)^{-2}; \quad
\qquad \text{and}\qquad \TNormS{\matV_{\rho-k}\transp \Omega \matS} 
\le \left( 1 + \sqrt{ \frac{\rho-k}{r} } \right)^{2}.
}
Combine these bounds with our generic equation and take square roots on both sides
of the resulting equation:
\vspace{-0.19in}
$$
\TNorm{\matA - \Pi_{\matC,k}^{2}(\matA)}
\leq \TNorm{\matA-\matA_k} + \TNorm{(\matA-\matA_k)} 
\left( 1 + \sqrt{ \frac{\rho-k}{r} } \right) \left( 1 - \sqrt{\frac{k}{r}} \right)^{-1}.
$$
We just proved one of the main results of this thesis (Theorem~\ref{theorem:intro1})! 
Very similar proof techniques are used in the rest of our results. 

\section*{Bibliographic Note}
We conclude this chapter by giving
a precise comparison of the results of this thesis with
conference or journal publications by the author and collaborators.
Chapter~\ref{chap2}, except Lemma~\ref{lem:eq}, is joint work 
with Petros Drineas and Malik Magdon Ismail in~\cite{BDM11a}. 
Sections \ref{chap317} and \ref{chap318} in Chapter~\ref{chap3} appeared
in~\cite{BDM11a} as well. Section \ref{chap319} is joint work with 
Anastasios Zouzias and Petros Drineas in~\cite{BZD10}. 
Section~\ref{chap3110} is joint work with Petros Drineas
in~\cite{BD09}. Portions of Section~\ref{chap32} appeared in~\cite{BDM11a}
(joint work with Petros Drineas and Malik Magdon Ismail)
and in~\cite{BMD08, BMD09a, BMD09b} (joint work with Michael W. Mahoney and Petros Drineas).
Portions of Section~\ref{chap33} appeared in~\cite{BD09}, while
Section~\ref{chap34} appeared in~\cite{BZD10}. 
Chapter~\ref{chap4} - except Theorem \ref{fastcssp} - is joint work 
with Petros Drineas and Malik Magdon Ismail in~\cite{BDM11a}. 
A preliminary version of Theorem \ref{fastcssp} appeared in~\cite{BMD08, BMD09a, BMD09b} and
is joint work with Michael W. Mahoney and Petros Drineas. This preliminary version
uses the exact SVD for the factorization in the first step of the algorithm; here,
we show that an asymptotically similar bound is possible by using an approximate SVD factorization. 
In Chapter~\ref{chap5}, the first part of Section \ref{chap51} is joint work with Petros Drineas and 
Malik Magdon Ismail in~\cite{BDM11b}, while Section \ref{chap52} is joint work with 
Petros Drineas in~\cite{BD09}. 
In Chapter~\ref{chap6}, a preliminary version of Section \ref{chap61} appeared in \cite{BMD09c} and
is joint work with Michael W. Mahoney and Petros Drineas. This preliminary version
uses the exact SVD for the factorization in the first step of the algorithm; here,
we show that an asymptotically similar bound is possible by using an approximate SVD factorization. 
Finally, Section~\ref{chap62} appeared in~\cite{BZD10} and is joint work with Anastasios Zouzias
and Petros Drineas.          

\chapter{PRELIMINARIES} \label{chap2}
\footnotetext[1]{Portions of this chapter previously appeared as:
C. Boutsidis, P. Drineas, and M. Magdon-Ismail, 
Near-Optimal Column-Based Matrix Reconstruction,
arXiv:1103.0995, 2011.}
This chapter introduces the notation that we use throughout this 
document and provides the necessary background from Linear Algebra
and Probability Theory. 

\section*{Basic Notation}  \label{chap21}
We use \math{\matA,\matB,\ldots} to denote matrices; 
\math{\a,\b,\ldots} to denote column vectors. 
\math{\matA=[\a_1,\ldots,\a_n] \in \R^{m \times n}} 
represents a matrix with columns $\a_1,\ldots,\a_n \in \R^{m}$. 
$\matI_{n}$ is the $n \times n$
identity matrix;  $\bm{0}_{m \times n}$ is the $m \times n$ matrix of zeros; $\bm{1}_n$ is the $n \times 1$ vector of ones; $\bm{e}_i$ is the standard basis (whose dimensionality will be clear from the context). 
$\matA_{ij}$ denotes the $(i,j)$-th element of $\matA$. Logarithms are base two.
We abbreviate ``independent identically
distributed'' to ``i.i.d'' and ``with probability'' to ``w.p''.

\section*{Sampling Matrices}  \label{chap22}
This whole dissertation is about sampling columns from matrices; 
here, we introduce the notation that we will use to describe such
a process. Let \math{\matA=[\a_1,\ldots,\a_n]} and
\math{\matC=[\a_{i_1},\ldots,\a_{i_r}]} be \math{r}
columns of~\math{\matA}. We can equivalently write 
\math{\matC=\matA\Omega}, where the \emph{sampling matrix} is
\math{\Omega=[\e_{i_1},\ldots,\e_{i_r}]} and \math{\e_i} are
standard basis vectors in 
\math{\R^n}. 
Let $\matS$ denote an $r \times r$ diagonal 
\emph{rescaling matrix} with non-zero entries; then,
$\matC = \matA \Omega \matS$ contains $r$ columns from $\matA$ rescaled
with the corresponding diagonal element of $\matS$. 
Notice that $\matA\Omega(\matA\Omega)^+ = 
\matA\Omega\matS\mat(\matA\Omega\matS)^+$, because rescaling
\math{\matC} does not change the subspace spanned
by its columns. So, in Chapter~\ref{chap4} all the results hold for $\matC = \matA \Omega$
as well (we stated the results for $\matC = \matA \Omega \matS$). In Chapters
\ref{chap5} and \ref{chap6} the rescaling can not be ignored. 

A \emph{permutation} matrix is a special
case of a sampling matrix where, for some permutation \math{\pi} of
\math{[1,\ldots,n]}, 
\math{\Pi=[\e_{\pi_1},\e_{\pi_2},\ldots,\e_{\pi_n}] \in \R^{n \times n}}; i.e.
$\matA \Pi \in \R^{m \times n}$ contains the columns of $\matA$ just permuted 
according to \math{\pi}.

\section*{Matrix norms}
 \label{chap23}
We use the Frobenius and the spectral norm of a matrix: 
$ \FNorm{\matA} = \sqrt{\sum_{i,j} \matA_{ij}^2}$ and 
$\TNorm{\matA} = \max_{\x:\TNorm{\x}=1}\TNorm{\matA \x}$, respectively. 
For any two matrices $\matA$ and $\matB$ of appropriate dimensions, \math{\TNorm{\matA}\le\FNorm{\matA}\le\sqrt{\rank(\matA)}\TNorm{\matA}}, $\FNorm{\matA\matB} \leq \FNorm{\matA}\TNorm{\matB}$, and $ \FNorm{\matA\matB} \leq \TNorm{\matA} \FNorm{\matB}$. The latter two properties are stronger versions of the standard submultiplicativity property:
$\XNorm{\matA\matB} \leq \XNorm{\matA}\XNorm{\matB}$. We will refer to these two stronger versions
as spectral submultiplicativity. 
The notation $\XNorm{\matA}$ indicates that an expression holds for both $\xi = 2$ and $\xi = F$.

\section*{Singular Value Decomposition} \label{chap24}
The Singular Value Decomposition (SVD) of the matrix $\matA$ with $\rank(\matA) = \rho$ is:
\begin{eqnarray*}
\label{svdA} \matA
         = \underbrace{\left(\begin{array}{cc}
             \matU_{k} & \matU_{\rho-k}
          \end{array}
    \right)}_{\matU_A \in \R^{m \times \rho}}
    \underbrace{\left(\begin{array}{cc}
             \Sigma_{k} & \bf{0}\\
             \bf{0} & \Sigma_{\rho - k}
          \end{array}
    \right)}_{\Sigma_\matA \in \R^{\rho \times \rho}}
    \underbrace{\left(\begin{array}{c}
             \matV_{k}\transp\\
             \matV_{\rho-k}\transp
          \end{array}
    \right)}_{\matV_\matA\transp \in \R^{\rho \times n}},
\end{eqnarray*}
with singular values \math{\sigma_1\ge\ldots\sigma_k\geq\sigma_{k+1}\ge\ldots\ge\sigma_\rho > 0}. We will often denote $\sigma_1$ as $\sigma_{\max}$ and $\sigma_{\rho}$ as $\sigma_{\min}$, and will use $\sigma_i\left(\matA\right)$ to denote the $i$-th singular value of $\matA$ when the matrix is not clear from the context. The matrices
$\matU_k \in \R^{m \times k}$ and $\matU_{\rho-k} \in \R^{m \times (\rho-k)}$ contain the left singular vectors of~$\matA$; and, similarly, the matrices $\matV_k \in \R^{n \times k}$ and $\matV_{\rho-k} \in \R^{n \times (\rho-k)}$ contain the right singular vectors of~$\matA$. It is well-known that $\matA_k=\matU_k \Sigma_k \matV_k\transp \in \R^{m \times n}$ minimizes \math{\XNorm{\matA - \matX}} over all
matrices \math{\matX \in \R^{m \times n}} of rank at most $k$. We use $\matA_{\rho-k} \in \R^{m \times n}$ to denote the matrix $\matA - \matA_k = \matU_{\rho-k}\Sigma_{\rho-k}\matV_{\rho-k}\transp \in \R^{m \times n}$. Also, $ \FNorm{\matA} =
\sqrt{ \sum_{i=1}^\rho\sigma_i^2(\matA) }$ and $\TNorm{\matA} = \sigma_1(\matA)$. The best rank $k$ approximation to $\matA$ satisfies $\TNorm{\matA-\matA_k} = \sigma_{k+1}(\matA)$ and 
$\FNorm{\matA-\matA_k} = \sqrt{\sum_{i=k+1}^{\rho}\sigma_{i}^2(\matA)}$.

Let $\matB \in \R^{m \times n}$ ($m \le n$) and $\matA=\matB\matB\transp \in \R^{m \times m}$; then, for all $i=1, ...,m$, $\lambda_{i}\left(\matA\right) = \sigma_{i}^2\left(\matB\right)$ denotes the $i$-th eigenvalue of $\matA$. We will also use $\lambda_{\min}\left(\matA\right)$ and $\lambda_{\max}\left(\matA\right)$ to denote the smallest and largest eigenvalue of $\matA$, respectively. Any matrix $\matA$ that can be written in this form  $\matA=\matB\matB\transp$ 
is called a Positive Semidefinite (PSD) matrix; clearly, all eigenvalues of a PSD matrix are non-negative. 

\section*{Perturbation Theory, Quadratic Forms, and PSD Matrices}

\begin{lemma} \label{lem:eq}
Let $\matV \in \R^{n \times k}$ with $n > k$ and $ \matV\transp \matV = \matI_k$. Let $0 < \epsilon < 1$.
Let $\matW \in \R^{n \times r}$ with $k \le r \le n $.
Then, the following six statements are \emph{equivalent}.
\begin{enumerate}

\item $ \TNorm{ \matV\transp \matW \matW\transp \matV - \matI_k } \le \epsilon$.

\item For all $i=1,...,k$: $ 1 - \epsilon \le  \lambda_i( \matV\transp \matW \matW\transp \matV ) \le 1 + \epsilon $.

\item For all $i=1,...,k$: $ 1 - \epsilon \le  \sigma^2_i( \matV\transp \matW  ) \le 1 + \epsilon $.

\item For any vector $\y \in \R^k$: $ (1-\epsilon) \y\transp\matV\transp \matV \y \le  
\y\transp  \matV\transp \matW \matW\transp \matV \y
\le (1 + \epsilon) \y\transp\matV\transp \matV \y $.

\item For any vector $\y \in \R^k$: $ (1-\epsilon) \TNormS{\matV \y}  \le  \TNormS{\matW\transp \matV \y} \le (1 + \epsilon) \TNormS{\matV \y} $.

\item $(1 - \epsilon) \matV\transp \matV \preceq  \matV\transp \matW \matW\transp \matV \preceq (1 + \epsilon) \matV\transp \matV$. (Partial ordering - see proof.)
 
\end{enumerate}

\end{lemma}
\begin{proof}
Assume that the first statement in the Lemma is true. 
We first recall the standard perturbation result for 
eigenvalues, which says that the eigenvalues of the 
real square symmetric matrices $\matX$ and $\matX + \matE$
satisfy: $ | \lambda_i(\matX) - \lambda_i(\matX + \matE) | \le \TNorm{\matE}$.  

The second statement of the Lemma follows from the first statement and the above
perturbation result with $ \matX =  \matV\transp \matW \matW\transp \matV$ and 
$ \matE = \matV\transp \matW \matW\transp \matV-\matI_k$. 

The third statement follows from the second statement by using the relation 
$\lambda_i( \matV\transp \matW \matW\transp \matV ) = \sigma^2_i( \matV\transp \matW )$.

To prove the fourth statement, we will use a property of the Rayleigh quotient of a square 
symmetric matrix $\matX$. For a vector $\y$ define $R(\matX, \y) = \frac{ \y\transp \matX \y}{ \y\transp \y}$.
It is well known that for all $\y$: $ \lambda_{min}(\matX) \le R(\matX, \y) \le \lambda_{max}(\matX) $. 
To conclude, use the later equation with $\matX = \matV\transp \matW \matW\transp \matV$ along 
with the second statement of the Lemma. 

The fifth statement follows from the fourth statement by using that for any vector $\x$, 
$\TNormS{\x} = \x\transp \x$. To conclude, use this twice for $\x = \matV \y$ 
and $\x = \matW\transp \matV \y$. 

In the sixth statement, for two matrices $\matX$ and $\matY$, $\matX \preceq \matY$ 
denotes the fact that the matrix $\matY - \matX$ is a Positive Semidefinite (PSD) matrix, i.e. the statement says that  both
$\matV\transp \matW \matW\transp \matV - (1 - \epsilon) \matV\transp \matV$ and
$(1 + \epsilon) \matV\transp \matV - \matV\transp \matW \matW\transp \matV$ are PSD. 
First, recall the definition of a PSD matrix: 
a square symmetric matrix $\matX$ is PSD if and only if for any vector $\y$: $\y\transp \matX \y \geq 0$. 
The left inequality follows by the left inequality of the fourth statement; similarly for the right.
\end{proof}

\section*{Moore-Penrose Pseudo-inverse} \label{chap25}

$\matA^+ = \matV_\matA \Sigma_\matA^{-1} \matU_\matA\transp \in \R^{n \times m}$ 
denotes the Moore-Penrose pseudo-inverse of $\matA \in \R^{m \times n}$ ($\Sigma_\matA^{-1}$ is the inverse of $\Sigma_\matA$),
i.e. the unique $n \times m$ matrix satisfying all four properties: 
$\matA = \matA \matA^+ \matA$, 
$\matA^+ \matA \matA^+ = \matA^+$,
$(\matA \matA^+ )\transp = \matA \matA^+$, and 
$(\matA^+ \matA )\transp = \matA^+\matA$.
By the SVD of $\matA$ and $\matA^+$, it is easy to verify 
that, for all $i=1,...,\rho = \rank(\matA) = \rank(\matA^+)$,
$\sigma_i(\matA^+) = 1 / \sigma_{\rho - i + 1}(\matA)$.
We will also use the following standard result.
\begin{lemma}
\label{lemma:pseudo} Let $\matA \in \R^{m \times n}, \matB \in \R^{n \times \ell}$; then,
\math{(\matA\matB)^+=\matB^+\matA^+} if at least one
of the three hold:
(\rn{1}) \math{\matA\transp\matA=\matI_n};
(\rn{2}) \math{\matB\transp\matB=\matI_{\ell}};
or, (\rn{3}) $\rank(\matA) = \rank(\matB) = n$.
\end{lemma}

\section*{Matrix Pythagoras Theorem } \label{chap26}
\begin{lemma}\label{lem:pyth}
If \math{\matX,\matY\in\R^{m\times n}} and
\math{\matX\matY\transp=\bm{0}_{m \times m}} or \math{\matX\transp\matY=\bm{0}_{n \times n}}, then
\eqan{
&\FNorm{\matX+\matY}^2 = \FNorm{\matX}^2+\FNorm{\matY}^2,& \\
&\max\{\TNorm{\matX}^2,\TNorm{\matY}^2\}\le
\TNorm{\matX+\matY}^2 \le \TNorm{\matX}^2+\TNorm{\matY}^2.&
}
\end{lemma}
\begin{proof}
Since \math{\matX\matY\transp=\bm{0}_{m \times m}}, \math{(\matX+\matY)(\matX+\matY)\transp=\matX\matX\transp + \matY\matY\transp}.
For $\xi = F$,
$$
\FNormS{\matX + \matY}
= \trace\left((\matX+\matY)(\matX+\matY)\transp\right)
= \trace\left(\matX\matX\transp + \matY\matY\transp\right)
= \FNormS{\matX} + \FNormS{\matY}.
$$
Let $\z$ be any vector in $\mathbb{R}^m$. For $\xi=2$,
$$
\TNormS{\matX+\matY}
= \max_{\TNorm{\z}=1}\z\transp (\matX+\matY)(\matX+\matY)\transp \z
= \max_{\TNorm{\z}=1}\left( \z\transp\matX\matX\transp \z +
\z\transp\matY\matY\transp \z\right).
$$
The bounds follow from the following two relations:
$$
 \max_{\TNorm{\z}=1}\left( \z\transp\matX\matX\transp \z + \z\transp\matY\matY\transp \z\right)
\leq
\max_{\TNorm{\z}=1}\z\transp\matX\matX\transp \z +\max_{\TNorm{\z}=1}\z\transp\matY\matY\transp \z
= \TNormS{\matX} + \TNormS{\matY};
$$
$$
\max_{\TNorm{\z}=1}( \z\transp\matX\matX\transp \z + \z\transp\matY\matY\transp \z)
\geq
\max_{\TNorm{\z}=1} \z\transp\matX\matX\transp \z
=
\TNormS{\matX},$$
since $\z\transp\matY\matY\transp \z$ is non-negative for any vector $\z$. 
We get the same lower bound with \math{\TNormS{\matY}} instead, which means
we can lower bound with \math{\max\{\TNormS{\matX},\TNormS{\matY}\}}.
The case with \math{\matX\transp\matY=\bm0_{n\times n}} can be proven similarly.
\end{proof}

\section*{Projection Matrices}  \label{chap28}
A square matrix $\matP \in \R^{n \times n}$ is a projection matrix if
$\matP^2=\matP$. For such a projection matrix and any matrix $\matX$: $$\XNorm{\matP \matX} \leq \XNorm{\matX}.$$
Also, if $\matP$ is a projection matrix, then, $\matI_{n} - \matP$ is a projection matrix. So, 
for any matrix $\matX$, both $\matX \matX^+$ and $\matI_{n} - \matX \matX^+$ are projection matrices.

\section*{Markov's inequality, the Union Bound, and Boosting} \label{chap29}

We use $\textbf{E}[x]$ to take the expectation of a random variable $x$
and $\textbf{Pr}[{\cal E}]$ to take the probability of a probabilistic event
${\cal E}$. 

Markov's inequality can be stated as follows: 
Let $x$ be a random variable taking non-negative values with 
expectation $\Expect{x}$. Then, for all $t > 0$, and with probability at least $1-t^{-1}$,
$$x \leq t \cdot \Expect{x}.$$
%

We will also use the so-called union bound. Given a set of probabilistic events ${\cal E}_1,{\cal E}_2,\ldots,{\cal E}_{n}$ holding with respective probabilities $p_1,p_2,\ldots,p_n$, the probability that all events hold (a.k.a., the probability of the union of those events) is upper bounded by $\sum_{i=1}^n p_i$, i.e.
$$ \textbf{Pr}[{\cal E}_1 \cup {\cal E}_2 \ldots \cup {\cal E}_n]  \le \sum_{i=1}^n p_i. $$

Sometimes we state our non-deterministic results in terms of their expected approximation behavior. An application of Markov's inequality gives the result with constant probability; then, standard boosting techniques suffice to make the failure probability arbitrarily small by repeating the algorithm many times and keeping the best result. Using this approach, to make the failure probability $\delta$ arbitrarily small, it suffices to repeat the algorithm
$ O(\log( 1 / \delta )) $ times and keep the best result. One should be careful though because
identifying the best solution might be an expensive task.

\section*{Random Sampling Techniques: Implementation Issues}
We often sample columns from matrices randomly based on a probability 
distribution over the columns. Once this probability distribution is
computed, then, the sampling process has nothing to do with the matrix and 
the problem is reduced to sampling a subset of indices based on this
distribution. Here, we discuss how to implement this process
in three different scenarios: (i) sampling with
replacement, (ii) sampling without replacement, and (iii) uniform sampling. 
In all three cases the setting is as follows: we are given indices $\{1, 2,...,n\}$
(that correspond to the columns of a matrix $\matA \in \R^{m \times n}$), a probability
distribution $p_1, p_2,...,p_n$ over these indices, and a sampling parameter $0 < r < n$.
We are asked to select $r$ indices.

\paragraph{Sampling with Replacement.}
This corresponds to performing $r$ i.i.d trials of the following random experiment:
throw a biased die with $n$ faces each one occurring w.p $p_i$.
One way to implement this process in a computer programming language
is as follows. First, generate $r$ numbers $\eta_1, \eta_2,..., \eta_r$ 
i.i.d from the normal distribution. Then, in $j=1,...,r$ rounds find
an index $i$ with $p_i > \eta_j$. This approach needs
$O(r + nr)$. This process can be implemented thought more efficiently
in $O(n + r \log(r))$. 

\paragraph{Sampling without Replacement.}
For each index $i=1,...,n$, one computes $q_i = \min\{1, r p_i \}$
and selects this particular index $i$ with probability $q_i$ 
(flip a biased coin for each index separately). This can be implemented,
for example, by the use of a Gaussian random number generator. For a fixed $i$,
we can generate a number $\eta_i \in \mathcal{N}(0,1)$ and if $q_i < \eta_i$, we select $i$,
otherwise we do not. The overall process needs $O(n)$ time, since we need 
$n$ multiplications for all $r p_i$, $n$ comparisons to compute the $q_i$'s, $O(n)$ to
compute $n$ i.i.d numbers from the normal distribution, and finally, $n$ comparisons
to compare these numbers with the $q_i$'s. 

\paragraph{Uniform Sampling with Replacement.}
Here, for $j=1,..,r$ i.i.d trials one has to implement the following experiment:
select the index $i$ from the set $\{1, 2,...,n\}$ with probability $p_i = \frac{1}{n}$. 
We can use a random number generator returning instances from the discrete uniform 
distribution. We only need $r$ i.i.d such numbers to sample $r$ indices, which takes
$O(r)$.

\section{Best rank \math{k} Approximation
$\Pi_{\matC,k}^\xi(\matA)$ within a Subspace}  \label{chap21}

Given $\matA \in \mathbb{R}^{m \times n}$, integer $k$, and $\matC \in \mathbb{R}^{m \times r}$ with $r > k$, we define the matrix $\Pi_{\matC,k}^{\xi}(\matA)\in \mathbb{R}^{m \times n}$ as the best approximation to $\matA$
(under the $\xi$-norm) within the column space of $\matC$ that has rank at most
$k$; so, $\Pi_{\matC,k}^{\xi}(\matA) \in \R^{m \times n}$
minimizes the residual
$\norm{\matA-\hat\matA}_\xi,$
over all
\math{\hat\matA \in \R^{m \times n}} in the column space of
\math{\matC} that have rank at most \math{k}.
In general, $\Pi_{\matC,k}^{2}(\matA) \neq \Pi_{\matC,k}^{F}(\matA)$. 
We can write
$\Pi_{\matC,k}^\xi(\matA) = \matC\matX^\xi$, where $\matX^\xi$:
$$
\matX^\xi = \argmin_{\Psi \in {\R}^{r \times n}:\rank(\Psi)\leq k}\XNormS{\matA-
\matC\Psi}.
$$
In order to compute (or approximate if exact computation is not obvious) 
$\Pi_{\matC,k}^{\xi}(\matA)$ given $\matA$, 
$\matC$, and $k$, we will use the following algorithm:
\begin{center}
\begin{algorithmic}[1]
\STATE Orthonormalize the columns of $\matC$ in $O(m r^2)$ time to construct $\matQ \in \R^{m \times r}$.
\STATE Compute
 $\left(\matQ\transp \matA\right)_k \in \R^{r \times n}$ via the SVD 
in \math{O(mnr+ nr^2)} time; $\left(\matQ\transp \matA\right)_k$  is 
the best rank-$k$ approximation of \math{\matQ\transp\matA}.
\STATE Return $\Pi_{\matC,k}^{\xi}(\matA) = \matQ\left(\matQ\transp \matA\right)_k \in \mathbb{R}^{m \times n}$ in $O(mnk)$ time.
\end{algorithmic}
\end{center}
Note that though  $\Pi_{\matC,k}^{\xi}(\matA)$ can depend on 
\math{\xi}, our algorithm computes the same matrix, independent
of \math{\xi}. 
The next lemma, which is essentially
 Lemma 4.3 in~\cite{CW09} together with
 a slight improvment of Theorem 9.3 in~\cite{HMT},
proves that this algorithm computes $\Pi_{\matC,k}^{F}(\matA)$ and a constant 
factor approximation to 
$\Pi_{\matC,k}^{2}(\matA)$. 
\begin{lemma}\label{lem:bestF}(See Appendix for the proof.)
Given $\matA \in {\R}^{m \times n}$, $\matC\in\R^{m\times r}$ 
and $k$,  the matrix
$\matQ\left(\matQ\transp \matA\right)_k \in \mathbb{R}^{m \times n}$
 described above (where \math{\matQ} is an orthonormal basis for the columns
of \math{\matC})
can be
computed in 
$O\left(mnr + (m+n)r^2\right)$ time and satisfies:
\begin{eqnarray*}
\norm{\matA-\matQ\left(\matQ\transp \matA\right)_k}_F^2 &=& \FNormS{\matA-\Pi_{\matC,k}^{F}(\matA)},\\
\norm{\matA-\matQ\left(\matQ\transp \matA\right)_k}_2^2 &\leq& 2\TNormS{\matA-\Pi_{\matC,k}^{2}(\matA)}.
\end{eqnarray*}
\end{lemma} 

\paragraph{Remark 1.} In the context of the definition of $\Pi_{\matC,k}^{\xi}(\matA)$,
$\matC$ can be any matrix, not necessarily a subset of the columns of $\matA$.

\paragraph{Remark 2.} An interesting open question is whether one can compute
$\Pi_{\matC,k}^{2}(\matA)$ exactly or obtain a better than a $2$-approximation as in Lemma~\ref{lem:bestF}. 

\section{Column-based Matrix Reconstruction through Matrix Factorization}\label{chap22}
This section presents the fundamental idea underlying all our results regarding low-rank
column-based matrix approximation (Chapter~\ref{chap4}). Lemma \ref{lem:genericNoSVD} below draws a connection
between matrix factorizations and column-based reconstruction. The factorizations that
we will see in this Lemma are of the following form $$\matA=\matB\matZ\transp+\matE,$$ where
\math{\matB\in\R^{m\times k}},
\math{\matZ\in\R^{n\times k}},
\math{\matE\in\R^{m\times n}}, 
and \math{\matZ} consists of orthonormal columns. 
Lemmas \ref{lem:generic}, \ref{tropp1}, and \ref{tropp2} discuss
algorithms to construct such factorizations. The 
proofs of the results of this section are given in the Appendix. Recall
that $k$ is the target rank for the approximation. 
\begin{lemma}
\label{lem:genericNoSVD}
Let $\matA = \matB \matZ\transp + \matE$, with $\matE\matZ = \bm{0}_{m \times k}$ and $\matZ\transp\matZ=\matI_{k}$.
Let $\matW\in\R^{n\times r}$ be any matrix such that $rank(\matZ\transp \matW) =
rank(\matZ)=k.$
Let $\matC = \matA \matW \in \R^{m \times r}$. Then,
\begin{equation*}
\norm{\matA - \matC \matC^+ \matA}_\xi^2 \leq
\norm{\matA - \Pi^{\xi}_{\matC,k}(\matA)}_\xi^2 \leq \XNormS{\matE} +
\norm{\matE\matW (\matZ\transp \matW)^+}_\xi^2.
\end{equation*}
\end{lemma}
\noindent View $\matC$ as a dimensionally-reduced or 
sampled sketch of $\matA$; $\matW$ is the dimension-reduction or sampling matrix,
for example, $\matW = \Omega \matS$, for some sampling and rescaling matrices $\Omega, \matS$,
such that $\matC$ contains (rescaled) columns of $\matA$.
In words, Lemma \ref{lem:genericNoSVD} argues that if the matrix $\matW$ preserves the rank of an approximate factorization of the original matrix $\matA$, then,
the reconstruction of $\matA$ from $\matC = \matA \matW$ has an error that is essentially proportional to the error of the approximate factorization.
The importance of this lemma is that it  indicates an algorithm for matrix reconstruction using a subset of the columns of $\matA$: first,
compute \emph{any} factorization of the form $\matA = \matB \matZ\transp + \matE$ satisfying the assumptions of the lemma; then, compute sampling and rescaling matrices \math{\matW = \Omega \matS} which satisfy the rank assumption and control the error $\norm{\matE \matW (\matZ\transp\matW)^+}_\xi$.

An immediate corollary of Lemma~\ref{lem:genericNoSVD} emerges by considering the SVD of $\matA$. More specifically, consider the following factorization of $\matA$: $\matA = \matA\matV_k\matV_k\transp + \left(\matA-\matA_k\right)$, where $\matV_k$ is the matrix of the top $k$ right singular vectors of $\matA$. In the parlance of Lemma~\ref{lem:genericNoSVD}, $\matZ = \matV_k$, $\matB = \matA \matV_k$, $\matE = \matA -\matA_k$, and clearly \math{\matE\matZ=\bm{0}_{m \times k}}.
\begin{lemma}
\label{lem:generic}
Let $\matW \in \mathbb{R}^{n \times r}$ be a matrix with $rank(\matV_k\transp\matW) = k$. Let \math{\matC=\matA\matW}; then,
$$
\norm{\matA - \matC \matC^+ \matA}_\xi^2 \leq
\norm{\matA - \Pi_{\matC,k}^{\xi}(\matA)}_\xi^2
\leq \XNormS{\matA-\matA_k} + \norm{(\matA-\matA_k) \matW (\matV_k\transp \matW)^+}_\xi^2.
$$
\end{lemma}
\noindent The above lemma will be useful for designing the deterministic (spectral norm and Frobenius norm) column-reconstruction algorithms of Theorems~\ref{theorem:intro1} and~\ref{theorem:intro2} in Chapter \ref{chap4}. However, computing the SVD is costly and thus we would like to design a factorization of the form $\matA = \matB \matZ\transp + \matE$ that is as good as the SVD, but can be computed much faster. The next two lemmas achieve this goal. The proposed algorithms are extensions of the algorithms presented in \cite{RST09,HMT}. We will use these factorizations to design fast column reconstruction algorithms in Theorems \ref{thmFast1}, \ref{thmFast2}, and \ref{thmFast3} in Chapter~\ref{chap4}. 
\begin{lemma}[Randomized fast spectral norm SVD]
\label{tropp1}
Given \math{\matA\in\R^{m\times n}} of rank $\rho$, a target rank $2\leq k < \rho$, and
$0 < \epsilon < 1$,
there exists an algorithm that
computes a factorization  $\matA = \matB \matZ\transp + \matE$, with $\matB = \matA \matZ$, $\matZ\transp\matZ = \matI_k$, and \math{\matE\matZ=\bm{0}_{m \times k}} such that
$$\Expect{\TNorm{\matE}} \leq \left(\sqrt{2}+\epsilon\right)
\TNorm{\matA - \matA_k}.$$
The proposed algorithm runs in 
$O\left(mnk\frac{\log\left( k^{-1}\min\{m,n\}\right)}{\log\left(1+\epsilon\right)}\right)$ time.
We will use the statement $\matZ = FastSpectralSVD(\matA, k, \epsilon)$ to denote this procedure. 
\end{lemma}
\begin{lemma}[Randomized fast Frobenius norm SVD]\label{tropp2}
Given \math{\matA\in\R^{m\times n}} of rank $\rho$, a target rank $2\leq k < \rho$, and $0 < \epsilon < 1$, there exists an algorithm
that  computes a factorization $\matA = \matB \matZ\transp + \matE$, with $\matB = \matA \matZ$, $\matZ\transp\matZ = \matI_k$, and \math{\matE\matZ=\bm{0}_{m \times k}}
such that
$$\Expect{\FNormS{\matE}} \leq (1+{\epsilon})
\FNormS{\matA - \matA_k}.$$
The proposed algorithm runs in $O\left(mnk\epsilon^{-1}\right)$ time.
We will use the statement $\matZ = FastFrobeniusSVD(\matA, k, \epsilon)$ to denote this procedure. 
\end{lemma}

\paragraph{Remark.} Notice that any $\beta$-approximation to $ \TNormS{ \matA - \Pi_{\matC,k}^{2}(\matA)}$ in 
Lemma~\ref{lem:bestF} implies a $\left( \sqrt{\beta} + \epsilon \right)$-approximation in 
Lemma~\ref{tropp1}; in particular a $\beta = \left( 1 + \epsilon \right)$-approximation 
would imply a relative error approximation in Lemma~\ref{tropp1}. To our best knowledge, 
that would be the first relative error fast low-rank approximation algorithm with 
respect to the spectral norm. 


\chapter{BACKGROUND AND RELATED WORK}\label{chap3} 
\footnotetext[2]{Portions of this chapter previously appeared as:
C. Boutsidis and P. Drineas, Random Projections for the Nonnegative
Least Squares Problem, Linear Algebra and its Applications, 431(5-7):760-771, 2009, as
C. Boutsidis, A. Zouzias, and P. Drineas, Random Projections for $k$-means Clustering, 
Advances in Neural Information Processing Systems (NIPS), 2010, as
C. Boutsidis, M.W. Mahoney, and P. Drineas,
An Improved Approximation Algorithm for the Column Subset Selection Problem,
Proceedings of the 20th Annual ACM-SIAM Symposium on Discrete Algorithms (SODA), 2009,
and as
C. Boutsidis, P. Drineas, and M. Magdon-Ismail, 
Near-Optimal Column-Based Matrix Reconstruction, arXiv:1103.0995, 2011.}

The primary goal of this chapter is to give a comprehensive overview
of the broad topic of matrix sampling algorithms. We achieve this
goal in three steps. 

First, in Section \ref{chap31}, we present (randomized and deterministic) techniques
for selecting columns from matrices. In total, we present ten such techniques.
Most of these techniques appeared in prior work; for example, the randomized method
of section \ref{chap314} became popular~\cite{DMM08,SS08} after 
it introduced in the celebrated work of Rudelson and Virshynin~\cite{RV07}, and
the deterministic technique of section \ref{chap315} corresponds to the
seminal work of Gu and Eisenstat on Strong Rank-Revealing QR Factorizations~\cite{GE96}.
Two of them  are novel techniques (see Sections \ref{chap317} and \ref{chap318}) and
two of them are not exactly techniques for column selection, rather they construct a small subset
of columns that are linear combinations of the columns of the input matrix
 (see Sections \ref{chap319} and \ref{chap3110}). 
\emph{ {\bf We use these ten elementary techniques in our algorithms in Chapters \ref{chap4}, \ref{chap5}, and \ref{chap6}. }}
Each of these ten elementary techniques will be supported by a Lemma that, in some sense, 
describes the quality of the sampled columns. These lemmas describe the
``spectral properties'' of the submatrices constructed with the
corresponding technique. These lemmas lie in the heart of the proofs 
presented in Chapters \ref{chap4}, \ref{chap5}, and \ref{chap6}. 

Second, we give a precice description of prior work
directly related to the three problems that we study in the present dissertation 
(see Sections \ref{chap32}, \ref{chap33}, and \ref{chap34}). 

Finally, we attempt to draw the ``big picture'' of the topic of matrix sampling algorithms.
We do so by presenting several related problems from this area in Section \ref{chap35}. 

\section{Sampling Techniques for Matrices} \label{chap31}
 
\subsection{Randomized Additive-Error Sampling}\label{chap311} 

Frieze, Kannan, and Vempala~\cite{FKV98} presented the first algorithm for fast column-based
low-rank matrix approximations. The algorithm of the lemma below is 
randomized and offers ``additive-error'' approximation guarantees, i.e. 
the approximation error $\FNormS{\matA - \Pi^F_{\matC,k}(\matA)}$ is upper
bounded by the ``optimal'' term $\FNormS{\matA-\matA_k}$ plus an additive
term which depends on $\FNormS{\matA}$. A relative error algorithm would
replace this term with $\FNormS{\matA-\matA_k}$,
delivering a much more accurate worst-case approximation bound.
The advantage of the algorithm of the following lemma is that its running time
is essentially linear on the dimensions of $\matA$.
\begin{lemma}\label{lem:kar}
Given a matrix \math{\matA\in\R^{m\times n}} of rank $\rho$, a target rank $k < \rho$, and an oversampling parameter $0 < r \le n$, 
there is a $O(mn+r\log(r))$ randomized algorithm to construct \math{\matC\in\R^{m\times r}}:
\mand{ 
\Expect{ \FNormS{\matA - \matC \matC^+ \matA} }
\le
\Expect{ \FNormS{\matA - \Pi^F_{\matC,k}(\matA)} }
\le
\FNormS{\matA-\matA_k}+\frac{k}{r}\FNormS{\matA}.
}
We will write $\matC = AdditiveSampling(\matA, r)$ to denote this randomized algorithm. 
\end{lemma}
Let $\a_i \in \R^m$ denotes the $i$-th column of $\matA$ as a column vector.
The algorithm mentioned in the lemma can be implemented as follows.
For $i=1,...,n$ compute:
\vspace{-.13in} 
$$p_i = \frac{ \TNormS{ \a_i } }{ \FNormS{\matA} }.$$
Now, construct a sampling matrix $\Omega \in \R^{n \times r}$ as follows. 
Initially, $\Omega = \bm{0}_{n \times r}$. Then, for every column $j=1,...,r$ of $\Omega$, 
independently, pick an index $i$ from the set $\{1,2,...,n\}$ with probability $p_i$ and 
set $\Omega_{ij} = 1$. Return $\matC = \matA \Omega$.
One needs $O(mn)$ time to compute the sampling probabilities and
$O(n+r\log(r))$ to choose the $r$ columns, in total $O(mn + r\log(r))$. Although, 
$AdditiveSampling$ is not used in later chapters, we included it 
in our discussion since it is the first matrix sampling algorithm;
also, it serves as a prequel to the presentation of the algorithm in
the next subsection, which we will use in Theorem \ref{thmFast3} in
Section \ref{chap42}. 
\vspace{-.2in}
\paragraph{Remark.} An interesting open question is whether there 
exists a deterministic algorithm achieving a similar ``additive-error'' 
approximation bound.

\subsection{Randomized Adaptive Sampling}\label{chap312}
As in~\cite{FKV98}, Desphande et al~\cite{DRVW06} continue on the topic of fast
column-based low-rank matrix approximations and present an
extension of the seminal result of~\cite{FKV98}. More specifically, they ask and answer in affirmative and constructively
the following question: given $\matA$ and an initial ``good'' subset
of columns $\matC_1$, is it possible to select columns from $\matA$
and improve the above additive-error algorithm? The following
lemma, which is Theorem 2.1 of~\cite{DRVW06}, presents an algorithm
that replaces the additive error term $\FNormS{\matA}$ in the bound
of Lemma~\ref{lem:kar} with the term $ \FNormS{\matA - \matC_1 \matC_1^+ \matA}$.
 
\begin{lemma}\label{oneround}
Given a matrix $\matA \in \R^{m \times n}$ of rank $\rho$, a matrix $\matC_1 \in \R^{m \times r}$ consisting 
of $r$ columns of $\matA$, a target rank $k < \rho$, and an oversampling parameter $0 < s \le n-r$, there is a
$O( m r \min\{m,r\} + m r n + s\log(s)  )$ 
randomized algorithm
to construct $\matC \in \R^{m \times (r+s)}$: 
\vspace{-0.15in}
$$
\Expect{ \FNormS{\matA - \matC \matC^+ \matA} }
\le
\Expect{ \norm{ \matA - \Pi_{\matC,k}^{F}(\matA) }_F^2 } 
\le 
\FNormS{ \matA - \matA_k } + \frac{k}{s} \norm{\matA - \matC_1 \matC_1^+ \matA}_F^2.$$
We write $\matC = AdaptiveSampling(\matA, \matC_1, s)$ to denote this randomized algorithm. 
\end{lemma}
The algorithm mentioned in the lemma is similar with the one
described in the previous section with the only difference 
being the sampling probabilities which now depend on the 
matrix $\matA - \matC_1 \matC_1^+ \matA$ instead of the matrix $\matA$.
More specifically, define the residual error matrix $\matB = \matA - \matC_1 \matC_1^+ \matA
\in \R^{m \times n}$. For $i=1,\ldots,n$, let
$$p_i = \frac{\TNormS{\b_{i}}}{\FNormS{\matB}},$$
where $\b_i$ is the $i$-th column of $\matB$. 
Construct a sampling matrix $\Omega \in \R^{n \times s}$ as follows. 
Initially, $\Omega = \bm{0}_{n \times s}$. Then, for every column $j=1,...,s$ of $\Omega$, 
independently, pick an index $i$ from the set $\{1,2,...,n\}$ with probability $p_i$ and 
set $\Omega_{ij} = 1$.
Let $\matC_2 = \matA \Omega \in \R^{m \times s}$ contains the $s$ sampled columns;
then, $\matC = [\matC_1\ \ \matC_2] \in \R^{m \times (r+s)}$ contains the columns 
of both $\matC_1$ and $\matC_2$, all of which are columns of \math{\matA}.
One needs $O(mr\min\{m,r\})$ to compute $\matC_1^+$, $O(mnr)$ to construct $\matB$,
$O(mn)$ to construct the $p_i$'s, and $O(n+s\log(s))$ to sample the additional $s$ columns;
in total $O( mr\min\{m,r\} + mnr + s\log(s) )$.
\vspace{-0.15in}
\paragraph{Remark.} Lemma~\ref{lem:kar} is a special case of Lemma~\ref{oneround}
with $\matC_1$ being an empty matrix.

\subsection{Randomized Volume Sampling}\label{chap313}

Desphande and collaborators~\cite{DRVW06, DV06, DR10} introduced 
a randomized technique that samples column submatrices from 
the input matrix with probabilities that are proportional to the 
volume of the simplex formed by the columns in the submatrix and the origin. 
Notice that this technique can either sample single columns (submatrices
with just one column) or sets of multiple columns. Similarly to the
techniques of the previous two sections, the columns returned by
volume sampling approaches offer provably accurate column-based
low-rank approximations with respect to the Frobenius norm. 
\cite{DRVW06, DV06, DR10} discuss in 
detail this line of research; here, we give a quick summary 
and highlight the most important results. 
\begin{lemma}[Theorem 1.3, \cite{DRVW06}]\label{lem:vol1}
Fix $\matA \in \R^{m \times n}$ of rank $\rho$ and target rank $k < \rho$.
For $i=1,...,\binom{n}{k}$, consider all possible matrices $\matC_1$, $\matC_2$, ..., $\matC_{\binom{n}{k}} \in \R^{m \times k}$ consisting
of $k$ columns of $\matA$, and probabilities
\vspace{-0.1in}
$$ p_i = \frac{\det{(\matC_i\transp \matC_i})}{ \sum_{j=1}^{\binom{n}{k}} \det{(\matC_j\transp \matC_j)}}.$$
In one random trial, pick the matrix $\matC_i$ with probability $p_i$; then:
$$\Expect{ \FNormS{ \matA - \matC_i \matC_i^+ \matA } } \le (k+1) \FNormS{\matA - \matA_k}.$$
\end{lemma}
\noindent Unfortunately, the above lemma does not indicate an efficient algorithm to compute
the probabilities $p_i$'s; so, from an algorithmic perspective, the result is not useful. 
Propositions 1 and 2 in~\cite{DV06} made progress by quickly approximating these probabilities at
the cost of replacing the approximation factor $(k+1)$ with $(k+1)!$. Recently, 
Theorem 7 in~\cite{DR10} presented an efficient randomized algorithm for computing 
these probabilities exactly in $O( k n m^3 \log(m) )$. A derandomization of this
randomized algorithm led to a deterministic column-sampling algorithm that, again, runs
in  $O( k n m^3 \log(m) )$ and achieves approximation error $(k+1)$. Finally,
by leveraging a random projection type result of~\cite{MZ08}, 
Theorem 9 in~\cite{DR10} presents an algorithm that, for any $0 < \epsilon < 1$,
runs in 
$ O( n m \log(n)k^2\epsilon^{-1} + n \log^3(n)k^7\epsilon^{-6}\log( k\epsilon^{-1} \log(n) ) )$,
and achieves approximation 
$\Expect{ \FNormS{ \matA - \matC \matC^+ \matA } } \le  (1+\epsilon) (k+1) \FNormS{\matA - \matA_k}.$
Finally, recently~\cite{GK11} volume sampling extended to sample any $r \ge k$ columns.

\subsection{Randomized Subspace Sampling }\label{chap314}
\begin{definition}[Random Sampling with Replacement~\cite{RV07}] \label{def:sampling}
Let $\matX \in \R^{n \times k}$ with $n > k$; $\x_i\transp \in \R^{1 \times k}$ 
denotes the $i$-th row of $\matX$ and $0 < \beta \leq 1$.
For $i=1,...,n,$ if $\beta=1$, then $p_i = (\x_i\transp \x_i) / \FNormS{\matX}$,
otherwise compute some $p_i \geq \beta (\x_i\transp \x_i) / \FNormS{\matX}$ with $ \sum_{i=1}^{n} p_i = 1$. 
Let $r$ be an integer with $ 1 \le r \le n$. 
Construct a sampling matrix $\Omega \in \R^{n \times r}$ and a rescaling matrix $\matS \in \R^{r \times r}$ as follows. Initially, $\Omega = \bm{0}_{n \times r}$ and $\matS=\bm{0}_{r \times r}$. 
Then, for every column $j=1,...,r$ of $\Omega$, $\matS$, independently, pick an index $i$ from the set $\{1,2,...,n\}$ with probability $p_i$ and set $\Omega_{ij} = 1$ and $\matS_{jj} = 1/\sqrt{p_i r}$. To denote this 
$O(nk + r\log(r))$ time randomized procedure we will write: 
\vspace{-0.16in}
$$[\Omega, \matS] = SubspaceSampling(\matX, \beta, r).$$ 
\end{definition}
\vspace{-0.16in}
It is interesting
to consider applying this technique for selecting columns from short-fat matrices of orthonormal 
rows. Drineas et. al~\cite{DMM06a} are credited for applying the method of \cite{RV07}
to such matrices. The term ``subspace sampling'' is from \cite{DMM06a}
and denotes the fact that the sampled columns ``capture'' the subspace
of interest. 
\begin{lemma}[Originally proved in \cite{RV07}]\label{lem:random} 
Let $\matV \in \R^{n \times k}$ with $n > k$ and $\matV\transp \matV = \matI_{k}$.
Let $0 < \beta \le 1$, $0 < \delta \le 1$, and $ 4 k \ln( 2 k / \delta) / \beta < r \leq n$.
Let $[\Omega, \matS] = SubspaceSampling(\matV, \beta,  r)$.
Then, for all $i=1,...,k$, w.p. at least $1 - \delta$:
\vspace{-0.1in} 
$$1 -  \sqrt{\frac{4 k \ln(2 k / \delta )}{ r \beta }}  \leq  \sigma_i^2(\matV\transp \Omega \matS)  
\leq 1+ \sqrt{\frac{4  k \ln(2k/\delta)}{ r \beta }}.
$$
\end{lemma}
\vspace{-0.1in}
\begin{proof}
In Theorem 2 of~\cite{Mag10}, set $\matS = \matI$  and 
replace $\epsilon$ in terms of $r, \beta,$ and $d$. The lemma is proved;
one should be careful to fit this into our notation. 
\end{proof}
\begin{lemma} [See Appendix for the proof] \label{lem:fnorm}
For any $\beta$, $r$, $\matX \in \R^{n \times k}$, and $\matY \in \R^{m \times n}$, let 
$[\Omega, \matS] = SubspaceSampling(\matX, \beta,  r)$; then, w.p. $1-\delta$:
$ \FNormS{ \matY \Omega \matS } \leq \frac{1}{\delta} \FNormS{ \matY }. $ 
\end{lemma}
\vspace{-0.3in}
\paragraph{Remark.} 
Quite recently, Zouzias~\cite{Zou11} proved the deterministic analog of Lemma \ref{lem:random}.
More specifically, Theorem 11 in~\cite{Zou11} proves that, on input $\matV \in \R^{n \times k}$ 
and $r > 30 k \ln(2k)$, there is a deterministic $\tilde{O}(n r k \log^2(k))$ algorithm that
returns $\Omega \in \R^{n \times r}$, $\matS \in \R^{r \times r}$ such that, for all $i=1,...,k$:
$1 -  \sqrt{30 k \ln(2 k ) / r }  \leq  \sigma_i^2(\matV\transp \Omega \matS)  
\leq 1+ \sqrt{30  k \ln(2k) / r }.
$
The notation $\tilde{O}(\cdot)$ hides $\log ( \log(1/ \epsilon))$ and $ \log ( \log (k) ) $ factors.

\subsection{Deterministic Sampling with the Strong RRQR~}\label{chap315}
Gu and Eisenstat~\cite{GE96} give a deterministic algorithm for selecting $k < n$ columns from generic
$m \times n$ matrices. Here, we are only interested in sampling 
$k$ columns from $k \times n$ short-fat matrices; so, we will
restate the main result of \cite{GE96} as Lemma \ref{lem:rrqrGE} below to fit our notation. 
(Lemma \ref{lem:rrqrGE} is not immediate from~\cite{GE96}, so we include a proof in the Appendix.) 
Lemma \ref{lem:rrqr} is a simple corollary of Lemma \ref{lem:rrqrGE}. 
Notice that our lemmas here guarantee the existence of a deterministic technique;
the actual description of this technique can be found as Algorithm 4 in~\cite{GE96}.

\begin{lemma}[Extension of Algorithm 4 of \cite{GE96} to short-fat Matrices]
\label{lem:rrqrGE}
Let $\matX\in\R^{n\times m}$ (\math{m\le n}) have rank \math{\rho_\matX},
and \math{1\le k\le \rho_\matX}. Let $f > 1$.
There is a deterministic $O( m n k \log_{f}( n ) )$ algorithm to construct 
a permutation matrix $\Pi\in\R^{n \times n}$, a matrix $\matQ_{\matX} \in\R^{m \times m}$ 
with orthonormal columns, and a matrix $\matR \in\R^{m \times n}$:
\mand{
\matX\transp
\Pi = 
\matQ \matR = \matQ_{\matX} \left(%
\begin{array}{cc}
\matA_{k} & \matB_{k} \\
 \bm{0}_{(m-k) \times k} & \tilde{\matC}_{k} \\
\end{array}
\right);
\matA_{k} \in \R^{k\times k};
\
\matB_{k} \in \R^{k\times (n-k)};
\tilde{\matC}_{k} \in \R^{(m-k)\times (n-k)};
}
such that for all $i=1,...,k$, and $j=k+1,...,m$:
\mand{ 
\frac{\sigma_i(\matX\transp)}{\sqrt{ f^2k(n-k) + 1 }} 
\leq \sigma_{i}(\matA_{k}) \leq \sigma_i(\matX\transp); 
 \sigma_{j}(\matX\transp) \le \sigma_{j-k}(\tilde{\matC}_{k}) \leq \sqrt{ f^2k(n-k) + 1 } \sigma_{j}(\matX\transp)
}
\end{lemma}

\begin{lemma}
\label{lem:rrqr}
Let \math{\matX \in \R^{n \times k}} with $n \geq k$; there exists a deterministic algorithm 
that runs in $O( n k^2 \log( n ) )$ time and constructs a sampling matrix $\Omega \in \R^{n \times k}$ such that
\mand{
\sigma_k(\matX\transp \Omega) \ge \frac{\sigma_k(\matX\transp)}{ \sqrt{ 4k(n-k) + 1 } }
\qquad\text{and}\qquad
\XNorm{\matX\transp \Omega}\leq \XNorm{\matX\transp}.}
We write $\Omega = RRQRSampling(\matX,k)$ to denote such a deterministic procedure.
\end{lemma}
\begin{proof}
Compute a QR factorization of $\matX$ with Lemma \ref{lem:rrqrGE} and
let $\Omega$ be the first $k$ columns of the matrix $\Pi$.
The bound for $\sigma_k(\matX\transp \Omega)$ follows by applying the first bound in Lemma~\ref{lem:rrqrGE} 
with \math{i=k} and \math{f=2}. 
The $\xi$-norm upper bound follows by spectral submultiplicativity and
$\TNorm{\Omega}=1$. The running time is from Lemma 
\ref{lem:rrqrGE} with $m=k$.
\end{proof}

\subsection{Deterministic Sampling with the Barrier Method I}\label{chap316}
\begin{lemma} [Single-set Spectral Sparsification \cite{BSS09}]
\label{lem:1set}
Let \math{\matV\in\R^{n\times k}} with \math{\matV\transp\matV=\matI_{k}}.
Let $r$ is an integer and assume $k < r \le n$.
One can construct (deterministically, in $O(r n k^2)$ time) 
a sampling matrix \math{\Omega\in\R^{n\times r}} and a 
positive diagonal rescaling matrix \math{\matS\in\R^{r\times r}} such 
that, for all $i=1,...,k$:
$$    1 -  \sqrt{ \frac{ k }{ r  } } 
\leq  \sigma_i(\matV\transp \Omega \matS)  
\leq  1 +  \sqrt{\frac{ k }{ r }}.
$$
\end{lemma}
Lemma \ref{lem:1set} is due to Batson, Srivastava, and Spielman~\cite{BSS09}. 
The deterministic algorithm promised 
in the lemma can be found in the (constructive) proof of Theorem 3.1 
in~\cite{BSS09}. In the next subsection, we present an important generalization of Lemma \ref{lem:1set}
to sample two orthonormal matrices simultaneously. Lemma \ref{lem:1set}
follows from Lemma \ref{lem:2setS} with $\matU = \matV$, so the
algorithm that we present in the proof of Lemma \ref{lem:2setS}
in the Appendix is the algorithm of Lemma \ref{lem:1set}.  
Finally, we restate Lemma \ref{lem:1set} in a different form, 
which is the form that we will use it in Section \ref{chap51}.
\begin{corollary} \label{cor:1set} Frame the hypothesis in Lemma \ref{lem:1set}. 
For those $\Omega, \matS$ and 
$\forall\y\in\R^k$:
$$ \left(1-\sqrt{\frac{k}{r}}\right)^2 \TNormS{\matV \y}
\leq \TNormS{\matS\transp \Omega\transp \matV \y}
\leq \left(1+\sqrt{\frac{k}{r}}\right)^2 
\TNormS{\matV \y}.$$
We write $[\Omega, \matS] = BarrierSamplingI(\matV, r)$ to denote such a procedure. 
\end{corollary}
\vspace{-0.33in}
\paragraph{Remark.} An improvement of this result (on the run time) appeared recently by 
Zouzias in~\cite{Zou11}. More specifically, in the parlance
of Lemma \ref{lem:1set}, Theorem 12 in~\cite{Zou11} shows that there is
a $ \tilde{O}( n r k \log^3(k) + k^2 \log(k) r ) $ deterministic algorithm with 
$ \left(1-\sqrt{\frac{k}{r}}\right)^3 \TNormS{\matV \y}
\leq \TNormS{\matS\transp \Omega\transp \matV \y}
\leq \left(1+\sqrt{\frac{k}{r}}\right)^3 
\TNormS{\matV \y}.$
The notation $\tilde{O}(\cdot)$ hides $\log ( \log(1/ \epsilon))$ and 
$\log ( \log (k) ) $ factors. The algorithm of Theorem 12 in~\cite{Zou11}
combines the algorithm of Theorem 11 in~\cite{Zou11}, which we mentioned
in the remark of Section \ref{chap314}, along with the algorithm of 
Lemma \ref{lem:1set}. More specifically: if $r > 30 k \ln(2k)$,
run the algorithm of Theorem 11 in~\cite{Zou11}; if $r < 30 k \ln(2k)$,
select $\hat{r} = 30 k \ln(2k)$ columns with the
algorithm of Theorem 11 in~\cite{Zou11} and then, by using the algorithm
of Lemma \ref{lem:1set}, down-sample this to exactly $r$ columns. This
improvement implies a run time improvement in Theorem \ref{lem:regression} of this thesis. 

\subsection{Deterministic Sampling with the Barrier Method II}\label{chap317}
Lemma \ref{lem:2setS} below guarantees the existence of a
deterministic algorithm for sampling columns from two matrices simultaneously.
The details of the corresponding algorithm can be found in
the (constructive) proof of that lemma that we give in the Appendix. 
Our algorithm generalizes Lemma~\ref{lem:1set}; in fact,
setting $\matU = \matV$ in our Lemma gives Lemma~\ref{lem:1set}, which is
Theorem 3.1 in~\cite{BSS09}.
The innovation here is that we develop an algorithm that samples columns
from two different matrices simultaneously. 
\begin{lemma} [Dual Set Spectral Sparsification]
\label{lem:2setS}
Let \math{\matV\in\R^{n\times k}} and \math{\matU\in\R^{n\times \ell}}
with \math{\matV\transp\matV=\matI_{k}} and 
\math{\matU\transp\matU=\matI_{\ell}}. Let $r$ is an integer
with \math{k < r \le n} and assume \math{k,\ell \le n}.
One can construct (deterministically, in $O(r n(k^2+\ell^2))$ time) 
a sampling matrix \math{\Omega\in\R^{n\times r}} and a 
positive diagonal rescaling matrix \math{\matS\in\R^{r\times r}} such 
that: 
\mand{
\sigma_{k}(\matV\transp \Omega \matS ) \ge 1 - \sqrt{\frac{k}{r}}; \quad
\qquad \text{and}\qquad \TNorm{\matU\transp \Omega \matS} 
\le 1 + \sqrt{ \frac{\ell}{r} }.
}
We write $[\Omega, \matS] = BarrierSamplingII(\matV,\matU, r)$ 
to denote such a procedure. 
\end{lemma}

\subsection{Deterministic Sampling with the Barrier Method III}\label{chap318}
The algorithm of Lemma \ref{lem:2setF} is similar to that of Lemma \ref{lem:2setS}.
The innovation here is to control the Frobenius norm of the sub-sampled matrix. 
A (constructive) proof of this result is given in the Appendix. 
\begin{lemma} [Dual Set Spectral-Frobenius Sparsification]
\label{lem:2setF}
Let \math{\matV\in\R^{n\times k}} with \math{\matV\transp\matV=\matI_{k}} and
\math{\matA\in\R^{\ell \times n}}. Let $r$ is an integer
with \math{k < r \le n} and assume \math{k \le n}.
One can construct (deterministically, in $O(r n k^2  + \ell n)$ time) 
a sampling matrix \math{\Omega\in\R^{n\times r}} and a 
positive diagonal rescaling matrix \math{\matS\in\R^{r\times r}} such 
that:   
\mand{
\sigma_{k}(\matV\transp \Omega \matS ) \ge 1 - \sqrt{\frac{k}{r}}; \quad
\qquad \text{and}\qquad \FNorm{\matA \Omega \matS} \le \FNorm{\matA}.}
We write $[\Omega, \matS] = BarrierSamplingIII(\matV,\matA, r)$ to denote such a procedure. 
\end{lemma}
\paragraph{Remark.}
Although in Lemmas \ref{lem:2setS} and \ref{lem:2setF} we assumed that the
matrices are orthonormal, this is not necessary
(see Lemmas \ref{theorem:2setGeneral} and \ref{theorem:2setGeneralF} in the Appendix).

\subsection{Random Projections}\label{chap319}
A classical result of Johnson and Lindenstrauss~\cite{JL84} states that, 
for any $ 0 < \epsilon < 1$, any set of $m$ points in $n$ dimensions 
(rows in $\matA \in \R^{m \times n}$) 
can be linearly projected into $r=O\left(\log (m) /\epsilon^2\right)$ dimensions 
while preserving all the pairwise distances of the points within a factor of $1\pm\epsilon$. 
More precisely,~\cite{JL84} showed the existence of a (random orthonormal) matrix 
$R \in R^{n \times r}$ such that, 
for all $i,j=1,...,m$, and with high probability (over the randomness of $R$, 
$\matA_{(i)}$ denotes the $i$-th row of $\matA$):
\vspace{-0.15in}
$$
(1- \epsilon) \norm{ \matA_{(i)} - \matA_{(j)} } 
\leq  \norm{ \matA_{(i)}R - \matA_{(j)}R } \leq  
(1 + \epsilon) \norm{ \matA_{(i)} - \matA_{(j)}}.
$$
Subsequent research simplified the proof of~\cite{JL84} by showing that such an
embedding can be generated using an $n \times r$ random Gaussian 
matrix $R$, i.e. a matrix whose entries are i.i.d. Gaussian random variables 
with zero mean and variance $1/\sqrt{r}$~\cite{IM98}. Recently,~\cite{AC06} 
presented the so-called Fast Johnson-Lindenstrauss Transform which describes
an $R$ such that $\matA R$ can be computed fast. In this thesis, we will use
a construction by Achlioptas: \cite{Ach03} proved that a rescaled random sign matrix,
i.e. a matrix whose entries have i.i.d values $\pm1/\sqrt{r}$ with probability $1/2$,  
satisfy the above equation. We use such random projection embeddings in Section \ref{chap62}. 
Here, we summarize some properties of such matrices that might be of independent interest. 
It is particularly interesting to apply such embeddings in short-fat matrices with orthonormal 
columns. Sarlos is credited for this idea in~\cite{Sar06}.
\begin{lemma}[See Appendix for the Proof]\label{lem:rpall}
Let $\matA \in \R^{m \times n}$ has rank $\rho$, $k < \rho$, and $0 < \epsilon < \frac{1}{3}$.
Let $R \in \R^{n \times r}$ is a rescaled random sign matrix constructed as we described above 
with $r = c_0 k \epsilon^{-2}$, where $c_0$ is a (sufficiently large) constant. 
\begin{enumerate}
    
    \item For all $i=1,...,k$ and w.p. $0.99$:  $ 1 - \epsilon \le \sigma_i(\matV_k\transp R) \le 1 +  \epsilon.$
    
    \item For any $\matX \in \R^{m \times n}$, $\matY \in \R^{n \times k}$, and r:
    $ \EE{ \FNorm{ \matX\matY - \matX R R\transp \matY}^2 } \leq \frac{2}{r} \FNorm{\matX}^2 \FNorm{\matY}^2. $
    
    \item For any $\matX \in \R^{m \times n}$ and $r$: $ \EE{\FNorm{\matX R}^2} = 
    \FNorm{\matX}^2$ and $\var{\FNorm{\matX R}} \leq  2 \FNorm{\matX}^4 / r.$
     
    \item W.p. $0.99$: $\norm{(\matV_k\transp R)^+ - (\matV_k\transp R)\transp }\ \leq\ 3 \epsilon.$
    
    \item For any $\matX \in \R^{m \times n}$ and w.p. $0.99$: $\FNorm{ \matX R }\ \leq\ \sqrt{(1+\epsilon)} \FNorm{\matX}.$
    
    \item W.p. $0.97$:  $\matA_k  = \matA R (\matV_k\transp  R)^+\matV_k\transp  + \matE;$
          $\matE \in \R^{m \times n}$ with $\FNorm{\matE} \leq 4 \epsilon \FNorm{\matA-\matA_k}$. 
\end{enumerate}
\end{lemma}

\subsection{Subsampled Randomized Hadamard Transform} \label{chap3110}
We give the definitions of the ``Normalized Walsh-Hadamard'' and the ``Subsampled Randomized Hadamard Transform'' matrices as well as a few basic facts for computations with such matrices. 
We use these matrices in Section \ref{chap52}.

\begin{definition} [Normalized Walsh-Hadamard Matrix] \label{walsh}
Fix an integer $m = 2^p$, for $p = 1,2,3, ...$. The (non-normalized) 
$m \times m$ matrix of the Hadamard-Walsh
transform is defined recursively as follows:
$$ \matH_m = \left[
\begin{array}{cc}
  \matH_{m/2} &  \matH_{m/2} \\
  \matH_{m/2} & -\matH_{m/2}
\end{array}\right],
\qquad \mbox{with} \qquad
\matH_2 = \left[
\begin{array}{cc}
  +1 & +1 \\
  +1 & -1
\end{array}\right].
$$
The $m \times m$ normalized matrix of the Hadamard-Walsh transform
is equal to $$\matH = m^{-\frac{1}{2}}H_m.$$
\end{definition}
\begin{definition} [Subsampled Randomized Hadamard Transform (SRHT) matrix]
\label{srht}
Fix integers $r$ and $m = 2^p$ with $r < m$ and $p = 1,2,3, ...$. 
A SRHT is an $r \times m$ matrix of the form
$$ \Theta = \matS\transp \Omega\transp  \matD \matH,$$ 
where
\begin{itemize}

\item $\matH \in \R^{m \times m}$ is a normalized Walsh-Hadamard matrix.  

\item $\matD \in \R^{m \times m}$ is a diagonal matrix constructed as follows: each diagonal element is 
a random variable taking values $\{+1, -1\}$ with equal probability.
    
\item $\Omega \in \R^{m \times r}$ is a sampling matrix constructed as follows: for $j=1,2,...,r$ $i.i.d$
random trials pick a vector $\e_j$ from the standard basis of $\R^m$ with probability $\frac{1}{m}$
and set the $j$-th column of $\Omega$ equal to that vector. 

\item $\matS \in \R^{r \times r}$ is a rescaling (diagonal) matrix containing the value $\sqrt{\frac{m}{r}}$.

\end{itemize}

\end{definition}
\begin{proposition} [Fast Matrix-Vector Multiplication, Theorem 2.1 in~\cite{AL08}] \label{fast}
Given $\x \in \R^m$ and integer $r < n$, one can compute the product $\Theta \x$ 
with at most $2 m \log(r + 1) )$ operations. 
\end{proposition}
\begin{lemma} [\cite{AC06}, Lemma 3 in \cite{DMMS07}]
\label{lem:HU} 
Let $\matU \in \R^{m \times k}$ has orthonormal columns. 
Let $\left(\matD \matH  \matU \right)_{(i)}$ denotes the $i$-th row 
of the matrix $\matD \matH  \matU \in \R^{m \times k}$ and ${\cal E}_i$ denotes
the probabilistic event that
$ \TNormS{ \left( \matD \matH \matU \right)_{(i)} }  \leq \frac{2 k \log(40 m k)}{m}$ (over 
the randomness of $\matD$):
\begin{eqnarray}
\textbf{Pr}[{\cal E}_1 \cup {\cal E}_2 \ldots \cup {\cal E}_m]  \le 0.95. 
\end{eqnarray}
\end{lemma}
It is interesting to consider applying the SRHT to orthonormal matrices. 
Drineas et al are credited for this idea in~\cite{DMMS07}. 
Let $\matU \in \R^{m \times k}$ has orthonormal columns and 
$m \gg k$. The following lemma studies
the singular values of the matrix $ \Theta \matU $. This result
is very similar with the result of Lemma \ref{lem:random} with
the only difference being the fact that before applying the 
$SubspaceSampling$ method on the rows of $\matU$ we pre-multiply
it with a randomized Hadamard Transform, i.e. we apply the 
$SubspaceSampling$ on the rows of the matrix $\matD \matH \matU$
so, the matrices $\Omega$ and $\matS$ in Definition \ref{srht} can be 
obtained by a special application of  $SubspaceSampling$:
$$ [ \Omega, \matS ] = SubspaceSampling( \matD \matH \matU, \frac{1}{2 \log(40 k m)} , r ).$$
Since $\beta < 1$, we need to specify the sampling probabilities $p_i$'s in
Definition \ref{def:sampling}:
$$ p_i = \frac{1}{m}  
     \geq \frac{1}{2 \log(40 k m)}  \frac{ \TNormS{ \left( \matD \matH\matU \right)_{(i)} } }{ k }
=         \beta                     \frac{ \TNormS{ \left( \matD \matH\matU \right)_{(i)} } }{ k }.$$
The inequality in this derivation is from Lemma \ref{lem:HU}. 
We formalize this discussion in Lemma \ref{preserves}, which can 
be viewed as the analog of Lemma \ref{lem:random}. We should note
that a mild improvement of Lemma \ref{preserves} can be found in~\cite{Tro11}. 
\begin{lemma} \label{preserves}
Let $\matU \in \R^{m \times k}$ with $m > k$ and $\matU\transp \matU = \matI_{k}$.
Let  $\beta = \frac{1}{2 \log(40 k m)}$, and $ 4 k \ln( 2 k / \delta) / \beta < r \leq m$.
Let $[\Omega, \matS] = SubspaceSampling(\matD \matH \matU, \frac{1}{2 \log(40 k m)}, r)$ and
$p_i = 1/m$ in Definition \ref{def:sampling}.
Then, for all $i=1,...,k$, w.p. at least $0.95 - \delta$:
\vspace{-0.18in} 
$$1 -  \sqrt{\frac{8 k \ln(2 k / \delta ) \log(40 k m)}{ r }}  
\leq  \sigma_i^2(\matU \transp \matD\transp \matH\transp \Omega \matS)  
= \sigma_i^2( \Theta \matU )
\leq 1+ \sqrt{\frac{8  k \ln(2 k/\delta) \log(40 k m)}{ r }}.
$$
\vspace{-0.22in}
\end{lemma}
\vspace{-0.32in} 
\paragraph{Remark.} From Lemma \ref{lem:eq}, we restate the latest result as: For any vector $\y \in \R^k$: 
$$ (1-\sqrt{\frac{8 k \ln(2 k / \delta ) \log(40 k m)}{ r }}) \TNormS{\matV \y}  
\le  \TNormS{\Theta \transp \matV \y} \le 
(1 + \sqrt{\frac{8 k \ln(2 k / \delta ) \log(40 k m)}{ r }}) \TNormS{\matV \y} .$$


\section{Prior Work: Low-rank Column-based Approximation} \label{chap32} 

\begin{table}
\begin{center}
\begin{tabular}{l|l|l}
   & Spectral norm ($\xi=2$) & Frobenius norm ($\xi=F$)\\
\hline
\textbf{$r=k$}
&  $\hat\alpha=\sqrt{\frac{n}{k}}$ \cite{DR10} & $\hat\alpha=\sqrt{k+1}$ \cite{DRVW06}   \\
\hline
\textbf{$r>k$}
& $\hat\alpha=\sqrt{\frac{n}{r}}$ (Section \ref{sec:lower})
& $ \hat\alpha=\sqrt{1 + \frac{k}{2r}}$ \cite{DV06}        \\
\hline
\end{tabular}
\caption{
Lower bounds \math{ \XNorm{\matA-\Pi_{\matC,k}^{\xi}(\matA)} / \XNorm{\matA-\matA_k} \ge \hat{\alpha} }.
}
\label{table:31}
\end{center}
\end{table}

We start with a discussion on lower bounds for low-rank column based matrix approximation. The above
table provides a summary on lower bounds for the ratio
$$\frac{\XNorm{\matA-\Pi_{\matC,k}^{\xi}(\matA)}}{\XNorm{\matA-\matA_k}} \ge \hat\alpha.$$
If this ratio is greater or equal to a specific value $\hat\alpha$, this 
means that there exists at least a matrix $\matC$ such that no algorithm
can find an approximation with a factor $\alpha$ which is strictly better than $\hat\alpha$. An algorithm is called near-optimal
if it - asymptotically - matches the best possible approximation bound. Notice that the
approximation ratio is presented with respect to the best rank $k$ low-rank approximation 
computed with the SVD; this case is certainly interesting, but replacing $\matA_k$ with 
$ \matC_{opt} \matC_{opt}^+ \matA $ would have implied more meaningful information for
the hardness of approximation of this column selection problem. Here, $\matC_{opt}$ satisfies:
$$ \matC_{opt} = \arg \min_{\matC } \XNorm{ \matA - \Pi_{\matC,k}^{\xi}(\matA) }.$$

Lower bounds of the above form were available for both spectral and Frobenius norm and $r = k$.
When $r > k$, a Frobenius norm bound was given in~\cite{DRVW06} but a spectral norm bound didn't appear
in prior work. Our Theorem~\ref{theorem:lower1} in the Appendix contributes a new lower bound 
for the spectral norm case when $r > k$. It is worth noting that any lower bound for the ratio 
\math{\XNorm{\matA-\matC\matC^+\matA}/\XNorm{\matA-\matA_k}} also implies a lower bound for 
\math{\XNorm{\matA-\Pi_{\matC,k}^{\xi}(\matA)}/\XNorm{\matA-\matA_k}}; the converse, however, is not true.
See Table \ref{table:31} for a summary of the lower bounds.

\subsection*{The Frobenius norm case ($\xi=F$)}
\begin{table}
\begin{center}
\begin{tabular}{l|l|l}
   & Spectral norm ($\xi=2$) & Frobenius norm ($\xi=F$)\\
\hline
\textbf{$r=k$}
&  $\alpha=\sqrt{4 k(n-k) +1}$ \cite{GE96} & $\alpha=\sqrt{k+1}$ \cite{DR10}   \\
\hline
\textbf{$k < r = o(k \log(k))$}
& -  & -         \\
\hline
\textbf{$r =\Omega(k \log(k))$}
& -
& $ \sqrt{ 1 + O\left(\frac{k}{r - k\log(k)} \right) }$ \cite{DMM06d,DR10,DV06}       \\
\end{tabular}
\caption{
Best Available algorithms (prior to our results) for the equation: 
\math{ \XNorm{\matA-\Pi_{\matC,k}^{\xi}(\matA)}  \le \alpha \XNorm{\matA-\matA_k}}.
We offer $\alpha = O\left(\sqrt{\rho/r}\right)$ for $\xi = 2$, $r > k$ (Theorem~\ref{theorem:intro1}); and
         $\alpha = \sqrt{1 + O\left(k/r\right)}$ for $\xi=F$, $r > 10k$ (Theorem~\ref{thmFast3}).
}
\label{table:32}
\end{center}
\end{table}

We summarize prior work by presenting algorithms with accuracy 
guarantees $\alpha$ with respect 
to the equation
\vspace{-.1in} 
$$\FNorm{\matA-\matC \matC^+\matA} \le  
\FNorm{\matA-\Pi_{\matC,k}^{F}(\matA)} 
\le \alpha \FNorm{\matA-\matA_k}.$$

\paragraph{The $r = k$ case.}
Theorem 8 in~\cite{DR10} describes a deterministic $O(k n m^3 \log(m))$ algorithm with approximation:
\vspace{-.1in} 
$$ \FNorm{\matA-\matC \matC^+\matA} = \FNorm{\matA-\Pi_{\matC,k}^{F}(\matA)} 
\le \sqrt{k+1}  \FNorm{\matA-\matA_k}.$$
This matches the lower bound presented in~\cite{DRVW06}. The
algorithm of~\cite{DR10} is based on the method of volume sampling that we discussed in Section \ref{chap312}.
\cite{DR10} also presented a faster randomized algorithm achieving a $(1+\epsilon)\sqrt{k+1}$ approximation (in expectation), running in 
$O(nm \log(n) k^{2}\epsilon^{-2} + n \log^{3}(n) \cdot k^{7}\epsilon^{-6} \log\left(k \epsilon^{-1} \log(n)\right))$ time (See Theorem 9 in~\cite{DR10}):
\vspace{-.1in} 
$$ \FNorm{\matA-\matC \matC^+\matA} = \FNorm{\matA-\Pi_{\matC,k}^{F}(\matA)} 
\le \left( 1 + \epsilon \right)  \sqrt{k+1}  \FNorm{\matA-\matA_k}.$$
Here $\epsilon$ is given as input and can be made arbitrary small $0 < \epsilon < 1/2$. It would have been interesting though to consider setting $\epsilon$ of the order $k$ such that to improve the run time of the algorithm but, to our best understanding, this is not possible, since Theorem 9 in~\cite{DR10} is based on a random projection type result from~\cite{MZ08} that is based on the construction of Achlioptas in~\cite{Ach03}, 
which breaks unless $\epsilon < 1$. 

\paragraph{The $k < r = o(k \log(k))$ case.} For this range of values of $r$,
we are not familiar with any available algorithm in the literature. 

\paragraph{The $r=\Omega(k\log k)$ case.} When $r$ is asymptotically larger than $k \log(k)$,
there are a few algorithms available that we summarize below. Recall that for general $r > k$,
the lower bound for Frobenius norm approximation is $\hat\alpha = \sqrt{1+\frac{k}{2r}}$. All
algorithms of this paragraph offer approximations of the order $\sqrt{ 1 + O( \frac{k \log(k)}{r})}$,
so they are optimal up to a factor $O\left(\log(k) \right)$. Our Theorem \ref{thmFast3} presents an algorithm which
asymptotically matches this lower bound; in fact it is optimal up to a constant $20$.  

\cite{DMM06d} presented a randomized algorithm that samples
columns from $\matA$ with probabilities proportional to
the Euclidean norms of the rows of $\matV_k$. In fact, the algorithm 
of \cite{DMM06d} applies the $SubspaceSampling$ method that we described 
in Section \ref{chap314} on the matrix $\matV_k\transp$ of the top $k$ right singular vectors $\matA$. 
For fixed $\matA, k$, and 
$r = \Omega(k \log(k)$, \cite{DMM06d} describes a $O( m n \min\{m,n\} + r \log(r) )$ time
algorithm that, with constant probability, guarantees
\vspace{-.1in}   
$$ \FNorm{\matA-\Pi^F_{\matC,k}(\matA)} \le \sqrt{
1 + O\left( \frac{ k \log(k)}{r}  \right) } \FNorm{\matA-\matA_k}.$$ 
\cite{Sar06} showed how to get the same result as~\cite{DMM06d} but with
improved running time $T(\tilde{\matV}_k) + O(nk + r \log(r))$, where $\tilde{\matV}_k \in \R^{n \times k}$ contains the right singular vectors of a rank $k$ matrix that approximates $\matA_k$ and can be computed in $o(mn\min\{m,n\})$
This idea of Sarlos~\cite{Sar06} is similar with our Lemma \ref{lem:genericNoSVD} in Section \ref{chap22}.

In~\cite{DV06}, the authors leveraged adaptive sampling (see Section~\ref{chap312}) and 
volume sampling (see Section~\ref{chap313}) to design a \math{O(mnk^2\log k + r\log(r))} time
randomized algorithm that, with constant probability, obtains approximation error:
\vspace{-.1in} 
$$ \FNorm{\matA-\Pi^F_{\matC,k}(\matA)} \le \sqrt{
 1 + O\left(  \frac{ k }{r - k^2 \log(k)}  \right) } \FNorm{\matA-\matA_k}.$$ 
Notice that this method, for arbitrarily $\epsilon > 0$, 
needs \math{r=O(k^2\log k + k\epsilon^{-1})} columns to obtain $\FNorm{\matA-\Pi^F_{\matC,k}(\matA)} \le 
\left( 1 + \epsilon \right) \FNorm{\matA-\matA_k},$
whereas \cite{DMM06d} needs $r=O(k \log(k) / \epsilon^2)$.
Finally, it is possible to combine the fast volume sampling approach in~\cite{DR10} (setting, for example, $\epsilon=1/2$) with \math{O(\log k)} rounds of adaptive sampling as described in~\cite{DV06} to achieve a relative error approximation using \math{r=O\left(k\log k +k\epsilon^{-1}\right)} columns in time
$O\left(mnk^2 \log(n) + n k^7 \log^{3}(n) \log\left(k \log(n)\right)\right)$. Also, a $(1+\epsilon)$-error
with \math{r=O(k\log k + k\epsilon^{-1})} columns is possible by combining~~\cite{DMM06d} with 
\math{O(\log k)} rounds of adaptive sampling as described in~\cite{DV06}. The run time of this combined
method is $O( mn\min\{m,n\} + r \log(r))$. Theorem \ref{thmFast3} of our work gives 
$(1+\epsilon)$-error with $O(k/\epsilon)$ columns in $O(mnk + nk^3 + n \log(r))$. 
Notice that our result considerably improves both running time and approximation 
accuracy of existing algorithms selecting $r > k$ columns. 

\subsection*{The spectral norm case ($\xi=2$)}
We summarize prior work by presenting algorithms with accuracy 
guarantees $\alpha$ with respect 
to the equation
\vspace{-0.12in}
$$\TNorm{\matA-\matC \matC^+\matA} \le  
\TNorm{\matA-\Pi_{\matC,k}^{2}(\matA)} \le 
\alpha \TNorm{\matA-\matA_k}.$$

\paragraph{The $r = k$ case.} 
The strongest bound for $r=k$ emerges from the Strong Rank Revealing QR (RRQR) algorithm of \cite{GE96} (specifically Algorithm 4 in \cite{GE96}), which, for
\math{f>1}, runs in $O(mnk \log_{f}n )$ time, and constructs $\matC \in \R^{m \times k}$:
\vspace{-0.12in}
$$\TNorm{\matA-\matC \matC^+\matA} =  
\TNorm{\matA-\Pi_{\matC,k}^{2}(\matA)} 
\le \sqrt{f^2 k(n-k) +1} \TNorm{\matA-\matA_k}.$$
One can also trivially get an $\alpha \sqrt{\rho-k}$-approximation in the spectral norm from
any algorithm that guarantees an $\alpha$-approximation in the Frobenius norm, because,
for any matrix $\matX$,
\math{\norm{\matX}_2\le \norm{\matX}_F\le \sqrt{\rank(\matX)}\norm{\matX}_2}; so,
\mand{
\TNorm{\matA - \matC\matC^+\matA} 
\le
\FNorm{\matA - \matC\matC^+\matA}
\le 
\alpha \FNorm{\matA-\matA_k} 
\le 
\alpha \sqrt{\rho-k} \TNorm{\matA-\matA_k}.
}
\mand{
\TNorm{\matA - \Pi_{\matC,k}^{2}(\matA) }
\le
\TNorm{\matA - \Pi_{\matC,k}^{F}(\matA) } 
\le 
\FNorm{\matA - \Pi_{\matC,k}^{F}(\matA)}
\le 
\alpha \FNorm{\matA-\matA_k} 
\le 
\alpha \sqrt{\rho-k}\TNorm{\matA-\matA_k}.
}
\paragraph{The $r > k$ case.} We are not aware of any bound
applicable to this domain other than those obtained by 
extending the available Frobenius norm bounds as indicated above.
\cite{RV07} presents an additive-error spectral norm algorithm but only for a special
case of matrices where the ratio $ \eta = \frac{\FNormS{\matA}}{\TNormS{A}}$ is ``small''. 
More specifically, Theorem 1.1 in~\cite{RV07}, for fixed $\matA$, $0 < \delta < 1$, 
$0 < \epsilon < 1$, constructs a matrix $\matC \in \R^{m \times r}$ with 
$r = \Omega( \frac{\eta}{\epsilon^4 \delta} \log( \frac{\eta}{\epsilon^4 \delta} )  )$
as $\matC = AdditiveSampling(\matA, r)$ 
(see Section \ref{chap311} for the description of $ AdditiveSampling$) such that w.p. 
$1 - \frac{2}{e^{O(\delta)}}$:
$$ \TNorm{ \matA -  \Pi_{\matC,k}^{2}(\matA) } \le \TNorm{ \matA -  \Pi_{\matC,k}^{F}(\matA) }  \leq \TNorm{\matA - \matA_k} + \epsilon \TNorm{\matA}.$$
(The term $\Pi_{\matC,k}^{F}(\matA)$ follows after reorganization of the notation
in~\cite{RV07}. The left inequality
follows by the optimality of $\Pi_{\matC,k}^{2}(\matA)$). 

\subsection*{Column-based approximations with the Rank-Revealing QR}

Finally, we comment on a line of research within the 
Numerical Linear Algebra community that, although does not primarily focus
on low-rank column based matrix approximations, does provide several algorithms
for such a task. This line of research focuses on spectral norm bounds for the $r = k$
case. We start with the definition of the so-called
\textsc{Rank Revealing QR (RRQR) factorization}.
\begin{definition}
\label{def:rrqr}
\textbf{(The RRQR factorization)}
Given a matrix $\matA \in R ^ {m \times n}$  ($m \geq n$) and an integer $k$
($k \leq n$), assume partial $\matQ R$ factorizations of the form:
\begin{eqnarray*}
\label{eq:rrqr}
\matA \Pi = \matQ R
      = \matQ\left(
            \begin{array}{cc}
                 R_{11} & R_{12} \\
                 0  & R_{22}
            \end{array}
         \right)   ,
\end{eqnarray*}
where $\matQ \in R^{m \times n}$ is an orthonormal matrix, $R \in R^{n
\times n}$ is upper triangular, $R_{11} \in R^{k \times k}$, $R_{12}
\in R^{k \times (n-k)}$, $R_{22} \in R^{(n-k) \times (n-k)}$, and
$\Pi \in R^{n \times n}$ is a permutation matrix. The above
factorization is called a RRQR factorization if it satisfies
\begin{eqnarray*}
\label{eq:rrqrb}
\frac{\sigma_{k}(A)} { p_1(k,n)} \leq &\sigma_{min}(R_{11})& \leq
\sigma_{k}(\matA)\\
\sigma_{k+1}(\matA) \leq &\sigma_{max}(R_{22})& \leq p_2(k,n)
\sigma_{k+1}(\matA),
\end{eqnarray*}
where $p_1(k,n)$ and $p_2(k,n)$ are functions bounded by low degree
polynomials in $k$,$n$.
\end{definition}
\begin{table}
\begin{center}
\begin{tabular}{|l|l l|l|l|}
\hline \hline
  \textbf{Method} & \textbf{Reference }  &                       &
\textbf{$p_2(k,n)$}     & \textbf{Time}                             \\
  \hline \hline
  Pivoted QR   &  \text{[Golub, '65]} & \cite{Gol65} &  $\sqrt{(n-k)}
2^{k}$           &        $O(mnk)$                      \\
  \hline
  High RRQR  & \text{[Foster, '86]} & \cite{Fos86}                 &
$\sqrt{n(n-k)}  2^{n-k}    $ &$O(mn^2)$    \\
  \hline
  High RRQR  & \text{[Chan, '87]} & \cite{Cha87}                  &
$\sqrt{n(n-k)}  2^{n-k}    $ & $O(mn^2)$ \\
  \hline
  RRQR   & \text{[Hong/Pan, '92]} &
\cite{HP92}                      & $\sqrt{k(n-k)+ k} $ & -    \\
  \hline
  Low RRQR   & \text{[Chan/Hansen, '94]} &
\cite{CH94}                   & $\sqrt{(k+1)n}  2^{k+1}  $ & $O(mn^2)$\\
  \hline
  Hybrid-I        & \text{[Chandr./Ipsen, '94]} &
\cite{CI94}                  & $\sqrt{(k+1)(n-k)}         $ & - \\
  Hybrid-II       &      &
\cite{CI94}                                                   &
$\sqrt{(k+1)(n-k)}          $ & - \\
  Hybrid-III     &      &
\cite{CI94}                                                   &
$\sqrt{(k+1)(n-k)}          $ & - \\
  \hline
  Algorithm 3   & \text{[Gu/Eisenstat, '96]} &
\cite{GE96}               & $\sqrt{k(n-k)+1}           $  & -  \\
  Algorithm 4   &          &
\cite{GE96}                                      &
$\sqrt{f^2k(n-k)+1}         $ & $O(kmn\log_{f}(n))$ \\
  \hline
  DGEQPY        & \text{[Bischof/ Orti, '98]} &
\cite{BQ98a}               & $O(\sqrt{(k+1)^2(n-k)})   $ & - \\
  DGEQPX        &                 & \cite{BQ98a}
& $O(\sqrt{(k+1)(n-k)})    $   & -   \\
  \hline
    SPQR   & \text{[Stewart, '99]}       &     \cite{Ste99}       & - &
-   \\
  \hline
  Algorithm 1       & \text{[Pan/Tang, '99]} & \cite{PT99}
& $O(\sqrt{(k+1)(n-k)})$    & - \\
  Algorithm 2       &                   & \cite{PT99}
& $O(\sqrt{(k+1)^2(n-k)})$ & - \\
  Algorithm 3       &                   & \cite{PT99}
& $O(\sqrt{(k+1)^2(n-k)}) $ & - \\
  \hline
  Algorithm 2       & \text{[Pan, '00]} &
\cite{Pan00}                           & $O(\sqrt{k(n-k)+ 1})$   &
-  \\
  \hline
\end{tabular}
\caption{
Deterministic Rank-Revealing QR Algorithms. 
\label{table:33}
}
\end{center}
\end{table}
The work of Golub on pivoted $QR$
factorizations~\cite{Gol65} was followed by much research addressing
the problem of constructing an efficient RRQR factorization. Most
researchers improved RRQR factorizations by focusing on improving
the functions $p_1(k,n)$ and $p_2(k,n)$ in
Definition~\ref{def:rrqr}. Let $\Pi_k \in \R^{n \times k}$ denote the first $k$ columns
of the permutation matrix $\Pi$. Then, if $\matC=\matA\Pi_k$:
\vspace{-0.12in}
$$
\XNorm{ \matA - \matC \matC^+ \matA } = \TNorm{\matA-\Pi_{\matC,k}^{2}(\matA)}  = \XNorm{ R_{22}},
$$
for both $\xi=2,F$.
Thus, for $\xi=2$, it follows that
\vspace{-0.12in}
$$
\TNorm{ \matA - \matC \matC^+ \matA } = \TNorm{\matA-\Pi_{\matC,k}^{2}(\matA)} 
 \leq  p_2(k,n) \sigma_{k+1}(\matA)
 = p_2(k,n)\TNorm{\matA-\matA_k}  ,
$$
i.e., any algorithm that constructs a RRQR factorization with provable guarantees also provides the same provable
guarantees for column-based matrix approximation. See Table~\ref{table:33} for a summary
of existing results. A dash implies that the running time
of the algorithm is not explicitly stated in the corresponding citation
or that it depends exponentially on $k$.  
In addition, $m \geq n$ and $f > 1$ for this table. We should note that 
our Theorem~\ref{thmCSSPs} gives an algorithm which is (up to a constant 4) as accurate
as the best algorithm of this table but faster by a factor $O(1/\log(k))$.

\section{Prior Work: Coresets for Least-Squares Regression} \label{chap33} 

We quickly review the definition of the least-squares problem.
The input is a matrix $\matA \in \R^{m \times n}$ ($m \gg n$) of rank $\rho = n$ and a
vector $\b \in \R^m$; the output is a vector $\x_{opt} \in \R^n$ that minimizes the distance
$\TNorm{\matA \x-\b}$, over all $\x$. If no additional constraints are placed on $\x$,
$\x_{opt} = \matA^+ \b$, which can be computed in $O(mn^2)$ time~\cite{GV89}. 
When $\x$ is forced to satisfy additional constraints, for example, all elements of $\x$ are nonnegative,
then, the regression problem does not have an analytical solution and
algorithms can be as complex as the constraints. The results of this section along with
our results from Chapter~\ref{chap5} are summarized in Table~\ref{table:34}.

\paragraph{Coresets for Constrained Regression.}

Dasgupta et. al.~\cite{DDHKM08} presented a randomized 
algorithm\footnote[3]{This algorithm actually provides coresets for the more general problem of minimizing $||\matA \x - \b ||_p$,
for $p=1,2,3,...$, over all $\x$ (possibly constrained) but since our focus is on $p=2$, we do not elaborate
further on the $p \ne 2$ case.}
that, with probability $0.5$, 
constructs a $(1+\epsilon)$-coreset of size
$r = \frac{1296 n}{\epsilon^2}( n \ln(\frac{36}{\epsilon}) + \ln(200))$, 
in time $O(m n \min\{m,n\} + T_c)$. $T_c$ is the time needed to compute
an exact solution of a constrained regression problem with inputs 
$\hat\matC \in \R^{\hat{r} \times n}$ and 
$\hat{\b}_c \in \R^{\hat{r}}$, 
which contain $\hat{r} = 82,944 n ( n \ln(288) + \ln(200))$ 
(rescaled) rows from $\matA$ and $\b$, respectively. The constraints in this smaller
problem are the same as those in the original problem.
$T_c$ depends on the specific constraints on $\x$. For example,
if no constraints are placed on $\x$, $T_c = O( \hat{r} n^2)$ via the SVD. 
The algorithm from~\cite{DDHKM08} runs in two stages. 
In the first stage, 
$$\hat{r} = 82,944 n ( n \ln(288) + \ln(200))$$ 
rows from $\matA, \b$ are selected to form $\hat\matC \in \R^{\hat{r} \times n}$, $\hat{\b}_c \in \R^{\hat{r} }$; 
then, 
in the second stage, exactly
$$r = \frac{1296 n}{\epsilon^2}( n \ln(\frac{36}{\epsilon}) + \ln(200))$$
rows are selected to construct the coreset
$\matC \in \R^{r \times n}$ and $\b_c \in \R^r$. 

In both stages, the sampling of the rows is done with the same probabilistic technique.
More specifically, for a fixed set of probabilities $p_1,...,p_m$, a diagonal 
matrix $\matQ_j \in \R^{m \times m}$ is constructed as follows ($j=1,2$ for the two stages). 
For $i=1,...,m$ the $i$-th diagonal element of $\matQ_j$ is equal to 
$1/\sqrt{p_i}$ with probability $p_i$ and $0$ with probability $1 - p_i$ (Bernoulli process).
In the first stage ($j=1$), the probabilities are 
$$p_i = \min\{1,  \frac{ \TNormS{ (\matU_{\matA})_{(i)} } }{n} \hat r \};$$
in the second stage ($j=2$), the probabilities are refined as: 
$$\hat{p_i} = \min\{ 1, \max\{ p_i, \frac{ z_i }{\TNormS{\z}} r \} \},$$ 
where
$\z = [z1, ..., z_m] = (\matA \hat{\x}_{opt}-\b)\transp \in \R^{1 \times m}$, with 
$\hat{\x}_{opt} =  \arg\min_{\x}|| \matQ_1\matA \x - \matQ_1\b ||_2$.
The final coreset $\matC \in \R^{r \times n}$, $\b_c \in \R^r$ 
corresponds to the non-zero rows and elements from 
$\matQ_2 \matA$ and $\matQ_2 \b$, respectively. We slightly abused notation,
since, in both stages,
the actual rows that are sampled are not exactly $\hat{r}$ and $r$.
It can be proved thought that - in expectation - 
$\hat{r}$ and $r$ rows are sampled from $\matA, \b$ in
the first and the second stage, respectively. 
\begin{table}
\begin{center}
\begin{tabular}{|l|l|l|l|l|l|}
\hline
Year & Ref.  & coreset size $r=$ & Time=$O(x), x = $ & $\delta$ & C/C \\
\hline
2008 & \cite{DDHKM08} & $1296 n( n \ln(\frac{36}{\epsilon}) +\ln(200))/\epsilon^2$ & - & $0.5$ & Y/Y \\
\hline
2011 & Thm~\ref{lem:regression}  &  $ 225(n+1)/\epsilon^2 $  & $mn^2 + mn^3/\epsilon^2$  & $0$ & Y/Y \\
\hline
2011 & Thm~\ref{lem:regression2} & $ 36(n+1)\ln(2(n+1)/\delta_0) /\epsilon^2$ & $mn^2+r\log(r)$ & $\delta_0$ & Y/Y \\
\hline
2006 & \cite{DMM06a} & $3492n^2\ln(3/\delta_1)/\epsilon^2$ & $mn^2 + r\log(r)$ & $\delta_1$ & N/Y \\
\hline
2008 & \cite{DMM08}  & $3200n^2 / \epsilon^2$ & $mn^2 + r\log(r)$ & $0.3$ & N/Y \\
\hline
2008 & \cite{DMM08}  & $c_1 n \log(n) / \epsilon^2$ & $mn^2$ & $0.3$ & N/Y \\
\hline
2006 & \cite{Sar06}  & $c_2 n \log(n)/\epsilon $        & $ mn\log(m) +t_0\log(n)/\epsilon $ & $0.7$ & N/N \\
\hline
2011 & \cite{Zou10}  & $c_3 n / \epsilon $            & $mn\log(m) + t_0/\epsilon$ & $0.7$ & N/N \\
\hline
2007 & \cite{DMMS07} & $\max\{ \xi_1 , 40 n \ln(40mn)/\epsilon \} $ & $mn\log(r)$ & $0.2$ & N/N \\
\hline
2007 & \cite{DMMS07} & $\max\{ c_4(118^2d+98^2) , 60n/\epsilon \} $ & $mn\log(m r t_1)$ & $0.2$ & N/N \\
\hline
2008 & \cite{RT08}   & $\left(\frac{(1+\epsilon)^2 + 1}{(1+\epsilon)^2 - 1}\right)^2 10(n+1)^2$ & $mn\log(r)$ & $0.1$ & N/N \\
\hline
2011 & Thm~\ref{lem:regression3}& $ 72(n+1)\ln(2(n+1)/\delta_2) \xi_2 /\epsilon^2$ & $mn\log(r)$ & $\delta_2$ & N/Y \\
\hline
\end{tabular}
\caption{ Coresets/``Coresets'' for (Constrained) Regression $\TNorm{\matA \x - \b}$; $\matA \in \R^{m \times n}$. 
C/C abbreviates Coreset/Constraint and N/Y (NO/YES) implies ``Coreset'' 
for Constrained Regression. $\rho = \rank(\matA) = n$.
$c_1$, $c_2$,..., are (unspecified) constants.  A dash means that the time
can not be specified exactly. $\delta$ is the failure probability. 
$\xi_2 = \log(40(n+1)m)$.
$\xi_1 =  48^2 n \ln(40mn) \ln(100^2 n \ln(40mn)) $;
 $t_0 = n^2\log^2(m)$; $t_1 = \frac{n \ln(mn)}{m}(\ln(m)+n)$.
}
\label{table:34}
\end{center}
\end{table}
  
\paragraph{Coresets for Unconstrained Regression.}
When no constraints are placed on $\x$, there are a few available
algorithms in the literature that we present below.

First, Theorem 3.1 in~\cite{DMM06a} describes a randomized algorithm
that, with probability at least $1 - \delta$, for any $0 < \delta < 1$,
constructs a $(1+\epsilon)$-coreset of size $r = 3492 n^2 \ln(1/\delta) / \epsilon^2$
in time $O( mn^2 + r \log(r) )$. The
algorithm of ~\cite{DMM06a} is as follows: (i) for $i=1,...,m$, 
compute probabilities $p_i$ that satisfy three conditions 
(let $\matU_{\matA}\transp                 = [ \u_1,..., \u_m ] \in \R^{n \times m}$, $\u_i \in \R^n$, 
and  $ (\matA\matA^+\b -\b)\transp = [ z_1,...,  z_m ] \in \R^{1 \times m}$):
\mand{  
p_i \ge \frac{1}{3} \frac{ \u_i\transp\u_i  }{ \rho }; 
\quad 
p_i \ge \frac{1}{3} \frac{ (\u_i\transp\u_i) z_i } { \sum_{j=1}^{m} (\u_j\transp\u_j) z_j };
\qquad  
p_i \ge \frac{1}{3} \frac{ z_i }{ \TNormS{\x_{opt}} };
}
(ii) sample $r$ rows from $\matA$ in $r$ i.i.d. trials, where in each trial the $i$-th row 
is sampled with probability $p_i$ (an appropriate rescaling also applies to the rows of $\matA$ and $\b$). 
Notice that this algorithm requires the computation of the matrix $\matU_{\matA}$ from the SVD and the optimum solution vector 
$\x_{opt} = \matA^+\b$. The factors $1/3$ in the above conditions can actually
be relaxed to any values $0 < \beta_1, \beta_2, \beta_3 \le 1$, with a change 
in the sampling complexity by a multiplicative factor $ \min\{\beta_1^2, \beta_2^2, \beta_1^2\} / 9$. 
An interesting open question is whether one can approximate these probabilities in $o(mn\min\{m,n\})$
time. Magdon-Ismail in~\cite{Mag10} made progress towards this direction.
Unfortunately, the algorithm of~\cite{Mag10}, which is based on random projections, 
returns probabilities that satisfy only the first condition, so, in
the parlance of Theorem 3.1 in~\cite{DMM06a},~\cite{Mag10} does not provide any useful 
result. As we will see below though, 
sampling with probabilities satisfying only the first condition gives a $(1+\epsilon)$-coreset
with comparable size but only
with constant probability. Notice that the algorithm of this paragraph fails with arbitrarily 
small probability $\delta$ and $r=O(\ln(1/\delta))$. 

Second, the first algorithm of Theorem 5 in~\cite{DMM08} presents a randomized algorithm that, for
any $0 < \beta \le 1$, w.p. $0.7$, constructs a $(1+\epsilon)$-coreset of size 
$r = 3200 n^2 / \beta \epsilon^2$ by running $[\Omega, \matS] = SubspaceSampling(\matU_\matA, \beta, \epsilon)$
as we described in Section~\ref{chap314}. This algorithm needs $O( mn^2 + r \log(r) )$. 
Inspecting carefully the $SubspaceSampling$ method, we see that the probabilities that are used to sample columns
from $\matU_{\matA}\transp$ must satisfy $p_i \ge \beta \frac{ \u_i\transp\u_i  }{ \rho }$,
which is the first of the three conditions that we saw in the previous paragraph.
One can compute these probabilities with $\beta = 1$ through the SVD. 
\cite{Mag10} showed how to approximate these probabilities in $o(mn^2)$ using random 
projections\footnote[4]{The algorithm that approximates the probabilities is randomized and, for any
$0 < \delta < 1$, does so with probability $1-  \delta$ 
in time $O( m n \log(m) + m n \log(n) \log(\frac{1}{\delta}) + \frac{m n^2}{ \log(m) } \log(\frac{1}{\delta}) )$.
So, for $\log(m) = o(n)$, the total run time is $o(mn^2)$.
}. 
An important observation about the algorithm in Theorem 5 in~\cite{DMM08} is that it doesn't need access
to $\b$ in order to compute the coreset. This is useful, because the same coreset can be used for multiple
$\b$. 

Third, the second algorithm of Theorem 5 in~\cite{DMM08} presents a randomized algorithm that, with
constant probability, constructs a $(1+\epsilon)$-coreset of size 
$r = O( n \log(n) / \beta \epsilon^2)$ in time $O( mn^2 )$. Here, the
sampling probabilities are the same as in the previous paragraph; the sampling though is done with
a slightly different way (without replacement). 
The algorithm of this paragraph is the best available in the literature with respect to the coreset
size and running time. Again, the result from~\cite{Mag10}
implies a running time of $o(mn^2)$ to approximate the probabilities.
\vspace{-.18in}
\paragraph{``Coresets'' for Unconstrained Regression.}
If we allow the coreset to be linear combinations of the rows 
of $\matA$ and $\b$, there are several methods that
construct $(1+\epsilon)$-``coresets'' in $o(mn^2)$ time. 
Sarlos in [Eqn. 2, Theorem 12, \cite{Sar06}] gives an algorithm to construct a 
$(1+\epsilon)$-``coreset'' of size $r = O(n \log(n)/\epsilon)$. 
\cite{Sar06} uses the Fast Random Projection method of~\cite{AC06}.
Along the lines in~\cite{Sar06}, Avner and Zouzias in [Eqn. 3.3, Theorem 3.3, \cite{Zou10}] showed that 
$r = O(n / \epsilon)$ rows are enough; the construction in \cite{Zou10} 
uses the method of~\cite{AC06} as well.  
Drineas et. al in~\cite{DMMS07} employ the Subsampled Randomized Hadamard Transform
that we described in Section~\ref{chap3110} to design faster algorithms for constructing
$(1+\epsilon)$-``coresets''. More specifically, [Algorithm 1, Theorem 2, \cite{DMMS07}] 
describes a method which does so with coreset size 
$r = \max\{ 48^2 n \ln(40mn) \ln(100^2 n \ln(40nd)) , 40 n \ln(40mn)/\epsilon \}$. 
The difference of this algorithm with ours (Theorem~\ref{lem:regression2}) is that our ``coreset''
works for arbitrary constraint regression.
Moreover, [Algorithm 2, Theorem 3, \cite{DMMS07}] requires 
$r = \max\{ c(118^2d+98^2) , 60n/\epsilon \}$, for a constant $c$; this algorithm is based
on a (sparse) random projection embedding. 
Along the same lines as \cite{DMMS07} - by using a different matrix than the Hadamard - 
Rokhlin and Tygert in [Lemma 2, \cite{RT08}] describe an algorithm which constructs a ``coreset'' of size 
$r = (\frac{(1+\epsilon)^2 + 1}{(1+\epsilon)^2 - 1})^210(n+1)^2$
w.p. $0.9$. Finally, \cite{AMT10} presents experiments with the randomized method 
of~\cite{DMMS07}. 
 
\section*{Nonnegative Least Squares Regression Problems}
Notice that so far we haven't discussed algorithms and running times
for solving the regression problem per se. Here, we elaborate
on a popular constrained regression problem called NNLS. The purpose
of this section is to give further motivation that coresets for constrained
regression problems are indeed important in several applications.
The Nonnegative Least Squares (NNLS) problem is a constrained
least-squares regression problem where the variables are allowed
to take only nonnegative values. 
NNLS is a quadratic optimization problem with linear inequality
constraints. As such, it is a convex optimization problem and thus
it is solvable (up to arbitrary accuracy) in polynomial time
\cite{Bjo96}. 
The motivation for NNLS problems in data mining and machine
learning stems from the fact that given least-squares regression
problems on nonnegative data such as images, text, etc., it is
natural to seek nonnegative solution vectors. (Examples of
data applications are described in \cite{CP07}.) NNLS is also
useful in the computation of the Nonnegative Matrix
Factorization~\cite{KP07a}, which has received considerable
attention in the past few years. Finally, NNLS is the core
optimization problem and the computational bottleneck in designing
a class of Support Vector Machines~\cite{SSL02}. Since modern
datasets are often massive, there is continuous need for faster,
more efficient algorithms for NNLS.
\vspace{-0.12in}
\paragraph{Existing Algorithms.}
We briefly review NNLS algorithms following the extensive review
in~\cite{CP07}. Such algorithms can be
divided into three general categories: ($i$) active set methods,
($ii$) iterative methods, and ($iii$) other methods. The approach
of Lawson and Hanson in \cite{LH74} seems to be the first
technique to solve NNLS problems. It is a typical example of an
active set method and is implemented as the function
\emph{lsqnonneg} in Matlab. Immediate followups to this work
include the technique of Bro and Jong~\cite{BJ97} which is
suitable for problems with multiple right hand sides, as well as
the combinatorial NNLS approach of Dax~\cite{Dax91}. The
Projective Quasi-Newton NNLS algorithm of \cite{KSD07} is an
example from the second category. It is an iterative approach
based on the Newton iteration and the efficient approximation of
the Hessian matrix. Numerical experiments in \cite{KSD07} indicate
that it is a very fast alternative to the aforementioned active
set methods. 
Finally, the sequential coordinate-wise approach
of~\cite{FHN05} and
interior point methods~\cite{PJV94} are also useful.

\paragraph{A Faster NNLS Algorithm.}
We show how the ``coreset'' construction algorithm of Theorem \ref{lem:regression2}
can be used to speed up existing NNLS algorithms. Recall that the
the construction of the ``coreset'' takes
$T_{cor} = O(m n \log ( n \log(m) / \epsilon^2 )).$ After
this step, we can employ a standard NNLS solver on the smaller
problem. The computational cost of the NNLS solver on the small
problem is denoted as $T_{NNLS}(r,n)$, with $r = O(n\ln(n)\log(nm)/\epsilon^2)$. 
Compare this with 
$T_{NNLS}(m,n)$ which is the time needed to solve the problem exactly. 
$T_{NNLS}(r,n)$ cannot be specified since theoretical running times for 
exact NNLS solvers are unknown. In the sequel we comment on the 
computational costs of some well defined segments of some NNLS solvers.
NNLS is a convex quadratic program:
\begin{eqnarray*}
\min_{x \in \mathbb{R}^n, \x\geq 0} \x^T \matQ \x - 2\q\transp \x, 
\end{eqnarray*}
where $\matQ = \matA\transp \matA \in \mathbb{R}^{n\times n}$ and $\q = \matA\transp \b \in
\mathbb{R}^n$. Computing $\matQ$ and $\q$ takes $O(mn^2)$ time, and
then, the time required to solve the above formulation of the NNLS
problem is independent of $m$. Using this formulation, our
algorithm would necessitate $T_{cor}$ time for the computation
of the ``coreset'', and then $\tilde{\matQ} = \matC\transp \matC$ and
$\tilde{\q} = \matC\transp \b_c$ can be computed in
$T_{MM} = O(rn^2)$ time (MM stands for Matrix Multiplication); 
given our choice of $r$, this implies
$ T_{MM} = O(n^3 \log(n m) / \epsilon^2).$ Overall, the standard
approach would take $O(mn^2)$ time to compute $\matQ$, whereas our
method would need only $T_{cor}$ + $T_{MM}$ time for the
construction of $\tilde{\matQ}$. Note, for example, that when
$m=O(n^2)$ and treating $\epsilon$ as a constant, $\tilde{\matQ}$ 
can be computed $O(n / \log(n m))$ faster.

On the other hand, many standard implementations of NNLS solvers
(and in particular those that are based on active set methods)
work directly on  $\matA$ and $\b$. A
typical cost of these implementations is of the order $O(mn^2)$
per iteration. Other approaches, for example the NNLS method of
\cite{KSD07}, proceed by computing matrix-vector products of the
form $\matA \u$, for an appropriate $n$-dimensional vector $u$, thus
cost typically $O(mn)$ time per iteration. In these cases, our
algorithm needs again $T_{cor}$ preprocessing time, but costs
only $O(rn^2)$ or $O(rn)$ time per iteration, respectively. Again,
the computational savings per
iteration are $O(n / \log(nm))$.

\begin{table*}[htdp]
\begin{center}
\begin{tabular}{|c|c|c|c|c|}
\hline
\textbf{Reference}    & \textbf{Description}     & \textbf{Dimensions} & \textbf{Time = $O(x), x =$} & \textbf{Error }          \\
\hline
 Folklore        & RP  & $O(\log(m)/\epsilon^2)$        & $mn \lceil \epsilon^{-2} \log(m)/ \log(d) \rceil$   & $1+\epsilon$ \\
\hline
\cite{DFKVV99}   & Exact SVD & $k$                                 & $mn\min\{m,n\}$ & $2$    \\
\hline
Theorem \ref{fastkmeans}       & RS & $O(k \log(k)/\epsilon^2)$ & $mnk\epsilon^{-1} + t_0$ & $3+\epsilon$ \\
\hline
Theorem \ref{thm:second_result}      & RP  & $O(k / \epsilon^2)$                & $mn \lceil \epsilon^{-2} k/ \log(n) \rceil$   & $2+\epsilon$ \\
\hline
Theorem \ref{thm:first_result}      & Approx. SVD  & $k$                & $mnk\epsilon^{-1}$   & $2+\epsilon$\\
\hline
\end{tabular}
\end{center}
\label{table:35} 
\caption{\small{Provably Accurate dimensionality reduction methods for
$k$-means clustering. RP stands for Random Projections, similarly
for RS and Random Sampling. The technique in the second row of the table is deterministic;
the others fail with, say, a constant probability. In the Random Projection methods the 
construction is done with random sign matrices and the mailman algorithm (see Section \ref{chap62}). 
$t_0 = k \log(k) \epsilon^{-2} \log(k \log(k) \epsilon^{-1})$.
}}
\end{table*}
\section{Prior Work: Feature Selection for $k$-means Clustering}\label{chap34} 
Dimensionality reduction encompasses the union of two different
approaches: feature selection, which embeds the points into a
low-dimensional space by selecting actual dimensions of the data,
and feature extraction, which finds an embedding by constructing
new artificial features that are, for example, linear combinations
of the original features. Let $\matA$ be an $m\times n$ 
matrix containing $m$ $n$-dimensional points ($\matA_{(i)}$ denotes the $i$-th point of the set), and let $k$ be the
number of clusters. 
We say that an embedding 
$f: \matA \to \R^{r}$ with $f(\matA_{(i)}) = \matC_{(i)}$ for all $i=1,...,m$ 
and some
$r < n$, preserves the clustering structure of $\matA$ within a factor
$\phi$, for some $\phi \geq 1$, if finding an optimal clustering in $\matC \in \R^{m \times r}$
and plugging it back to $\matA \in \R^{m \times n}$ is only a factor of $\phi$ worse than finding the optimal
clustering directly in $\matA$. 

Prior efforts on designing provably accurate dimensionality
reduction methods for $k$-means include:
(i) random projections, where one projects the input points into $ r = O(\log(m)/ \epsilon^2)$ dimensions such
that, the clustering structure is preserved within a factor of $\phi=1+\epsilon$; and
(ii) the Singular Value Decomposition (SVD), where one constructs 
$\matC = \matU_k \Sigma_k \in \R^{m \times k}$ such that
the clustering structure is preserved within a factor of $\phi=2$.
We summarize prior work and our results in Table~\ref{table:35}.
Finally, other techniques, 
for example the Laplacian scores~\cite{HCN06} or the Fisher scores~\cite{FS75}, 
are very popular in applications (see also~\cite{GE03}). 
However, they lack a theoretical analysis;
so, a discussion of those techniques is beyond 
the scope of this thesis.


\section{Similar Studies}\label{chap35} 
There are several more related problems that received 
considerable attention over the last decade or so. We
discuss some of these problems below. The presentation
is careful to highlight the differences and the similarities
with the algorithms and the results presented in this
thesis. Other sources that give the big picture of this
area include the surveys in~\cite{HMT, Kan10}. 

\subsection*{Matrix Multiplication}                 
Approximate Matrix Multiplication studies quick approximations of the product $\matA \matB$,
given $\matA$, $\matB$. Let $\matA \in \R^{m \times n}$ and $\matB \in \R^{n \times p}$. 
Then, a standard BLAS implementation takes $O(m n p)$ time to compute the product $\matA \matB$
exactly. Algorithms that approximate the product $\matA \matB$ in $o(m n p)$ time are of considerable
interest. An idea that received considerable attention
is to do that by first sampling a small subset of the columns of $\matA$ and the corresponding
rows from $\matB$, and then approximate the product $\matA \matB$ with $\tilde{\matA} \tilde{\matB}$;
Let $r>0$ is the number of sampled columns and rows from $\matA$ and $\matB$, respectively. Then
$\tilde{\matA} \in \R^{m \times r}$ and  $\tilde{\matB} \in \R^{r \times n}$ are the subsampled matrices.
Notice that the product $\tilde{\matA} \tilde{\matB}$ takes $O(m r p)$ time; so, one needs to choose $r$
as well as the technique to construct $\tilde{\matA}, \tilde{\matB}$ such that the overall running time
is $o(m n p)$. There are several such approaches in existing literature; we
describe some of them below. 
\cite{DK01, DKM06a} described a randomized sampling based algorithm
with approximation:
$$ \Expect{ \FNormS{ \matA \matB - \tilde{\matA} \tilde{\matB} } }
\le \frac{1}{r} \FNormS{\matA} \FNormS{\matB}.$$
The method of~\cite{DK01, DKM06a} proceeds as follows. First, compute probabilities:
$$ p_i = \frac{\TNorm{ \matA^{(i)} } \TNorm{\matB_{(i)}} }{ \sum_{j=1}^{n} \TNorm{ \matA^{(j)} } \TNorm{\matB_{(j)}} }.$$
Then, construct a sampling matrix $\Omega \in \R^{n \times r}$ and a rescaling matrix $\matS \in \R^{r \times r}$ as follows. Initially, $\Omega = \bm{0}_{n \times r}$ and $\matS=\bm{0}_{r \times r}$. 
Then, for every column $j=1,...,r$ of $\Omega$, $\matS$, independently, pick an index $i$ from the set $\{1,2,...,n\}$ with probability $p_i$ and set $\Omega_{ij} = 1$ and $\matS_{jj} = 1/\sqrt{p_i r}$. Finally,
set $\tilde{\matA} = \matA \Omega \matS$ and $\tilde{\matB} = \matS\transp \Omega\transp \matB$.
It is worth noting that constructing $\Omega, \matS$ by using any set of probabilities $p_1,...,p_n$ gives:
$$ \Expect{ \FNormS{ \matA \matB - \tilde{\matA} \tilde{\matB} } }
\le  \sum_{j=1}^n \frac{ \TNormS{ \matA^{(i)} } \TNormS{\matB_{(i)}} }{ r p_j }  - \frac{1}{r} \FNormS{\matA \matB}.$$
In fact, \cite{DK01, DKM06a} select the appropriate probabilities to get good bounds for the residual error. 
In the above approach, computing the probabilities takes $O(mn + np)$, sampling with replacement necessitates 
$O(n + r\log(r))$, and computing the product $\tilde{\matA} \tilde{\matB}$ takes $O(m r p)$. Overall, this 
approaches runs in $O( mn + np + mpr + r\log(r) )$. 

The approach that we outlined above works for matrices of arbitrary dimensions and this is actually
quite impressive. To our best knowledge, the above approach is the best
sampling-based method for approximate matrix multiplication with respect to the Frobenius norm.
Notice thought that this algorithm is randomized. An interesting open question is whether there
exists a deterministic algorithm with comparable approximation bounds. Although such a deterministic
algorithm might be too costly for approximating the product $\matA \matB$ in $o(m n p)$, we believe
that such a  result will open new directions in designing deterministic algorithms for
other problems involving subsampling, e.g. low-rank matrix approximation and $k$-means clustering. 

For approximations with respect to the spectral norm we refer the reader to
\cite{RV07, Zou10, Mag10, HKZ11}. We comment carefully on the best of these
results, which is statement (ii) in Theorem 3.2 in~\cite{Zou10}. The setting
is similar with the example that we outlined above for the Frobenius norm case.
The construction of $\tilde{\matA}, \tilde{\matB}$ is done with the same
algorithm as well. The results for the spectral norm though do not hold for
arbitrary matrices. It is required that the matrices have low stable rank.
We define the stable rank of a matrix $\matA$ as
$$ sr(\matA) = \frac{ \FNormS{\matA} }{ \TNormS{\matA} }. $$
Assume that $sr(\matA), sr(\matB) \le \tilde{\rho}$. 
Let $\epsilon > 0$ and $r = \Omega( \tilde{\rho} \log(\tilde{\rho} / \epsilon^2) / \epsilon^2 )$.
Then, statement (ii) in Theorem 3.2 in~\cite{Zou10} argues that with probability
at least $1 - \frac{1}{poly(\tilde{\rho})}$:
$$ \TNorm{ \matA \matB  - \tilde{\matA} \tilde{\matB}  } \le \epsilon \TNorm{\matA} \TNorm{\matB}.$$

\paragraph{Deterministic Symmetric Multiplication.} 
Notice that the algorithm just mentioned above is also randomized. Deterministic approaches for the
same problem will be particularly important. The celebrated deterministic sparsification result of~\cite{BSS09} gives
such a bound for short-and-fat matrices with orthonormal rows. We stated this result in Section~\ref{chap316}. 
This bound can be generalized to short-and-fat matrices with any set of rows (with an additional
$O(nk^2)$ cost for the SVD of the matrix - see below), but for now let us focus on
the simple case. Let $\matV \in \R^{n \times k}$ with $n > k$ and $\matV\transp \matV = \matI_k$. Let $r > k$. Then,
deterministically in $O( n k^2 r )$:
$$ \TNorm{ \matV\transp \Omega \matS \matS\transp \Omega\transp \matV - \matV\transp \matV } \le \sqrt{ \frac{9 k}{r} } \TNorm{\matV} \TNorm{\matV}.$$
To see this, recall Lemma~\ref{lem:1set}. Take squares on the right hand side of the equation in that lemma and observe that:
$$ \sigma_i^2( \matV\transp \Omega \matS ) \le 1 + \sqrt{\frac{9k}{r}}.$$   
The result follows by using the first and the third statements of Lemma~\ref{lem:eq} with $\epsilon = \sqrt{\frac{9k}{r}}$.
To extend this to arbitrary $\matA \in \R^{n \times k}$ with $n > k$, is suffices to compute the SVD of 
$\matA$ and apply the above algorithm to the matrix containing the left singular vectors of $\matA$. More specifically,
let the SVD of $\matA$ is $\matA = \matU_{\matA} \Sigma_{\matA} \matV_{\matA}\transp$ with
$\matU_{\matA} \in \R^{n \times k}$, $\Sigma_{\matA} \in \R^{k \times k}$, and
$\matV_{\matA} \in \R^{k \times k}$. Now, consider the following derivations:
\begin{eqnarray*}
\TNorm{ \matA\transp \Omega \matS \matS\transp \Omega\transp \matA - \matA\transp \matA} &=&
 \TNorm{ \matV_{\matA} \Sigma_{\matA} \matU_{\matA}\transp \Omega \matS \matS\transp \Omega\transp \matU_{\matA} \Sigma_{\matA} \matV_{\matA}\transp 
 - \matV_{\matA} \Sigma_{\matA} \matU_{\matA}\transp \matU_{\matA} \Sigma_{\matA} \matV_{\matA}\transp} \\ &=&
\TNorm{ \matV_{\matA} \left( \Sigma_{\matA} \matU_{\matA}\transp \Omega \matS \matS\transp \Omega\transp \matU_{\matA} \Sigma_{\matA} 
 -  \Sigma_{\matA} \matU_{\matA}\transp \matU_{\matA} \Sigma_{\matA} \right) \matV_{\matA}\transp}  \\ &\le&
\TNorm{  \Sigma_{\matA} \matU_{\matA}\transp \Omega \matS \matS\transp \Omega\transp \matU_{\matA} \Sigma_{\matA} 
 -  \Sigma_{\matA} \matU_{\matA}\transp \matU_{\matA} \Sigma_{\matA}}  \\ &=&
\TNorm{  \Sigma_{\matA}  \left(\matU_{\matA}\transp \Omega \matS \matS\transp \Omega\transp \matU_{\matA} 
 -   \matU_{\matA}\transp \matU_{\matA}  \right) \Sigma_{\matA}}  \\ &\le&
 \TNorm{  \matA }  \TNorm{ \matU_{\matA}\transp \Omega \matS \matS\transp \Omega\transp \matU_{\matA}  
 -  \matU_{\matA}\transp \matU_{\matA}  } \TNorm{ \matA} 
\end{eqnarray*}
Clearly, it suffices to approximate the product $\matU_{\matA}\transp \matU_{\matA}$. Overall, in $O( nk^2 + nk^2 r )$: 
$$ \TNorm{ \matA\transp \Omega \matS \matS\transp \Omega\transp \matA - \matA\transp \matA } \le \sqrt{ \frac{9 k}{r} } \TNorm{\matA} \TNorm{\matA}.$$ 

\paragraph{Deterministic Asymmetric Multiplication.}
The above deterministic result applies to the so-called symmetric matrix multiplication problem, i.e. for multiplication involving the product $\matV\transp\matV$. A simple modification of this result suffices to extend this
to the asymmetric case. Let $\matV_1 \in \R^{n_1 \times k}$ and $\matV_2 \in \R^{n_2 \times k}$. 
Let also $\max\{n_1, n_2 \} > 2k$. For now let us assume that $\matV_1\transp \matV_1 = \matV_2\transp \matV_2 = \matI_k$
(the general case is discussed later). For simplicity, let us also assume that $n_1 = n_2 = n$
(otherwise, padding the matrix with the smaller dimension with zeros resolves this issue). Consider the
matrix $ \matA = [ \matV_1;  \matV_2 ] \in \R^{n \times 2k}$ and apply the symmetric multiplication result
(for general matrices) to this $\matA$ with $r > 2k$: 
$$ \TNorm{  
\left[
\begin{matrix}
\left( \matV_1\transp\Omega\matS\matS\transp\Omega\transp\matV_1 - \matV_1\transp\matV_1 \right) & 
\left( \matV_1\transp\Omega\matS\matS\transp\Omega\transp\matV_2 - \matV_1\transp\matV_2 \right)  \\
\left( \matV_2\transp\Omega\matS\matS\transp\Omega\transp\matV_1 - \matV_2\transp\matV_1 \right) & 
\left( \matV_2\transp\Omega\matS\matS\transp\Omega\transp\matV_2 - \matV_2\transp\matV_2 \right)
\end{matrix}
\right]
}  
\le \sqrt{ \frac{18 k}{r} } \TNorm{\matA} \TNorm{\matA} .$$
From the interlacing property of singular values:
$$ \TNorm{\matV_1\transp\Omega\matS\matS\transp\Omega\transp\matV_2 - \matV_1\transp\matV_2} \le
\TNorm{  
\left[
\begin{matrix}
\left( \matV_1\transp\Omega\matS\matS\transp\Omega\transp\matV_1 - \matV_1\transp\matV_1 \right) & 
\left( \matV_1\transp\Omega\matS\matS\transp\Omega\transp\matV_2 - \matV_1\transp\matV_2 \right)  \\
\left( \matV_2\transp\Omega\matS\matS\transp\Omega\transp\matV_1 - \matV_2\transp\matV_1 \right) & 
\left( \matV_2\transp\Omega\matS\matS\transp\Omega\transp\matV_2 - \matV_2\transp\matV_2 \right)
\end{matrix}
\right]
}.
$$
Also, $\TNorm{\matA} \le \sqrt{2}$. Overall, in $O( \max\{n_1, n_2\} k^2 r )$:
$$ \TNorm{\matV_1\transp\Omega\matS\matS\transp\Omega\transp\matV_2 - \matV_1\transp\matV_2} \le \sqrt{ \frac{72 k}{r} } \TNorm{\matV_1} \TNorm{\matV_2}.$$

To extend this to general matrices $\matA \in \R^{n_1 \times k}$, $\matB \in \R^{n_2 \times k}$,
not necessarily matrices of orthonormal columns, one needs to proceed as follows. First, 
compute the SVD of $\matA$, $\matB$ in $O(n_1 k ^2)$ and $O(n_2 k^2)$, respectively. 
Let the SVD of $\matA$ is $\matA = \matU_{\matA} \Sigma_{\matA} \matV_{\matA}\transp$ with
$\matU_{\matA} \in \R^{n_1 \times k}$, $\Sigma_{\matA} \in \R^{k \times k}$, and
$\matV_{\matA} \in \R^{k \times k}$. Let also the SVD of $\matB$ is 
$\matB = \matU_{\matB} \Sigma_{\matB} \matV_{\matB}\transp$ with
$\matU_{\matB} \in \R^{n_2 \times k}$, $\Sigma_{\matB} \in \R^{k \times k}$, and
$\matV_{\matB} \in \R^{k \times k}$. Observe that:
$$ \TNorm{ \matA\transp \Omega \matS \matS\transp \Omega\transp \matB - \matA\transp \matB} \le
\TNorm{  \matA }  \TNorm{ \matU_{\matA}\transp \Omega \matS \matS\transp \Omega\transp \matU_{\matB}  
 -  \matU_{\matA}\transp \matU_{\matB}  } \TNorm{ \matB} 
$$
%
Clearly, it suffices to approximate $\matU_{\matA}\transp \matU_{\matB}$.
Overall, in  $O(k^2(n_1 + n_2 + \max\{n_1, n_2\} r))$:
$$ \TNorm{ \matA\transp \Omega \matS \matS\transp \Omega\transp \matB - \matA\transp \matB } \le \sqrt{ \frac{72 k}{r} } \TNorm{\matA} \TNorm{\matB}.$$

\subsection*{Low-rank Matrix Approximation}         
Although in this thesis we focused on fast column-based low
rank matrix approximations, such fast low-rank approximations
can be obtained by other paths than sampling columns from the input matrix $\matA$.
Examples of this kind of approaches include \cite{Sar06,Har06,HMT} (Frobenius norm)
and~\cite{LWMRT07, MRT10, RST09, HMT} (spectral norm). The main idea in all these papers, except~\cite{Har06}, is
to compute an SVD of the matrix $\matA R$ or $R\transp \matA$, where $R$ is a random projection matrix as
we described in Section~\ref{chap319}. \cite{Har06} computes low-rank approximations based
on ideas from computational geometry. 

A fast low-rank approximation in the spectral norm with impressive approximation guarantees
appeared in~\cite{RST09}. The algorithm of~\cite{RST09}
was analyzed more carefully in~\cite{HMT}, and slightly more carefully in the current thesis;
so, Lemma~\ref{troppextension0} in the Appendix provides, to my best understanding of 
existing literature and results, the most tight analysis of the technique of~\cite{RST09}.

Existing algorithms for fast low-rank approximations in the Frobenius norm are accurate and fast. 
For example, the method of Sarlos~\cite{Sar06} computes a rank-$k$ matrix with $(1+\epsilon)$-error 
and failure probability $0 < \delta < 1$ in  
$O( mnk\epsilon^{-1} \log(\frac{1}{\delta}) + (m+n)k\epsilon^{-2}\log(\frac{1}{\delta}) )$. 
The method in~\cite{Har06} does so in $O(m n k(\epsilon^{-1} + k) \log(k/( \epsilon \delta)))$. 
Below, we contribute a new analysis of an algorithm that employs the Hadamard Transform that we 
discussed in Section~\ref{chap3110}. Our analysis delivers a $(1+\epsilon)$-error with constant
probability in time
$$ O\left( m \cdot n \cdot k \cdot \ln(k) \cdot \log(k n) \cdot \epsilon^{-1} 
+ (m+n) \cdot (k^2 \cdot \ln^2(k) \cdot \log^2(k n) \cdot \epsilon^{-2}) \right).$$
A preliminary analysis of this algorithm appeared in Theorem 11.2 of~\cite{HMT} and Theorem 1 of~\cite{NDT09}. 
\cite{HMT} gives an approximation error that is not tight; \cite{NDT09} claims a running time that, 
as far as I can understand, is not correct. More specifically, an extra $O(m n d)$ term (in the notation of~\cite{NDT09}, $d$ plays the role of $r$ in our notation) in the running time seems to be necessary. 
This is because the computation of the best rank-$k$ matrix $\Pi_{\matC, k}^{F}(\matA)$ takes at least $O(mnd)$. 
\cite{NDT09} quotes~\cite{DRVW06xx}, where this supposed to be done in $O(md^2)$,
but this connection is not clear to the author. 



\paragraph{Relative-error rank $k$ Approximation with the Hadamard Transform.}
\begin{theorem}
Fix $\matA \in \R^{m \times n}$ of rank $\rho$, target rank $k < \rho$,
and an accuracy parameter $0 < \epsilon < 1/2$. Construct a rank $k$ matrix   $\Xi \in \R^{m \times n}$ 
as follows:
\begin{algorithmic}[1]
\STATE Let $r = 200 \cdot k \cdot \ln(40 k) \cdot \log(40 k n) / \epsilon$.
\STATE Using definition~\ref{srht} construct a SRHT matrix $\Theta \in \R^{r \times n}$.
\STATE Construct the matrix $\matC = \matA \Theta\transp$.
\STATE Using the algorithm of Section~\ref{chap22} construct $\Pi_{\matC,k}^{F}(\matA)$.
\STATE Return $\Xi = \Pi_{\matC,k}^{F}(\matA) \in \R^{m \times n}$ of rank at most $k$. 
\end{algorithmic}
Then, with probability at least $0.7$:
$$ \FNormS{ \matA - \Xi } \le \left( 1 + \epsilon \right) \FNormS{\matA - \matA_k}.$$
The algorithm runs in $O\left( mn k \ln(k) \log(k n)\epsilon^{-1} + (m+n)(k^2 \ln^2(k) \log^2(k n)\epsilon^{-2}) \right)$.
\end{theorem}
\begin{proof}
We first comment on running time. Step 3 takes $O( m n \log(r) )$ (Lemma~\ref{fast}).
Step 4 takes $O( mnr + (m+n)r^2 )$ (from Lemma~\ref{lem:bestF}). Our choice of $r$ gives
the overal running time. We continue by manipulating the term $\FNormS{ \matA - \Xi }$.
We would like to apply Lemma~\ref{lem:generic} with $\matW = \Theta\transp$ and $\xi=F$.
First, notice that Lemma~\ref{preserves} gives:
$$1 -  \sqrt{\frac{8 k \ln(2 k / \delta ) \log(40 k n)}{ r }}  
\leq  \sigma_i^2( \Theta \matV_k )
\leq 1+ \sqrt{\frac{8  k \ln(2 k/\delta) \log(40 k n)}{ r }}.
$$
Now, our choice of $r$ implies ($\delta=0.05$) that w.p. at least $0.9$:
 $$1 -  \frac{\sqrt{\epsilon}}{\sqrt{25}}
\leq  \sigma_i^2( \Theta \matV_k )
\leq 1+ \frac{\sqrt{\epsilon}}{\sqrt{25}}.
$$
The assumption on $\epsilon < 1/2$ and the left hand side of this inequality imply that $\rank(\matV_k\transp \Theta\transp) = k$;
so, we can apply Lemma~\ref{lem:generic} (recall that $\Xi = \Pi_{\matC,k}^{F}(\matA)$):
$$
\FNormS{\matA - \Xi } \le \FNormS{\matA-\matA_k} + \FNormS{(\matA-\matA_k) \Theta\transp (\matV_k\transp \Theta\transp)^+}.
$$
We will return to this generic equation later. First, we
prove three results of independent interest.

\paragraph{First result of independent interest.}
Recall that by Lemma \ref{preserves} and our choice of $r$, for all $i=1,...,k$ and w.p. $0.9$:
$1 - \frac{\sqrt{\epsilon}}{\sqrt{25}}  \leq  \sigma_i^2(\matV_k\transp \Theta\transp)   \leq 1 + \frac{\sqrt{\epsilon}}{\sqrt{25}}.$
Let $\matX = \matV_k\transp \Theta\transp \in \R^{k \times r}$ with SVD:
$\matX = \matU_{\matX} \Sigma_{\matX} \matV_\matX\transp$.
Here,  $\matU_{\matX} \in \R^{k \times k}$, 
$\Sigma_{\matX} \in \R^{k \times k}$, and $\matV_\matX \in \R^{r \times k}$.
By taking the SVD of $\matX^+$ and $\matX\transp$:
$$
\TNorm{ (\matV_k\transp \Theta\transp)^+ - (\matV_k\transp \Theta\transp)\transp } = 
\TNorm{ \matV_{\matX} \Sigma_{\matX}^{-1} \matU_{\matX}\transp - \matV_{\matX}\Sigma_{\matX} \matU_{\matX}\transp } = 
= \TNorm{ \Sigma_{\matX}^{-1} - \Sigma_{\matX} },$$ 
since $\matV_{\matX}$ and
$\matU_{\matX}\transp $ can be dropped without changing any unitarily
invariant norm. 
Let $\matY = \Sigma_{\matX}^{-1} - \Sigma_{\matX} \in \R^{k \times k}$ be diagonal;
Assuming that, for all
$i=1,...,k$, $\tau_i(\matY)$ denotes the $i$-th diagonal element
of $\matY$:
$ \tau_i(\matY)  = \frac{ 1 - \sigma_i^2(\matX)  }{ \sigma_{i}(\matX) }.$
Since $\matY$ is a diagonal matrix:
\[\TNorm{ \matY }\ =\ \max_{1 \leq i \leq k} \abs{\tau_i(\matY)} \ =\ 
\max_{1 \leq i \leq k} \frac{ \abs{ 1 - \sigma_i^2(\matX)}}{ \sigma_{i}(\matX) } \le \frac{\frac{\sqrt{\epsilon}}{\sqrt{25}}}{\sqrt{1-\frac{\sqrt{\epsilon}}{\sqrt{25}}}} .\]
The inequality follows by using the bounds for $\sigma_{i}^2(\matX)$ from above. 
The failure probability is $0.1$ because the bounds for $\sigma_{i}^2(\matX)$
fail with this probability.  Overall, we proved that
$$
\TNorm{ (\matV_k\transp \Theta\transp)^+ - (\matV_k\transp \Theta\transp)\transp } 
\le \frac{\frac{\sqrt{\epsilon}}{\sqrt{25}}}{\sqrt{1-\frac{\sqrt{\epsilon}}{\sqrt{25}}}}
$$

\paragraph{Second result of independent interest.}
Consider the term: $\FNormS{(\matA-\matA_k) \Theta\transp \Theta \matV_k }$. We would like to upper
bound this term. Recall that $\Theta\transp = \matH \matD \Omega \matS \in \R^{n \times r}$. 
Eqn.~(4) of Lemma 4 of \cite{DKM06a} gives a result for the above matrix-multiplication-type term
and any set of probabilities $p_1, p_2,...,p_n$ (set $\matE = (\matA-\matA_k) \matH \matD$, $\matZ = \matD\transp \matH\transp \matV_k$) :
$$ \Expect{ \FNormS{ (\matA-\matA_k) \matH \matD \matD\transp \matH\transp \matV_k - 
(\matA-\matA_k) \matH \matD \Omega \matS \matS\transp \Omega\transp \matD\transp \matH\transp \matV_k } }$$ 
$$ \le
\sum_{i=1}^{n} \frac{ \TNormS{\matE^{(i)}} \TNormS{ \matZ_{(i)} } }{r p_i} - \frac{1}{r} \FNormS{\matE \matZ}. $$
First, notice that $\matE \matZ = \bm{0}_{m \times k}$. Our choice of $p_i$'s is:
$$ p_i = \frac{1}{n}  
     \geq \frac{1}{2 \log(40 k n)}  \frac{ \TNormS{ \left( \matD \matH\matV_k \right)_{(i)} } }{ k }.
$$
By using this inequality and rearranging:
$$ \Expect{ \FNormS{ (\matA-\matA_k) \Theta\transp \Theta \matV_k } } \le
\frac{2 k \log(40 k n)}{r} \FNormS{(\matA-\matA_k) \matH \matD} = \frac{2 k \log(40 k n)}{r} \FNormS{\matA-\matA_k},$$
since $\matH \matD$ can be dropped without changing the Frobenius norm. 
Finally, apply Markov's inequality to the random variable 
$x = \FNormS{(\matA-\matA_k) \Theta\transp \Theta \matV_k }$ to get that with probability $0.9$
$$ \FNormS{ (\matA-\matA_k) \Theta\transp \Theta \matV_k }  \le \frac{20 k \log(40 k n)}{r} \FNormS{\matA-\matA_k},$$

\paragraph{Third result of independent interest.} We would like to compute an upper bound for the
term $\FNormS{(\matA-\matA_k) \Theta\transp}$. Replace $\Theta\transp = \matH \matD \Omega \matS \in \R^{n \times r}$.
Then, Lemma~\ref{lem:fnorm} on the random variable $x = \FNormS{(\matA-\matA_k) \Theta\transp}$ implies that with
probability $0.9$:
$$ \FNormS{(\matA-\matA_k) \Theta\transp} \le 10 \FNormS{(\matA-\matA_k) \matH \matD}. $$
Notice that $\matH \matD$ can be dropped without changing the Frobenius norm; so, w.p. $0.9$:
$$ \FNormS{(\matA-\matA_k) \Theta\transp} \le 10 \FNormS{\matA-\matA_k}. $$

\paragraph{Back to the generic equation.} Equipped with the above bounds, 
we are ready to conclude the proof of the theorem. We continue by manipulating
our generic equation as follows:
$
\FNormS{\matA - \Xi } \le
$
\vspace{-.2in}
\begin{eqnarray*}
&\leq&  \FNormS{\matA-\matA_k} + \FNormS{(\matA-\matA_k) \Theta\transp (\matV_k\transp \Theta\transp)^+}  \\
&\leq&  \FNormS{\matA-\matA_k} + 2\FNormS{(\matA-\matA_k) \Theta\transp \Theta \matV_k} + 
2\FNormS{(\matA-\matA_k) \Theta\transp  ( (\matV_k\transp \Theta\transp)^+ - (\matV_k\transp \Theta\transp)\transp )} \\
&\leq& \FNormS{\matA-\matA_k} + 2\FNormS{(\matA-\matA_k) \Theta\transp \Theta \matV_k} + 
2\FNormS{(\matA-\matA_k) \Theta\transp} \TNormS{( (\matV_k\transp \Theta\transp)^+ - (\matV_k\transp \Theta\transp)\transp )} \\
&\leq& 
\FNormS{\matA-\matA_k} + 2 \frac{20 k \log(40 k n)}{r} \FNormS{\matA-\matA_k} + 
2\cdot 10 \FNormS{\matA-\matA_k} \frac{\frac{\epsilon}{25}}{1-\frac{\sqrt{\epsilon}}{\sqrt{25}}} \\
&\leq&  \FNormS{\matA-\matA_k} +   \left(   \frac{40}{376\ln(40)} + \frac{20/25}{1 - 1/\sqrt{2 \cdot 25}} \right) \epsilon  \FNormS{\matA-\matA_k} 
 \le (1 + 0.986 \cdot \epsilon) \FNormS{\matA-\matA_k} 
\end{eqnarray*}
The failure probability follows by a union bound on all the probabilistic events
involved in the proof of the theorem.
\end{proof}

\subsection*{Column/row based Matrix Approximation} 
This thesis focused on low-rank approximation of matrices expressed as 
$$\matA \approx \matC \matC^+\matA,$$
with $\matC$ containing columns of $\matA$. Factorizations of the form
$$\matA \approx \matC_{1} (\matC_1^+ \matA \matC_2^+) \matC_2\transp$$
are also of considerable interest. Here, $\matC_1$ contains columns of $\matA$ and $\matC_2$ contains columns
of $\matA\transp$. We refer the reader to~\cite{DMM08} for applications of such column-row approximations. 
In terms of algorithms, both relative error~\cite{DMM08} and additive error~\cite{DK03,DKM06c} bounds are available
(in the Frobenius norm).
The relative error algorithm of~\cite{DMM08} employs the randomized technique that we described in Section~\ref{chap314}.
The additive error algorithm of~\cite{DK03,DKM06c} employs the randomized technique that we described in Section~\ref{chap311}. Notice that all~\cite{DK03,DKM06c,DMM08} provide randomized algorithms with Frobenius norm bounds and the 
setting is that one is given $\matA, k$ and an oversampling parameter $r > k$; the algorithm then returns $\matC_1$, $\matC_2$ containing $r$ columns and rows, respectively. 

Deterministic column-row decompositions with spectral norm bounds by selecting exactly $r = k$ columns/rows are
described in~\cite{HP97, Pan03, MG03, GM04} and~\cite{Tyr96, Tyr00, GTZ97a, GTZ97b}. The algorithms in
\cite{Tyr96, Tyr00, GTZ97a, GTZ97b} construct $\matC_1$ by running the method that we described in Section~\ref{chap315}
on $\matV_k\transp$ and  $\matC_2$ by running the same deterministic method on $\matU_k\transp$. It is quite interesting
that in these papers appeared a preliminary version of Lemma~\ref{lem:generic} that we described in Section~\ref{chap22}.
The techniques of~\cite{HP97, Pan03, MG03, GM04} discuss the so-called Rank-Revealing LU factorization. It can be proved
(well, with a little effort) that a RRLU factorization implies a factorization of the form 
$\matA \approx \matC_{1} (\matC_1^+ \matA \matC_2^+) \matC_2\transp$ with provable approximation bounds. 
(Recall that in Section~\ref{chap32} we saw that a Rank-Revealing QR factorization implies a factorization of the 
form $\matA \approx \matC \matC^+\matA$ with provable approximation bounds.)

\subsection*{Subspace Approximation}                
The column-based low-rank matrix approximation problem
(in the Frobenius norm) studied in this thesis is as follows.
Fix $\matA \in \R^{m \times n}$, $k$, and $r > k$. The goal is to find a
set of $r$ columns from $\matA$ that contain a $k$-dimensional
subspace which is as good as the $k$-dimensional subspace of the
Singular Value Decomposition. We focused on algorithms with approximations 
of the form:
$$ \FNorm{\matA-\Pi_{\matC,k}^{F}(\matA)} \le \alpha \FNorm{\matA-\matA_k}.$$
Let $\tilde{\matU}_k \in \R^{m \times k}$ be the best $k$-subspace in the column space of $\matC$,
i.e. $ \Pi_{\matC,k}^{F}(\matA) = \tilde{\matU}_k \tilde{\matU}_k\transp\matA $.
Also, let $\matA$ has column representation $\matA = [ \a_1,\a_2...,\a_n]$ and
recall that $\matA_k = \matU_k \matU_k\transp \matA$ with $\matU_k \in \R^{m \times k}$.
The above equation is equivalent to:
$$ \left( \sum_{i=1}^{n} \left( \TNorm{ \a_i - \tilde{\matU}_k \tilde{\matU}_k\transp \a_i } \right)^2 \right)^{\frac{1}{2}}
\le
\alpha  \left( \sum_{i=1}^{n} \left( \TNorm{ \a_i - \matU_k \matU_k\transp \a_i } \right)^2 \right)^{\frac{1}{2}}. 
$$
Define the function 
$$ d(\x, \matU, p) = \left( \PNorm{ \x - \matU \matU \transp \x } \right)^2,$$
to be the $p$-norm distance of the vector $\x \in \R^m$ from the subspace $\matU \in \R^{m \times k}$. 
The $p$-norm of a vector $\x = [x_1,..., x_m]$ is defined as $\PNorm{\x} = \left( \sum_{i=1}^m |x_i|^p\right)^\frac{1}{p}$.
So, the equation corresponding to our problem can be revised as:
$$ \left( \sum_{i=1}^{n} d(\a_i, \tilde{\matU}_k, 2) \right)^\frac{1}{2} \le
\alpha \left( \sum_{i=1}^{n} d(\a_i, \matU_k, 2) \right)^\frac{1}{2}.
$$
In words, we seek a subset of $r$ points (columns from $\matA$) that contain
a $k$-subspace $\tilde{\matU}_k$ that is as good as the optimal $k$-subspace from the SVD. Replacing $p = 2$
with $p=1, 3, 4, ...$ corresponds to the more general subspace approximation problem
that received considerable attention as well. We refer the reader to~\cite{SV07,DV07,DTV09,FL11}
and references therein for background and motivation for this generalized problem. 
In the general case ($p \ne 2$), SVD does not provide an analytical expression for the 
best $k$-subspace but one still
seeks algorithms with approximations of the form:
$$ \left( \sum_{i=1}^{n} d(\a_i, \tilde{\matH}_{(p,k)}, p) \right)^\frac{1}{p} \le
\alpha \left( \sum_{i=1}^{n} d(\a_i, \matH_{(p,k)}, p) \right)^\frac{1}{p}.
$$
Here $\matH_{(p,k)} \in \R^{m \times k}$ is the best $k$-subspace with respect to the $p$ norm:
$$ \matH_{(p,k)} = \arg \min_{ \matH \in \R^{m \times k} } \sum_{i=1}^{n} d(\a_i, \matH, p); $$
and $\tilde{\matH}_{(p,k)}$ is the best $k$-subspace within the column space of $\matA$, i.e 
$\tilde{\matH}_{(p,k)} = \matC \matX \in \R^{m \times k}$, with $\matX \in \R^{r \times k}$ as:
$$ \matX = \arg \min_{ \matX \in \R^{r \times k}, \matH = \matC \matX } \sum_{i=1}^{n} d(\a_i, \matH, p). $$
It is worth mentioning that the results of this line of research for $p=2$ are not
better than the results we presented in this thesis. For example, \cite{SV07} describes
a randomized algorithm that finds a subset of $r = O(k^2 \epsilon^{-1} \log(k \epsilon^{-1}))$
points with corresponding approximation $\alpha = 1 + \epsilon$. Unhappily, this algorithm runs in
time exponential in $k \epsilon^{-1}$. \cite{DV07} presents several interesting results that can
be viewed as extensions of the methods that we presented in Sections~\ref{chap311} - \ref{chap313}
for the general subspace approximation problem. For example, Theorem 5 from~\cite{DV07} is
the analog of the additive-error algorithm that we discussed in Section~\ref{chap311}.
\begin{theorem}
Fix $\matA$, $k$, $p$, $0 < \epsilon, \delta < 1$. There is a randomized algorithm that samples
$r = O\left( k (\frac{2k}{\delta})^p \frac{k}{\delta} \log(\frac{k}{\delta}) \right)$
columns from $\matA$ such that w.p.
$1 - \frac{1}{k}$ there is a $k$-dimensional subspace 
$\matH_S \in \R^{m \times k}$ within the column space of the sampled columns:
$$ \left( \sum_{i=1}^{n} d(\a_i, \matH_S, p) \right)^{\frac{1}{p}} 
\le \left( \sum_{i=1}^{n} d(\a_i, \matH_{(p,k)}, p) \right)^{\frac{1}{p}} + 
\epsilon \left( \sum_{i=1}^{n} \PNorm{\a_i} \right)^{\frac{1}{p}}
$$
\end{theorem}
It would be interesting to understand whether the other techniques that we used in this 
thesis for subspace approximation in the $p=2$ norm (Frobenius norm) can be extended to 
the more general case of $p \ne 2$.

\subsection*{Element-wise Matrix Sparsification} 
Unlike most of the problems that we saw so far, where one
is asked to select a subset of columns (or rows) from
the input matrix, Element-wise Matrix Sparsification
studies matrix approximations by
sampling individual elements from the  matrix.
Let $\matA \in \R^{m \times n}$ be the input matrix,
$k$ be the target rank for the approximation, and $r \gg k$
be the number of sampled elements from $\matA$; we would like
to construct $\tilde{\matA} \in \R^{m \times n}$ with  $r$
non-zero entries such that $\XNorm{\matA - \tilde{\matA}_k }$ is
as close to $\XNorm{\matA - \matA_k}$ as possible. 
$\tilde{\matA}_k \in \R^{m \times n}$ is the best rank $k$ approximation
of $\tilde{\matA}$ computed with the SVD. Notice that $\tilde{\matA}$ is
a sparse matrix, so an SVD on this matrix will be faster than an SVD on $\matA$, which is dense.
We should note thought that the later claim can not be proved theoretically; in practice
though it is well known that algorithms for computing the SVD, such as the Lanczos iteration
or the power iteration, operate much faster on sparse matrices. So, matrix sparsification 
does offer yet another way of fast rank $k$ 
approximations to matrices. 
Matrix Sparsification pioneered by Achlioptas and McSherry in~\cite{AM01}.
We quote Theorem 3 from~\cite{AM01}, which gives an idea of the
approximations that can be achieved using this approach.
\begin{theorem}
Fix $\matA \in \R^{m \times n}$ with $76 \le m \le n$. Let $\beta = \max_{i,j}|\matA_{ij}|$.
For any $p > 0$, define $\tau_{ij} = p \frac{(\matA_{ij})^2}{\beta}$ and 
$p_{ij} = \max\{ \tau_{ij}, \sqrt{ \tau_{ij} \cdot 8^4 \log^4(n) /n } \}$.  
Let $\tilde{\matA}$ be a random $m \times n$ matrix
whose entries are i.i.d. as $\tilde{\matA}_{ij} = \matA_{ij}/p_{ij}$ w.p. $p_{ij}$, and 
$\tilde{\matA}_{ij} = 0$ w.p. $1 - p_{ij}$. Then, w.p. $1 - e^{- 19\log^4(n)}$:
$$ \TNorm{\matA - \tilde{\matA}_k} \le \TNorm{\matA - \matA_k} + 2 \frac{4 \beta \sqrt{n}}{\sqrt{p}};$$
$$ \FNorm{\matA - \tilde{\matA}_k} \le \FNorm{\matA - \matA_k} + \frac{4 b \sqrt{k n}}{\sqrt{p}} + 
2 \sqrt{ \frac{4 b \sqrt{k n}}{\sqrt{p}} \FNorm{\matA_k} }.$$
Moreover, if $\tilde{r}$ is the random variable counting the number
of non-zero elements in $\tilde{\matA}$, then 
$$\Expect{\tilde{r}} \le p \cdot \FNormS{\matA}\beta^{-2}  + 4096 \log^4(n) m.$$
\end{theorem}
We refer the reader to~\cite{DZ11} for an updated discussion on existing
literature for this topic. To our best knowledge, all existing algorithms 
for matrix sparsification are randomized. Recently, \cite{Zou11} presented the first
deterministic algorithm.

\subsection*{Graph Sparsification}                  
Given a dense graph $\cl G$, there are applications where it is required to
approximate the graph with another sparse graph $\cl H$, i.e. a graph with
a much smaller number of edges than the original graph. We refer the
reader to \cite{ST08b, SS08, BSS09, Zou11, KL11} and references therein
for a discussion of such applications. There are many notions of graph
sparsification. What it really means that a graph approximates another graph?
Motivated by applications on solving linear equations with Laplacian matrices,
\cite{ST08b} introduced the so-called spectral notion of graph sparsification.
Assume that $\matA_{\cl G}$ and $\matA_{\cl H}$ are the Laplacian matrices of
the graphs $\cl G$ and  $\cl H$, respectively. Then, the graph  $\cl H$ approximates
the graph  $\cl G$ ``spectrally'' if the eigenvalues of  $\matA_{\cl H}$ are within
relative error accuracy from the eigenvalues of  $\matA_{\cl G}$ ($0 <\epsilon<1$): 
$$ (1 - \epsilon) \lambda_i( \matA_{\cl G} ) \le \lambda_i( \matA_{\cl H} ) \le  (1 + \epsilon) \lambda_i( \matA_{\cl G} ).$$
\cite{SS08} showed that in order to construct such spectral graph sparsifiers, it suffices to preserve,
after column sub-sampling, the singular values of a special short-fat matrix with orthonormal rows. The number
of columns of this matrix is the same with the number of edges in the graph, so edge selection
corresponds to column selection to this special matrix. After this result in~\cite{SS08}, it is immediate
that one can use column sampling algorithms for graph sparsification. For example, \cite{SS08}
uses the randomized technique that we described in Section~\ref{chap314} to construct sparsifiers
with $r = O( m \log(m) / \epsilon^2 )$ edges; $m$ is the number of vertices in the graph. 
\cite{BSS09} uses the method of Section~\ref{chap316} to construct sparsifiers with $r = O(m / \epsilon^2)$
edges deterministically. \cite{Zou11} combines and improves upon the ideas  of~\cite{SS08,BSS09} to
construct sparsifiers with $r = O(m / \epsilon^2)$ edges deterministically but faster than~\cite{BSS09}.
Finally,~\cite{KL11} discusses the construction of such spectral sparsifiers in the streaming model
of computation.

\subsection*{Linear Equation Solving}               
In a series of papers~\cite{ST08c, KMST10, KM07, KMP10, KMP11}, there
were developed fast approximation algorithms for solving systems of Linear
Equations with Laplacian matrices. Let a graph has $m$ vertices and $n$ edges;
this corresponds to a Laplacian matrix $\matA \in \R^{m \times m}$ with $O(n)$ 
non-zero entries. Solving a system with this matrix takes $O(m^3)$. 
In a breakthrough paper~\cite{ST08c}, Spielman and Teng showed how to do that
approximately in $o(m^3)$ time. Subsequent research improved upon the work
of ~\cite{ST08c}. Currently, the best such method is~\cite{KMP11} which
solves this system approximately in 
$$O\left( n \cdot \log(m) \cdot \log(\frac{1}{\epsilon}) \cdot \log^2(\log (n)) \right).$$
The basic idea of all these methods is to sparsify the graph (i.e. sparsify
the laplacian matrix) and then use a standard method, such as the Conjugate Gradient
method, on the sparsified Laplacian. Clearly, the sparsity of the new Laplacian matrix
improves the computational efficiency of standard methods. The real meat in
this approach thought is that spectral sparsification of graphs implies relative
error approximation to the solution vector of the linear system. We refer
the reader to~\cite{ST08c} for the corresponding details. 

\subsection*{Coresets for $k$-means Clustering}   
One of the three problems studied in this thesis is
feature selection for $k$-means clustering. The
idea is to select a subset of the features and
by using only this subset obtain a partition of
the points that is as good as the partition that
would have been obtained by using all the features.
A complementary line of research~\cite{PK05, HM04, FS06, FMC07, ADK09}
approaches the $k$-means problem by sub-sampling the points
of the dataset. The idea here is to select a small subset of
the points and by using only this subset obtain a partition
for all the points that is as good as the partition that would
have been obtained by using all the points. 
\cite{PK05, HM04, FS06, FMC07, ADK09} offer algorithms for $(1+\epsilon)$
approximate partitions. Note that we were able to give only constant
factor approximations. For example, \cite{FS06} shows the existence
of an $(1+\epsilon)$-approximate coreset of size $r = O(k^3 / \epsilon^{n+1})$ 
($n$ is the number of features). \cite{FMC07} provides a coreset of size $r = poly(k, \epsilon^{-1})$.
The techniques used for all these coresets are different from the techniques we used
for feature selection. It would be interesting to understand whether the techniques from
\cite{PK05, HM04, FS06, FMC07, ADK09} are useful for feature selection as well. In particular,
it appears that there is potential to obtain relative error feature selection $k$-means algorithms
by using such approaches. 


\subsection*{Trace Approximation} 
Let $\matA \in \R^{n \times n}$ is a square matrix. The trace of $\matA$, 
denoted with $\trace(\matA) \in \R$, equals the sum of its diagonal elements. 
So, computing the trace of an explicit matrix is a simple operation which takes $O(n)$ time. 
There are applications though (see, for example, \cite{AT11}) where one needs to compute 
the trace of an implicit matrix, i.e. a matrix $f(\matA) \in \R^{n \times n}$; 
$f$ is some function on $\matA$, for example, $f(\matA) = \matA^3 + 2 \cdot \matA + \matI_n$. 
Computing the trace of $f(\matA)$, given $\matA$, is a rather expensive task. In such cases,
quick approximations to the exact value \trace(f(\matA)) are of interest. 

Avron and Toledo~\cite{AT11} wrote a very influential paper on
approximation algorithms for estimating the trace of implicit matrices. Such algorithms were known before
to perform well in several real applications. On the negative side, there was no theoretical analysis
for the performance of these algorithms. On an effort to close this theory-practice gap, 
\cite{AT11} provided a theoretical analysis of several existing trace approximation algorithms. 
Below, we present the basic idea of these algorithms and the type of approximations that
can be achieved. For simplicity, following the discussion in~\cite{AT11}, 
we assume that one is interested in estimating the trace of $\matA$. 
Replacing $\matA$ with some function of $\matA$ doesn't require any different analysis.  

The main idea of existing trace approximation algorithms is to return an approximation 
$\hat{\trace}(\matA)$ that is computed as follows:
$$ \hat{\trace}(\matA) = \frac{1}{p} \sum_{i=1}^{p} \z_i\transp \matA \z_i.$$
Here, $p >0$ is an integer; clearly, the largest we choose $p$, the better the approximation is.
Of course, choosing a large $p$ affects the running time of the method. 
The vectors $\z_i \in \R^n$ are random vectors chosen from a probability distribution. Different probability
distributions give different trace approximation algorithms. For example, Hutchinson's
method~\cite{Hut99} suggests that a specific $\z_i$ has entries where each one is chosen i.i.d as a Rademacher random
variable (each entry equals $\pm 1$ with the same probability).  
Hutchinson~\cite{Hut99} proved that his estimator is unbiased: 
$\Expect{ \z\transp \matA \z } = \trace(\matA)$;
$\z \in \R^n$ is a random variable chosen as described above. 
\cite{AT11} proved
the following bound for the approximation $\hat{\trace}(\matA)$. Let $\matA$ is a PSD symmetric matrix; then,
for any  $0< \delta, \epsilon < 1$, it suffices to choose  
$p \ge \frac{6}{\epsilon^2} \ln(2 n / \delta)$ random vectors $\z_i$ such
that w.p. $1 - \delta$:
$$ \abs{ \hat{\trace}(\matA) - \trace(\matA)} \le \epsilon \cdot \trace(\matA). $$
An immediate application of fast trace estimators is on counting triangles in a graph. 
Tsourakakis~\cite{Tso08} observed that in an undirected graph $\cl G$ with adjacency
matrix representation $\matA$, the number of triangles $T_3$ in the graph equals
$$T_3 = \frac{1}{6}\trace(\matA^3).$$ Avron in~\cite{Avr10} used his estimators from~\cite{AT11}
to quickly approximate the number of triangles in large real graphs.

\subsection*{Fiedler Vector Approximation}            

The discussion here assumes some familiarity with graphs, Laplacian matrices, and
spectral clustering. We refer the reader to~\cite{Spi10} for the necessary background.  
Fiedler~\cite{Fie73} made an outstanding contribution to the topic of clustering data. 
Given $m$ points $\x_1,...,\x_m$ in some Euclidean space of dimension $n$, \cite{Fie73} suggests
that one can obtain a $2$-clustering of the points by using the signs of the eigenvector
corresponding to the second smallest eigenvalue of the Laplacian Matrix of the Graph
that corresponds to these points (the idea of spectral clustering~\cite{SM00} basically
extends this result to any $k > 2$ partition by working with multiple eigenvectors). 
To recognize the outstanding contribution of~\cite{Fie73}, 
this eigenvector is known as Fiedler vector. The graph mentioned above has $m$ vertices, 
and for any two points $\x_i,\x_j$ the corresponding edge denotes the distance between the
points; for example, the weight
in that edge is $e^{\TNormS{\x_i - \x_j}}$. Computing this eigenvector takes $O(m^3)$ 
through the SVD of the Laplacian matrix $\matA \in \R^{m \times m}$. Since data are
getting larger and larger, faster (approximate) algorithms for the Fiedler vector are
of particular interest. 

A result of Mihail~\cite{Mih89} indicates that any vector that approximates the Rayleigh quotient
of the second smallest eigenvalue of the Laplacian matrix can be used to partition the points. 
Influenced by the result of Mihail, Spielman and Teng in~\cite{ST08c} give 
the following definition for an approximate Fiedler vector.
\begin{definition}[Approximate Fiedler Vector] For a Laplacian matrix $\matA$ and $0 < \epsilon < 1$, 
$\v \in \R^m$ is an $\epsilon$-approximate
Fiedler vector if $\v$ is orthogonal to the all-ones vector and
$$ \frac{\v \matA \v}{\v\transp \v} \le (1 + \epsilon) \frac{f \matA f}{f\transp f} = (1+\epsilon) \lambda_{n-1}(\matA).$$
Here $f \in \R^m$ denotes the Fiedler vector. 
\end{definition}
Theorem 6.2 in~\cite{ST08c} gives a randomized algorithm to compute such an $\epsilon$-approximate Fiedler vector w.p. $1 - \delta$, in time
$$ O\left( \nnz(\matA) \log^{27}(\nnz(\matA)) \log(\frac{1}{\delta}) \log(\frac{1}{\epsilon}) \epsilon^{-1}  \right).$$
This method is based on the fast solvers for linear systems with Laplacian matrices developed in the same paper.
Another method to compute an $\epsilon$-approximate Fiedler vector described recently by Trevisan in~\cite{Tre11}.
This method uses the power iteration and computes such an approximate vector with constant probability in time
$$ O\left( \nnz(\matA) \log(\nnz(\matA)/\epsilon)   \epsilon^{-1}  \right).$$

\subsection*{Data Mining Applications}  
Most of the matrix sampling algorithms described in this thesis
are motivated by applications involving the analysis of large
datasets. Mahoney~\cite{Mah10} gives a nice overview of how
such algorithms are useful in the analysis of social network data
and data arising fron human genetics applications. Column sampling
algorithms for analyzing genetics data are also described in~\cite{PZBCR07}.
A few more interesting case studies can be found, for example, in~\cite{MMD06, SXZF07, BSA08, Sgoura}.
For a comprehensive treatment of several other data applications we refer the reader to~\cite{MMDS06, MMDS08, MMDS10},
which describe the topics of three successful Meetings 
(Workshops on Algorithms for Modern Massive Data Sets) focusing 
exactly on applications of matrix sampling algorithms to data mining problems.         

\chapter{LOW-RANK COLUMN-BASED MATRIX APPROXIMATION}
\label{chap4}
\footnotetext[5]{Portions of this chapter previously appeared as:
C. Boutsidis, P. Drineas, and M. Magdon-Ismail, 
Near-Optimal Column-Based Matrix Reconstruction,
arXiv:1103.0995, 2011.}
We present our results on low-rank column-based matrix approximation. 
The objects of this problem are
the $m \times n$ matrix $\matA$ of rank $\rho$, the target rank $k < \rho$, and the number of sampled columns $k \le r \le n$. 
$\matC \in \R^{m \times r}$ 
contains $r$ columns from $\matA$ and 
$\Pi^{\xi}_{\matC,k}(\matA) \in \R^{m \times n}$
is the best rank $k$ approximation to $\matA$ (under the $\xi$-norm) within the column space of $\matC$.
Recall that, for $r > k$: $$\XNorm{\matA-\matC\matC^+\matA} \le \XNorm{\matA-\Pi^{\xi}_{\matC,k}(\matA)};$$
and when $r = k$: $\XNorm{\matA-\matC\matC^+\matA} = \XNorm{\matA-\Pi^{\xi}_{\matC,k}(\matA)}$. The goal
is to design algorithms that construct $\matC$ with ``small'' $\alpha$ and guarantee: 
$$\XNorm{\matA-\matC\matC^+\matA} \le \XNorm{\matA-\Pi^{\xi}_{\matC,k}(\matA)} \le \alpha \XNorm{\matA-\matA_k}.$$

The structure of the present chapter is as follows. In 
Sections \ref{chap41} and \ref{chap42}, we assume $r > k$
and present results for spectral norm and Frobenius norm,
respectively. Section \ref{chap43} presents our algorithms 
for the Column Subset Selection Problem ($r = k$) and a novel 
algorithm for an Interpolative Decomposition of a matrix. 
We give the proofs of all the results presented in this chapter 
in Section \ref{chap44}. All the proofs in this section are obtained by combining
Lemmas \ref{lem:genericNoSVD}, \ref{lem:generic}, \ref{tropp1}, 
and \ref{tropp2} of Section \ref{chap22} along with some results
of Section \ref{chap31}.

\section{Spectral Norm Approximation ($r > k, \xi = 2$)}\label{chap41}
By combining the deterministic exact SVD of Lemma \ref{lem:generic} in Section \ref{chap22}
along with the deterministic  technique of section \ref{chap317}, 
we get our first result on deterministic column-based 
matrix reconstruction with respect to the spectral norm. 
\begin{algorithm}[t]
\begin{framed}
\textbf{Input:} $\matA\in\R^{m\times n}$ of rank $\rho$, target rank $k < \rho$, and oversampling parameter $r > k$. \\
\noindent \textbf{Output:} \math{\matC \Omega \matS \in \R^{m \times r}} with $r$ (rescaled) columns from $\matA$.
\begin{algorithmic}[1]
\STATE Compute the matrices $\matV_k$ and $\matV_{\rho-k}$ with the SVD. 
\STATE $[\Omega, \matS] = BarrierSamplingII(\matV_k, \matV_{\rho-k},r)$ (Lemma \ref{lem:2setS} in Section \ref{chap317}).
\STATE Return $\matC = \matA \Omega \matS \in \R^{m \times r}$. 
\end{algorithmic}
\caption{Deterministic spectral norm reconstruction with $r > k$.}
\label{alg:chap41}
\end{framed}
\end{algorithm}
\begin{algorithm}[t]
\begin{framed}
\textbf{Input:} $\matA\in\R^{m\times n}$ of rank $\rho$, target rank $2 \le k < \rho$, and oversampling parameter $r > k$. \\
\noindent \textbf{Output:} \math{\matC \Omega \matS \in \R^{m \times r}} with $r$ (rescaled) columns from $\matA$.
\begin{algorithmic}[1]
\STATE $\matZ = FastSpectralSVD(\matA, k, 1)$ (Lemma \ref{tropp1} in section \ref{chap22}).
\STATE $[\Omega, \matS] = BarrierSamplingII( \matZ, \matI_n, r)$ (Lemma \ref{lem:2setS} in Section \ref{chap317}).
\STATE Return $\matC = \matA \Omega \matS \in \R^{m \times r}$. 
\end{algorithmic}
\caption{Fast spectral norm reconstruction with $r > k$.}
\label{alg:chap42}
\end{framed}
\end{algorithm}
\begin{theorem}[Deterministic spectral norm reconstruction]
\label{theorem:intro1}
Given \math{\matA\in\R^{m \times n}} of rank \math{\rho},
a target rank \math{k < \rho}, and an oversampling parameter $r > k$,
Algorithm~\ref{alg:chap41} (deterministically in
\math{T(\matV_k, \matV_{\rho-k}) + O\left(rn\left(k^2+\left(\rho-k\right)^2\right)\right)}) 
constructs $\matC\in\R^{m \times r}$:
$$
\TNorm{\matA - \Pi^2_{\matC,k}(\matA)} \leq
\left(1 + \frac{1 + \sqrt{ (\rho-k)/r }}{ 1 - \sqrt{ k/r }  }  \right)
\TNorm{\matA - \matA_k}=
O\left(\sqrt{{\rho}/{r}}\right)
\TNorm{\matA - \matA_k}.$$
\end{theorem}
%
\noindent The asymptotic multiplicative error of the above theorem matches a lower bound that we prove in Section~\ref{sec:lower}. This is the first spectral reconstruction algorithm with asymptotically optimal guarantees for arbitrary $r > k$. Previous work presented almost near-optimal algorithms only for $r=k$~\cite{GE96}. We note that in the proof of this theorem in Section~\ref{chap44}, we will present an algorithm that achieves a slightly worse error bound (essentially replacing $\rho$ by $n$ in the accuracy guarantee), but only needs to compute the top $k$ right singular vectors of 
$\matA$ (i.e., the matrix $\matV_k$). 

Next, we describe an algorithm that gives (up to a tiny constant) the same
bound as Theorem \ref{theorem:intro1} but is considerably more efficient. 
In particular, there is no need to compute the right singular vectors of $\matA$;
all we need are approximations that we can obtain through Lemma \ref{tropp1} of section \ref{chap22}.
Randomization is the penalty to be paid for the improved computational efficiency. 
\begin{theorem}[Fast spectral norm reconstruction]
\label{thmFast1}
Given \math{\matA\in\R^{m\times n}} of rank $\rho$, a target rank $2\leq k < \rho$, 
and an oversampling parameter $r > k$, Algorithm~\ref{alg:chap42}
(randomly in $O\left(mnk\log\left( k^{-1}\min\{m,n\}\right)+nrk^2\right)$) 
constructs $\matC\in\R^{m \times r}$ such that:
$$  \Expect{ \TNorm{\matA - \Pi_{\matC,k}^2(\matA)} } \leq \left(\sqrt{2}+1\right) \left(1 + \frac{1 + \sqrt{ n/r }}{1 - \sqrt{ k/r }}\right)
\TNorm{\matA-\matA_k}=O\left(\sqrt{{n}/{r}}\right)\TNorm{\matA-\matA_k}.$$
\end{theorem}

\section{Frobenius Norm Approximation ($r > k, \xi = F$) }\label{chap42}

By combining the deterministic exact SVD of Lemma \ref{lem:generic} of section \ref{chap22}
along with 
the deterministic technique of section \ref{chap318}, 
we get a deterministic column-based matrix reconstruction algorithm in the Frobenius norm. 
\begin{algorithm}[t]
\begin{framed}
\textbf{Input:} $\matA\in\R^{m\times n}$ of rank $\rho$, target rank $k < \rho$, and oversampling parameter $r > k$. \\
\noindent \textbf{Output:} \math{\matC \Omega \matS \in \R^{m \times r}} with $r$ (rescaled) columns from $\matA$.
\begin{algorithmic}[1]
\STATE Compute the matrix $\matV_k$  with the SVD. 
\STATE $[\Omega, \matS] = BarrierSamplingIII( \matV_k, \matA - \matA \matV_k \matV_k\transp, r)$ 
(Lemma \ref{lem:2setF} in Section \ref{chap318}).
\STATE Return $\matC = \matA \Omega \matS \in \R^{m \times r}$. 
\end{algorithmic}
\caption{Deterministic Frobenius norm reconstruction with $r > k$.}
\label{alg:chap43}
\end{framed}
\end{algorithm}
\begin{theorem}[Deterministic Frobenius norm reconstruction]
\label{theorem:intro2}
Given \math{\matA\in\R^{m\times n}} of rank $\rho$, a target rank $k <\rho$, 
and an oversampling parameter $r > k$, Algorithm~\ref{alg:chap43}
(deterministically in \math{T(\matV_k) + O\left(mn + nrk^2\right)})
constructs
$\matC\in\R^{m \times r}$ 
such that:
$$ \FNorm{\matA - \Pi_{\matC,k}^F(\matA)} \leq
\sqrt{ 1 + \frac{1}{ \left( 1 - \sqrt{ k/r } \right)^2 } } \FNorm{\matA - \matA_k} \le \sqrt{ 2 + O\left(  \frac{k}{r}  \right)  }
\FNorm{\matA - \matA_k}.$$
\end{theorem}
By combining the randomized approximate SVD of Lemma \ref{tropp2} of section \ref{chap22} along with 
the deterministic technique of section \ref{chap318}, 
we get a considerably faster randomized column-based matrix reconstruction algorithm in the Frobenius norm. 
\begin{algorithm}[t]
\begin{framed}
\textbf{Input:} $\matA\in\R^{m\times n}$ of rank $\rho$, target rank $2 \le k < \rho$, and oversampling parameter $r > k$. \\
\noindent \textbf{Output:} \math{\matC \Omega \matS \in \R^{m \times r}} with $r$ (rescaled) columns from $\matA$.
\begin{algorithmic}[1]
\STATE $\matZ = FastFrobeniusSVD(\matA, k, 0.1)$ (Lemma \ref{tropp2} in section \ref{chap22}).
\STATE $[\Omega, \matS] = BarrierSamplingIII( \matZ, \matA - \matA \matZ \matZ\transp, r)$ 
(Lemma \ref{lem:2setF} in Section \ref{chap318}).
\STATE Return $\matC = \matA \Omega \matS \in \R^{m \times r}$. 
\end{algorithmic}
\caption{Fast Frobenius norm reconstruction with $r > k$.}
\label{alg:chap44}
\end{framed}
\end{algorithm}
\begin{theorem}[Fast Frobenius norm reconstruction]
\label{thmFast2}
Given \math{\matA\in\R^{m\times n}} of rank $\rho$, a
target rank $2\leq k < \rho$, and an oversampling parameter $r > k$,
Algorithm~\ref{alg:chap44} (randomly in $O\left(mnk+nrk^2\right)$)
constructs $\matC\in\R^{m \times r}$ such that:
$$
 \Expect{\FNorm{\matA - \Pi_{\matC,k}^F(\matA)}} \leq
\sqrt{ 1.1 + \frac{1.1}{ \left( 1 - \sqrt{ k/r } \right)^2 } } \FNorm{\matA - \matA_k} \le \sqrt{ 3 + O\left(  \frac{k}{r}  \right)  }
\FNorm{\matA - \matA_k}.$$
\end{theorem}
Previous work presented deterministic near-optimal algorithms only for $r=k$ \cite{DR10}. 
We are not aware of any deterministic algorithm for the $r>k$ case; previous work presents only
randomized algorithms that fail unless $r = \Omega(k \log(k))$~\cite{DMM06a,DV06}. 

Both Theorems \ref{theorem:intro2} and \ref{thmFast2} guarantee constant factor approximations
to the error $\FNorm{ \matA - \matA_k }$. In both cases, for arbitrary small $\epsilon > 0$,
if $r = O( k / \epsilon)$: 
$$\FNormS{\matA - \Pi_{\matC,k}^F(\matA)} \le  \left(\beta + \epsilon \right) \FNormS{\matA - \matA_k},$$
(In Theorem \ref{theorem:intro2}, $\beta=2$; in Theorem \ref{thmFast2}, $\beta=3$). We stated the
result here with the squares since it is stronger and is necessary to conclude that $r = O( k / \epsilon)$ columns give constant factor
approximations. Manipulating the bound without the squares yields $r = O( k / \epsilon^2)$ columns, which is weaker. 
The stronger bound with the squares is possible and can be found in the proofs of the corresponding
theorems.

Constant factor approximations are interesting but not optimal.
Here, we describe an algorithm that guarantees $(1+\epsilon)$-error Frobenius norm approximation. We do so
by combining the algorithm of Theorem~\ref{thmFast2} with one round of adaptive sampling~\cite{DV06},
i.e. with the randomized technique described in section \ref{chap312}.
This is the first relative-error approximation algorithm for Frobenius norm reconstruction that uses a linear  number of columns
in $k$ (the target rank).  Previous work~\cite{DMM06b,Sar06} achieves relative error with \math{O(k\log k+ k\epsilon^{-1})} columns. Our result is asymptotically optimal, matching the \math{\Omega(k/\epsilon)} lower bound in \cite{DV06}.
\begin{algorithm}[t]
\begin{framed}
\textbf{Input:} $\matA\in\R^{m\times n}$ of rank $\rho$, target rank $2 \le k < \rho$, and oversampling parameter $r > 10k$. \\
\noindent \textbf{Output:} \math{\matC \Omega \matS \in \R^{m \times r}} with $r$ (rescaled) columns from $\matA$.
\begin{algorithmic}[1]
\STATE $\matZ = FastFrobeniusSVD(\matA, k, 0.1)$ (Lemma \ref{tropp2} in section \ref{chap22}).
\STATE $[\Omega, \matS] = BarrierSamplingIII( \matZ, \matA - \matA \matZ \matZ\transp, 4k)$ 
(Lemma \ref{lem:2setF} in Section \ref{chap318}).
\STATE $\matC_2 = AdaptiveSampling(\matA, \matA \Omega, r-4k)$ (Section \ref{chap312}).
\STATE Return $\matC = [\matC_2; \matA \Omega ] \in \R^{m \times r}$. 
\end{algorithmic}
\caption{Fast relative-error Frobenius norm reconstruction.}
\label{alg:chap45}
\end{framed}
\end{algorithm}
\begin{theorem}[Fast relative-error Frobenius norm reconstruction]
\label{thmFast3}
Given \math{\matA\in\R^{m\times n}} of rank $\rho$, a target rank $2\leq k < \rho$, and an oversampling parameter $r > 10k$,
Algorithm~\ref{alg:chap45} (randomly in $O\left(mnk + nk^3 + n \log(r) \right)$)
constructs
$\matC\in\R^{m \times r}$
such that:
\mand{\Expect{ \FNorm{\matA - \Pi_{\matC,k}^F(\matA)}} \leq \sqrt{1+ \frac{6k}{r-4k}}\FNorm{\matA-\matA_k}.
}
\end{theorem}
%
\noindent In the proof of the theorem we get: 
$\Expect{ \FNormS{\matA - \Pi_{\matC,k}^F(\matA)}} \leq (1+\frac{6k}{r-4k})\FNormS{\matA-\matA_k}.
$
For any given $\epsilon > 0$, choosing $r \geq 4k + 6k / \epsilon \ge 10k / \epsilon$ columns gives
$\Expect{ \FNormS{\matA - \Pi_{\matC,k}^F(\matA)}} \leq (1+\epsilon)\FNormS{\matA-\matA_k}.$
Taking square roots on both sides of this equation and observing that $\sqrt{1+\epsilon} \le 1+\epsilon$:
$\sqrt{\Expect{ \FNormS{\matA - \Pi_{\matC,k}^F(\matA)}}} \leq (1+\epsilon)\FNorm{\matA-\matA_k}.$
Finally, Holder's inequality implies: $\Expect{ \FNorm{\matA - \Pi_{\matC,k}^F(\matA)} }  
\le 
\sqrt{ \Expect{ \FNormS{\matA - \Pi_{\matC,k}^F(\matA)}} }$. Overall,
\mand{\Expect{ \FNorm{\matA - \Pi_{\matC,k}^F(\matA)}} \leq (1+\epsilon)\FNorm{\matA-\matA_k}.}

\subsection*{Note on Running Times.} 
All our running times are stated in terms of the number of
operations needed to compute the matrix $\matC$, and
for simplicity we assume that \math{\matA} is dense;
 if \math{\matA} is sparse, additional savings might be possible.
Our accuracy guarantees are stated in terms of the optimal matrix
$\Pi_{\matC,k}^{\xi}(\matA)$, which would require additional time to compute.
For the Frobenius norm ($\xi=F$), the
computation of $\Pi_{\matC,k}^{F}(\matA)$ is straightforward,
and only requires an
additional $O\left(mnr + \left(m+n\right)r^2\right)$ time (see the
discussion in Section~\ref{chap21}).
However, for the spectral norm ($\xi=2$),
we are not aware of any algorithm to compute
$\Pi_{\matC,k}^{2}(\matA)$ exactly.
In Section~\ref{chap21} we presented a simple approach
that computes $\hat\Pi_{\matC,k}^{2}(\matA)$,
a constant-factor approximation
to $\Pi_{\matC,k}^{2}(\matA)$, in $O\left(mnr + \left(m+n\right)r^2\right)$
time.
Our bounds in Theorems~\ref{theorem:intro1} and~\ref{thmFast1} could now be restated in terms of the error $\norm{\matA-\hat\Pi_{\matC,k}^{2}(\matA)}_2$;
the only change in the accuracy guarantees would be a multiplicative increase of $\sqrt{2}$ (from Lemma~\ref{lem:bestF}).
%


\section{The Column Subset Selection Problem ($r = k, \xi = 2, F$)} \label{chap43}

Our focus has been on selecting $r > k$ columns; however, selecting exactly 
\math{k} columns is also of considerable interest
\cite{BMD09a,DR10,Ipsen}. Selecting exactly $k$ columns from $\matA$
is known as the Column Subset Selection Problem (CSSP). We present
novel algorithms for the CSSP below. 

The basic idea of our algorithms is to use the algorithms in Theorems \ref{thmFast1} and \ref{thmFast2}
to select \math{r = dk} columns, with, for example, some small constant $d=2,3,4...$, 
and then down sample this to exactly \math{k} columns using the deterministic RRQR technique
of section \ref{chap315}. We will use exactly this approach to obtain $k$ columns for 
Frobenius reconstruction (Theorem \ref{thmCSSPf}).
For spectral reconstruction, selecting $k$ columns using the deterministic RRQR technique
of section \ref{chap315} suffices to give near-optimal results (Theorem \ref{thmCSSPs}). We also
present a two-step algorithm that gives bounds for both the spectral and the Frobenius norm
but uses the randomized technique of Section \ref{chap314} in the first step (Theorem \ref{fastcssp}). 
For the exact description of the algorithms in the three theorems below see the corresponding 
(constructive) proofs in Section \ref{chap44}.  

\begin{theorem}[Randomized Spectral CSSP]
\label{thmCSSPs}
Given \math{\matA\in\R^{m\times n}} of rank $\rho$ and a
target rank $2 \leq k < \rho$, 
there is a $O(m n k \log(\min(m,n)/k)  + nk^2\log n)$  randomized algorithm 
to construct 
\math{\matC\in\R^{m\times k}} such that
$$ \Expect{ \TNorm{\matA - \matC\matC^+\matA} } \leq 4 \sqrt{4k(n-k)+1)} \TNorm{\matA-\matA_k}.$$ 
\end{theorem}
\noindent To put our result into perspective, the 
best known algorithm for spectral-norm CSSP 
\cite[Algorithm 4]{GE96}, is deterministic, 
runs in $O( k m n \log(n) )$, and gives:
$$ \TNorm{\matA - \matC\matC^+\matA} \leq \sqrt{(4 k(n-k) +1 )} \TNorm{\matA - \matA_k}.$$
We are worse by a constant $4$ but faster by $O(1/\log(1/k))$.
\begin{theorem}[Randomized Frobenius CSSP]
\label{thmCSSPf}
Given \math{\matA\in\R^{m\times n}} of rank $\rho$ and a
target rank $2\leq k < \rho$, 
there is an $O(m n k + n k^3 + k^3\log(k)) )$ randomized algorithm to construct
\math{\matC\in\R^{m\times k}} such that
$$ \Expect{
\FNorm{\matA - \matC \matC^+ \matA} }\leq 
9 k 
\FNorm{\matA-\matA_k}.
$$
\end{theorem}
\noindent To put our result into perspective, 
the best known algorithm for Frobenius-norm  CSSP \cite[Theorem 9, $\epsilon=1/4$]{DR10},
which runs in $O(mn  k^{2}\log(m)  + m k^{7}
\log^{3}(m) \log(k  \log m))$, gives:
$$ \Expect{ \FNorm{\matA - 
\matC\matC^+\matA} } \leq \sqrt{1.25 (k + 1) }  \FNorm{\matA - \matA_k},$$
We are  worse  by $O(k)$ but faster by at least $O(1 / (k \log(m)))$. 

\begin{theorem} [Randomized Spectral/Frobenius CSSP]
\label{fastcssp}
Given \math{\matA\in\R^{m\times n}} of rank $\rho$, $1 \leq k < \rho$,  and 
$0 < \delta < 1$, 
there is an algorithm to construct $\matC \in \R^{m \times k}$ in
$O\left(m n k +  k^3 \cdot \ln(k / \delta) + k \cdot \ln(k / \delta) \cdot \log\left(k \cdot \ln(k/\delta) ) \right)  \right)$ 
such that w.p. $1 - 3\delta$:
$$\XNorm{\matA-\matC\matC^+\matA} \le  
\frac{26 k \sqrt{\ln(2k / \delta)}}{\delta} \FNorm{\matA-\matA_k} \le \frac{26 k \sqrt{\ln(2k / \delta)}}{\delta} 
\sqrt{\rho-k} \TNorm{\matA -\matA_k} .$$
\end{theorem}
For Frobenius norm, treating $\delta$ as a constant, the approximation error is $O( k \sqrt{\log(k)} )$;
for spectral norm, we get a somewhat unusual bound, i.e. the right hand side contains
the term $\FNorm{ \matA - \matA_k } \le \sqrt{\rho-k} \TNorm{ \matA - \matA_k }$.

\subsection*{Interpolative Decompositions} \label{sec:id}
We conclude this section by describing a novel algorithm for an Interpolative Decomposition of a matrix. 
We start by defining such a decomposition. 
\begin{lemma}[Interpolative Decomposition (ID) - Lemma 1 in \cite{LWMRT07}]\label{lem:id}
Let $\matA \in \R^{m \times n}$ and $k \leq \min\{m,n\}$; then, there is a matrix $\matC \in \R^{m \times k}$ containing columns
of $\matA$ and a matrix $\matX \in \R^{k \times n}$ such that 
\begin{enumerate}
\item some subset of the columns of $\matX$ makes up the $\matI_k$.
\item no entry of $\matX$ has absolute value greater than one.
\item $\TNorm{\matX} \le p_1(k,n)$, with $ p_1(k,n)= \sqrt{k(n-k)+1}$.
\item the $k$-th singular value of $\matX$ is at least one.
\item if $k=m$ or $k=n$, then $\matA = \matC \matX$.
\item if $k < m, n$, then: $ \TNorm{\matA - \matC \matX} \le p_2(k,n) \TNorm{\matA - \matA_k};$ 
$p_2(k,n) = \sqrt{k(n-k)+1}$.
\end{enumerate}
\end{lemma}
Notice that an ID is similar with the decompositions studied in this thesis with the only difference
being the matrix $\matX$; we typically choose $\matX = \matC^+ \matA$ but this choice might not be
numerically stable when the singular values of $\matC$ are very small. An ID introduces several 
properties for the matrix $\matX$ making sure that the approximation $\matA \approx \matC \matX$ is numerically stable. 
An algorithm for computing such an ID is in observation 3.3 of \cite{MRT11} 
(All \cite{CGMR05, LWMRT07, MRT11, HMT} study interpolative
decompositions as defined in Lemma \ref{lem:id});
this algorithm runs in $O(k m n \log(n))$ time and computes a factorization that slightly sacrifices
properties $(2)$, $(3)$, and $(6)$ from the above definition. More specifically, it computes $\matC, \matX$
such that: in $(2)$, no entry of $\matX$ has an absolute value greater than $2$; 
in $(3)$, $p_1(k,n) = \sqrt{4 k(n-k)+1}$;
similarly, in $(6)$, $p_2(k,n) =  \sqrt{4 k(n-k)+1}$.  
A quick modification of the algorithm of Theorem \ref{thmCSSPs} (see the proof for the details of the algorithms) gives a novel 
$O\left(m n k \log(k^{-1} \min(m,n)) + nk^2\log(n) \right)$ time randomized algorithm:
\begin{algorithmic}[1]
\STATE Via Lemma~\ref{tropp1}   of section \ref{chap22},  let $ \matZ  = FastSpectralSVD(\matA, k, 0.5)$;
\STATE Via Lemma~\ref{lem:rrqr} of section \ref{chap315}, let $ \Omega = RRQRSampling(\matZ,k)         $; 
\STATE Return \math{\matC=\matA\Omega} and $\matX = (\matZ\transp \Omega)^+ \matZ\transp.$
\end{algorithmic}
\paragraph{Analysis.}
Without loss of generality, assume that $\Omega$ samples the first $k$ columns of $\matZ\transp$. Let
$ \matZ =  \left[
\matZ_1;
\matZ_2
\right],$
with $\matZ_1 \in \R^{k \times k}$ and $\matZ_2 \in \R^{(n-k) \times k}$. 
From Lemma \ref{lem:rrqrGE}, $\matZ\transp \Pi = \matQ [\matA_k, \matB_k]$; so
$\matZ_1\transp = \matQ \matA_k$ and $\matZ_2\transp = \matQ \matB_k$. By using
this notation, 
$$\matX = [ (\matQ \matA_k)^+ (\matQ \matA_k), (\matQ \matA_k)^+ (\matQ \matB_k) ] =  [ \matI_k, \matA_k^+ \matB_k ].$$
Now, we comment on all six properties of the above definition for the output $\matC, \matX$ from our method. 
\begin{enumerate}

\item In Lemma \ref{lem:id}, (1) is obviously satisfied by the above choice of $\matX$.

\item Theorem 3.2 of \cite{GE96}, with $f=2$, shows that no entry in $\matX$ has an absolute value greater than $2$.

\item From Lemma \ref{lem:rrqr}: 
$\TNorm{\matX} = \TNorm{ (\matZ\transp \Omega)^+ \matZ\transp } \le 
\TNorm{ (\matZ\transp \Omega)^+} \leq p_1(k,n) \TNorm{ (\matZ\transp)^+} 
= \sqrt{4k(n-k)}+1 \frac{1}{\sigma_k(\matZ\transp)} = (\sqrt{4k(n-k)}+1) 1 = \sqrt{4k(n-k)}+1 $.

\item $\sigma_{min}(\matX) > 1$ because
$$ \sigma_{min}(\matX) = \sigma_k((\matZ\transp \Omega)^+ \matZ\transp) = \sigma_k((\matZ\transp \Omega)^+) = \frac{1}{\sigma_{max}((\matZ\transp \Omega))},$$
and the fact that $\sigma_{max}(\matZ\transp \Omega) < \sigma_{max}(\matZ\transp) = 1$, by the
interlacing property of the singular values.

\item 
Via Lemma \ref{tropp1}, we computed a factorization  
$\matA = \matA \matZ \matZ\transp + \matE$. 
Lemma~\ref{prop3} argues that $\rank(\matZ) = k$ with probability one.
Also, since in this case $\matA = \matA_k$, $\matE = \bm{0}_{n \times n}$ w.p. 1. 
Plug this factorization into Lemma \ref{lem:genericNoSVD}
and construct the matrix $\matW$ as we described in Theorem \ref{thmCSSPs}.
Theorem 7.2 in \cite{GE96} gives that $\rank(\matZ\transp \matW) = k$.
So, $\FNormS{\matA - \matC \matC^+\matA } = 0$, which also means
that $\FNormS{\matA - \matC \matX } = 0$, i.e. $\matA = \matC \matX$. 

\item Property (6) is satisfied with $p_2 = 4 \sqrt{4k(n-k)+1}$; to
see this, notice that the bound for $\TNorm{\matA - \matC \matC^+ \matA}$ also holds for 
$\TNorm{\matA - \matC \matX}$ because this is the way we proved Lemma \ref{lem:genericNoSVD}.

\end{enumerate}

\clearpage

\section*{Proofs} \label{chap44}

\subsection*{Proof of Theorem~\ref{theorem:intro1}}
\begin{algorithmic}[1]
\STATE Via the SVD compute the matrices $\matV_k$, $\matV_{\rho-k}$ of the right singular vectors of $\matA$.
\STATE Via Lemma~\ref{lem:2setS} of section \ref{chap317}, let $ [\Omega, \matS] = BarrierSamplingII(\matV_k,\matV_{\rho-k},r)$. 
\STATE Return \math{\matC=\matA\Omega\matS} with rescaled columns of $\matA$.
\end{algorithmic}
Lemma \ref{lem:2setS} guarantees that $\sigma_{k}(\matV_k\transp\Omega\matS)\geq 1-\sqrt{k/r}>0$ (assuming $r>k$), and so $\rank(\matV_k\transp\Omega\matS)=k$. Also,
\math{\sigma_1(\matV_{\rho-k}\transp\Omega\matS)=
\norm{\matV_{\rho-k}\transp\Omega\matS}_2\leq 1+\sqrt{(\rho-k)/r}
}.
Applying Lemma~\ref{lem:generic}, we get
\eqan{
\TNormS{\matA - \Pi_{\matC,k}^2(\matA)}
&\leq&
\TNormS{\matA - \matA_k}+
\TNormS{(\matA-\matA_k) \Omega  \matS (\matV_k\transp \Omega\matS)^+}\\
&\le&
\TNormS{\matA - \matA_k}+
\TNormS{(\matA -\matA_k)\Omega  \matS}\TNormS{(\matV_k\transp\Omega\matS)^+}\\
&=&
\TNormS{\matA - \matA_k}+
\TNormS{\matU_{\rho-k}\Sigma_{\rho-k}\matV_{\rho-k}\transp
\matS}\TNormS{(\matV_k\transp \Omega\matS)^+}\\
&\le&
\TNormS{\matA - \matA_k}+
\TNormS{\Sigma_{\rho-k}}\TNormS{\matV_{\rho-k}\transp \Omega
\matS}\TNormS{(\matV_k\transp \Omega \matS)^+}\\
&\le&
\TNormS{\matA - \matA_k}\left(1+\frac{(1+\sqrt{(\rho-k)/r})^2}
{(1-\sqrt{k/r})^2}\right),
}
where the last inequality follows because
\math{\TNorm{\Sigma_{\rho-k}}=\TNorm{\matA - \matA_k}} and
\math{\TNorm{(\matV_k\transp \Omega \matS)^+}=
1/\sigma_{k}(\matV_k\transp \Omega \matS)\le 1/(1-\sqrt{k/r})}.
Theorem \ref{theorem:intro1} now follows by taking square roots of
both sides.
The running time is equal to the time needed to compute
$\matV_k$ and $\matV_{\rho-k}$ (denoted as $T(\matV_k, \matV_{\rho-k}$) 
plus the running time of the algorithm in Lemma~\ref{lem:2setS}. 
\qedsymb

\paragraph{A faster spectral norm deterministic algorithm.}
 Our next theorem describes a deterministic algorithm for spectral norm reconstruction that only needs to compute \math{\matV_k} and will serve as a prequel to the proof of Theorem~\ref{thmFast1}. The accuracy guarantee of this theorem is essentially identical to the one in Theorem~\ref{theorem:intro1}, with $\rho-k$ being replaced by $n$.
\begin{theorem}\label{theorem:spectralIn}
Given \math{\matA\in\R^{m\times n}} of rank $\rho$, 
a target rank $k < \rho$, and an oversampling parameter $r > k$,
there exists a $T(\matV_k)  + O(nrk^2)$ deterministic algorithm to construct
$\matC\in\R^{m \times r}$ such that
$$\TNorm{\matA - \Pi_{\matC,k}^2(\matA)} \leq
\left( 1 + \frac{1 + \sqrt{ n/r }}{1 - \sqrt{ k/r }}  \right)
\TNorm{\matA-\matA_k} = O\left( \sqrt{\frac{n}{r}} \right) \TNorm{\matA-\matA_k}.$$
\end{theorem}
\begin{proof}
The proof  is very similar to the proof of Theorem~\ref{theorem:intro1}, so we only highlight the differences. 
We first give the algorithm.
\begin{algorithmic}[1]
\STATE Via the SVD compute the matrix $\matV_k$ of the right singular vectors of $\matA$.
\STATE Via Lemma~\ref{lem:2setS} of section \ref{chap317}, let $ [\Omega, \matS] = BarrierSamplingII(\matV_k,\matI_n,r)$. 
\STATE Return \math{\matC=\matA\Omega\matS} with rescaled columns of $\matA$.
\end{algorithmic}
Lemma \ref{lem:2setS} guarantees that \math{\norm{\matI_n\matS}_2\le 1+\sqrt{n/r}}. We now replicate the proof of Theorem~\ref{theorem:intro1} up to the point where
\math{\norm{(\matA-\matA_k)\Omega\matS(\matV_k\transp\Omega\matS)^+}_2^2} is bounded. We continue as follows:
\mand{
\norm{(\matA-\matA_k)\Omega\matS(\matV_k\transp \Omega\matS)^+}_2^2
\le
\TNormS{(\matA-\matA_k)}\TNormS{\matI_n\Omega\matS}
\norm{(\matV_k\transp\Omega\matS)^+}_2^2.
}
The remainder of the proof now follows the same line as in Theorem~\ref{theorem:intro1}. 
To analyze the running time of the proposed algorithm, we need to look more closely at 
Lemma~\ref{lemma:intro1} and the related Algorithm~\ref{alg:2set} and note that since one set of input vectors consists of the standard basis vectors, Algorithm~\ref{alg:2set} runs in $O(nrk^2)$ time (see the discussion of the running time for 
this algorithm). The total running time is the time needed to compute $\matV_k$ plus $O(nrk^2)$.
\end{proof}

\subsection*{Proof of Theorem~\ref{thmFast1}} 
\begin{algorithmic}[1]
\STATE Via Lemma \ref{tropp1} of Section \ref{chap22}, let $\matZ = FastSpectralSVD(\matA, k, 1)$.
\STATE Via Lemma \ref{lem:2setS} of section \ref{chap317}, let $[\Omega, \matS] = BarrierSamplingII( \matZ, \matI_n, r)$ 
\STATE Return \math{\matC=\matA\Omega\matS} with rescaled columns of $\matA$.
\end{algorithmic}
In order to prove Theorem \ref{thmFast1} we will follow the proof of Theorem~\ref{theorem:intro1} using Lemma~\ref{tropp1} (a fast matrix factorization) instead of Lemma~\ref{lem:generic} (the exact SVD of $\matA$). More specifically, instead of using the top \math{k} right singular vectors of \math{\matA} (the matrix \math{\matV_k}), we use the matrix \math{\matZ \in \mathbb{R}^{n \times k}} of Lemma~\ref{tropp1}. The proof of Theorem~\ref{thmFast1} is now identical to the proof of Theorem~\ref{theorem:spectralIn}, except for using Lemma~\ref{lem:genericNoSVD} instead of Lemma~\ref{lem:generic} in the first step of the proof:
\eqan{
\TNormS{\matA - \Pi^2_{\matC,k}(\matA)}
&\leq&
\TNormS{\matE}+
\TNormS{\matE\Omega \matS (\matZ\transp\Omega\matS)^+}
=
\TNormS{\matE}+
\TNormS{\matE\matI_n\Omega\matS (\matZ\transp\Omega\matS)^+}\\
&\le&
\TNormS{\matE}\left(1+
\TNormS{\matI_n\Omega\matS}
\TNormS{(\matZ\transp\Omega\matS)^+}\right),
}
where \math{\matE} is the residual error from the matrix factorization of Lemma~\ref{tropp1}. Taking square roots of both sides, we get
\eqan{
\norm{\matA - \Pi^2_{\matC,k}(\matA)}_2
&\le&
\norm{\matE}_2\left(1+
\norm{\matI_n\Omega\matS}_2
\norm{(\matZ\transp\Omega\matS)^+}_2\right).
}
We can now use the bounds guaranteed by Lemma~\ref{lem:2setS} for \math{\norm{\matI_n\matS}_2} and \math{\norm{(\matZ\transp\matS)^+}_2}, to
obtain a bound in terms of \math{\norm{\matE}_2}. Finally, since \math{\matE} is a random variable, taking expectations and applying the bound of Lemma~\ref{tropp1} concludes the proof of the theorem. The overall running time is equal to the time needed to compute the matrix $\matZ$ from Lemma~\ref{tropp1} plus an additional $O(nrk^2)$ time as in Theorem~\ref{theorem:spectralIn}.
\qedsymb

\subsection*{Proof of Theorem \ref{theorem:intro2}}  
\begin{algorithmic}[1]
\STATE Via the SVD compute the matrix $\matV_k$ of the right singular vectors of $\matA$.
\STATE Via Lemma \ref{lem:2setF} of section \ref{chap318}, let 
$[\Omega, \matS] = BarrierSamplingIII( \matV_k, \matA - \matA \matV_k \matV_k\transp, r)$.
\STATE Return \math{\matC=\matA\Omega\matS} with rescaled columns of $\matA$.
\end{algorithmic}
We follow the proof of Theorem~\ref{theorem:intro1} 
in the previous section up to the point where we need to bound the term
\math{\norm{(\matA-\matA_k)\Omega\matS (\matV_k\transp \Omega\matS)^+}^2_F}.
By spectral submultiplicativity,
\mand{
\FNorm{(\matA-\matA_k) \Omega \matS (\matV_k\transp \Omega\matS)^+}^2
\le
\FNorm{(\matA-\matA_k) \Omega\matS}^2
\TNorm{(\matV_k\transp \Omega\matS)^+}^2.
}
To conclude, we apply Lemma~\ref{lem:2setF} of Section \ref{chap318}
to bound the two terms in the right-hand side of the above inequality
and take square roots on the resulting equation.
The running time of the proposed algorithm is equal to the time needed to compute $\matV_k$ plus the time needed to compute $\matA-\matA \matV_k \matV_k\transp$ (which is equal to \math{O(mnk)} given $\matV_k$), plus the time needed to run the algorithm of Lemma~\ref{lem:2setF}, which is equal to \math{O\left(nrk^2+mn\right)}.
\qedsymb

\subsection*{Proof of Theorem \ref{thmFast2}}
We will follow the proof of Theorem~\ref{theorem:intro2}, but, as with the proof of Theorem~\ref{thmFast1}, instead of using the top \math{k} right singular vectors of \math{\matA} (the matrix \math{\matV_k}), we will use the matrix \math{\matZ} of Lemma~\ref{tropp2} that is computed via a fast factorization. 
\begin{algorithmic}[1]
\STATE Via Lemma \ref{tropp2} of section \ref{chap22}, let $\matZ = FastFrobeniusSVD(\matA, k, 0.1)$;
\STATE Via Lemma \ref{lem:2setF} of section \ref{chap318}, let 
$[\Omega, \matS] = BarrierSamplingIII( \matZ, \matA - \matA \matZ \matZ\transp, r)$.
\STATE Return \math{\matC=\matA\Omega\matS} with rescaled columns of $\matA$.
\end{algorithmic}
The proof of Theorem \ref{thmFast2} is now identical to the proof of Theorem~\ref{theorem:intro2}, except for using Lemma~\ref{lem:genericNoSVD} instead of Lemma~\ref{lem:generic}. Ultimately, we obtain
\eqan{
\FNormS{\matA - \Pi^F_{\matC,k}(\matA)}
&\leq&
\FNormS{\matE}+
\TNormS{\matE \Omega\matS (\matZ\transp \Omega\matS)^+}\\
&\le&
\FNormS{\matE}+
\FNormS{\matE\Omega\matS}\TNormS{(\matZ\transp \Omega\matS)^+}\\
&\le&
\left(1+\left(1-\sqrt{k/r}\right)^{-2}\right)\FNormS{\matE}.
}
The last inequality follows from the bounds of Lemma~\ref{lem:2setF}.
The theorem now follows by taking the expectation of both sides, 
using Lemma~\ref{tropp2} to bound ${\bf E}[\norm{\matE}_F^2]$,
taking squares roots on both sides of the resulting equation, 
and using Holder's inequality (i.e Lemma~\ref{prop0}
with $x = \FNormS{\matA - \Pi^F_{\matC,k}(\matA)}$) to get 
$ \Expect{ \FNorm{\matA - \Pi^F_{\matC,k}(\matA)}
}  \le \sqrt{ \Expect{ \FNormS{\matA - \Pi^F_{\matC,k}(\matA)} } } $.
The overall running time is derived by replacing the time needed to compute \math{\mat\matV_k} in Theorem~\ref{theorem:intro2} with the time needed to compute the fast approximate
factorization of Lemma~\ref{tropp2}.
\qedsymb

\subsection*{Proof of Theorem~\ref{thmFast3}}
\begin{algorithmic}[1]
\STATE Via Lemma \ref{tropp2} of section \ref{chap22}, let $\matZ = FastFrobeniusSVD(\matA, k, 0.1)$.
\STATE Via Lemma \ref{lem:2setF} of section \ref{chap318}, let $[\Omega, \matS] = BarrierSamplingIII( \matZ, \matA - \matA \matZ \matZ\transp, 4k)$.
\STATE Via the technique of section \ref{chap312}, let $\matC_2 = AdaptiveSampling(\matA, \matA \Omega, r-4k )$; 
\STATE Return $\matC$ containing the columns of both $\matA \Omega \matS$ and $\matC_2$
\end{algorithmic}
First, let $\hat{r} = 4 k $ and
$$
c_0=\left(1+0.1\right)\left(1+\frac{1}{(1-\sqrt{k/\hat{r}})^2}\right) = \frac{11}{5}.
$$
Notice that the first two steps of the algorithm correspond to running
the algorithm of Theorem~\ref{thmFast2} to sample \math{\hat{r}=4k} columns of $\matA$ and form the matrix \math{\matC_1 = \matA \Omega \matS}. The third step corresponds to running
the adaptive sampling algorithm of Lemma~\ref{oneround} with \math{\matB=\matA-\matC_1\matC_1^+\matA} and
sampling a further  $s = \ceil{r - 4k}$ columns of $\matA$ to form the matrix $\matC_2$.
Let \math{\matC = [\matC_1\ \ \matC_2]\in\R^{m \times r}}
contain all the
sampled columns. We will analyze the expectation
\math{\Expect{\norm{\matA-\Pi^F_{\matC,k}(\matA)}_F^2}}.

Using the bound of Lemma~\ref{oneround}, we first compute the expectation
with respect to \math{\matC_2} conditioned on $\matC_1$:
$$
\Exp_{\matC_2}\left[\left.\norm{\matA - \Pi^F_{\matC,k}(\matA)}_F^2
\right|\matC_1\right]
\le \FNormS{\matA-\matA_k}+\frac{k}{s} \FNormS{\matB}.
$$
We now compute the expectation with respect to \math{\matC_1} (only
\math{\matB} depends on \math{\matC_1}):
\mld{\Exp_{\matC_1}\left[\Exp_{\matC_2}
\left[\left.\FNormS{\matA - \Pi^F_{\matC,k}(\matA)}
\right|\matC_1\right]
\right]
\le
\FNormS{\matA-\matA_k}+\frac{k}{s} \Exp_{\matC_1}\left[\FNormS{\matA-\matC_1\matC_1^+\matA}\right].
\label{eq:iterated}
}
By the law of iterated expectation, the left hand side is exactly equal to the term \math{\Expect{\FNormS{\matA-\Pi^F_{\matC,k}(\matA)}}}.
We now use the accuracy guarantee of Theorem~\ref{thmFast2} and
our choice of $c_0$:
$$\Exp_{\matC_1}\left[\FNormS{\matA-\matC_1\matC_1^+\matA}\right]
\le \Exp_{\matC_1}\left[\FNormS{\matA-\Pi^F_{C_1,k}(\matA)}\right] \le c_0
\FNormS{\matA-\matA_k}.$$
Using this bound in \r{eq:iterated}, we obtain:
\mand{
\Expect{\FNormS{\matA-\Pi^F_{\matC,k}(\matA)}}
\le
\left(1+\frac{c_0k}{s}\right)\FNormS{\matA-\matA_k}.
}
Using $c_0 \le 6$ and our choice of $s$ gives: 
\mand{
\Expect{\FNormS{\matA-\Pi^F_{\matC,k}(\matA)}}
\le
\left(1+\frac{6 k}{r - 4k}\right)\FNormS{\matA-\matA_k}.
}
Taking square roots on both sides of this equation:
\mand{
\sqrt{\Expect{\FNormS{\matA-\Pi^F_{\matC,k}(\matA)}}}
\le
\sqrt{1+ \frac{6 k}{r - 4k}}\FNorm{\matA-\matA_k}.
}
To wrap up, use Holder's inequality as:
$$ \Expect{ \FNorm{\matA-\Pi^F_{\matC,k}(\matA)} } \le \sqrt{\Expect{\FNormS{\matA-\Pi^F_{\matC,k}(\matA)}}} $$

The time needed to compute the matrix $\matC$ is the sum of two terms:
the running time of Theorem~\ref{thmFast2} (which is
\math{O(m n k + n \hat{r} k^2)}),
plus the time of Lemma~\ref{oneround} 
(which is \math{O( m \hat{r} \min\{m,n\} + mn\hat{r} + n + s \log(s) )}). 
The total run time $O(mnk + nk^3 + n \log(r))$ follows because
$\hat{r} = 4k$ and $s = r - 4k < r  \le n$. 
\qed


\subsubsection*{Proof of Theorem \ref{thmCSSPs}}
\begin{algorithmic}[1]
\STATE Using
Lemma~\ref{tropp1} in section \ref{chap22}, let $\matZ = FastSpectralSVD(\matA, k, 0.5)$. 
\STATE Using Lemma~\ref{lem:rrqr} in section \ref{chap315}, let \math{\Omega=RRQRSampling(\matZ,k)}.
\STATE Return \math{\matC=\matA\Omega}.
\end{algorithmic}
Lemma~\ref{lem:rrqr} implies that $\rank(\matZ\transp)
=\rank(\matZ\transp\Omega)=k$, so 
we can apply Lemma \ref{lem:genericNoSVD}:
$$
\TNormS{\matA - \matC\matC^+\matA} 
\leq \TNormS{\matE} + 
\TNormS{ \matE \Omega (\matZ\transp  \Omega)^+}.
$$
By submultiplicativity, 
$
\TNormS{\matE  \Omega (\matZ\transp \Omega)^+} 
\leq 
\TNormS{\matE}\TNormS{\Omega} 
\TNormS{(\matZ\transp \Omega)^+}$.
Since \math{\Omega} is a subset of a permutation matrix,
\math{\TNormS{\Omega}=1}.
By Lemma~\ref{lem:rrqr}, 
\math{\TNormS{(\matZ\transp\Omega)^+}\le
(4k(r-k)+1)\TNormS{(\matZ\transp)^+}}. Using 
\math{\TNormS{(\matZ\transp)^+}=1}:
$$
\TNormS{\matA - \matC\matC^+\matA} 
\leq \TNormS{\matE} + 
\TNormS{\matE}(4k(r-k)+1).
%
$$
Taking square roots on the last expression: $ \TNorm{\matA - \matC\matC^+\matA} 
\leq \TNorm{\matE} ( 1 + \sqrt{(4k(r-k)+1)} ) $. Finally, taking expectations
over the randomness of $\matE = \matA - \matA \matZ \matZ\transp$ and using
Lemma~\ref{tropp1} with $\epsilon=0.5$ gives the result:
$$ \TNorm{\matA - \matC\matC^+\matA} \leq \TNorm{\matE} ( 1 + \sqrt{(4k(r-k)+1)} ) \le 4  \sqrt{(4k(r-k)+1)}
\TNorm{\matA - \matA_k}.$$
The run time follows by the run time of Lemma~\ref{tropp1} and Corollary~\ref{lem:rrqr}.
\qedsymb

\subsubsection*{Proof of Theorem \ref{thmCSSPf}}
\begin{algorithmic}[1]
\STATE Using Lemma~\ref{tropp2} in Section \ref{chap22}, let $\matZ = FastFrobeniusSVD(\matA, k, 1/2)$. 
\STATE Using Lemma \ref{lem:2setF} of Section \ref{chap318}, let
$[\Omega_1, \matS] = BarrierSamplingIII(\matZ, \matA - \matA \matZ \matZ\transp, 4 k)$. 
Define $\matX\transp = \matZ\transp \Omega_1 \matS$; so
\math{\matX\in\R^{4k \times k}}.
\STATE Via Lemma \ref{lem:rrqr} of Section \ref{chap315}, let
\math{\Omega_2=RRQRSampling(\matX,k)}; \math{\Omega_2\in\R^{4k \times k}}.
\STATE Return \math{\matC=\matA \Omega_1 \matS \Omega_2}.
\end{algorithmic}
Lemma~\ref{lem:2setF} implies that
\math{\rank(\matZ\transp\Omega_1\matS)=\rank(\matZ\transp)=k};
Lemma~\ref{lem:rrqr} implies that $\rank(\matX\transp)
=\rank(\matX\transp\Omega_2)$.
Thus, letting \math{\matW=\Omega_1\matS\Omega_2},
\math{\rank(\matZ\transp\matW)=k}. So, 
we can apply Lemma \ref{lem:genericNoSVD}:
$$
\FNormS{\matA - \matC\matC^+\matA} 
\leq 
\FNormS{\matE} + 
\FNormS{ \matE \Omega_1 \matS \Omega_2 (\matZ\transp  \Omega_1 \matS \Omega_2)^+}.
$$
By spectral submultiplicativity, 
$$
\FNormS{\matE  \Omega_1 \matS \Omega_2 (\matZ\transp \Omega_1 \matS \Omega_2)^+} 
\leq 
\FNormS{\matE \Omega_1 \matS} \TNormS{\Omega_2} 
\TNormS{(\matZ\transp \Omega_1 \matS \Omega_2)^+}.$$
By Lemma~\ref{lem:2setF}, 
\math{ 
\FNormS{\matE \Omega_1 \matS} 
\le
\FNormS{\matE}};
since \math{\Omega_2} is a subset of a permutation matrix,
\math{\TNormS{\Omega_2}=1}.
By Lemma~\ref{lem:rrqr}, with \math{\matX\transp=\matZ\transp \Omega_1 \matS},
\math{\TNormS{(\matX\transp \Omega_2)^+}\le
(4k(4k-k)+1)\TNormS{(\matX\transp)^+}}, and again by 
Lemma~\ref{lem:2setF},
\math{\TNormS{(\matX\transp)^+}\le1/(1-\sqrt{1/4})^2}.
Combining all these results together and using the bound from Lemma \ref{tropp1} with $\epsilon=1/2$, we have:
\eqan{
\FNormS{\matA - \matC\matC^+\matA} 
&\leq& \FNormS{\matE}
\left(1+ \frac{4k(4k-k)+1}{(1-\sqrt{1/4})^2}\right),\\
&\le&
53 \FNormS{\matE}.
}
Taking expectations on the later equation using Lemma~\ref{tropp2} with $\epsilon = 0.5$:
$$  \Expect{\FNormS{\matA - \matC\matC^+\matA} } \le 53 (1 + 0.5) \Expect{ \FNormS{\matA - \matA_k} }.$$
Taking square roots on this equation and using Holder's inequality concludes the proof. 
Finally, the total run time follows
by combining the run time of 
Lemma~\ref{tropp2}, Lemma~\ref{lem:2setF}, and Lemma~\ref{lem:rrqrGE},
which are \math{O(mnk)}, \math{O(nk^3 + mn)}, and 
\math{O(k^3\log(k))}. 
\qedsymb

\subsection*{Proof of Theorem \ref{fastcssp}}
\begin{algorithmic}[1]
\STATE Via Lemma \ref{tropp2} of Section \ref{chap22}, let $ \matZ = FastFrobeniusSVD(\matA, k, 0.5)$;
\STATE Via Lemma \ref{lem:random} of Section \ref{chap314}, 
let $[\Omega_1, \matS] = RandomSampling(\matZ, 1, 8 k \ln(\frac{2k}{\delta}))$;
Define $\matX\transp = \matZ\transp \Omega_1 \matS$; so
\math{\matX\in\R^{8 k \ln(\frac{2k}{\delta}) \times k}}.
\STATE Via Lemma \ref{lem:rrqr} of Section \ref{chap314},  
let $ \Omega_2 = RRQRSampling( \matX\transp )$.
\STATE Return $\matC = \matA \Omega_1 \matS \Omega_2$. 
\end{algorithmic}
Lemma~\ref{lem:2setF} implies that w.p. at least $1 - \delta$
\math{\rank(\matZ\transp\Omega_1\matS)=\rank(\matZ\transp)=k};
Lemma~\ref{lem:rrqr} implies that $\rank(\matX\transp)
=\rank(\matX\transp\Omega_2)$, for $\matX\transp = \matZ\transp \Omega_1 \matS$.
Thus, letting \math{\matW=\Omega_1\matS\Omega_2},
\math{\rank(\matZ\transp\matW)=k}. So, 
we can apply Lemma \ref{lem:genericNoSVD}:
$$
\XNormS{\matA - \matC\matC^+\matA} 
\leq 
\XNormS{\matE} + 
\XNormS{ \matE \Omega_1 \matS \Omega_2 (\matZ\transp  \Omega_1 \matS \Omega_2)^+}.
$$
By spectral submultiplicativity and the fact that $\TNormS{\matE \Omega_1 \matS \Omega_1
} \leq \FNormS{\matE \Omega_1 \matS \Omega_1}$: 
$$
\XNormS{\matE  \Omega_1 \matS \Omega_2 (\matZ\transp \Omega_1 \matS \Omega_2)^+} 
\leq 
\FNormS{\matE \Omega_1 \matS } \TNormS{\Omega_2} 
\TNormS{(\matZ\transp \Omega_1 \matS \Omega_2)^+}.$$
By Lemma \ref{lem:fnorm}, with $\matY = \matE$ and failure probability $\delta$: 
$\FNormS{ \matE\Omega\matS} \leq \frac{1}{\delta} \FNormS{ \matE}$, and 
since \math{\Omega_2} is a subset of a permutation matrix,
\math{\TNormS{\Omega_2}=1}.
By Lemma~\ref{lem:rrqr}, with \math{\matX\transp=\matZ\transp \Omega_1 \matS},
$$\TNormS{(\matX\transp \Omega_2)^+}\le  (4k(  8k\ln(2k/\delta) - k  )+1)  \TNormS{(\matX\transp)^+}
.$$
Also, from Lemma \ref{lem:random}, $\TNormS{(\matZ\transp \Omega \matS)^{+} } \leq 1/0.3$  w.p. at least $1-\delta$.
Combine all these results together, use $\XNormS{\matE} \le \FNormS{\matE}$
 and apply a union bound to get - so far - that w.p. at least 
$1 - 3 \delta$,
$$ \XNormS{\matA - \matC\matC^+\matA} \leq \frac{428 k^2 \ln(2k/\delta)}{\delta} \FNormS{\matE}.$$
Using Lemma \ref{tropp2} with $\epsilon = 0.5$: $\Expect{\FNormS{\matE} } \le 1.5 \FNormS{\matA - \matA_k}$.
Applying Markov's inequality on the random variable $x = \FNormS{\matE}$, we conclude that 
w.p. $1 - \delta$: $\FNormS{\matE} \le 1.5 / \delta \FNormS{\matA - \matA_k}$. Replacing this to the
above equation and taking square roots on both sides
we get that w.p. $1 - 4\delta$
$$ \XNorm{\matA - \matC\matC^+\matA} \leq  \frac{26 k \sqrt{\ln(2k/\delta)}}{\delta} \FNorm{\matA - \matA_k}.$$
Finally, the run time of the algorithm is 
  $O(mnk )$ 
+ $O( n +  k \log(k / \delta) \log( k \log(k / \delta) ) )$ 
+ $O(k^3 \ln(k/\delta) )$, from the first, second, and third step, respectively.  
\qedsymb

\chapter{CORESET CONSTRUCTION IN LEAST-SQUARES REGRESSION}
\label{chap5}
\footnotetext[6]{Portions of this chapter previously appeared as:
C. Boutsidis and P. Drineas, Random Projections for the Nonnegative
Least Squares Problem, Linear Algebra and its Applications, 431(5-7):760-771, 2009,
and as: C. Boutsidis, P. Drineas, M. Magdon-Ismail, Rich Coresets for Constrained Linear Regression
Manuscript, 2011.}
\section{Coreset Construction with Deterministic Sampling}\label{chap51}

Given $\matA\in\R^{m\times n}$ ($m \gg n$) of rank $\rho$, $\b\in\R^m$, and $\mathcal{D} \subseteq \R^n$,
the regression problem asks to find \math{\x_{opt}\in\cl D} for which
$ \TNorm{\matA \x_{opt} - \b}\le \TNorm{\matA \x - \b}$, for all $\x\in\cl D$; the domain
$\mathcal{D}$ represents the constraints on 
\math{\x} and can be arbitrary.
A coreset of size $r < m$ is  
$ \matC = \matS\transp \Omega\transp \matA $,  $ \b_c = \matS\transp \Omega\transp \b $,
for some sampling and rescaling matrices $\Omega \in \R^{m \times r}$, $ \matS \in \R^{r \times r}$. 
We assume that $\epsilon > 0$, which denotes the coreset approximation, is given as input. The goal
is to construct a coreset of size $r$, with $r$ being as small as possible and,
$
\TNormS{\matA \tilde{\x}_{opt} - \b} 
\le \left( 1 + \epsilon \right) \TNormS{\matA \x_{opt} - \b}; 
\tilde\x_{opt}=\argmin_{\x\in\cl D}\norm{\matS\transp \Omega\transp(\matA\x-\b)}_2^2.
$
\begin{theorem} 
\label{lem:regression}
Given $\matA \in \R^{m \times n}$ of rank $\rho=n$, $\b \in \R^m$, 
and~\math{ 0 \le \epsilon \le \frac{1}{3}}, Algorithm~\ref{alg:SimpReg}
constructs
\math{\Omega\in\R^{m\times r}} and \math{\matS\in\R^{r\times r}} with $r = \ceil{225 (n + 1) \epsilon^{-2} }$
in time $O( m n^2 + m n^3 / \epsilon^2)$ 
such that \math{\tilde\x_{opt} \in \R^n} obtained from  \math{\matS\transp\Omega\transp\matA,\ \matS\transp\Omega\transp\b} satisfies
$$
\TNormS{\matA \tilde{\x}_{opt} - \b} 
 \le \left( 1 + \epsilon \right) \TNormS{\matA \x_{opt} - \b}.
$$
\end{theorem}
\begin{algorithm}[t]
\begin{framed}
\textbf{Input:} $\matA\in\R^{m\times n}$ of rank $\rho=n$, $\b \in \R^m$, and \math{0 < \epsilon < 1/3}. \\
\noindent \textbf{Output:} \math{\Omega \in \R^{m \times r}}, \math{\matS \in \R^{r \times r}}
with $r = \ceil{225 (n + 1) \epsilon^{-2} }$.
\begin{algorithmic}[1]
\STATE Compute the matrix \math{\matU_{\matY}\in\R^{m\times k}}
of the top $k$ left singular vectors of \math{\matY=[\matA,\b] \in \R^{m \times (n+1)}};
\math{k = \rank(\matY) \le \rank(\matA) + 1 = \rho + 1}.
\STATE Let $r = \ceil{225 (n + 1) \epsilon^{-2} }$.
\STATE 
{\bf return} \math{[\Omega,\matS]=BarrierSamplingI(\matU_{\matY},r)}.
(Corollary~\ref{cor:1set})
\end{algorithmic}
\caption{Deterministic coreset for constrained regression.}
\label{alg:SimpReg}
\end{framed}
\end{algorithm}
\begin{proof} Let 
$ \matY = [\matA,\b]\in\R^{m\times(n+1)}$, and let its SVD: 
\math{\matY=\matU_{\matY}\Sigma_{\matY}\matV_{\matY}\transp}, where
\math{\matU_{\matY}\in\R^{m\times k}}, \math{\Sigma_{\matY}\in\R^{k \times k}}
and
\math{\matV_{\matY}\in\R^{(n+1)\times k}}.
Here, $k = \rank(\matY) \le \rho +1 = \rank(\matA) + 1$. 
Let $r > \rho+1$ be a sampling parameter whose exact value will be specified later. 
Let
$[\Omega, \matS] = BarrierSamplingI(\matU_{\matY}, r)$.
Let \math{\y_1,\y_2\in\R^{k}} defined as
\mand{
\y_1=\Sigma_{\matY}\matV_{\matY}\transp
\left[
\begin{matrix}
\x_{opt}\\ 
-1
\end{matrix}
\right],
\qquad\hbox{and}\qquad
\y_2 = 
\left[
\begin{matrix}
\tilde\x_{opt}\\
 -1
\end{matrix}
\right].
}
Note that 
\math{\matU_{\matY}\y_1=\matA\x_{opt}-\b},
\math{\matU_{\matY}\y_2=\matA\tilde\x_{opt}-\b},
\math{\matS\transp\Omega\transp\matU_{\matY}\y_1=\matS\transp\Omega\transp(\matA\x_{opt}-\b)},
and
\math{\matS\transp\Omega\transp\matU_{\matY}\y_2=\matS\transp\Omega\transp(\matA\tilde\x_{opt}-\b)}.
We need to bound \math{\norm{\matU_{\matY}\y_2}} in terms of \math{\norm{\matU_{\matY}\y_1}}:
\mand{
 \left(1-\sqrt{\frac{k}{r}}\right)^2 
\TNormS{\matU_{\matY} \y_2}
\mathop{\buildrel{(a)}\over{\leq}}
\TNormS{\matS\transp\Omega\transp\matU_{\matY}\y_2}
\mathop{\buildrel{(b)}\over{\leq}}
\TNormS{\matS\transp\Omega\transp\matU_{\matY}\y_1}
\mathop{\buildrel{(c)}\over{\leq}}
 \left(1+\sqrt{\frac{k}{r}}\right)^2 
\TNormS{\matU_{\matY} \y_1}.
}
(a) and (c) use Corollary \ref{cor:1set};
(b) follows because \math{\tilde\x_{opt}} is optimal for the
coreset regression.
After reorganization, using 
\math{k \le \rho + 1}, and assuming $r > 9(\rho+1)$,
$$
\TNormS{\matA \tilde{\x}_{opt} - \b} 
\le
\frac{\left(1+\sqrt{\frac{\rho+1}{r}}\right)^2}{ \left(1-\sqrt{\frac{\rho+1}{r}}\right)^{2}}\TNormS{\matA \x_{opt} - \b}
\le
  \left(  1 + 15 \sqrt{ \frac{ \rho+1 }{ r }} \right) \TNormS{\matA \x_{opt} - \b}  .  
$$
Setting $r = \ceil{\frac{225 (\rho + 1)}{\epsilon^2}}$ gives the bound in the Theorem. 
The overall running time is the sum of two terms: the time to 
compute \math{\matU_{\matY}} via the SVD and the time to run the $BarrierSamplingI$ method
on \math{\matU_{\matY}} (Lemma~\ref{cor:1set} with 
\math{k \le \rho+1}).
\end{proof}
Theorem \ref{lem:regression} improves on [Theorem 3.1, \cite{DMM06a}]
and [Theorem 5, \cite{DMM08}] which require coresets of size
$r = 3492 n^2 \ln(3/\delta)/\epsilon^2$ and 
$r = O(n \log(n) / \epsilon^2)$ to achieve $(1+\epsilon)$-error 
w.p. $1-\delta$ and $0.5$, respectively. The coresets of 
\cite{DMM06a,DMM08} are useful only for unconstrained regression. 
Both \cite{DMM06a,DMM08} and our result use the SVD,
so they are computationally comparable. 
It is interesting to note here that Algorithm \ref{alg:SimpReg} 
at Theorem \ref{lem:regression} applies the method of Section 
\ref{chap316} to the left singular vectors of $[\matA; \b] \in 
\R^{m \times (n+1)}$ whereas \cite{DMM06a, DMM08} 
work with the left singular vectors of $\matA \in \R^{m \times n}$
(see Section \ref{chap33}). 
Next, by applying $SubspaceSampling$ on the left singular
vectors of $\matY$ we obtain a coreset with high
probability and arbitrary constraints.

\section{Coreset Construction with Randomized Sampling}
The coreset construction algorithm of this section is similar with
the one presented in Section \ref{chap51} with the only difference
being the use of the randomized technique of Section \ref{chap314}
instead of the deterministic technique of Section \ref{chap316} that
we used in Algorithm \ref{alg:SimpReg}. Algorithm \ref{alg:SimpReg1}
finds a coreset with high probability for arbitrary constrained regression.
Previous randomized coreset construction algorithms~\cite{DMM06a, DMM08} can not
handle arbitrary constraints and succeed only with constant probability~\cite{DMM08}.
 
The algorithm of this section will serve as a prequel to the 
``coreset'' construction algorithm that we will present in Section \ref{chap52}.
The algorithm of Section \ref{chap52} is very similar with Algorithm \ref{alg:SimpReg1}
presented here, with the only difference being the fact that, before applying the 
$SubspaceSampling$ method on the left singular vectors of the matrix containing
both $\matA$ and $\b$, we will pre-multiply $\matA$ and $\b$ with the 
Randomized Hadamard Transform that we presented in Section \ref{chap3110}. 
The algorithm of Section \ref{chap52} works for arbitrary constraint regression
but succeeds only with constant probability, which is due to Lemma \ref{lem:HU} of
Section \ref{chap3110}.

\begin{theorem} 
\label{lem:regression2}
Given $\matA \in \R^{m \times n}$ of rank $\rho=n$, $\b \in \R^m$, $0 < \delta < 1$,
and~\math{ 0 \le \epsilon \le \frac{1}{3}}, Algorithm~\ref{alg:SimpReg1}
constructs
\math{\Omega\in\R^{m\times r}}, \math{\matS\in\R^{r\times r}} with 
$r = \ceil{\frac{ 36 (n + 1) \ln( 2 (n+1) / \delta ) }{\epsilon^2}}$
in time 
$O\left( mn^2 +   n \ln( n / \delta ) \cdot \epsilon^{-2} \cdot \log( n \ln( n / \delta ) \epsilon^{-1} ) \right)$ 
such that \math{\tilde\x_{opt}} obtained from  \math{\matS\transp\Omega\transp\matA,\ \matS\transp\Omega\transp\b} satisfies w.p. $1 - \delta$
$$
\TNormS{\matA \tilde{\x}_{opt} - \b} 
 \le \left( 1 + \epsilon \right) \TNormS{\matA \x_{opt} - \b};
$$
\end{theorem}
\begin{algorithm}[t]
\begin{framed}
\textbf{Input:} $\matA\in\R^{m\times n}$ of rank $\rho=n$, $\b \in \R^m$, $0 < \delta < 1$, 
and \math{0 < \epsilon < 1/3}. \\
\noindent \textbf{Output:} \math{\Omega \in \R^{m \times r}}, \math{\matS \in \R^{r \times r}}
with $r = \ceil{\frac{ 36 (n+ 1) \ln( 2 (n+1) / \delta ) }{\epsilon^2}}$.
\begin{algorithmic}[1]
\STATE Compute the matrix \math{\matU_{\matY}\in\R^{m\times k}}
of the top $k$ left singular vectors of \math{\matY=[\matA,\b] \in \R^{m \times (n+1)}};
\math{k = \rank(\matY) \le \rank(\matA) + 1 = \rho + 1}.
\STATE Let $r = \ceil{\frac{ 36 (n + 1) \ln( 2 (n+1) / \delta ) }{\epsilon^2}}$.
\STATE 
{\bf return} \math{[\Omega,\matS]=SubspaceSampling(\matU_{\matY},1, r)}
(Lemma~\ref{lem:random})
\end{algorithmic}
\caption{Randomized coreset for constrained regression.}
\label{alg:SimpReg1}
\end{framed}
\end{algorithm}
\begin{proof} Let 
$ \matY = [\matA,\b]\in\R^{m\times(n+1)}$, and let its SVD: 
\math{\matY=\matU_{\matY}\Sigma_{\matY}\matV_{\matY}\transp}, where
\math{\matU_{\matY}\in\R^{m\times k}}, \math{\Sigma_{\matY}\in\R^{k \times k}}
and
\math{\matV_{\matY}\in\R^{(n+1)\times k}}.
Here $k = \rank(\matY) \le \rho + 1 = \rank(\matA) + 1$. For the failure probability $\delta$ of the Theorem, let
$r > 4 k \ln(2k / \delta)$ be a sampling parameter whose exact value we will be specified later. 
Let
$$[\Omega, \matS] = SubspaceSamplingI(\matU_{\matY}, 1, r).$$
Let \math{\y_1,\y_2\in\R^{k}} defined as
\mand{
\y_1=\Sigma_{\matY}\matV_{\matY}\transp
\left[
\begin{matrix}
\x_{opt}\\ 
-1
\end{matrix}
\right],
\qquad\hbox{and}\qquad
\y_2 = 
\left[
\begin{matrix}
\tilde\x_{opt}\\
 -1
\end{matrix}
\right].
}
Note that 
\math{\matU_{\matY}\y_1=\matA\x_{opt}-\b},
\math{\matU_{\matY}\y_2=\matA\tilde\x_{opt}-\b},
\math{\matS\transp\Omega\transp\matU_{\matY}\y_1=\matS\transp\Omega\transp(\matA\x_{opt}-\b)},
and
\math{\matS\transp\Omega\transp\matU_{\matY}\y_2=\matS\transp\Omega\transp(\matA\tilde\x_{opt}-\b)}.
We need to bound \math{\norm{\matU_{\matY}\y_2}} in terms of \math{\norm{\matU_{\matY}\y_1}}:
\mand{
 \left(1-\sqrt{\frac{4 k \ln(2k/\delta)}{r}}\right) 
\TNormS{\matU_{\matY} \y_2}
\mathop{\buildrel{(a)}\over{\leq}}
\TNormS{\matS\transp\Omega\transp\matU_{\matY}\y_2}
\mathop{\buildrel{(b)}\over{\leq}}
\TNormS{\matS\transp\Omega\transp\matU_{\matY}\y_1}
\mathop{\buildrel{(c)}\over{\leq}} 
}
$$
 \left(1+\sqrt{\frac{4 k \ln(2k/\delta)}{r}}\right)
\TNormS{\matU_{\matY} \y_1}.
$$
(a) and (c) use Lemma \ref{lem:random};
(b) follows because \math{\tilde\x_{opt}} is optimal for the
coreset regression.
After reorganization, using 
\math{k \le \rho + 1}, and assuming $r > 36 k \ln( 2(\rho+1) / \delta )$,
$$
\TNormS{\matA \tilde{\x}_{opt} - \b} 
\le
  \left(  1 + 3 \sqrt{ \frac{ 4 (\rho+1) \ln( 2(\rho+1) / \delta ) }{ r }} \right) \TNormS{\matA \x_{opt} - \b}  .  
$$
Setting 
$$r = \ceil{\frac{ 36 (\rho + 1) \ln( 2 (\rho+1) / \delta ) }{\epsilon^2}}$$ 
gives the bound in the Theorem. 
The overall running time is the sum of two terms: the time to compute \math{\matU_{\matY}} via the SVD and the time to run the $SubspaceSampling$ method on \math{\matU_{\matY}} (Lemma~\ref{lem:random} with 
\math{k \le \rho+1}).
\end{proof}


\section{``Coreset'' Construction with the Hadamard Transform}\label{chap52}
Recall the discussion in Section \ref{chap3110}. First, notice that, since $\matH \matD$
is a square orthonormal matrix, for every $\x$,
$ \TNormS{\matA \x - \b} = \TNormS{ \matH \matD \matA \x - \matH \matD \b}.$
So, it suffices to approximate a solution to the residual 
$\TNormS{ \matH \matD \matA \x - \matH \matD \b}$. In what follows, we essentially
apply the randomized method of the previous section to the regression problem involving
$\matH \matD \matA$ and $\matH \matD \b$ instead of $\matA$ and $\b$. Our algorithm
here constructs coresets for the problem with $\matH \matD \matA$ and $\matH \matD \b$; 
unfortunately, this corresponds to ``coresets'' for the original problem involving $\matA$
and $\b$. The crux in the analysis of this section is that we do not need to compute
the left singular vectors of the matrix $\matY = [\matA,\b]\in\R^{m\times(n+1)}$, since,
due to Lemma \ref{lem:HU}, after multiplying $\matA$ and $\b$ with 
$ \matH \matD $, the norms of the rows of the left singular vectors of $\matY$ 
are essentially known; so, applying Definition \ref{def:sampling} and
Lemma \ref{lem:random} with uniform probabilities $q_i = 1/m$ suffices to 
preserve the singular values of the sub-sampled matrix $\matU_{\matY}$, which
is all we need to achieve in order to prove Theorem \ref{lem:regression2}, the main
result of this section. 
\begin{algorithm}[t]
\begin{framed}
\textbf{Input:} $\matA\in\R^{m\times n}$ of rank $\rho=n$, $\b \in \R^m$, $0 < \delta < 1$, 
and \math{0 < \epsilon < 1/3}. \\
\noindent \textbf{Output:} \math{\Omega \in \R^{m \times r}}, \math{\matS \in \R^{r \times r}};
$r = \ceil{\frac{ 72 (n + 1) \ln( 2 (n+1) / \delta ) \log(40 (n+1)m) }{\epsilon^2}}$
\begin{algorithmic}[1]
\STATE Compute the matrix \math{\matU_{\matY}\in\R^{m\times k}}
of the top $k$ left singular vectors of \math{\matY=[\matA,\b] \in \R^{m \times (n+1)}};
\math{k = \rank(\matY) \le \rank(\matA) + 1 = \rho + 1}.
\STATE Let $r = \ceil{\frac{ 72 (n + 1) \ln( 2 (n+1) / \delta ) \log(40 (n+1)m) }{\epsilon^2}}$
\STATE 
{\bf return} \math{[\Omega,\matS]=SubspaceSampling(\matU_{\matY}, \frac{1}{2 \log(40km)}, r)}
(Lemma~\ref{lem:random})
\end{algorithmic}
\caption{Randomized ``coreset'' for constrained regression.}
\label{alg:SimpReg2}
\end{framed}
\end{algorithm}
\begin{theorem} 
\label{lem:regression3}
Given $\matA \in \R^{m \times n}$ of rank $\rho=n$, $\b \in \R^m$, $0 < \delta < 0.95$,
and~\math{ 0 \le \epsilon \le \frac{1}{3}}, Algorithm~\ref{alg:SimpReg2}
constructs
\math{\Omega\in\R^{m\times r}}, \math{\matS\in\R^{r\times r}} with 
$r = \ceil{\frac{ 72 (n + 1) \ln( 2 (n+1) / \delta ) \log(40 (n+1)m) }{\epsilon^2}}$
in time 
$$O\left( m n \log\left( n \cdot \ln(n/\delta) \cdot \log( n m) \cdot \epsilon^{-1}   \right)  \right)$$
such that \math{\tilde\x_{opt} \in \R^n} obtained from  \math{\matS\transp\Omega\transp\matH \matD\matA,\ \matS\transp\Omega\transp\matH \matD\b} satisfies w.p. $0.95 - \delta$
$$
\TNormS{\matA \tilde{\x}_{opt} - \b} 
 \le \left( 1 + \epsilon \right) \TNormS{\matA \x_{opt} - \b}.
$$
\end{theorem}
\begin{proof} Let $ \matY = [ \matA, \b]\in\R^{m\times(n+1)}$, and let its SVD: 
\math{\matY=\matU_{\matY}\Sigma_{\matY}\matV_{\matY}\transp}, where
\math{\matU_{\matY}\in\R^{m\times k}}, \math{\Sigma_{\matY}\in\R^{k \times k}}
and
\math{\matV_{\matY}\in\R^{(n+1)\times k}}.
Let $ \hat{\matY} = [\matH \matD \matA,\matH \matD \b] = \matH \matD \matY \in\R^{m\times(n+1)}$, and let its SVD: 
\math{\hat{\matY}=\matU_{\hat{\matY}}\Sigma_{\hat{\matY}}\matV_{\hat{\matY}}\transp}, where
\math{\matU_{\hat{\matY}}\in\R^{m\times k}}, \math{\Sigma_{\hat{\matY}}\in\R^{k \times k}}
and
\math{\matV_{\hat{\matY}}\in\R^{(n+1)\times k}}.
It is important to note that $ \matU_{\hat{\matY}} = \matH \matD \matU_{\matY}$, 
$\Sigma_{\hat{\matY}} = \Sigma_{\matY}$, and $\matV_{\hat{\matY}} = \matV_{\matY}$.
Here $k = \rank(\hat{\matY}) \rho + 1 = \rank(\matA) + 1$. For the failure probability $\delta$ of the Theorem, let
$r > 4 k \ln(2k / \delta)$ be a sampling parameter whose exact value we will be specified later. 
Let $[\Omega, \matS] = SubspaceSamplingI(\matU_{\hat{\matY}}, \frac{1}{2 \log(40km)}, r)$ and
\mand{
\y_1=\Sigma_{\hat{\matY}}\matV_{\hat{\matY}}\transp
\left[
\begin{matrix}
\x_{opt}\\ 
-1
\end{matrix}
\right] \in\R^{k},
\qquad\hbox{and}\qquad
\y_2 = 
\left[
\begin{matrix}
\tilde\x_{opt}\\
 -1
\end{matrix}
\right] \in\R^{k}.
}
Note that 
\math{\matU_{\hat{\matY}}\y_1=\matH \matD\matA\x_{opt}-\matH \matD\b},
\math{\matU_{\hat{\matY}}\y_2=\matH \matD\matA\tilde\x_{opt}-\matH \matD\b},
\math{\matS\transp\Omega\transp\matU_{\hat{\matY}}\y_1=\matS\transp\Omega\transp(\matH \matD\matA\x_{opt}-\matH \matD\b)},
and
\math{\matS\transp\Omega\transp\matU_{\hat{\matY}}\y_2=\matS\transp\Omega\transp(\matH \matD\matA\tilde\x_{opt}-\matH \matD\b)}.
Now: \mand{
 \left(1-\sqrt{\frac{8 k \ln(2k/\delta) \log(40km) }{r}}\right) 
\TNormS{\matU_{\hat{\matY}} \y_2}
\mathop{\buildrel{(a)}\over{\leq}}
\TNormS{\matS\transp\Omega\transp\matU_{\hat{\matY}}\y_2}
\mathop{\buildrel{(b)}\over{\leq}}
\TNormS{\matS\transp\Omega\transp\matU_{\hat{\matY}}\y_1}
\mathop{\buildrel{(c)}\over{\leq}} 
}
$$
 \left(1+\sqrt{\frac{8 k \ln(2k/\delta) \log(40km) }{r}}\right)
\TNormS{\matU_{\hat{\matY}} \y_1}.
$$
(a) and (c) use Lemma \ref{preserves};
(b) follows because \math{\tilde\x_{opt}} is optimal for the
coreset regression.
After reorganization, using \math{k \le \rho + 1}, and assuming $r > 72 k \ln(2k/\delta) \log(40km) $,
$$
\TNormS{\matH \matD\matA \tilde{\x}_{opt} - \matH \matD\b} 
\le
  \left(  1 + 3 \sqrt{\frac{8 k \ln(2k/\delta) \log(40km) }{r}} \right) 
  \TNormS{\matH \matD\matA \x_{opt} - \matH \matD\b}  .  
$$
Setting 
$$r = \ceil{\frac{ 72 (\rho + 1) \ln( 2 (\rho+1) / \delta ) \log(40 k m) }{\epsilon^2}}$$ 
and removing $\matH \matD$ from
both sides gives the bound in the Theorem. 
The overall running time is the sum of two terms: the time to compute $ \hat{\matY} =  \matH \matD [ \matA, \b] $
(Lemma~\ref{fast}) and the time to run the $SubspaceSampling$ method on \math{\matU_{\hat{\matY}}} (Lemma~\ref{lem:random} with \math{k \le \rho+1}): 
$O\left( m n \log\left( n \cdot \ln(n/\delta) \cdot \log( n m) \cdot \epsilon^{-1}   \right)  \right).$ 
\end{proof}         

\chapter{FEATURE SELECTION IN $K$-MEANS CLUSTERING}
\label{chap6}
\footnotetext[7]{Portions of this chapter previously appeared as:
C. Boutsidis, M.W. Mahoney and P. Drineas, Unsupervised Feature
Selection for the $k$-means Clustering Problem, Advances in Neural Information Processing Systems (NIPS), 2009,
and as: C. Boutsidis, A. Zouzias and P. Drineas, Random Projections for $k$-means Clustering, 
Advances in Neural Information Processing Systems (NIPS), 2010.
}
Consider $m$ points $\mathcal{P} = \{ p_1, p_2, ..., p_m \} \in \R^n$, and integer
$k$ denoting the number of clusters. The objective of 
$k$-means is to find a $k$-partition of 
$\mathcal{P}$ such that points that are ``close'' to each other belong to the same cluster and points 
that are ``far'' from each other belong to different clusters. 
A $k$-partition of $\mathcal{P}$ is a collection 
$\cl S=\{\mathcal{S}_1, \mathcal{S}_2, ..., \mathcal{S}_k\}$
of \math{k} non-empty pairwise disjoint
sets which covers \math{\cl P}.
Let $s_j=|\mathcal{S}_j|$ be the size of $\mathcal{S}_j$. 
For each set \math{S_j}, let \math{\bm\mu_j\in\R^n} be its centroid 
(the mean point):
\math{\bm\mu_j=(\sum_{p_i\in S_j}p_i)/s_j}.
The $k$-means objective function is 
$$
\mathcal{F}(\mathcal{P}, \cl S) = 
\sum_{i=1}^m\norm{p_i-\bm\mu(p_i)}_2^2,
$$
where \math{\bm\mu(p_i)} is the centroid of the cluster to which 
\math{p_i} belongs.
The goal of $k$-means is to find a partition
\math{\cl S} which minimizes \math{\cl F}.

For the remainder of this chapter, we will switch to a more convenient
linear algebraic formulation of the $k$-means clustering problem.
Define the data matrix
\math{\matA \in \R^{m \times n}}, which has the data points for
 its rows:
\math{\matA\transp=[p_1,\ldots,p_n]}.
We represent a clustering \math{\cl S} by its cluster indicator matrix
$\matX \in \R^{m \times k}$.
Each column \math{j=1,\ldots,k} of \math{\matX} represents a cluster. 
Each row \math{i=1,\ldots,m}
indicates the
cluster membership of the point \math{p_i}. So, 
\math{\matX_{ij}=1/\sqrt{s_j}} if and
only if data point \math{p_i} is in cluster \math{S_j}. Every row of $\matX$ has 
exactly one non-zero element, corresponding to the cluster the data 
point belongs to. There are \math{s_j} non-zero elements in column 
\math{j} which indicates the data points belonging to cluster 
\math{S_j}. 
To see that the two formulations are equivalent, 
$$ \cl F(\matA,\matX) = 
\norm{\matA - \matX \matX\transp \matA}_F^2
= \sum_{i=1}^{m} \norm{ p_i\transp - p_i\transp \matX\transp\matA }_2^2
= \sum_{i=1}^m\norm{p_i\transp-\bm\mu(p_i)\transp}_2^2
= \mathcal{F}(\mathcal{P}, \cl S).
$$ 
After some
elementary algebra, one can verify that
for $i=1,...,m,$ $p_{i}\transp \matX\transp\matA=\bm\mu(p_i)\transp$.
Using this formulation, the goal of $k$-means is to find an indicator 
matrix $\matX$ which minimizes \math{\norm{\matA - \matX \matX\transp \matA}_F^2}.

To evaluate the quality of different clusterings, 
without access to a ``ground truth'' partitioning (labels),
we will use the \math{k}-means objective function. Given some
clustering \math{\hat\matX}, we are interested
in the ratio \math{\cl F(\matA,\hat\matX)/\cl F(\matA,\matX_{opt})},
where \math{\matX_{opt}} is the optimal clustering. The choice 
of evaluating a clustering this way is not new: 
all~\cite{OR99,KSS04,HM04,PK05,FS06,ORSS06,AV07} provide results along the same lines. 

Below, we give the formal definitions of the $k$-means problem and a $k$-means 
approximation algorithm. Recall that our primarily goal in this thesis is to develop
techniques that select features  from the data; we do not design algorithms to cluster the data per se. 
A $k$-means approximation algorithm is useful in our discussion since it will be used 
to evaluate the quality of the clusterings that can be obtained after our feature selection techniques, 
so we include this definition as well. 
\begin{definition}
\textsc{[The k-means clustering problem]}
Given $\matA \in \mathbb{R}^{m \times n}$ (representing $m$
points -- rows -- described with respect to $n$ features --
columns) and a positive integer $k$ denoting the number of
clusters, find the indicator matrix $\matX_{opt} \in \R^{m \times k}$:
$$
\matX_{opt} = \arg \min_{\matX \in \cal{X}} \FNormS{\matA - \matX \matX\transp \matA}.
$$
The optimal value of the $k$-means clustering objective is
$$
\cl F(\matA,\matX_{opt}) = \min_{\matX \in \cal{X}} \FNormS{\matA - \matX \matX\transp \matA} = \FNormS{\matA - \matX_{opt} \matX_{opt}\transp\matA} = \cl F_{opt}.$$
In the above, $\cal{X}$ denotes the set of all $m \times k$ indicator matrices $\matX$.
\end{definition}
\begin{definition} \label{def:approx}
\textsc{[k-means approximation algorithm]}
An algorithm is a ``$\gamma$-approximation'' for the $k$-means clustering problem ($\gamma \geq 1$) if
it takes inputs $\matA$ and $k$, and returns an indicator matrix $\matX_{\gamma}$ such that
w.p. $1 - \delta_{\gamma}$:
\vspace{-.15in}
$$
\FNormS{\matA - \matX_{\gamma} \matX_{\gamma}\transp \matA} \leq \gamma \min_{\matX \in
\cal{X}} \FNormS{\matA - \matX \matX\transp \matA} = \gamma \cl F(\matA,\matX_{opt}) = \gamma \cl F_{opt}.
$$
\end{definition}
\vspace{-.1in}
\noindent 
An example of such an algorithm is~\cite{KSS04}
with $\gamma = 1+\epsilon$ ($0 < \epsilon < 1$), 
and $\delta_{\gamma}$ some constant in $(0,1)$. 
This method runs in $O(mn 2^{(k/\epsilon)^{O(1)}})$.


\begin{algorithm}[t]
\begin{framed}
\textbf{Input:} Dataset $\matA\in\R^{m\times n}$, number of clusters $k$, and \math{0 < \epsilon < \frac{1}{3}}. \\
\noindent \textbf{Output:} \math{\matC  \in \R^{m \times r}} with $r = O(k \log(k) / \epsilon^2)$ rescaled features.
\begin{algorithmic}[1]
\STATE Let $ \matZ = FastFrobeniusSVD(\matA, k, \epsilon)$    (Lemma~\ref{tropp2}).
\STATE Let $r = c_0 \cdot 4 k \ln(200k)/\epsilon^2$ ($c_0$ is a sufficiently large constant).
\STATE Let $[\Omega, \matS] = SubspaceSampling(\matZ,  1, r)$ (Section~\ref{chap31}).
\STATE Return $\matC = \matA \Omega \matS \in \R^{m \times r}$ with $r$ rescaled columns from $\matA$. 
\end{algorithmic}
\caption{Randomized Feature Selection for $k$-means Clustering.}
\label{alg:chap61}
\end{framed}
\end{algorithm}

\section{Feature Selection with Randomized Sampling}\label{chap61}
Given $\matA, k$, and $0 < \epsilon < 1/3$, Algorithm~\ref{alg:chap61}
is our main algorithm for feature selection in $k$-means clustering. 
In a nutshell, construct the matrix $\matZ$ with the (approximate) top-$k$ right
singular vectors of $\matA$ and select $r = O(k \log(k)/\epsilon^2)$ columns 
from $\matZ\transp$ with the randomized technique of Section~\ref{chap314}.
One can replace the first step in Algorithm~\ref{alg:chap61} with the exact SVD of
$\matA$; the result though is asymptotically the
same as the one we will present in Theorem~\ref{fastkmeans}. 
Working with the approximate singular vectors $\matZ$ gives a considerably faster algorithm. 
\begin{theorem}\label{fastkmeans}
Let $\matA \in \R^{m \times n}$ and $k$ are inputs of the $k$-means clustering problem.
Let $\epsilon \in (0,1/3)$ and, by using Algorithm~\ref{alg:chap61} 
in $O( m n k/\epsilon + k \ln(k)/\epsilon^2\log(k \ln(k)/\epsilon) )$
construct features $\matC \in \R^{m \times r}$ with
$r = O(k \log(k) / \epsilon^2)$. Run any
$\gamma$-approximation $k$-means algorithm on $\matC, k$ and construct $\matX_{\tilde{\gamma}}$.
Then w.p. $0.21 - \delta_{\gamma}$:
$$
\FNorm{\matA - \matX_{\tilde{\gamma}} \matX_{\tilde{\gamma}}\transp  \matA}^2
\leq \left(1+(2+\epsilon)\gamma\right) \FNorm{\matA - \matX_{opt}
\matX_{opt}\transp  \matA}^2.
$$
\end{theorem}
\noindent To prove Theorem \ref{fastkmeans}, we first need  
Lemma~\ref{lem:rsall}, which we prove in the Appendix. 
\begin{lemma}\label{lem:rsall}
Fix $\matA$, $k$, $0 < \epsilon < 1/3$ and  $0 < \delta < 1$.
Via Lemma~\ref{tropp2}, 
for some $r$, let $[\Omega, \matS] = SubspaceSampling(\matZ,  1, r).$
From Lemma~\ref{tropp2}: $\matA = \matA\matZ\matZ\transp + \matE$. Then:
\begin{enumerate}

\item For any $r > 0$ and w.p. $1 - \delta$:
$ \FNormS{ \matE \Omega \matS \matS\transp \Omega\transp \matZ } \le \frac{k}{\delta r} \FNormS{ \matE }$.

\item Let $r =  4 k \ln(2k/\delta)/\epsilon^2$; then
w.p.  $1 - \delta$:
$\TNorm{(\matZ\transp \Omega \matS)^+ - (\matZ\transp \Omega \matS)\transp} \le \frac{\epsilon}{\sqrt{1-\epsilon}}$ 

\item For some $\tilde{\matE} \in \R^{m \times n}$, let
$ \matA \matZ \matZ\transp = \matA \Omega \matS (\matZ\transp \Omega \matS)^+ \matZ\transp + \tilde{\matE}$;
then, if $r = 4 k \ln(2k/\delta)/\epsilon^2$, w.p. $1 - 3\delta$:
$ \FNorm{ \tilde{\matE} } \le \frac{ 1.6 \epsilon}{\sqrt{\delta}}  \FNorm{\matE}$. 
\end{enumerate}
\end{lemma}
\begin{proof} (of Theorem~\ref{fastkmeans})
We start by manipulating the term 
$\FNormS{\matA -\matX_{\tilde{\gamma}} \matX_{\tilde{\gamma}}\transp \matA}$. 
Replacing $\matA = \matB \matZ\transp + \matE$ and using
Matrix Pythagoras:
\vspace{-.12in}
\begin{eqnarray}
\label{eqn:f1tt} \FNormS{\matA - \matX_{\tilde{\gamma}} \matX_{\tilde{\gamma}}\transp
\matA} = \underbrace{\FNormS{(\matI_{m } - \matX_{\tilde{\gamma}}
\matX_{\tilde{\gamma}}\transp) \matB \matZ\transp}}_{\theta_3^2} + \underbrace{\FNormS{(\matI_{m }
- \matX_{\tilde{\gamma}} \matX_{\tilde{\gamma}}\transp) \matE }}_{\theta_4^2}.
\end{eqnarray}
We first bound the second term of eqn.~(\ref{eqn:f1tt}). Since
$\matI_{m }-\matX_{\tilde{\gamma}}\matX_{\tilde{\gamma}}\transp$ is a projector matrix,
it can be dropped without increasing the Frobenius norm. From Lemma~\ref{tropp2} and
Markov's inequality w.p. $0.99$: $\FNormS{\matE} \le (1+100\epsilon) \FNormS{\matA -\matA_k}.$
Note also that $\matX_{opt}\matX_{opt}\transp\matA$ has rank at most $k$; 
so, overall, w.p. $0.99$:
$$
\theta_4^2 \le (1+100\epsilon)\FNormS{\matA - \matA_k} \le (1+100\epsilon)\FNormS{\matA - \matX_{opt}\matX_{opt}\transp\matA}=   (1+100\epsilon) F_{opt}.
$$
We now bound the first term in eqn.~(\ref{eqn:f1tt}):
\vspace{-.12in}
\begin{eqnarray}
\label{t0cc} \theta_3
&\leq&  \FNorm{(\matI_{m } -\matX_{\tilde{\gamma}}\matX_{\tilde{\gamma}}\transp)
\matA \Omega \matS (\matZ\transp \Omega \matS)^+\matZ\transp} + \FNorm{\tilde{\matE}} \\
\label{t1cc}
&\leq&  \FNorm{(\matI_{m } - \matX_{\tilde{\gamma}}\matX_{\tilde{\gamma}}\transp)\matA \Omega \matS} 
\TNorm{(\matZ\transp \Omega \matS)^+} + \FNorm{\tilde{\matE}} \\
\label{t2cc}
&\leq& \sqrt{\gamma} \FNorm{(\matI_{m } - \matX_{opt}\matX_{opt}\transp)\matA \Omega \matS} 
\TNorm{(\matZ\transp \Omega \matS)^+}  + \FNorm{\tilde{\matE}} 
\end{eqnarray}
In eqn.~(\ref{t0cc}), we used the third statement of Lemma \ref{lem:rsall}, the triangle
inequality, and the fact that $\matI_m -
\tilde{\matX}_{\gamma}\tilde{\matX}_{\gamma}\transp$ is a projector matrix and
can be dropped without increasing a unitarily invariant norm. In
eqn.~(\ref{t1cc}), we used spectral submultiplicativity and the fact that $\matZ\transp$ 
can be dropped without changing the spectral norm. In eqn.~(\ref{t2cc}), we replaced
$\matX_{\tilde{\gamma}}$ by $\matX_{opt}$ and the factor $\sqrt{\gamma}$ appeared
in the first term. To better understand this step, notice that
$\matX_{\tilde{\gamma}}$ gives a $\gamma$-approximation to the optimal
$k$-means clustering of the matrix $\matC = \matA \Omega \matS$, so any other $m \times
k$ indicator matrix (e.g. $\matX_{opt}$):
$$\FNormS{\left(\matI_{m } - \matX_{\tilde{\gamma}} \matX_{\tilde{\gamma}}\transp\right) \matA \Omega \matS}
\leq \gamma \min_{\matX \in \cal{X}} \FNormS{(\matI_{m } - \matX \matX\transp) \matA \Omega \matS} \leq
\gamma \FNormS{\left(\matI_{m} - \matX_{opt} \matX_{opt}\transp\right) \matA \Omega \matS}.$$ 
By using Lemma \ref{lem:fnorm} with $\delta = 3/4$, Lemma~\ref{lem:random},
and the union bound on these two probabilistic events, w.p. $1 - \frac{3}{4} - \delta$:
\vspace{-0,13in}
$$ \FNorm{(\matI_{m} -\matX_{opt}\matX_{opt}\transp)\matA \Omega \matS} 
\TNorm{(\matZ\transp \Omega \matS)^+} \leq  \sqrt{ \frac{4}{3-3\epsilon} F_{opt}}.$$
We are now in position to bound $\theta_3$; set $\delta=0.01$. Assuming $1 \leq \gamma$: 
$$ \theta_3 \leq \left(  \sqrt{\frac{4}{3-3\epsilon}} + \frac{1.6 \epsilon \sqrt{1+100\epsilon}}{\sqrt{0.01}}\right) \sqrt{\gamma } \sqrt{F_{opt}} 
\leq \left(  \sqrt{2} + 94 \epsilon  \right) \sqrt{\gamma } \sqrt{F_{opt}}.
$$
This bound holds w.p. $1 - \frac{3}{4} - 3 \cdot 0.01 - 0.01$ due to the third statement of Lemma~\ref{lem:rsall}
and the fact that $\FNormS{\matE} \le (1+100\epsilon) \FNormS{\matA -\matA_k}$ w.p. $0.99$.
The last inequality follows from our choice of $\epsilon < 1/3$ and elementary algebra. Taking squares on both sides:
$$ \theta_{3}^2 \le (2 + 9162 \epsilon) \gamma F_{opt}.$$  
Overall (assuming $\gamma \ge 1$):
$$ \FNormS{\matA -\matX_{\tilde{\gamma}} \matX_{\tilde{\gamma}}\transp \matA} \le \theta_{3}^2 + \theta_{4}^2 \le
(2 + 9162 \epsilon) \gamma F_{opt} + (1+100\epsilon) F_{opt} \le F_{opt} + ( 2 + 10^3 \epsilon  )\gamma F_{opt}.
$$ 
Rescaling $\epsilon$ accordingly ($c_0= 10^6$) gives the bound in the Theorem. The failure probability
follows by a union bound on all the probabilistic events involved in the proof
of this theorem. Indeed, $\frac{3}{4} + 3 \cdot 0.01 + 0.01 + \delta_{\gamma}= 0.79 + \delta_{\gamma}$ 
is the overall failure probability. 
\end{proof}

\section*{Existence of $O(k/\epsilon^2)$ ``good'' features}
Theorem~\ref{fastkmeans} proved that 
$O( k \log(k) / \epsilon^2 )$ features suffice to preserve the
clustering structure of a dataset within a factor of $3+\epsilon$.
The technique that we used to select the features breaks for $r = o(k\log(k))$
(due to Lemma~\ref{lem:random} in Section~\ref{chap314}).
Recently, in~\cite{BM11}, by using more sophisticated methods 
(extensions of the techniques of Sections \ref{chap317} and \ref{chap318}
to sample columns from multiple matrices simultaneously),
we proved the existence of $O(k/\epsilon^2)$ features with approximation $8+\epsilon$.
Unhappily, to find these features, the (deterministic) algorithm in~\cite{BM11} requires knowledge 
of the optimum partition $\matX_{opt}$, which is not realistic in any practical application. 
On the positive side, we managed to get a (randomized) algorithmic version with $O(k / \epsilon^2)$ features but this 
gives only a $O(\log(k))$ approximation. An interesting open question is whether one can get
constructively $(1+\epsilon)$ or some constant-factor 
approximation with $O(k/\epsilon^2)$ or even $O(k/\epsilon)$ features
(hopefully with a ``small'' hidden constant). 

\section{Feature Extraction with Random Projections}\label{chap62}
 
\begin{algorithm}[t]
\begin{framed}
\textbf{Input:} Dataset $\matA\in\R^{m\times n}$, number of clusters $k$, and \math{0 < \epsilon < \frac{1}{3}}. \\
\noindent \textbf{Output:} \math{\matC  \in \R^{m \times r}} with $r = O(k / \epsilon^2)$ artificial features.
\begin{algorithmic}[1]
\STATE Set  $r =  c_0 k /\epsilon^2$ for a sufficiently large constant $c_0$.
\STATE Compute a random $n \times r$ matrix $R$ as follows. For all $i=1,...,n$, $j=1,...,r$ (i.i.d)
   \[ R_{ij} = \begin{cases}
       +1/\sqrt{r}, \text{w.p. 1/2},\\
      -1/\sqrt{r}, \text{w.p. 1/2}.
\end{cases} \]
\STATE Compute $\matC = \matA R $ with the Mailman Algorithm (see text). 
\STATE Return $\matC\in \R^{m \times r}$. 
\end{algorithmic}
\caption{Randomized Feature Extraction for $k$-means Clustering.}
\label{alg:chap62}
\end{framed}
\end{algorithm} 
We prove that any set of $m$ points in $n$
dimensions (rows in a matrix $\matA \in \R^{m \times n}$) can be projected into $r =
O(k / \epsilon^2)$ dimensions, for any $\epsilon \in (0,1/3)$, in
$O(m n \lceil \epsilon^{-2} k/ \log(n) \rceil )$ time, such that, with 
constant probability, the optimal $k$-partition of the points
is preserved within a factor of $2+\epsilon$. The projection is
done by post-multiplying $\matA$ with an $n \times r$ random
matrix $R$ having entries $+1/\sqrt{r}$ or $-1/\sqrt{r}$ with equal probability.
More specifically, on input $\matA, k$, and $\epsilon$, we construct $\matC \in \R^{m \times r}$
with Algorithm~\ref{alg:chap62}.

{\bf Running time.} The algorithm needs $O(m k /\epsilon^2)$ time to generate $R$; then,
the product $\matA R$ can be naively computed in $O(mnk/\epsilon^2)$. 
One though can employ the so-called mailman algorithm for matrix
multiplication~\cite{LZ09} and compute the product $\matA R$ in $O(m n \lceil \epsilon^{-2} k/ \log(n) \rceil )$. 
Indeed, the mailman algorithm computes (after preprocessing) a matrix-vector product of any $n$-dimensional vector (row of $\matA$) with an 
$n \times \log(n)$ sign matrix in $O(n)$ time. 
Reading the input $n \times \log n$ sign matrix requires $O(n\log n)$ time. However, in our case we only consider 
multiplication with a random sign matrix, therefore we can avoid the preprocessing step by directly computing 
a random correspondence matrix as discussed in~\cite[Preprocessing Section]{LZ09}.
By partitioning the columns of our $n \times r$ matrix $R$ into 
$\lceil r/\log(n)\rceil$ blocks, the claim follows. 

{\bf Analysis.} Theorem \ref{thm:second_result} is our quality-of-approximation-result 
regarding the clustering that can be obtained with the features returned
from the above algorithm. 
Notice that if $\gamma = 1$, the distortion is at most $2+\epsilon$, as advertised.
If the $\gamma$-approximation algorithm is~\cite{KSS04} the overall approximation factor
would be $(1 + (1+\epsilon)^2)$ with running time
$O( m n \lceil \epsilon^{-2} k/ \log(n) \rceil + 2^{(k/\epsilon)^{O(1)}} m k / \epsilon^2 )$.
\begin{theorem}\label{thm:second_result}
Let $\matA \in \R^{m \times n}$ and $k$ are inputs of the $k$-means clustering problem.
Let $\epsilon \in (0,1/3)$ and construct features $\matC \in \R^{m \times r}$ with
$r = O(k / \epsilon^2)$
by using Algorithm~\ref{alg:chap62} in $O(m n \lceil \epsilon^{-2} k/ \log(n) \rceil )$. 
Run any $\gamma$-approximation $k$-means algorithm on $\matC, k$ and construct $\matX_{\tilde{\gamma}}$.
Then w.p. $0.95 - \delta_{\gamma}$:
$$
\FNorm{\matA - \matX_{\tilde{\gamma}} \matX_{\tilde{\gamma}}\transp  \matA}^2
\leq \left(1+(1+\epsilon)\gamma\right) \FNorm{\matA - \matX_{opt}
\matX_{opt}\transp  \matA}^2.
$$
\end{theorem}

\begin{proof}
We start by manipulating the term 
$\FNorm{\matA - \matX_{\tilde{\gamma}} \matX_{\tilde{\gamma}}\transp  \matA}^2$. 
Replacing $\matA = \matA_k + \matA_{\rho-k}$ and
using Matrix Pythagoras:
\begin{eqnarray}\label{eqn:f1}
\FNorm{\matA - \matX_{\tilde{\gamma}} \matX_{\tilde{\gamma}}\transp \matA}^2\ =\ \underbrace{\FNorm{(\matI_m -
\matX_{\tilde{\gamma}} \matX_{\tilde{\gamma}}\transp ) \matA_k}^2}_{\theta_1^2}\ +\ \underbrace{\FNorm{(\matI_m - \matX_{\tilde{\gamma}}
\matX_{\tilde{\gamma}}\transp )\matA_{\rho-k}}^2}_{\theta_2^2}.
\end{eqnarray}
We first bound the second term of Eqn.~\eqref{eqn:f1}. Since
$\matI_m-\matX_{\tilde{\gamma}}\matX_{\tilde{\gamma}}\transp $ is a projector
matrix, it can be dropped without increasing the Frobenius norm. So, by using
this and the fact that $\matX_{opt}\matX_{opt}\transp \matA$ has rank at most $k$:
\begin{eqnarray} \label{eqn:f2}
\theta_2^2\ \leq\ \FNorm{\matA_{\rho-k}}^2\  = \FNorm{\matA - \matA_k}^2\ \leq\ \FNorm{ \matA -
\matX_{opt}\matX_{opt}\transp \matA }^2.
\end{eqnarray}
We now bound the first term of Eqn.~\eqref{eqn:f1}:
\begin{eqnarray}
\label{t0} \theta_1
&\leq&  \FNorm{(\matI_m -\matX_{\tilde{\gamma}}\matX_{\tilde{\gamma}}\transp )\matA R(\matV_k R)^+\matV_k\transp }\ +\ \FNorm{\matE} \\
\label{t1}
&\leq&  \FNorm{(\matI_m -\matX_{\tilde{\gamma}}\matX_{\tilde{\gamma}}\transp )\matA R} \FNorm{(\matV_k R)^+}\ +\ \FNorm{\matE}\\
\label{t2}
&\leq& \sqrt{\gamma} \FNorm{(\matI_m - \matX_{opt}\matX_{opt}\transp )\matA R} \FNorm{(\matV_k R)^+}\  +\ \FNorm{\matE} \\
\label{t3}
&\leq& \sqrt{\gamma}  \sqrt{(1+\epsilon)} \FNorm{(\matI_m - \matX_{opt}\matX_{opt}\transp )\matA}  \frac{1}{1-\epsilon}\ +\ 4 \epsilon \FNorm{(\matI_m -\matX_{opt}\matX_{opt}\transp )\matA} \\
\label{t4}
&\leq& \sqrt{\gamma}  (1+2.5 \epsilon) \FNorm{(\matI_m -\matX_{opt}\matX_{opt}\transp )\matA}\ + \sqrt{\gamma} \ 4 \epsilon \FNorm{(\matI_m -\matX_{opt}\matX_{opt}\transp )\matA}\\
\label{t5}
&\leq& \sqrt{\gamma} ( 1 + 6.5 \epsilon) \FNorm{(\matI_m
-\matX_{opt}\matX_{opt}\transp )\matA}
\end{eqnarray}
In Eqn.~\eqref{t0}, we used the sixth statement of Lemma \ref{lem:rpall}, the
triangle inequality for matrix norms, and the fact that $\matI_m -
\tilde{\matX}_{\gamma}\tilde{\matX}_{\gamma}\transp $ is a projector matrix
and can be dropped without increasing the Frobenius norm.
In Eqn.~\eqref{t1}, we used spectral submultiplicativity 
and the fact that $\matV_k\transp $ can be dropped
without changing the spectral norm. 
In Eqn.~\eqref{t2}, we replaced
$\matX_{\tilde{\gamma}}$ by $\matX_{opt}$ and the factor $\sqrt{\gamma}$
appeared in the first term. To better understand this step, notice
that $\matX_{\tilde{\gamma}}$ gives a $\gamma$-approximation to the
optimal $k$-means clustering of the matrix $\matC$, and any other $m
\times k$ indicator matrix (for example, the matrix $\matX_{opt}$)
satisfies
\begin{equation*}
\FNorm{\left(\matI_m - \matX_{\tilde{\gamma}} \matX_{\tilde{\gamma}}\transp
\right) \matC}^2 \leq\ \gamma\ \min_{\matX \in \cal{X}} \FNorm{(\matI_m - \matX
\matX\transp ) \matC}^2\ \leq \gamma \FNorm{\left(\matI_m - \matX_{opt}
\matX_{opt}\transp \right) \matC}^2.
\end{equation*}
In Eqn.~\eqref{t3}, we used the fifth statement of Lemma~\ref{lem:rpall} with
$ \matX= (\matI - \matX_{opt} \matX_{opt}\transp )\matA$, the first statement of 
Lemma \ref{lem:rpall} and the optimality of SVD. 
In Eqn.~\eqref{t4}, we used the fact that
$\gamma \geq 1$ and that for any $\epsilon \in (0,1/3)$ it is
$(\sqrt{1+\epsilon})/(1-\epsilon) \leq 1+ 2.5\epsilon$. 
Taking squares in
Eqn.~\eqref{t5}, we get
\[ \theta_1^2\ \leq\ \gamma ( 1 + 28 \epsilon ) \FNorm{(\matI_m
-\matX_{opt}\matX_{opt}\transp )\matA}^2.\]
Rescaling $\epsilon$ accordingly gives the bound in the theorem.
The failure probability $0.05 + \delta_{\gamma}$ follows by a
union bound on all four probabilistic events involved in the proof
of the theorem which are the $\gamma$-approximation $k$-means
algorithm with failure probability $\delta_{\gamma}$ and the
first, fourth, and sixth statements of Lemma~\ref{lem:rpall}
with failure probabilities $0.01$, $0.01$, and $0.03$, respectively. 
\end{proof}

\begin{algorithm}[t]
\begin{framed}
\textbf{Input:} Dataset $\matA\in\R^{m\times n}$, number of clusters $k$, and \math{0 < \epsilon < 1}. \\
\noindent \textbf{Output:} \math{\matC \in \R^{m \times k}} with $k$ artificial features.
\begin{algorithmic}[1]
\STATE Let $ \matZ = FastFrobeniusSVD(\matA, k, \epsilon)$    (Lemma~\ref{tropp2}).
\STATE Return $\matC = \matA \matZ \in \R^{m \times k}$.
\end{algorithmic}
\caption{Randomized Feature Extraction for $k$-means Clustering.}
\label{alg:chap63}
\end{framed}
\end{algorithm}
\section{Feature Extraction with Approximate SVD}\label{chap63}
Finally, we present a feature extraction algorithm that employs the SVD
to construct $r = k$ artificial features. Our method and proof technique
are the same with those of~\cite{DFKVV99} with the only difference being the fact that we
use a fast approximate randomized SVD from Lemma~\ref{tropp2} as opposed
to the expensive exact deterministic SVD. Our choice gives a considerably
faster algorithm with approximation error similar to the approach in~\cite{DFKVV99}.
\begin{theorem}\label{thm:first_result}
Let $\matA \in \R^{m \times n}$ and $k$ are inputs of the $k$-means clustering problem.
Let $\epsilon \in (0,1)$ and construct features $\matC \in \R^{m \times k}$ 
by using Algorithm~\ref{alg:chap63} in $O(m n k / \epsilon )$ time. 
Run any $\gamma$-approximation $k$-means algorithm on $\matC, k$ and construct $\matX_{\tilde{\gamma}}$.
Then w.p. $0.99 - \delta_{\gamma}$:
$$
\FNorm{\matA - \matX_{\tilde{\gamma}} \matX_{\tilde{\gamma}}\transp  \matA}^2
\leq \left(1+(1+\epsilon)\gamma\right) \FNorm{\matA - \matX_{opt}
\matX_{opt}\transp  \matA}^2.
$$
\end{theorem}
\begin{proof} 
We start by manipulating the term 
$\FNormS{\matA -\matX_{\tilde{\gamma}} \matX_{\tilde{\gamma}}\transp \matA}$. 
Replacing $\matA = \matB \matZ\transp + \matE$ and using
Matrix Pythagoras:
\begin{eqnarray}
\label{eqn:f1tt1} \FNormS{\matA - \matX_{\tilde{\gamma}} \matX_{\tilde{\gamma}}\transp
\matA} = \underbrace{\FNormS{(\matI_{m } - \matX_{\tilde{\gamma}}
\matX_{\tilde{\gamma}}\transp) \matB \matZ\transp}}_{\theta_3^2} + \underbrace{\FNormS{(\matI_{m }
- \matX_{\tilde{\gamma}} \matX_{\tilde{\gamma}}\transp) \matE }}_{\theta_4^2}.
\end{eqnarray}
In the proof of Theorem~\ref{fastkmeans} we argued that w.p. $0.99$:
$$
\theta_4^2 \le (1+100\epsilon)\FNormS{\matA - \matA_k} \le (1+100\epsilon)\FNormS{\matA - \matX_{opt}\matX_{opt}\transp\matA}=   (1+100\epsilon) F_{opt}.
$$
We now bound the first term in eqn.~(\ref{eqn:f1tt1}):
\vspace{-.12in}
\begin{eqnarray}
\label{cb1} \theta_3
&\leq&  \FNorm{(\matI_{m } -\matX_{\tilde{\gamma}}\matX_{\tilde{\gamma}}\transp)
\matA \matZ \matZ\transp }  \\
\label{cb2}
&\leq&  \FNorm{(\matI_{m } - \matX_{\tilde{\gamma}}\matX_{\tilde{\gamma}}\transp)\matA \matZ}  \\
\label{cb3}
&\leq& \sqrt{\gamma} \FNorm{(\matI_{m } - \matX_{opt}\matX_{opt}\transp)\matA \matZ}\\
\label{cb4}
&\leq& \sqrt{\gamma} \FNorm{(\matI_{m } - \matX_{opt}\matX_{opt}\transp)\matA} 
\end{eqnarray}
In eqn.~(\ref{cb1}), we replaced $\matB = \matA \matZ$.
In eqn.~(\ref{cb2}), we used spectral submultiplicativity and the fact that $\TNorm{\matZ\transp}=1$.
In eqn.~(\ref{cb3}),  we replaced
$\matX_{\tilde{\gamma}}$ by $\matX_{opt}$ and the factor $\sqrt{\gamma}$ appeared
in the first term (similar argument as in the proof of Theorem~\ref{fastkmeans}).
In eqn.~(\ref{cb4}), we used spectral submultiplicativity and the fact that $\TNorm{\matZ}=1$.
 Overall (assuming $\gamma \ge 1$):
$$ \FNormS{\matA -\matX_{\tilde{\gamma}} \matX_{\tilde{\gamma}}\transp \matA} \le \theta_{3}^2 + \theta_{4}^2 \le
 \gamma F_{opt} + (1+100\epsilon) F_{opt} \le F_{opt} + (1 + 10^2 \epsilon  )\gamma F_{opt}.
$$ 
Rescale $\epsilon$ accordingly and use the union bound to wrap up.
\end{proof}

\chapter{FUTURE DIRECTIONS}\label{chap7}

This thesis presented a few aspects of the world of ``Matrix Sampling Algorithms''. 
A popular problem in this line of research is low-rank approximations of matrices
by using a small subset of their columns: given $\matA$ and $r$, among all $\binom{n}{r}$
matrices $\C$ with $r$ columns from $\matA$ find one with 
\vspace{-0.1in}
$$\matC^* = \argmin\nolimits_{\matC} \FNormS{\matA - \matC \matC^+\matA} .$$
We learned various approaches, from Additive-error algorithms as in Section \ref{chap311} 
to more sophisticated greedy techniques as in Chapter \ref{chap4}. 
We argued that the predominant goal of a ``Matrix Sampling Algorithm'' is to 
preserve the spectral structure of the top right
singular vectors of the matrix (Sections \ref{chap314} to \ref{chap318}). 

In another - parallel - line of research\footnote[8]{
``Topics in Sparse Approximations'' from Joel Tropp~\cite{Tro05} and 
``Topics in Compressed Sensing'' from Deanna Needell~\cite{Nee09}
give a nice overview of this area.}, 
``Sparse Approximation Algorithms'' are used for the following combinatorial problem:
given $\matA \in \R^{m \times n}$, $\b \in \R^m$, 
and $r < n$,
among all $\binom{n}{r}$ vectors \math{\x_r \in \R^n} with at most $r$ non-zero 
entries, find one with
\vspace{-0.1in}
$$\x_{r}^* = \argmin\nolimits_{\norm{\x}_0 \le r} \TNormS{\matA \x- \b} .$$
In words, this problem asks to find a subset of columns from $\matA$
(those columns that correspond to the non-zero elements in $\x_{r}^*$),
such that the solution vector obtained by using these columns is the best possible
among all $\binom{n}{r}$ such solutions. 

Notice that, in both problems, the goal is to find the  ``best'' columns
from the input matrix. What makes the problems different is the objective function
that these columns attempt to minimize. So, arguably, the problems from these two 
communities are not ``far'' from each other. Unfortunately, it appears that there 
is little, if no, overlap in algorithmic approaches. An interesting avenue for future 
research is to explore the applicability of 
``Matrix Sampling Algorithms'' for sparse approximation problems and vice versa. For example,
\begin{enumerate}
\item Are ``Matrix Sampling Algorithms'', such as those we presented in Section \ref{chap31},
useful in the context of sparse approximations? 
\item Are ``Sparse Approximation Algorithms'', such as, for example, the orthogonal matching pursuit
and the basis pursuit~\cite{Tro04},
useful in the context of low-rank column-based matrix approximations? 
\end{enumerate}
Towards this end, in~\cite{BDM11c} we made a little progress by leveraging the randomized
technique of Section \ref{chap314} and the deterministic technique of Section \ref{chap317}
to design novel sparse approximation algorithms. We now present
the result of~\cite{BDM11c}. 
\begin{theorem} \label{theorem1}
Fix $\matA \in \R^{m \times n}$, $\b \in \R^{m}$, and target rank $0 < k < n$.
For $r > 144 k \ln(20 k)$, let $[\Omega_1, \matS_1] = SubspaceSampling( \matV_k, 1, r)$,
and let $\x_r \in \R^n$ is the $r$-sparse vector obtained from $\matA \Omega_1$, $\b$. Then, with constant probability:
\vspace{-0.1in}
$$ 
\TNorm{\matA \x_r-\b}
\le
\TNorm{ \matA \matA_k^+\b  -\b} + \sqrt{\frac{36 k \ln(20 k)}{ r }} \frac{ \FNorm{\matA - \matA_k} }{ \sigma_k(\matA) }\norm{\b}_2. $$
For $ r > 36 k $, let $[\Omega_2, \matS_2]= BarrierSamplingIII( \matV_k, (\matA - \matA\matV_k\matV_k\transp)\transp, r)$,
and let $\x_r \in \R^n$ is the $r$-sparse vector obtained from $\matA \Omega_2$, $\b$. Then,
$$ 
\TNorm{\matA \x_r-\b}
\le
\TNorm{ \matA \matA_k^+\b  -\b} + (1+\sqrt{\frac{9 k }{ r }}) \frac{ \FNorm{\matA - \matA_k} }{ \sigma_k(\matA) }\norm{\b}_2. $$
\end{theorem}
Notice that the solution vectors obtained by our algorithms are evaluated with respect to the solution vector
obtained by the so-called Truncated SVD Regularized approach, $\x_k=\A_k^+\b$ ~\cite{Han87}. 
This is certainly interesting because the theorem essentially says that one can replace the $k$-SVD solution with a sparse solution with almost $k$ non-zero entries and get a comparable (up to an additive error term) performance. 
It would have been nice though to obtain bounds with respect to the optimum solution of the regression problem, 
$\x_{opt} = \matA^+ \b$. Although this is a much more harder bound to obtain, we believe that it is not quixotic
to seek algorithms with such properties. To obtain such bounds, it might require the development of new 
greedy-like methods (such as those of Sections \ref{chap316}, \ref{chap317}, and \ref{chap318}) that attempt to
optimize directly the objective of the underlying sparse approximation problem.

\paragraph{Low-rank Column-based Matrix Approximations.}
Consider approximations for the equation: 
$ \XNorm{\matA-\Pi_{\matC,k}^{\xi}(\matA)} \le \alpha \XNorm{\matA-\matA_k}.$ 
For $\xi = F$ and $r = k$, the problem is closed, modulo running time, since
\cite{DR10} provides a deterministic algorithm matching the known lower bound $\alpha = \hat\alpha = \sqrt{k+1}$.
For $\xi = 2$ and $r = k$, current algorithms are optimal up to a factor $O(k^2)$.
\cite{GE96} provides a deterministic algorithm with $\alpha = O(\sqrt{kn})$; the
lower bound though indicates that $\hat\alpha = \sqrt{n/k}$ might be possible, so there is some
hope to improve on this result. For $\xi = 2$ and $r > k$, the problem is closed, 
modulo running time and constants. We provided a deterministic algorithm with $\alpha = O(\sqrt{n/r})$
and the lower bound we proved is $\hat\alpha = \sqrt{n/r}$. For $\xi = 2$ and $r > k$, 
we provided a randomized algorithm which is optimal up to a constant $20$. It would be
nice though to design a deterministic algorithm with a better constant. 
\cite{GK11} provides a deterministic algorithm with approximation:
$$ \XNorm{\matA-\matC \matC^+\matA}  \le \sqrt{ \frac{r+1}{r+1-k}} \XNorm{\matA-\matA_k} $$
It would be nice to understand whether this result can be extended to the rank $k$
matrix $\Pi_{\matC,k}^{F}(\matA)$ as well. Moreover, a faster algorithm
than~\cite{GK11} should be possible. We believe that there exists such an algorithm running in SVD time.

\paragraph{Coreset Construction in Least-Squares Regression.}
We believe that the deterministic algorithm that we presented 
in Section~\ref{chap61} is optimal up to a factor $\frac{1}{\epsilon}$. 
It is known that relative error coresets of size $O(k)$ are not possible. 
Relative error coresets of size $O(k/\epsilon)$ should be possible.  
We hope such algorithms will become available soon; if not,
it would be nice to develop lower bounds for this problem.

\paragraph{Feature Selection in $k$-means.}
The feature selection algorithm of
Section~\ref{chap61} gives a constant factor approximation
to the optimal partition. 
We believe that relative error approximations should be
possible. Further, a deterministic approach would get
rid of the $\log(k)$ factor from the number of sampled
features, which is there due to Coupon Collector issues. 
We hope that deterministic relative-error algorithms with 
$O(k/\epsilon)$ features will become available soon;  
if not, it would be nice to understand the limitations
of this problem by developing lower bounds.           
\specialhead{BIBLIOGRAPHY}
\begin{singlespace}
\small{
\bibliographystyle{abbrv}
\bibliography{RPI_BIB}
}
\end{singlespace}

\appendix    
\addtocontents{toc}{\parindent0pt\vskip12pt APPENDICES} 

\chapter{TECHNICAL PROOFS}
\label{app:proofs}

\section{Proof of Lemma~\ref{lem:bestF}}
Our proof for the Frobenius norm case is a mild modification of the proof of Lemma 4.3 of~\cite{CW09}. First, note that $\Pi^{F}_{\matC,k}(\matA) = \Pi^{F}_{\matQ,k}(\matA)$, because $\matQ \in \R^{m \times r}$ is an orthonormal basis for the column space of~$\matC$. Thus,
$$\FNormS{\matA-\Pi^{F}_{\matC,k}(\matA)} = \FNormS{\matA-\Pi^{F}_{\matQ,k}(\matA)} = \min_{\Psi:\rank(\Psi)\leq k}\FNormS{\matA-\matQ\Psi}.$$
Now, using matrix-Pythagoras and the orthonormality of $\matQ$,
\vspace{-0.15in}
$$
\FNormS{\matA-\matQ\Psi}=\FNormS{\matA-\matQ\matQ\transp\matA+\matQ(\matQ\transp\matA-\Psi)}
=\FNormS{\matA-\matQ\matQ\transp\matA}
+\FNormS{\matQ\transp\matA-\Psi}.
$$
Setting $\Psi = (\matQ\transp\matA)_k$ minimizes the above quantity over all rank-$k$ matrices $\Psi$. Thus, combining the above results,
$\FNormS{\matA-\Pi^{F}_{\matC,k}(\matA)} = \FNormS{\matA - \matQ\left(\matQ\transp\matA\right)_k}$.

We now proceed to the spectral-norm part of the proof, which combines ideas from Theorem 9.3 of~\cite{HMT} and matrix-Pythagoras. Consider the following derivations:
\eqan{
\TNormS{\matA-\matQ\left(\matQ\transp\matA\right)_k}
&=&
\TNormS{\matA-\matQ\matQ\transp\matA+
\matQ\left(\matQ\transp\matA-(\matQ\transp\matA)_k\right)}\\
&\le&
\TNormS{\matA-\matQ\matQ\transp\matA}+
\TNormS{\matQ\matQ\transp\matA-(\matQ\matQ\transp\matA)_k}\\
&\buildrel(a)\over\le&
\TNormS{\matA-\Pi_{\matQ,k}^2(\matA)}+
\TNormS{\matA-\matA_k}\\
&\le&
2\TNormS{\matA-\Pi_{\matQ,k}^2(\matA)}.
}
The first inequality follows from the simple fact that $\left(\matQ\matQ\transp\matA\right)_k=\matQ\left(\matQ\transp\matA\right)_k$ and
matrix-Pythagoras; the first term in (a) follows because \math{\matQ\matQ\transp\matA} is the (unconstrained, not necessarily of rank at most $k$) best approximation to \math{\matA} in the column space of~\math{\matQ}; the second term in (a) follows because $\matQ\matQ\transp$ is a projector matrix, so:%
\vspace{-0.15in}
$$\TNormS{\matQ\matQ\transp\matA-(\matQ\matQ\transp\matA)_k}=\sigma_{k+1}^2(\matQ\matQ\transp\matA)\le\sigma_{k+1}^2(\matA)
=\TNormS{\matA-\matA_k}.$$
The last inequality follows because \math{\TNormS{\matA-\matA_k}\le\TNormS{\matA-\Pi_{\matQ,k}^2(\matA)}}. 

\section{Proof of Lemma~\ref{lem:genericNoSVD}}
The optimality of $\Pi^{\xi}_{\matC,k}(\matA)$ implies that 
$$\XNormS{\matA - \Pi^{\xi}_{\matC,k}(\matA)}
\leq \XNormS{\matA - \matX}$$ over all matrices $\matX \in \mathbb{R}^{m \times n}$ of rank at most $k$ in the column space of \math{\matC}.
Consider the matrix $$\matX = \matC\left(\matZ\transp \matW\right)^+ \matZ\transp.$$ 
Clearly \math{\matX} is in the column space of $\matC$ and $\rank(\matX)\le k$ because
\math{\matZ\in\R^{n\times k}}. Proceed as follows:
\eqan{
\XNormS{\matA - \matC(\matZ\transp \matW)^+ \matZ\transp}
&=&
\XNormS{\underbrace{\matB \matZ\transp + \left(\matA - \matB \matZ\transp\right)}_{\matA} - \underbrace{\left(\matB \matZ\transp + (\matA-\matB \matZ\transp)\right)\matW}_{\matC=\matA\matW} (\matZ\transp\matW)^+\matZ\transp }\\
&=&
\XNormS{\matB \matZ\transp - \matB \matZ\transp \matW (\matZ\transp\matW)^+\matZ\transp
+ \matE+\matE\matW (\matZ\transp\matW)^+\matZ\transp}\\
&\buildrel{(a)}\over{=}&
\XNormS{\matE+\matE\matW(\matZ\transp\matW)^+\matZ\transp}
\\
&\buildrel{(b)}\over{\leq}& \XNormS{\matE} + \XNormS{\matE\matW(\matZ\transp\matW)^+\matZ\transp}.
}
\math{(a)} follows because, by assumption, $$\rank(\matZ\transp \matW)=k,$$ and thus 
$$(\matZ\transp \matW) (\matZ\transp \matW)^+=\matI_{k},$$ which implies
$$ \matB \matZ\transp - \matB (\matZ\transp \matW) (\matZ\transp \matW)^+ \matZ\transp = \bm{0}_{m \times n}.$$
\math{(b)} follows by matrix-Pythagoras because
$$\matE\matW(\matZ\transp\matW)^+\matZ\transp \matE^T = \bm{0}_{m \times n}$$ 
(recall that $\matE = \matA-\matB\matZ^T$ and $\matE\matZ = \bm{0}_{m \times k}$ by assumption). 
The lemma follows by spectral submultiplicativity because $\matZ$ has orthonormal columns, hence  $\TNorm{\matZ}=1$.

\section{Proof of Lemma \ref{tropp1}}
Consider the following algorithm, described in Corollary 10.10 of~\cite{HMT}. The algorithm takes as inputs a matrix $\matA \in \R^{m \times n}$ of rank $\rho$, an integer $2 \leq k <\rho$, an integer $q \geq 1$, and an integer $p \geq 2$. Set $r = k+p$ and construct the matrix
$\matY \in \R^{m \times r}$ as follows:
\begin{enumerate}
\item Generate an $n \times r$ standard Gaussian matrix $\matR$ whose entries are i.i.d. $\mathcal{N}(0,1)$ variables.
\item Return $\matY = (\matA \matA\transp)^q \matA \matR \in \R^{m \times r}$.
\end{enumerate}
The running time of the above algorithm is $ O(mnrq )$.
Corollary 10.10 of~\cite{HMT} presents the following bound:
$$\Expect{\TNorm{\matA - \matY \matY^+ \matA }} \leq
\left(1 + \sqrt{\frac{k}{p-1}} + \frac{e \sqrt{k+p}}{p} \sqrt{ \min\{ m,n\} -k }
\right)^{\frac{1}{2q+1}}\TNorm{\matA - \matA_k},$$
where $e = 2.718\ldots$. 
This  result is not immediately applicable to the construction of a factorization of the form $\matA =
\matB \matZ\transp + \matE$ with $\matZ \in \R^{n \times k}$; 
the problem is with the dimension of $\matY \in R^{m \times r}$, since r > k\footnote{
We should note that \cite{HMT} contains a result that is directly applicable
for a factorization $\matB \matZ\transp$ with $\matZ \in \R^{n \times k}$; 
this is in eqn.~(1.11) in \cite{HMT}; one can use this result and get an approximation $2+\epsilon$
in Lemma \ref{tropp1}; the new analysis improves that to  $\sqrt{2}+\epsilon$. Also, it is possible 
to use the matrix $\matY$ to compute such a factorization since $\matZ$ can have more than $k$ columns. 
This will sacrifice the approximation bounds of our algorithms; for example, assume that we construct 
a matrix $\matZ$ with $2k$ columns; then, one needs to sample at least $r > 2k$ columns to guarantee the
rank assumption and the error bound, for example, in Theorem \ref{theorem:intro2} 
will become $1 + 1 / (1 - \sqrt{2k/r})$. Note that such a $\matZ$
can be computed by applying the algorithm on $\matA\transp$; then $\matY = \matZ$. Orthonormalizing
$\matZ$ is also not necessary (this requires a slightly different argument in Lemma \ref{lem:genericNoSVD}
that we do not discuss here). 
}. 
Lemma \ref{troppextension0} below, which  strengthens Corollary 10.10 in \cite{HMT}, argues that the matrix $\Pi_{\matY,k}^2(\matA)$ ``contains'' the desired factorization $\matB \matZ\transp$. Recall that while we cannot compute $\Pi_{\matY,k}^2(\matA)$ efficiently, we can compute a constant-factor approximation, which is sufficient for our purposes. The proof of Lemma~\ref{troppextension0} is very similar to the proof of Corollary 10.10 of~\cite{HMT}, with the only difference being our starting point: instead of using Theorem 9.1 of~\cite{HMT} we use Lemma~\ref{lem:generic} of our work. To prove Lemma~\ref{troppextension0}, we will 
need several results for standard Gaussian matrices,
projection matrices, and H\"{o}lder's inequality. 
The following seven lemmas are all borrowed from \cite{HMT}.

\begin{lemma} [Proposition 10.1 in \cite{HMT}] \label{prop1b}
Fix matrices $\matX$, $\matY$, and draw a standard Gaussian matrix $\matR$ of appropriate dimensions. Then,
\mand{
\qquad \Expect{ \TNorm {\matX \matR \matY} }  \leq \TNorm{\matX} \FNorm{\matY} + \FNorm{\matX} \TNorm{\matY}.
}
\end{lemma}
\begin{lemma} [Proposition 10.2 in \cite{HMT}] \label{prop2b}
For $k, p \geq 2$, draw a standard Gaussian matrix $\matR \in \R^{k \times (k+p)}$. Then,
\mand{
\qquad \Expect{ \TNorm{\matR^+} }  \leq \frac{e \sqrt{k+p}}{p},
}
where $e=2.718\ldots$.
\end{lemma}

\begin{lemma} [Proposition 10.1 in \cite{HMT}] \label{prop1a}
Fix matrices $\matX$, $\matY$, and a standard Gaussian matrix $\matR$ of appropriate dimensions. Then,
\mand{
\Expect{ \FNormS{\matX \matR \matY} } = \FNormS{\matX} \FNormS{Y}.
}
\end{lemma}

\begin{lemma} [Proposition 10.2 in \cite{HMT}] \label{prop2a}
For $k, p \geq 2$, draw a standard Gaussian matrix $\matR \in \R^{k \times (k+p)}$. Then,
\mand{
 \Expect{ \FNormS{\matR^+} }  =\frac{k}{p-1}.
}
\end{lemma}
\begin{lemma} [proved in \cite{HMT}] \label{prop3}
For integers $k,p \geq 1$, and a standard Gaussian $\matR \in \R^{k \times (k+p)}$ the rank of $\matR$ is equal to $k$ with probability one.
\end{lemma}
\begin{lemma} [Proposition 8.6 in \cite{HMT}] \label{projection}
Let $\matP$ be a projection matrix. For any matrix $\matX$ of appropriate dimensions and an integer $q\ge 0$,
$$ \TNorm{\matP \matX} \le \left( \TNorm{ \matP (\matX \matX\transp)^q \matX  }  \right)^{\frac{1}{2q+1}} $$
\end{lemma}
\begin{lemma} [H\"{o}lder's inequality] \label{prop0}
Let $x$ be a positive random variable. Then, for any $h \ge 1$,
$$\Expect{ x } \leq \left( \Expect{ x^h } \right)^{\frac{1}{h}}.$$
\end{lemma}
\noindent The following lemma provides an alternative definition for $\Pi_{\matC,k}^{\xi}(\matA)$ which will be useful in subsequent proofs.
Recall from Section~\ref{chap21} that we can write
$\Pi_{\matC,k}^\xi(\matA)= \matC\matX^\xi$, where
$$
\matX^\xi = 
\argmin_{\Psi \in \mathbb{R}^{r \times n}:\rank(\Psi)\leq k}\XNormS{\matA-\matC\Psi}.
$$
The next lemma basically says that $\Pi_{\matC,k}^\xi(\matA)$ is the projection
of \math{\matA} onto the rank-\math{k} subspace spanned by 
\math{\matC\matX^\xi}, and that no other subspace in the column space of
\math{\matC}   is better.
\begin{lemma}\label{lem:PPD1}
For $\matA \in {\R}^{m \times n}$ and $\matC \in {\R}^{m \times r}$,  
integer $r> k$, let
$\Pi_{\matC,k}^\xi(\matA)= \matC\matX^\xi$, and 
$\matY\in \mathbb{R}^{r \times n}$ be any matrix of rank at most $k$.
Then, 
$$\XNormS{\matA-\matC\matX^\xi} = \XNormS{\matA-(\matC\matX^\xi)
(\matC\matX^\xi)^+\matA} \leq \XNormS{\matA-(\matC\matY)(\matC\matY)^+\matA},$$
where $\matY \in \mathbb{R}^{r \times n}$ is any matrix of rank at most $k$.
\end{lemma}

\begin{proof}
The second inequality will follow from the optimality of 
\math{\matX^\xi} because 
\math{\matY(\matC\matY)^+\matA} has rank at most \math{k}. So we only need to
prove the first equality. Again, by the optimality of 
\math{\matX^\xi} and  because 
\math{\matX^\xi(\matC\matX^\xi)^+\matA} has rank at most \math{k},
$\XNormS{\matA-\matC\matX^\xi} \le \XNormS{\matA-(\matC\matX^\xi)
(\matC\matX^\xi)^+\matA}$. To get the reverse inequality, we will use
matrix-Pythagoras as follows:
\begin{eqnarray*}
\XNormS{\matA - \matC\matX^\xi} 
&=& 
\XNormS{\left(\matI_m-(\matC\matX^\xi)(\matC\matX^\xi)^+\right)\matA-
\matC\matX^\xi(\matI_n-(\matC\matX^\xi)^+\matA)}\\
&\ge& 
\XNormS{\left(\matI_m-(\matC\matX^\xi)(\matC\matX^\xi)^+\right)\matA}.
\end{eqnarray*}
\end{proof}
\begin{lemma} [Extension of Corollary 10.10 of \cite{HMT}] \label{troppextension0}
Let $\matA$ be a matrix in $\R^{m \times n}$ of rank $\rho$,  let $k$ be an integer satisfying $2 \leq k < \rho$, and let $r = k+p$ for some integer $p \geq 2$.
Let $\matR \in \R^{n \times r}$ be a standard Gaussian matrix (i.e., a matrix whose entries are drawn in i.i.d. trials from $\mathcal{N}(0,1)$).
Define $\matB = (\matA \matA\transp)^q \matA$ and compute $\matY = \matB \matR$. Then, for any $q \geq 0$,
$$\Expect{\TNorm{\matA - \Pi_{\matY,k}^2(\matA) } } \leq
\left(1 + \sqrt{\frac{k}{p-1}} + \frac{e \sqrt{k+p}}{p} \sqrt{ \min\{ m,n\} -k } \right)^{\frac{1}{2q+1}}
\TNorm{\matA - \matA_k}$$
\end{lemma}
\begin{proof}
Let $\Pi_{\matY,k}^2(\matA) =\mat Y\matX_1$ and $\Pi_{\matY,k}^2(\matB) 
= \matY\matX_2$, where
\math{\matX_1} is optimal for \math{\matA} and \math{\matX_2} for 
\math{\matB}.
From Lemma~\ref{lem:PPD1},
$$\TNorm{\matA-\Pi_{\matY,k}^2(\matA)} = 
\TNorm{(\matI_m-(\matY\matX_1)(\matY\matX_1)^+)\matA} \leq
\TNorm{(\matI_m-(\matY\matX_2)(\matY\matX_2)^+)\matA}.$$
From Lemma~\ref{projection} and using the fact that $\matI_m-(\matY\matX_2)(\matY\matX_2)^+$ is a projection,
\begin{eqnarray*}
\TNorm{(\matI_m-(\matY\matX_2)(\matY\matX_2)^+)\matA}&\leq&
\TNorm{\left(\matI_m-\left(\matY\matX_2\right)\left(\matY\matX_2\right)^+\right)\left(\matA\matA\transp\right)^q\matA}^{\frac{1}{2q+1}}\\
&=&
\TNorm{\matB-\left(\matY\matX_2\right)\left(\matY\matX_2\right)^+\matB}^{\frac{1}{2q+1}}\\
&=&
\TNorm{\matB-\Pi_{\matY,k}^2(\matB)}^{\frac{1}{2q+1}},
\end{eqnarray*}
where the last step follows from Lemma~\ref{lem:PPD1}.
We conclude that
\begin{equation*}
\TNorm{\matA-\Pi_{\matY,k}^2(\matA)} \leq \TNorm{\matB-\Pi_{\matY,k}^2(\matB)}^{\frac{1}{2q+1}}.
\end{equation*}
\math{\matY} is generated using a random \math{\matR}, so
taking expectations and applying H\"{o}lder's inequality, we get
\begin{equation}\label{eqn:PPD22}
\Expect{\TNorm{\matA-\Pi_{\matY,k}^2(\matA)}} \leq \left(\Expect{\TNorm{\matB-\Pi_{\matY,k}^2(\matB)}}\right)^{\frac{1}{2q+1}}.
\end{equation}
We now focus on bounding the term on 
the right-hand side of the above equation. 
Let the SVD of $\matB$ be $\matB = \matU_{\matB} \Sigma_{\matB} \matV_{\matB}\transp$, with 
the top rank $k$ factors from the SVD of $\matB$ being
 $\matU_{\matB,k}$, $\Sigma_{\matB,k}$, and $\matV_{\matB,k}$ and
 the  corresponding trailing factors being
$\matU_{\matB,\tau}$, $\Sigma_{\matB,\tau}$ and $\matV_{\matB,\tau}$. 
Let $\rho_\matB$ be the rank of $\matB$.
Let
\mand{
       \Omega_1 = \matV_{\matB,k}\transp \matR \in \R^{k \times r} \qquad \mbox{and} \qquad
\qquad \Omega_2 = \matV_{\matB,\tau}\transp \matR \in \R^{(\rho_{\matB}-k) \times r}.
}
The Gaussian distribution is rotationally invariant, so $\Omega_1$, $\Omega_2$ are also standard Gaussian matrices which are
stochastically independent because $\matV_{\matB}\transp$ can be extended to a
full rotation. Thus, $\matV_{\matB,k}\transp \matR$ and $\matV_{\matB,\tau}\transp\matR$ also have entries that are i.i.d. $\mathcal{N}(0,1)$ variables. We now apply Lemma~\ref{lem:generic} to reconstructing
\math{\matB},  with $\xi = 2$ and $\matW = \matR$.
The rank requirement in
 Lemma~\ref{lem:generic} is satisfied because,
from Lemma~\ref{prop3}, the rank of $\Omega_1$ is equal to $k$ (as it is a 
standard normal matrix), and thus the matrix $\matR$ 
satisfies the rank assumptions of Lemma \ref{lem:generic}. We get that,
with probability 1,
$$\TNormS{ \matB - \Pi^{2}_{\matY,k}(\matB) }
\leq  \TNormS{\matB-\matB_k} + \TNormS{(\matB-\matB_k)\matR(\matV_{\matB,k}\transp\matR)^+ }
\le
\TNormS{\Sigma_{\matB,\tau}} +  
\TNormS{\Sigma_{\matB,\tau} \Omega_2 \Omega_1^+ }.
$$
Using \math{\sqrt{x^2+y^2}\le x+y}, we conclude that
\math{
\TNorm{ \matB - \Pi^{2}_{\matY,k}(\matB) }\le
\TNorm{\Sigma_{\matB,\tau}} +  
\TNorm{\Sigma_{\matB,\tau} \Omega_2 \Omega_1^+ }
}.
We now need to take the expectation with respect to 
\math{\Omega_1,\Omega_2}. We first take the expectation with respect to
\math{\Omega_2}, conditioning on \math{\Omega_1}. We then take the expectation
w.r.t. \math{\Omega_1}. Since only the second term is stochastic, using
Lemma~\ref{prop1b}, we have:
\begin{eqnarray}
{\bf E}_{\Omega_2}\left[ \TNorm{ \Sigma_{\matB,\tau} \Omega_2 \Omega_1^+ }|\Omega_1\right]
\nonumber &\leq&   \TNorm{ \Sigma_{\matB,\tau}} \FNorm{\Omega_1^+ }
+ \FNorm{ \Sigma_{\matB,\tau}} \TNorm{\Omega_1^+ }.
\end{eqnarray}
We now take the expectation with respect to \math{\Omega_1}.
To bound the term \math{\Expect{\TNorm{\Omega_1^+ }}}, we use Lemma~\ref{prop2b}.
To bound the term \math{\Expect{\FNorm{\Omega_1^+ }}}, we first use H\"{o}lder's 
inequality to bound \math{\Expect{\FNorm{\Omega_1^+ }}\le
\Expect{\FNormS{\Omega_1^+ }}^{1/2}}, and then we use 
Lemma~\ref{prop2a}.
Since \math{ \FNorm{ \Sigma_{\matB,\tau}}\le\sqrt{\min(m,n)-k}
\TNorm{ \Sigma_{\matB,\tau}}},  
collecting our results together, we obtain:
\eqan{
\Expect{\TNorm{\matB-\Pi_{Y,k}^2(\matB)}}
\le
\left(1+\sqrt{\frac{k}{p-1}}+\frac{e \sqrt{k+p}}{p}\sqrt{\min(m,n)-k}
\right)
\TNorm{\matB-\matB_k}
.}
To conclude, combine with eqn.(~\ref{eqn:PPD22}) and note that
\math{\TNorm{\matB-\matB_k}=\TNorm{\matA-\matA_k}^{2q+1}}.
\end{proof}
We now have all the necessary ingredients to prove Lemma \ref{tropp1}. Let $\matY$ be the matrix of Lemma~\ref{troppextension0}. Set
$p = k$ and $$q = \ceil{ \frac{ \log\left( 1+
\sqrt{\frac{k}{k-1}} + \frac{e \sqrt{2k}}{k} \sqrt{ \min\{ m,n\} -k } \right)
}{2 \log\left(1+\epsilon/\sqrt{2}\right)  - 1/2 }},$$
so that
$$ \left(1 + \sqrt{\frac{k}{p-1}} + \frac{e \sqrt{k+p}}{p} \sqrt{ \min\{ m,n\} -k } \right)^{\frac{1}{2q+1}} \leq 1 + \frac{\epsilon}{\sqrt{2}}.$$
Then, 
\mld{\Expect{\TNorm{\matA - \Pi_{\matY,k}^2(\matA) }} \leq \left(1 + \frac{\epsilon}{\sqrt{2}}\right) \TNorm{\matA - \matA_k}.
\label{eq:expS}
}
Given \math{\matY}, let \math{\matQ} be an orthonormal basis for its column 
space. Then,
using the algorithm of Section~\ref{chap21} 
and applying Lemma~\ref{lem:bestF} we can construct 
the matrix $\matQ\left(\matQ\transp\matA\right)_k$ such that
\mand{
\TNorm{\matA -  \matQ\left(\matQ\transp\matA\right)_k} \leq \sqrt{2} \TNorm{\matA - \Pi_{\matY,k}^2(\matA) }.
}
Clearly, $\matQ\left(\matQ^T\matA\right)_k$ is a rank $k$ matrix; 
let $\matZ \in {\R}^{n \times k}$ denote the 
matrix containing the right singular vectors of 
$\matQ\left(\matQ^T\matA\right)_k$, so 
\math{\matQ\left(\matQ^T\matA\right)_k=\matX\matZ\transp}. 
Note that $\matZ$ is equal to the right singular vectors of the 
matrix $\left(\matQ\transp\matA\right)_k$ 
(because $\matQ$ has orthonormal columns), and so $\matZ$ 
has already been computed at the second step of the 
algorithm of Section~\ref{chap21}. Since
$\matE=\matA - \matA \matZ \matZ\transp$ and  
$\TNorm{\matA-\matA\matZ\matZ\transp}\le
\TNorm{\matA-\matX\matZ\transp}$
for any \math{\matX}, we have
$$\TNorm{\matE} \leq \TNorm{\matA - \matQ\left(\matQ\transp\matA\right)_k}
\le\sqrt{2} \TNorm{\matA - \Pi_{\matY,k}^2(\matA) }.$$
Note that by construction, $\matE\matZ=\bm0_{m\times k}$.
The running time follows by adding the running time of 
the algorithm at the beginning of this section and the 
running time of the algorithm of Lemma~\ref{chap21}.

\section{Proof of Lemma \ref{tropp2}}
Consider the following algorithm, described in Theorem 10.5 of~\cite{HMT}. The algorithm takes as inputs a matrix $\matA \in \R^{m \times n}$ of rank $\rho$, an integer $2 \leq k < \rho$, and an integer $p \geq 2$. Set $r = k+p$ and construct the matrix
$\matY \in \R^{m \times r}$ as follows:
\begin{enumerate}
\item Generate an $n \times r$ standard Gaussian matrix $\matR$ whose entries are i.i.d. $\mathcal{N}(0,1)$.
\item Return $\matY = \matA \matR \in \R^{m \times r}$.
\end{enumerate}
The running time of the above algorithm is $ O(mnr)$.
Theorem 10.5 in \cite{HMT} presents the following bound:
$$\Expect{\FNorm{\matA - \matY\matY^+\matA }} \leq \left( 1 + \frac{k}{p-1} \right)^{\frac{1}{2}} \FNorm{\matA - \matA_k}.$$
The above result is not immediately applicable to the construction of a factorization of the form $\matA =
\matB \matZ\transp + \matE$ (as in Lemma~\ref{tropp2}) because $\matY$ contains $r > k$ columns.
\begin{lemma} [Extension of Theorem 10.5 of~\cite{HMT}] \label{troppextension}
Let $\matA$ be a matrix in $\R^{m \times n}$ of rank $\rho$, let $k$ be an integer satisfying $2 \leq k < \rho$, and let $r = k+p$ for some integer $p \geq 2$.
Let $\matR \in \R^{n \times r}$ be a standard Gaussian matrix (i.e., a matrix whose entries are drawn in i.i.d. trials from $\mathcal{N}(0,1)$)
and compute $\matY = \matA \matR$. Then,
$$\Expect{\FNormS{\matA - \Pi_{\matY,k}^F(\matA) }} \leq
\left( 1 + \frac{k}{p-1} \right) \FNormS{\matA - \matA_k}.$$
\end{lemma}
\begin{proof}
We construct the matrix $\matY$ as described in the beginning of this section. Let the rank of $\matA$ be $\rho$ and let $\matA = \matU_{\matA} \Sigma_{\matA} \matV_{\matA}\transp$ be the SVD of $\matA$. Define
\mand{\Omega_1 = \matV_{k}\transp \matR \in \R^{k \times r} \qquad \mbox{and} \qquad
\qquad \Omega_2 = \matV_{\rho-k}\transp \matR \in \R^{(\rho-k) \times r}.
}
The Gaussian distribution is rotationally invariant, so $\Omega_1$, $\Omega_2$ are also standard Gaussian matrices which are
stochastically independent because $\matV\transp$ can be 
extended to a full rotation.
Thus, $\matV_{k}\transp \matR$ and $\matV_{\rho-k}\transp\matR$ 
also have entries that are i.i.d. $\mathcal{N}(0,1)$ variables.
We now apply Lemma~\ref{lem:generic} to
reconstructing \math{\matA}, with $\xi = F$ and $\matW = \matR$. 
Recall that from Lemma~\ref{prop3}, the rank of $\Omega_1$ is equal to $k$ with
probability 1, and thus the matrix $\matR$ satisfies the rank assumptions 
of Lemma~\ref{lem:generic}. We have that, with probability~1,
$$\FNormS{ \matA - \Pi_{\matY,k}^F(\matA) } \le 
\FNormS{\matA - \matA_k} + \FNormS{ \Sigma_{\rho-k} \Omega_2 \Omega_1^+} ,$$
where 
$\matA - \matA_k = \matU_{\rho-k}  \Sigma_{\rho-k} \matV_{\rho-k}\transp$. To conclude, we take the expectation on both sides, and since only
the second term on the right hand side is stochastic, we bound as follows:
\eqan{
\Expect{\FNormS{\Sigma_{\rho-k} \Omega_2 \Omega_1^+}}
&\buildrel{(a)}\over{=}&
{\bf E}_{\Omega_1}\left[{\bf E}_{\Omega_2}\left[
\FNormS{ \Sigma_{\rho-k} \Omega_2 \Omega_1^+  }|
\Omega_1   \right]\right]\\
&\buildrel{(b)}\over{=}&
{\bf E}_{\Omega_1}\left[\FNormS{ \Sigma_{\rho-k}} \FNormS{ \Omega_1^+} \right] \\
&\buildrel{(c)}\over{=}&
\FNormS{ \Sigma_{\rho-k}} \Expect{ \FNormS{ \Omega_1^+} } \\
&\buildrel{(d)}\over{=}&
\frac{k}{p-1} \FNormS{\Sigma_{\rho-k}}.
}
\math{(a)} follows from the law of iterated expectation; \math{(b)} follows from Lemma \ref{prop1a}; \math{(c)} follows because $\FNormS{\Sigma_{\rho-k}}$ is a constant; \math{(d)} follows from Lemma \ref{prop2a}. We conclude the proof by noting that $\FNorm{\Sigma_{\rho-k}} = \FNorm{\matA - \matA_k}$.
\end{proof}
\noindent We now have all the necessary ingredients to conclude the proof of Lemma \ref{tropp2}. Let \math{\matY} be the matrix of 
Lemma~\ref{troppextension},
and let \math{\matQ} be an orthonormal basis for its column
space. Then,
using the algorithm of Section~\ref{chap21} 
and applying Lemma~\ref{lem:bestF} we can construct the matrix 
$\matQ\left(\matQ\transp\matA\right)_k$ such that
$$\FNormS{\matA -  \matQ\left(\matQ\transp\matA\right)_k} = 
\FNormS{\matA - \Pi_{\matY,k}^F(\matA) }.$$
Clearly, $ \matQ\left(\matQ^T\matA\right)_k$ is a rank $k$ matrix; let 
$\matZ \in {\R}^{n \times k}$ be 
the matrix containing the right singular vectors of 
$\matQ\left(\matQ^T\matA\right)_k$, so 
\math{\matQ\left(\matQ^T\matA\right)_k=\matX\matZ\transp}.
Note that  $\matZ$ is equal to the right singular vectors of the 
matrix $\left(\matQ\transp\matA\right)_k$ (because $\matQ$ has orthonormal columns), and thus $\matZ$ has already been computed at the second step 
of the algorithm of Section~\ref{chap21}.  Since
\math{\matE=\matA - \matA \matZ 
\matZ\transp} and  
$$\FNorm{\matA-\matA\matZ\matZ\transp}\le
\FNorm{\matA-\matX\matZ\transp}$$ for any \math{\matX}, we have
$$\FNormS{\matE} \leq \FNormS{\matA - \matQ\left(\matQ\transp\matA\right)_k}
=\FNormS{\matA - \Pi_{\matY,k}^2(\matA) }.$$
To conclude, take expectations on both sides,
use Lemma~\ref{troppextension} to bound the term \math{\Expect{\FNormS{\matA - \Pi_{\matY,k}^2(\matA) }}},
and set $$p = \ceil{ \frac{k}{\epsilon } + 1}$$ to obtain:
$$\Expect{\FNormS{\matE}} \leq \left( 1 + \frac{k}{p-1} \right) \FNormS{\matA - \matA_k}
\leq \left(1+\epsilon\right)\FNormS{\matA - \matA_k},$$
By construction,  $\matE\matZ=\bm0_{m\times k}$. 
The running time follows by adding the running time of the algorithm at the beginning of this section and the running time of the algorithm of Section~\ref{chap21}.

\section{Proof of Lemma~\ref{lem:fnorm}}
\noindent Let 
$$x = \FNormS{ \matY \Omega \matS }$$ be a random variable
with nonnegative values. Assume that the following equation is true:
$$ \Expect{ \FNormS{ \matY \Omega \matS } } = \FNormS{ \matY }.$$
Applying Markov's inequality to this equation gives the bound in the lemma.
All that it remains to prove now is the above assumption.
Let $$\matX = \matY \Omega \matS \in \R^{m \times r},$$ and for $t=1,...,r,$ let $\matX^{(t)}$
denotes the $t$-th column of $\matX = \matY \Omega \matS$. We manipulate
the term $\Expect{ \FNormS{ \matY \Omega \matS } }$ as follows:
\eqan{
\Expect{ \FNormS{ \matY \Omega \matS } } &\buildrel{(a)}\over{=}&  \Expect{ \sum_{t=1}^{r} \TNormS{\matX^{(t)}} }\\ &\buildrel{(b)}\over{=}&   
\sum_{t=1}^{r} \Expect{ \TNormS{\matX^{(t)}} } \\ &\buildrel{(c)}\over{=}&  
\sum_{t=1}^{r} \sum_{j=1}^n p_j \frac{ \TNormS{\matY^{(j)}} }{r p_j} \\ &\buildrel{(d)}\over{=}& 
\frac{1}{r} \sum_{t=1}^{r} \FNormS{\matY} \\ &=& \FNormS{\matY}
}
\math{(a)} follows by the definition of the Frobenius norm of $\matX$.
\math{(b)} follows by the linearity of expectation.
\math{(c)} follows by our construction of $\Omega, \matS$.
\math{(d)} follows by the definition of the Frobenius norm of $\matY$.

\clearpage

\section{Proof of Lemma~\ref{lem:rrqrGE}}

\begin{proof}
Let \math{\matY\transp=[\matX,\bm0_{n-m\times n}]}, so that
\math{\matY\in\R^{n\times n}} is a square matrix whose first \math{m}
rows are \math{\matX\transp} and whose last \math{n-m} rows are zero.
Run Algorithm 4 of
\cite{GE96} with inputs $\matY$ and $k$.
The bounds for the singular values are 
exactly those of Theorem 3.2 of \cite{GE96}, which hold because the
singular values of \math{\matY} and \math{\matX} are identical.
The only concern is the running time.  
The run time directly from \cite{GE96} is \math{O(n^2k\log_{f} n)}, which 
treats \math{\matY} as a full \math{n\times n} matrix; but
we argue below that the algorithm can be implemented more efficiently in  
\math{O(m n k \log_{f} n)} time because we can effectively ignore the
padding with zeros in all computations. 
First, we observe that the \math{\matQ\matR} factorization of 
\math{\matY\Pi} for any permutation \math{\Pi} can always be written
\mand{
\matY\Pi=\matQ\matR=
\left[
\begin{matrix}
\matQ_\matX&\bm0_{m\times n-m}\\
\bm0_{n-m\times m}&\matI_{n-m\times n-m}
\end{matrix}
\right]
\left[
\begin{matrix}
\matA_k&\matB_k\\
\bm0_{m-k\times k}&\bar\matC_k\\
\bm0_{n-m\times k}&\bm0_{n-m\times n-k}
\end{matrix}
\right]
}
where \math{\matQ_\matX\in\R^{m\times m}} and 
\math{\matC_k=
\left[\begin{matrix}\bar\matC_k\\ \bm0_{n-m\times n-k}\end{matrix}\right]};
effectively we only need to perform a \math{\matQ\matR} factorization
on the upper unpadded part, so:
\mand{
\matX\transp\Pi=
\matQ_\matX
\left[
\begin{matrix}
\matA_k&\matB_k\\
\bm0_{m-k\times k}&\bar\matC_k\\
\end{matrix}
\right]
}
This has an important consequence to Algorithm 4 in \cite{GE96}. 
In what follows we assume the reader is familiar with the notation in \cite{GE96}.
There are two basic steps in the algorithm. The first is to 
compute a function such as \math{\rho(\matR,k)} to determine which two 
columns to permute. This function depends only on
\math{\gamma_*(C_k)=\gamma_*(\bar\matC_k)} by construction, because the 
padded zeros contribute nothing to these norms.
The second step is to refactorize and obtain the new \math{\matQ} and
\math{\matR}, or more specifically to update 
\math{\omega_*(\matA_k)}, \math{\gamma_*(\matC_k)} and
\math{\matA_k^{-1}\matB_k}. As observed above, we just need to refactorize 
\math{\matX\transp\Pi}, which means we simply need to run the the efficient 
update steps in 
\cite{GE96} for \math{\matA_k}, \math{\matB_k} and \math{\bar\matC_k}. We are
effectively running the algorithm on \math{\matX}, ignoring the padding 
completely. So, from 
\cite{GE96} the run time is \math{O(mnk\log_{f} n)}, as claimed.
\end{proof}



\section{Proof of Lemma \ref{lem:2setS}}

We first restate the Lemma in a slightly different notation 
(actually this is the notation used in \cite{BSS09} so we chose 
to be consistent with existing literature).

\begin{lemma}[Dual Set Spectral Sparsification.] \label{lemma:intro1}
Let \math{\cl V=\{\v_1,\ldots,\v_n\}} and \math{\cl U=\{\u_1,\ldots,\u_n\}}
be two equal cardinality decompositions of the identity, where
\math{\v_i\in\R^{k}} ($k < n$), \math{\u_i\in\R^\ell} ($\ell \leq n$), and:
$\sum_{i=1}^n\v_i\v_i\transp=\matI_{k}$ and $\sum_{i=1}^n\u_i\u_i\transp=\matI_{\ell}$.
Given an integer \math{r} with \math{k < r \le n}, there exists a set of weights
\math{s_i\ge 0} ($i=1,\ldots,n$) {at most \math{r} of which are non-zero}, such that
\mand{
\lambda_{k}\left(\sum_{i=1}^ns_i\v_i\v_i\transp\right)
\ge
\left(1 - \sqrt{\frac{k}{r}}\right)^2
\qquad\mbox{and}
\qquad \lambda_{1}\left(\sum_{i=1}^ns_i\u_i\u_i\transp\right)
\le \left(1 + \sqrt{ \frac{\ell}{r} }\right)^2.}
The weights $s_i$ can be computed deterministically in $O\left(r n \left(k^2+\ell^2\right) \right)$ time.
\end{lemma}
\noindent In matrix notation, let $\matU$ and $\matV$
be the matrices whose \textit{rows} are the vectors $\u_i$ and $\v_i$ respectively. We can now construct the sampling 
and rescaling matrices
$\Omega \in \mathbb{R}^{n \times r}, \matS \in \mathbb{R}^{r \times r}$ as follows: for $i=1,\ldots,n$, if $s_i$ is non-zero then include $\e_i$ as a column of $\Omega$ and $\sqrt{s_i}$ as the $i$-th diagonal element of $\matS$; here $\e_i$ is the $i$-th standard basis vector\footnote[9]{Note that we slightly abused notation: indeed, the number of columns of $\Omega$ is less than or equal to $r$, since at most $r$ of the weights are non-zero. Here, we use $r$ to also denote the actual number of non-zero weights, which is equal to the number of columns of the matrix $\Omega$.}. Using this matrix notation, the above lemma guarantees that
$
\sigma_{k}\left(\matV\transp\Omega\matS\right)
\ge
1 - \sqrt{\frac{k}{r}}$
and
$\sigma_{1}\left(\matU\transp\Omega\matS\right)
\le 1 + \sqrt{ \frac{\ell}{r} }$.
Clearly, $\Omega$ and $\matS$ may be viewed matrices that sample and rescale 
$r$ \textit{rows} of $\matU$ and $\matV$ (columns of $\matU\transp$ and $\matV\transp$),
namely the rows that correspond to non-zero weights $s_i$. 
Next, we prove Lemma~\ref{lemma:intro1}.

\subsection*{Proof of  Lemma~\ref{lemma:intro1}} \label{sec:proofpure1}

Lemma~\ref{lemma:intro1} generalizes Theorem 3.1 in~\cite{BSS09}. Indeed, setting
\math{\cl V=\cl U} reproduces the spectral sparsification result of
 Theorem 3.1 in~\cite{BSS09}. We will provide a constructive proof of the lemma and we start by describing the algorithm that computes the weights $s_i$, $i=1,\ldots,n$.
%
\begin{algorithm}[t]
\begin{framed}
\textbf{Input:}
     \begin{itemize}
          \item $ \cl V=\{\v_1,\ldots,\v_n\}$, with $\sum_{i=1}^{n}\v_i\v_i\transp=\matI_{k}$ ($k \leq n$)
          \item $ \cl U=\{\u_1,\ldots,\u_n\}$, with $\sum_{i=1}^{n}\u_i\u_i\transp=\matI_{\ell}$ ($\ell \leq n$)
          \item integer $r$, with $k < r < n$
     \end{itemize}

     \textbf{Output:} A vector of weights $\s=[s_1,\ldots,s_n]$, with $s_i \geq 0$ and at most $r$ non-zero $s_i$'s.

\begin{enumerate}

  \item Initialize $\s_0 = \textbf{0}_{n \times 1}$, $\matA_0 = \textbf{0}_{k \times k}$, $\matB_0 = \textbf{0}_{\ell \times \ell}$.

  \item For $\tau = 0,...,r-1$

        \begin{itemize}

           \item Compute $\scl_\tau$ and $\scu_\tau$ from eqn.~(\ref{eqn:PD1}).

            \item Find an index $j$ in $\left\{1,\ldots,n\right\}$ such that
			 \eqar{
			 \label{eqn:algo}
			 U(\u_j,\delta_\scu,\matB_\tau,\scu_\tau) \le L(\v_j,\delta_\scl, \matA_\tau,\scl_\tau).
				}		
    \item Let
    \begin{equation}\label{eqn:weight}
    t^{-1} = \frac{U(\u_j,\delta_\scu,\matB_\tau,\scu_\tau) + L(\v_j,\delta_\scl, \matA_\tau,\scl_\tau)}{2}.
    \end{equation}

    \item Update the \math{j}th 
component of \math{\s}, \math{\matA_\tau} and \math{\matB_\tau}:
\mld{
\s_{\tau+1}[j] =\s_{\tau}[j] + t,\ \matA_{\tau + 1} = \matA_{\tau} + t \v_j \v_j\transp,\ \hbox{and}\ \matB_{\tau+1} = \matB_{\tau} + t \u_j \u_j\transp.
\label{eq:update}
}

         \end{itemize}

  \item Return $\s = r^{-1}\left(1-\sqrt{k/r}\right)\cdot \s_r$.

\end{enumerate}
\caption{Deterministic Dual Set Spectral Sparsification.}
\label{alg:2set}
\end{framed}
\end{algorithm}

\paragraph{The Algorithm.}
The fundamental idea underlying Algorithm~\ref{alg:2set} is the greedy selection of vectors that satisfy a number of desired properties in each step. These properties will eventually imply the eigenvalue bounds of Lemma~\ref{lemma:intro1}. We start by defining several quantities that will be used in the description of the algorithm and its proof. First, fix two constants:
\vspace{-0.1in}
$$
\delta_\scl = 1;
\qquad
\delta_\scu = \frac{1+\sqrt{\frac{\ell}{r}}}{1-\sqrt{\frac{k}{r}}}.
$$
Given $k$, $\ell$, and $r$ (all inputs of Algorithm~\ref{alg:2set}), and a parameter $\tau = 0,\ldots,r-1$, define two
parameters $\scl_\tau$ and $\scu_\tau$ as follows:
\vspace{-0.1in}
\begin{equation}\label{eqn:PD1}
\scl_\tau = r\left(\frac{\tau}{r}-\sqrt{\frac{k}{r}}\right) = \tau - \sqrt{rk};
\scu_\tau = \frac{(\tau-r)\left(1+\sqrt{\frac{\ell}{r}}\right)
               +r\left(1+\sqrt{\frac{\ell}{r}}\right)^2}
               {1-\sqrt{\frac{k}{r}}}=\delta_{\scu}\left(\tau + \sqrt{\ell r}\right).
\end{equation}
We next define the lower and upper 
functions \math{\phil(\scl,\matA)} ($\scl\in\R$ and $\matA\in\R^{k \times k}$) and
$\phiu(\scu,\matB)$ ($\scu\in\R$ and $\matB\in\R^{\ell \times \ell}$) as follows:
\vspace{-0.1in}
\begin{equation}\label{eqn:PD2}
\phil(\scl, \matA) =  \sum_{i=1}^k\frac{1}{\lambda_i(\matA)-\scl};
\qquad
\phiu(\scu, \matB) =  \sum_{i=1}^\ell\frac{1}{\scu-\lambda_i(\matB)}.
\end{equation}
Let $L(\v, \delta_\scl, \matA, \scl)$  be a function with four inputs (a vector $\v \in \mathbb{R}^{k\times 1}$, $\delta_\scl \in \mathbb{R}$, a matrix $\matA \in \mathbb{R}^{k \times k}$, and $\scl \in \mathbb{R}$):
\vspace{-0.1in}
\begin{equation}\label{eqn:PD3}
L(\v, \delta_\scl, \matA, \scl)  =  \frac{\v\transp(\matA-(\scl + \delta_\scl)\matI_k)^{-2}\v}
{\phil(\scl + \delta_\scl, \matA)-\phil(\scl,\matA)} -\v\transp(\matA-(\scl + \delta_\scl)\matI_k)^{-1}\v.
\end{equation}
Similarly, let $U(\v, \delta_\scu, \matB, \scu)$  be a function with four inputs (a vector $\u \in \mathbb{R}^{\ell\times 1}$, $\delta_\scu \in \mathbb{R}$, a matrix $\matB \in \mathbb{R}^{\ell \times \ell}$, and $\scu \in \mathbb{R}$):
\vspace{-0.1in}
\begin{equation}\label{eqn:PD4}
U(\u, \delta_\scu, \matB, \scu )  = \frac{\u\transp((\scu + \delta_\scu)\matI_\ell-\matB)^{-2}\u}
{\phiu(\scu, \matB)-\phiu(\scu + \delta_\scu,\matB)}
+\u\transp((\scu + \delta_\scu)\matI_\ell-\matB)^{-1}\u.
\end{equation}
Algorithm \ref{alg:2set} runs in $r$ steps. The vector of weights $\s_0$ is initialized to the all-zero vector. At each step $\tau=0,\ldots,r-1$, the
algorithm selects a pair of vectors $(\u_j, \v_j)$ that satisfy eqn.~(\ref{eqn:algo}), computes the associated weight $t$ from eqn.~(\ref{eqn:weight}), and updates two matrices and the vector of weights as specified
in eqn. \r{eq:update}.

\paragraph{Running time.}

The algorithm runs in $r$ iterations. In each iteration, we evaluate the functions $U(\u,\delta_\scu,\matB,\scu)$ and $L(\v,\delta_\scl, \matA,\scl)$ at most $n$ times. Note that all $n$ evaluations for both functions need at most $O(k^3 + nk^2 + \ell^3 + n\ell^2)$ time, because the matrix inversions can be performed once for all $n$ evaluations. Finally, the updating step needs an additional $O(k^2 + \ell^2)$ time. Overall, the complexity of the algorithm is of the order $O(r(k^3 + nk^2 + \ell^3 + n\ell^2 + k^2 + \ell^2) ) = O\left(rn \left(k^2+\ell^2\right)\right)$.

Note that when \math{\cl U} is the standard basis 
(\math{\cl U=\{\e_1,\ldots,\e_n\}} and \math{\ell=n}),
the computations can be done much more efficiently:
the
eigenvalues of \math{\matB_\s} need not be computed 
explicitly (the expensive step), 
since they are available by inspection, being equal to the
weights \math{\s_\tau}. 
In the function \math{U(\u,\delta_{\scu},\matB,\scu)},
the functions \math{\phiu} (given the eigenvalues) need only be computed
once per iteration, in \math{O(n)} time; the remaining terms can be computed
in \math{O(1)} because (for example)
\math{\e_i\transp((\scu+\delta_\scu)\matI-\matB)^{-2}\e_i
=(\scu+\delta_\scu-\s[i])^{-2}}.
The run time, in this case, drops to \math{O(rnk^2)}, since all the operations
on \math{\cl U} only contribute \math{O(rn)}.

\paragraph{Proof of Correctness.}
We prove that the output of Algorithm \ref{alg:2set} satisfies Lemma \ref{lemma:intro1}. Our proof is similar to the proof of Theorem 3.1~\cite{BSS09}. The main difference is that we need to accommodate two different sets of vectors. Let $\matW \in \R^{m \times m}$ be a positive semi-definite matrix with eigendecomposition
\mand{
\matW=\sum_{i=1}^{m}\lambda_i(\matW) \u_i\u_i\transp
}
and recall the functions $\phil(\scl,\matW)$, $\phiu(\scu,\matW)$, $L(\v,\delta_\scl,\matW,\scl)$, and $U(\v,\delta_\scu,\matW,\scu)$ defined in eqns.~(\ref{eqn:PD2}), (\ref{eqn:PD3}), and~(\ref{eqn:PD4}). We now quote two lemmas
proven  in~\cite{BSS09} using the Sherman-Morrison-Woodbury identity;
these lemmas allow one 
to control the smallest and largest eigenvalues of
 $\matW$ under a rank-one perturbation.
\begin{lemma}
\label{lemma:sp1}
Fix \math{\delta_\scl>0}, \math{\matW \in \R^{m \times m}}, \math{\v \in \R^m}, and \math{\scl < \lambda_{m}(\matW)}. If \math{t > 0} satisfies
\mand{
t^{-1}\le L(\v,\delta_\scl,\matW,\scl)}
then \math{\lambda_{m}(\matW+t\v\v\transp) \ge \scl+\delta_\scl}.
\end{lemma}
\begin{lemma}
\label{lemma:sp2}
Fix \math{\delta_\scu>0}, \math{\matW \in R^{m \times m}}, \math{\v \in R^m}, and
\math{\scu > \lambda_{1}(\matW)}. If \math{t} satisfies
\mand{
t^{-1}\ge U(\v,\delta_\scu,\matW,\scu),}
then \math{\lambda_{1}(\matW+t\v\v\transp) \le\scu+\delta_\scu}.
\end{lemma}
Now recall that Algorithm~\ref{alg:2set} runs in \math{r} steps. Initially, all $n$ weights are set to zero. Assume that at the $\tau$-th step ($\tau = 0,\ldots,r-1$) the vector of weights \math{\s_{\tau} = 
\left[\s_{\tau}[1],\ldots, \s_{\tau}[n]\right]} has been constructed and let
$$\matA_\tau=\sum_{i=1}^n \s_{\tau}[i]\v_i\v_i\transp 
\qquad \mbox{and} \qquad \matB_\tau=\sum_{i=1}^n\s_{\tau}[i]\u_i\u_i\transp.$$
Note that both matrices $\matA_\tau$ and $\matB_\tau$ are positive 
semi-definite. We claim the following lemma
which guarantees that the algorithm is well-defined.
%
\begin{lemma}
\label{lemma:feasible}
At the $\tau$-th step, for all $\tau=0,\ldots,r-1$, there exists an index $j$ in $\left\{1,\ldots,n\right\}$ such that setting the weight $t > 0$ as in eqn.~(\ref{eqn:weight}) satisfies
\begin{equation}\label{eqn:PD5}
U(\u_j,\delta_\scu,\matB_\tau,\scu_\tau) \le t^{-1} \le L(\v_j,\delta_\scl,\matA_\tau,\scl_\tau).
\end{equation}
\end{lemma}

\paragraph{Proof of Lemma~\ref{lemma:feasible}.}
In order to prove Lemma~\ref{lemma:feasible} we will use the following averaging argument.
\begin{lemma}\label{lemma:feasiblebound}
At any step $\tau=0,\ldots,r-1$,
\mand{
\sum_{i=1}^n U(\u_i,\delta_\scu,\matB_\tau,\scu_\tau)
\le
1-\sqrt{\frac{k}{r}}
\le
\sum_{i=1}^n L(\v_i,\delta_\scl,\matB_\tau,\scl_\tau).
}
\end{lemma}

\begin{proof}
For notational convenience, let \math{\phiu_\tau=\phiu(\scu_\tau,\matB_\tau)} and let \math{\phil_\tau=\phil(\scl_\tau,\matA_\tau)}. At \math{\tau=0}, \math{\matB_0=\bm0} and \math{\matA_0=\bm0} and thus \math{\phiu_0=\ell/\scu_0} and \math{\phil_0=-k/\scl_0}. Focus on the $\tau$-th step and assume that the algorithm has run correctly up to that point. Then,
\math{\phiu_\tau\le\phiu_0} and \math{\phil_\tau\le\phil_0}. Both are true at $\tau=0$ and, assuming that the algorithm has run correctly until the $\tau$-th step, Lemmas~\ref{lemma:sp1} and~\ref{lemma:sp2} guarantee that \math{\phiu_\tau} and \math{\phil_\tau} are non-increasing.

First, consider the upper bound on \math{U}. In the following derivation, $\lambda_i$ denotes the $i$-th eigenvalue of $\matB_{\tau}$. Using
\math{\trace(\u\transp\matX\u)=\trace(\matX\u\u\transp)} and
\math{\sum_{i}\u_i\u_i\transp=\matI_\ell}, we get
\eqan{
\sum_{i=1}^n U(\u_i,\delta_\scu,\matB_\tau,\scu_\tau)&=&
\frac
{\trace\left[(\scu_{\tau+1}\matI_{\ell}-\matB_\tau)^{-2}\right]}
{\phiu_\tau-\phiu(\scu_{\tau+1},\matB_\tau)}
+\phiu(\scu_{\tau+1},\matB_\tau)\\
&=&
\frac
{\sum_{i=1}^\ell\frac{1}{(\scu_{\tau+1}-\lambda_i)^2}}
{\delta_{\scu}\sum_{i=1}^\ell\frac{1}
{(\scu_{\tau+1}-\lambda_i)(\scu_{\tau}-\lambda_i)}}
+
\sum_{i=1}^\ell\frac{1}{(\scu_{\tau+1}-\lambda_i)}\\
&=&
\frac{1}{\delta_\scu}+\phiu_\tau
-\frac{1}{\delta_\scu}\left(1-
\frac
{\sum_{i=1}^\ell\frac{1}{(\scu_{\tau+1}-\lambda_i)^2}}
{\sum_{i=1}^\ell\frac{1}
{(\scu_{\tau+1}-\lambda_i)(\scu_{\tau}-\lambda_i)}}
\right) \\ &-& \delta_\scu
\sum_{i=1}^\ell\frac{1}{(\scu_{\tau}-\lambda_i)(\scu_{\tau+1}-\lambda_i)}\\
&\le&
\frac{1}{\delta_\scu}+\phiu_0.
}
The last line follows because the last two
terms are negative (using the fact that \math{\scu_{\tau+1}>\scu_\tau>\lambda_i}) and \math{\phiu_\tau\le\phiu_0}. Now,
using \math{\phiu_0=\delta_\scu\sqrt{r\ell}} and the definition of
\math{\delta_\scu} the upper bound follows:
\mand{
\frac{1}{\delta_\scu}+\phiu_0=\frac{1}{\delta_\scu}+
\frac{\ell}{\delta_\scu\sqrt{r\ell}}=\frac{1}{\delta_\scu}
\left(1+\sqrt{\frac{\ell}{r}}\right)=1-\sqrt{\frac{k}{r}}.}
In order to prove the lower bound on \math{L} we use a similar argument. Let $\lambda_i$ denote the $i$-th eigenvalue of $\matA_{\tau}$. Then,
\eqan{
\sum_{i=1}^n L(\v_i,\delta_\scl,\matA_\tau,\scl_\tau)&=&
\frac
{\trace\left[(\matA_\tau-\scl_{\tau+1}\matI_k)^{-2}\right]}
{\phil(\scl_{\tau+1},\matA_\tau)-\phil_\tau}
-\phil(\scl_{\tau+1},\matA_\tau)\\
&=&
\frac
{\sum_{i=1}^k\frac{1}{(\lambda_i-\scl_{\tau+1})^2}}
{\delta_{\scl}\sum_{i=1}^k\frac{1}
{(\lambda_i-\scl_{\tau+1})(\lambda_i-\scl_{\tau})}}
-
\sum_{i=1}^k\frac{1}{(\lambda_i-\scl_{\tau+1})}\\
&=&
\frac{1}{\delta_\scl}-\phil_\tau
+
\underbrace{
\frac{1}{\delta_\scl}\left(
\frac
{\sum_{i=1}^k\frac{1}{(\lambda_i-\scl_{\tau+1})^2}}
{\sum_{i=1}^k\frac{1}
{(\lambda_i-\scl_{\tau+1})(\lambda_i-\scl_{\tau})}}
-1
\right)-\delta_\scl
\sum_{i=1}^k\frac{1}{(\lambda_i-\scl_{\tau})(\lambda_i-\scl_{\tau+1})}
}_{\cl E}\\
&\ge&
\frac{1}{\delta_\scl}-\phil_0 + \cl E.
}
Assuming \math{\cl E\ge 0} the claim follows immediately because \math{\delta_\scl=1} and \math{\phil_0=-k/\scl_0=k/\sqrt{rk}=\sqrt{k/r}}. Thus, we only need to show that \math{\cl E\ge 0}. From the Cauchy-Schwarz inequality, for \math{a_i,b_i\ge 0},
\math{\left(\sum_i a_i b_i\right)^2\le\left(\sum_ia_i^2b_i\right)\left(\sum_i b_i\right)} and thus
\begin{eqnarray}
\nonumber \cl E\sum_{i=1}^k\frac{1}
{(\lambda_i-\scl_{\tau+1})(\lambda_i-\scl_{\tau})}
&=&\frac{1}{\delta_{\scl}}
\sum_{i=1}^k\frac{1}{(\lambda_i-\scl_{\tau+1})^2(\lambda_i-\scl_{\tau})}
-
\delta_\scl
\left(\sum_{i=1}^k\frac{1}{(\lambda_i-\scl_{\tau})
(\lambda_i-\scl_{\tau+1})}\right)^2\\
\nonumber &\ge&
\frac{1}{\delta_{\scl}}\sum_{i=1}^k\frac{1}{(\lambda_i-\scl_{\tau+1})^2(\lambda_i-\scl_{\tau})}
-
\end{eqnarray}
\vspace{-.33in}
\begin{eqnarray}
\delta_\scl
\sum_{i=1}^k\frac{1}{(\lambda_i-\scl_{\tau+1})^2(\lambda_i-\scl_{\tau})}
\sum_{i=1}^k\frac{1}{\lambda_i-\scl_\tau} = 
\left(\frac{1}{\delta_{\scl}}-\delta_\scl\phil_\tau\right)
\sum_{i=1}^k\frac{1}{(\lambda_i-\scl_{\tau+1})^2(\lambda_i-\scl_{\tau})}.
\end{eqnarray}
To conclude our proof, first note that \math{\delta_\scl^{-1}-\delta_\scl\phil_\tau\ge \delta_\scl^{-1}-\delta_\scl\phil_0=1-\sqrt{k/r}>0}
(recall that \math{r>k}). Second, \math{\lambda_i>\scl_{\tau+1}} because
\mand{
\lambda_{\min}(\matA_\tau)>\scl_\tau+\frac{1}{\phil_\tau}
\ge \scl_\tau+\frac{1}{\phil_0}=\scl_\tau+\sqrt{\frac{r}{k}}
>\scl_\tau+1=\scl_{\tau+1}.
}
Combining these two observations with the later equation, we conclude that $\cl E \geq 0$.
\end{proof}
\noindent Lemma~\ref{lemma:feasible} now follows from Lemma~\ref{lemma:feasiblebound} because the two inequalities must hold simultaneously for at least one index \math{j}.
\noindent Once an index $j$ and a weight $t>0$ have been computed, Algorithm~\ref{alg:2set} updates the $j$-th weight in the vector of weights $\s_{\tau}$ to create the vector of weights $\s_{\tau+1}$. Clearly, at each of the $r$ steps, only one element of the vector of weights is updated. Since $\s_0$ is initialized to the all-zeros vector, after all $r$ steps are completed, at most $r$ weights are non-zero. The following lemma argues that $\lambda_{\min}(\matA_\tau)$ and $\lambda_{\max}(\matB_\tau)$ are bounded.
\begin{lemma}
\label{lemma:bounds}
At the $\tau$-th step, for all $\tau=0,\ldots,r-1$, $\lambda_{\min}(\matA_\tau)\ge\scl_\tau$ and
$\lambda_{\max}(\matB_\tau)\le \scu_\tau$.
\end{lemma}
\begin{proof}
Recall eqn.~(\ref{eqn:PD1}) and observe that \math{\scl_0=-\sqrt{rk}<0} and \math{\scu_0=\delta_\scu\sqrt{r\ell}>0}. Thus, the lemma holds at \math{\tau=0}. It is also easy to verify that \math{\scl_{\tau+1}=\scl_{\tau}+\delta_\scl}, and, similarly, \math{\scu_{\tau+1}=\scu_{\tau}+\delta_\scu}. Now, at the $\tau$-step, given an index $j$ and a corresponding weight $t > 0$ satisfying
eqn.~(\ref{eqn:PD5}), Lemmas \ref{lemma:sp1} and \ref{lemma:sp2} imply that
\eqan{\lambda_{\min}(\matA_{\tau+1})=
\lambda_{\min}(\matA_{\tau}+t\v_j\v_j\transp)\ge \scl_{\tau}+
\delta_{\scl}=\scl_{\tau+1};\\
\lambda_{\min}(\matB_{\tau+1}) = \lambda_{\max}(\matB_{\tau}+t\u_j\u_j\transp)\le \scu_{\tau}+
\delta_{\scu}=\scu_{\tau+1}.
}
The lemma now follows by simple induction on $\tau$.
\end{proof}
\noindent We are now ready to conclude the proof of Lemma~\ref{lemma:intro1}. By Lemma~\ref{lemma:bounds}, at the $r$-th step,
\mand{
\lambda_{\max}(\matB_r)\le
\scu_r \qquad\text{ and }\qquad
\lambda_{\min}(\matA_r)\ge\scl_r.
}
Recall the definitions of $\scu_r$ and $\scl_r$ from eqn.~(\ref{eqn:PD1}) and note that they are both positive and well-defined because \math{r>k}. Lemma \ref{lemma:intro1} now follows after
rescaling the vector of weights $\s$ by \math{r^{-1}\left(1-\sqrt{k/r}\right)}. Note that the rescaling does not change the number of non-zero elements of \math{\s}, but does rescale all the eigenvalues of \math{\matA_r} and \math{\matB_r}.

\section{Proof of Lemma~\ref{lem:2setF}}
Again, we first state Lemma~\ref{lem:2setF} in a slightly different notation. 
\begin{lemma}[Dual Set Spectral-Frobenius Sparsification.]
\label{lemma:intro2}
Let \math{\cl V=\{\v_1,\ldots,\v_n\}} be a decomposition of the identity, where \math{\v_i\in\R^{k}} ($k < n$) and
$\sum_{i=1}^n\v_i\v_i\transp=\matI_{k}$; let \math{\cl A=\{\a_1,\ldots,\a_n\}} be an arbitrary set
of vectors, where \math{\a_i\in\R^{\ell}}. Then,
given an integer \math{r} such that \math{k < r \le n}, there exists a set of weights \math{s_i\ge 0} ($i=1\ldots n$), {at most \math{r} of which are non-zero}, such that
\mand{
\lambda_{k}\left(\sum_{i=1}^ns_i\v_i\v_i\transp\right)
\ge
\left(1 - \sqrt{\frac{k}{r}}\right)^2
\qquad\text{and}
\qquad \trace\left(\sum_{i=1}^n s_i\a_i\a_i\transp\right)
\le \trace\left(\sum_{i=1}^n \a_i\a_i\transp\right)
}
The weights $s_i$ can be computed deterministically in $O\left(rnk^2+n\ell\right)$ time.
\end{lemma}
\noindent In matrix notation 
(here $\matA$ denotes the matrix whose \textit{columns} are the vectors $\a_i$), the above lemma argues that
$
\sigma_{k}\left(\matV\transp\Omega\matS\right)
\ge
1 - \sqrt{\frac{k}{r}}$ and
$\FNorm{\matA\Omega\matS}
\le \FNorm{\matA}$.
$\Omega$, $\matS$ may be viewed as matrices that sample and rescale $r$ \textit{columns} of $\matA$ and 
$\matV\transp$, namely the columns that correspond to non-zero weights $s_i$. Next we prove 
Lemma~\ref{lemma:intro2}. 

\subsection*{Proof of Lemma~\ref{lemma:intro2}}
In this section we will provide a constructive proof of Lemma~\ref{lemma:intro2}. Our proof closely follows the proof of Lemma \ref{lemma:intro1}, so we will only highlight the differences. We first discuss modifications to Algorithm~\ref{alg:2set}. First of all, the new inputs are \math{\cl V=\{\v_1,\ldots,\v_n\}} and \math{\cl A=\{\a_1,\ldots,\a_n\}}. The output is a set of $n$ non-negative weights $s_i$, at most $r$ of which are non-zero. We define the parameters
\mand{
\delta_\scl=1;
\qquad \delta_\scu=\frac{\sum_{i=1}^n \TNormS{\a_i}}{1-\sqrt{\frac{k}{r}}};
\qquad \scl_\tau=\tau-\sqrt{rk};
\qquad \scu_\tau=\tau \delta_\scu,
}
for all \math{\tau= 0,\ldots,r-1}. Let $\s_{\tau}$ denote the vector of weights at the $\tau$-th step of Algorithm~\ref{alg:2set} and initialize $\s_{0}$ and $\matA_0$ as in Algorithm~\ref{alg:2set} ($\matB_0$ will not be necessary). We now define the function $U_F\left(\a,\delta_\scu\right)$, where $\a \in \mathbb{R}^\ell$ and $\delta_\scu \in \mathbb{R}$:
\begin{equation}\label{eqn:UF}
U_F\left(\a,\delta_\scu\right) = \delta_\scu^{-1} \a^T \a.
\end{equation}
Then, at the $\tau$-th step, the algorithm will pick an index $j$ and compute a weight $t>0$ such that
\mld{
U_F(\a_j,\delta_\scu)\le t^{-1}\le
L(\v_j,\delta_\scl,\matA_\tau,\scl_\tau).
\label{eq:updateF}
}
The algorithm updates the vector of weights $\s_{\tau}$ and the matrix
$$\matA_\tau=\sum_{i=1}^n s_{\tau,i}\v_i\v_i\transp.$$
It is worth noting that the algorithm does not need to update the matrix
$$\matB_\tau=\sum_{i=1}^ns_{\tau,i}\a_i\a_i\transp,$$
because the function $U_F$ does not need $\matB_{\tau}$ as input. To prove the correctness of the algorithm we need the following two intermediate lemmas.
\begin{lemma}\label{lemma:feasibleF}
At every step \math{\tau=0,\ldots,r-1} there exists an index $j$ in $\left\{1,\ldots,n\right\}$ that satisfies eqn.~(\ref{eq:updateF}).
\end{lemma}
\begin{proof}
The proof is very similar to the proof of Lemma~\ref{lemma:feasible} (via Lemma~\ref{lemma:feasiblebound}) so we only sketch the differences. First,
note that the dynamics of \math{L} have not been changed and thus the lower bound for the average of $L(\v_j,\delta_\scl,\matA_\tau,\scl_\tau)$ still holds.
We only need to upper bound the average of \math{U_F(\a_i,\delta_\scu)} as in Lemma~\ref{lemma:feasiblebound}. Indeed,
\mand{
\sum_{i=1}^n U_F(\a_i,\delta_\scu)
=\delta_\scu^{-1}\sum_{i=1}^n\a_i\transp\a_i
=
\delta_\scu^{-1}\sum_{i=1}^n\TNormS{\a_i}
= 1-\sqrt{\frac{k}{r}},
}
where the last equality follows from the definition of \math{\delta_\scu}.
\end{proof}
\begin{lemma}\label{lemma:sp2F}
Let $W \in \mathbb{R}^{\ell \times \ell}$ be a symmetric positive semi-definite matrix, let $\a \in \mathbb{R}^{\ell}$ be a vector, and let
\math{\scu \in \mathbb{R}} satisfy $\scu > \trace(\matW)$. If \math{t > 0} satisfies
\mand{U_F\left(\a,\delta_\scu\right) \le
t^{-1},
}
then \math{\trace\left(\matW+t\v\v\transp\right)\le\scu+\delta_\scu}.
\end{lemma}
\begin{proof}
Using the conditions of the lemma and the definition of $U_F$ from eqn.~(\ref{eqn:UF}),
\eqan{
\trace(\matW+t\a\a\transp)-\scu-\delta_\scu,
&=&
\trace(\matW)-\scu+t\a\transp\a-\delta_\scu,\\
&\le&
\trace(\matW)-\scu < 0,
}
which concludes the proof of the lemma.
\end{proof}
\noindent We can now combine Lemmas~\ref{lemma:sp1} and \ref{lemma:sp2F} to prove that at all steps $\tau=0,\ldots,r-1$,
$$\lambda_{\min}(\matA_\tau)\ge\scl_\tau \qquad \mbox{and}\qquad
\trace(\matB_\tau)\le \scu_\tau.
$$
Note that after all $r$ steps of the algorithm are completed, \math{\scl_r=r\left(1-\sqrt{k/r}\right)} and
\math{\scu_r=r\left(1-\sqrt{k/r}\right)^{-1}\sum_{i=1}^n\TNormS{\a_i}}. A simple rescaling now concludes the proof.
The running time of the (modified) Algorithm~\ref{alg:2set} is $O\left(nrk^2+n\ell \right)$, where the latter term emerges from the need to compute the function $U_F(\a_j,\delta_\scu)$ for all $j=1,\ldots,n$ \textit{once} throughout the algorithm.

\section{Generalizations of Lemmas \ref{lem:2setS} and \ref{lem:2setF}}

\begin{lemma}
\label{theorem:2setGeneral}
Let \math{\matX\in\R^{n\times k}} and \math{\matY\in\R^{n\times \ell}} with
respective ranks \math{\rho_\matX,\rho_\matY};
let \math{\rho_{\matX} < r \le n}.
One can deterministically 
construct in $O(r n (\rho_{\matX}^2+\rho_{\matY}^2))$ time 
a sampling matrix \math{\Omega\in\R^{n\times r}} and a 
positive diagonal rescaling matrix \math{\matS\in\R^{r\times r}} such 
that
\mand{
\TNorm{(\matX\transp \Omega \matS )^+} \le \left(1 - \sqrt{\frac{\rho_{\matX}}{r}} \right) \TNorm{(\matX\transp)^+}   
\qquad \text{and} \qquad 
\TNorm{\matY\transp \Omega \matS}      \le \left( 1 + \sqrt{ \frac{\rho_{\matY}}{r} } \right) \TNorm{\matY\transp}.
}
\end{lemma}
\begin{proof}
Let the SVD of $\matX$ is $\matX = \matU_{\matX} \Sigma_{\matX} \matV_{\matX}\transp$,
with $\matU_{\matX} \in \R^{n \times \rho_\matX}$, $\Sigma_{\matX} \in \R^{\rho_\matX \times \rho_\matX}$,
and $\matV_{\matX} \in \R^{k \times \rho_\matX}$. 
Let the SVD of $\matY$ is $\matY = \matU_{\matY} \Sigma_{\matY} \matV_{\matY}\transp$,
with $\matU_{\matY} \in \R^{n \times \rho_\matY}$, $\Sigma_{\matY} \in \R^{\rho_\matY \times \rho_\matY}$,
and $\matV_{\matY} \in \R^{\ell \times \rho_\matY }$. 
Let $[\Omega,\matS]=BarrierSamplingII(\matU_\matX,\matU_\matY,r).$
By Lemma  \ref{lem:2setS},
\math{
\sigma_{\min}(\matU_\matX\transp\Omega\matS )\ge(1-\sqrt{\rho_{\matX}/{r}})
}, which implies 
\math{\TNorm{(\matU_\matX\transp \Omega \matS )^+}\le
(1-\sqrt{\rho_{\matX}}/{r})} and
\math{\rank(\matU_\matX\transp\Omega\matS)=\rho_\matX}. Also
\math{\TNorm{\matU_\matY\transp \Omega\matS } \le (1+\sqrt{\rho_{\matY}/{r}})}
because
\math{
\sigma_{\max}
(\matU_\matY\transp\Omega\matS )\le(1+\sqrt{\rho_{\matY}/{r}})
}. Thus,
\eqan{
\TNorm{(\matX\transp \Omega \matS )^+}
&=&
\TNorm{(\matV_{\matX} \Sigma_{\matX} \matU_{\matX}\transp \Omega \matS)^+ } 
{\buildrel (a)\over =}
\TNorm{(\matU_{\matX}\transp \Omega \matS)^+ (\matV_{\matX} \Sigma_{\matX})^+}\\
&\le&
\TNorm{(\matU_{\matX}\transp \Omega \matS)^+}
\TNorm{ (\matV_{\matX} \Sigma_{\matX})^+};}
$$
\TNorm{ \matY \transp \Omega\matS }  = 
 \TNorm{ \matV_\matY \Sigma_\matY \matU_\matY\transp \Omega\matS } \leq 
 \TNorm{ \matV_\matY \Sigma_\matY }  \TNorm{ \matU_\matY\transp \Omega\matS }. 
$$
(The inequalities follow from submultiplicativity, and (a) uses
Lemma ~\ref{lemma:pseudo}.)
 To conclude, observe that 
\math{\TNorm{(\matV_\matX\Sigma_\matX)^+}=\TNorm{(\matX\transp)^+}},
and 
\math{\TNorm{\matV_\matY\Sigma_\matY}=\TNorm{\matY\transp}}.
\end{proof}

\begin{lemma}
\label{theorem:2setGeneralF}
Let \math{\matX\in\R^{n\times k}} and \math{\matY\in\R^{\ell \times n}} with
respective ranks \math{\rho_\matX,\rho_\matY};
let \math{\rho_{\matX} < r \le n}.
One can deterministically 
construct in $O(r n \rho_{\matX}^2 + n \ell )$ time 
a sampling matrix \math{\Omega\in\R^{n\times r}} and a 
positive diagonal rescaling matrix \math{\matS\in\R^{r\times r}} such 
that
\mand{
\TNorm{(\matX\transp \Omega \matS )^+} \le \left(1 - \sqrt{\frac{\rho_{\matX}}{r}} \right) \TNorm{(\matX\transp)^+}   \qquad \text{and}\qquad 
\FNorm{\matY \Omega \matS}      \le  \TNorm{\matY}.
}
\end{lemma}
\begin{proof}
Set \math{[\Omega,\matS]=BarrierSamplingIII(\matU_{\matX},\matY,r)} and follow the 
argument in Lemma \ref{theorem:2setGeneral}.
\end{proof}


\section{Proof of Lemma \ref{lem:rpall}}

We start with the following definition, which is Definition 1 in~\cite{Sar06}
(properly stated to fit our notation).
\begin{definition}[Johnson-Lindenstrauss Transform]
\label{lem:jlt}
Let $\matA \in \R^{m \times n}$, parameters $0 < \delta, \epsilon < 1$ and function $f$.
A random matrix $R \in \R^{n \times r}$ with $r = \Omega( \frac{\log m}{\epsilon^2} f(\delta) )$
forms a Johnson-Lindenstrauss transform with parameters $\epsilon, \delta, m$ or 
$JLT(\epsilon, \delta, m)$ for short, if with probability at least $1 - \delta$
for all rows $\matA_{(i)} \in \R^{1 \times n}$ of $\matA$:
$$ (1 - \epsilon) \TNormS{ \matA_{(i)} } \le \TNormS{ \matA_{(i)} R } \le (1 + \epsilon) \TNormS{ \matA_{(i)} }.$$
\end{definition}

We continue with Theorem $1.1$ of \cite{Ach03} (properly stated to fit our notation
and after minor algebraic manipulations), which indicates that the (rescaled) 
sign matrix $\R$ of the lemma corresponds to a $JLT(\epsilon, \delta, m)$ transform
as defined in Definition~\ref{lem:jlt} with $f(\delta) = \log(\frac{1}{\delta})$.
\begin{theorem} [Achlioptas~\cite{Ach03}]
Let $\matA \in \R^{m \times n}$ and $0 < \epsilon < 1$. Let $R \in \R^{n \times r}$ be a rescaled random sign matrix 
with  $r = \frac{36}{\epsilon^2} \log(\frac{1}{\delta}) \log m$.
Then with probability $1 - \delta$:
$$ (1 - \epsilon) \TNormS{ \matA_{(i)} } \le \TNormS{ \matA_{(i)} R } \le (1 + \epsilon) \TNormS{ \matA_{(i)} }.$$
\end{theorem}

\paragraph{Statement 1.}
The first statement in our lemma proved in Corollary 11 of~\cite{Sar06}. More specifically,
this corollary indicates that if $R$ is a $JLT(\epsilon, \delta, m)$ with 
$r = O( k/\epsilon^2 \log(\frac{1}{\delta}))$, then, the singular values of the subsampled $\matV_k\transp R$
are within relative error accuracy from the singular values of $\matV_k$ with probability at least $1 - \delta$.
The failure probability $0.01$ in our lemma follows by assuming $c_0$ sufficiently large. 

\paragraph{Statement 2.}
This matrix multiplication bound follows from Lemma~$6$ of~\cite{Sar06}. The second claim
of this lemma says that for $\matX \in \R^{m \times n}$ and $\matY \in \R^{n \times k}$,
if $R \in \R^{n \times r}$ is a matrix with i.i.d rows, each one containing 
four-wise independent zero-mean $\{1/\sqrt{r}, -1/\sqrt{r}\}$ entries, then
$$ \Expect{ \FNormS{ \matX \matY - \matX R R\transp \matY } } \le \frac{2}{r} \FNormS{ \matX } \FNormS{ \matY }.$$
The statement in our lemma follows because our rescaled sign matrix satisfies the four-wise independence assumption.
Actually, our rescaled sign matrix has i.i.d rows with each one having $r$-wise 
independent zero-mean $\{1/\sqrt{r}, -1/\sqrt{r}\}$ entries. Assuming $r \ge 4$
the claim follows. 

\paragraph{Statement 3.}
These bounds appeared in Lemma~8 in~\cite{Sar06}. 
Again, \cite{Sar06} assumes that the matrix has i.i.d rows, each one containing 
four-wise independent zero-mean $\{1/\sqrt{r}, -1/\sqrt{r}\}$ entries.
The statement in our lemma follows because our rescaled sign matrix satisfies 
the four-wise independence assumption.

\paragraph{Statement 4.}
Let $\matX = \matV_k\transp R \in \R^{k \times r}$
with SVD $\matX = \matU_{\matX} \Sigma_{\matX} \matV_\matX\transp $.
Here,  $\matU_{\matX} \in \R^{k \times k}$, 
and $\Sigma_{\matX} \in \R^{k \times k}$, and $\matV_\matX \in \R^{r \times k}$.
Consider taking the SVD of $(\matV_k\transp R)^+$ and $(\matV_k\transp R)\transp $:
$$ \TNorm{(\matV_k\transp R)^+ - (\matV_k\transp R)\transp} =
    \TNorm{\matV_{\matX} \Sigma_{\matX}^{-1} \matU_{\matX}\transp - \matV_{\matX} \Sigma_{\matX} \matU_{\matX}\transp  } =
    \TNorm{\matV_{\matX}(\Sigma_{\matX}^{-1} - \Sigma_{\matX}) \matU_{\matX}\transp} = $$
$$ =    \TNorm{\Sigma_{\matX}^{-1} - \Sigma_{\matX}},
  $$
since $\matV_{\matX}$ and
$\matU_{\matX}\transp $ can be dropped without changing the spectral norm. 

Let $\matY = \Sigma_{\matX}^{-1} - \Sigma_{\matX} \in \R^{k \times k}$ be 
a diagonal matrix. Then, for all
$i=1,...,k$:
$$\matY_{ii}  = \frac{ 1 - \sigma_i^2(\matX)  }{ \sigma_{i}(\matX) }.$$
Since $\matY$ is a diagonal matrix:
$$
\TNorm{ \matY } = 
\max_{1 \leq i \leq k} \abs{\matY_{ii} } =
\max_{1 \leq i \leq k} \abs{\frac{ 1 - \sigma_i^2(\matX)}{ \sigma_{i}(\matX) } } \le
\max_{1 \leq i \leq k} \frac{ 1 - \sigma_i^2(\matX)}{ \sigma_{i}(\matX) } \le
\frac{ 2 \epsilon - \epsilon^2}{1-\epsilon} \le \frac{ 2 \epsilon }{1-\epsilon} \le 3 \epsilon.
$$
In the first inequality, $ \abs{ \sigma_{i}(\matX) } = \sigma_{i}(\matX)$, because
the singular values are nonnegative numbers. Also, 
$\abs{ 1 - \sigma_i^2(\matX) } \le 1 - \sigma_i^2(\matX)$, because
$1 - \sigma_i^2(\matX) \ge 0$, since $\sigma_i^2(\matX) < 1$
(from our choice of $\epsilon$ and the left hand side of the
bound for the singular values from the first statement of the Lemma). 
The second inequality follows by the bound for the singular values of $\matX$
from the first statement of the Lemma. The last inequality follows by the
assumption that $0 < \epsilon < \frac{1}{3}$. 

\paragraph{Statement 5.}
Notice that there exists a sufficiently large constant $c_0$ such that $r \geq c_0 k / \epsilon^2$.
Setting $x = \FNorm{ \matX R }^2$, using the third statement of the
Lemma, the fact that $k \geq 1$, and Chebyshev's
inequality we get
\begin{eqnarray*}
\textbf{Pr}[ |x - \EE{x }| \geq \epsilon \FNorm{\matX}^2 ] 
\le \frac{\var{x}}{\epsilon^2 \FNorm{\matX}^4}
\le \frac{2\FNorm{\matX}^4}{r\epsilon^2 \FNorm{\matX}^4} 
\le \frac{2}{c_0k} 
\le 0.01.
\end{eqnarray*}
The last inequality follows by assuming $c_0$ sufficiently large.
Finally, taking square root on both sides concludes the proof.

\paragraph{Statement 6.} Let
$\matE = \matA_k  - (\matA R) (\matV_k\transp  R)^+ \matV_k\transp  \in \R^{m \times n}$. 
By setting $\matA = \matA_k + \matA_{\rho-k}$ and using the triangle
inequality for the Frobenius matrix norm:
\[ \FNorm{\matE}\ \leq\ \FNorm{\matA_k - \matA_k R (\matV_k\transp R)^+ \matV_k\transp }\ +\ \FNorm{ \matA_{\rho-k}R (\matV_k\transp R)^+ \matV_k\transp  }.\]
The first statement of the Lemma  implies that
$\text{rank}(\matV_k\transp R) = k$ thus $(\matV_k\transp R) (\matV_k\transp R)^+ =
\matI_k$. Replacing
$\matA_k = \matU_k \Sigma_k \matV_k\transp $ and setting $(\matV_k\transp R)
(\matV_k\transp R)^+ = \matI_k$ we get that
$$ 
\FNorm{\matA_k - \matA_k R (\matV_k\transp R)^+ \matV_k\transp } 
= 
\FNorm{\matA_k - \matU_k\Sigma_k \matV_k\transp R (\matV_k\transp R)^+ \matV_k\transp } = 
\FNorm{\matA_k - \matU_k \Sigma_k \matV_k\transp } = 0. $$
To bound the second term above, we drop $\matV_k\transp $, add and
subtract $\matA_{\rho-k}R (\matV_k\transp R)\transp   $,
and use the triangle inequality and spectral submultiplicativity:
\begin{eqnarray*}
\FNorm{ \matA_{\rho-k} R (\matV_k\transp R)^+ \matV_k\transp  } & \leq & \FNorm{ \matA_{\rho-k}R (\matV_k\transp R)\transp  }\ +\ \FNorm{ \matA_{\rho-k}R ( (\matV_k\transp R)^+ - (\matV_k\transp R)\transp )} \\
                                               & \leq & \FNorm{ \matA_{\rho-k}R R\transp  \matV_k  }\ +\ \FNorm{ \matA_{\rho-k}R} \TNorm{ (\matV_k\transp R)^+ - (\matV_k\transp R)\transp                                           }.
\end{eqnarray*}
Now, we will bound each term individually. A crucial observation
for bounding the first term is that $\matA_{\rho-k}\matV_k=
\matU_{\rho-k}\Sigma_{\rho-k}\matV_{\rho-k}^{\top}\matV_k=\mathbf{0}_{m \times k}$ by
orthogonality of the columns of $\matV_k$ and $\matV_{\rho-k}$. This term
now can be bounded using the second statement of the Lemma
with $\matX = \matA_{\rho - k}$ and $\matY = \matV_k$. This statement, 
assuming $c_0$ sufficiently large, an application of Markov's inequality on the random variable
$x = \FNormS{\matA_{\rho-k}R R\transp  \matV_k  - \matA_{\rho-k}\matV_k}$, and taking square root
on both sides of the resulting equation, give that
w.p. $0.99$,
$$
    \FNorm{\matA_{\rho-k}R R\transp  \matV_k }\ \leq\ 0.5 \epsilon \FNorm{\matA_{\rho-k}}.
$$    
The second two terms can be bounded using the first statement of the Lemma
and the third statement of the Lemma with $\matX = \matA_{\rho-k}$. 
Hence by applying a union bound, we get that w.p. at least $0.97$:
\begin{eqnarray*}
\FNorm{ \matE }  & \leq & \FNorm{ \matA_{\rho-k}R R\transp  \matV_k  } + \FNorm{ \matA_{\rho-k}R} \norm{ (\matV_k\transp R)^+ - (\matV_k\transp R)\transp } \\
                           & \leq &  0.5 \epsilon \FNorm{ \matA_{\rho-k}} + \sqrt{(1+\epsilon)} \FNorm{\matA_{\rho-k}} \cdot 3 \epsilon \\
                           & \leq &  0.5 \epsilon \FNorm{ \matA_{\rho-k}} +  3.5 \epsilon \FNorm{\matA_{\rho-k}} \\
                           &  = &  4 \epsilon\cdot \FNorm{\matA_{\rho-k}}.
\end{eqnarray*}
The last inequality holds thanks to our choice of $\epsilon \in
(0,1/3)$.

\section{Proof of Lemma \ref{lem:rsall}}

\paragraph{Statement 1.}
Notice that this is a matrix-multiplication-type term involving the multiplication
of the matrices $\matE$ and $\matZ$. The sampling and rescaling matrices $\Omega, \matS$
indicate the subsampling of the columns and rows of $\matE$ and $\matZ$, respectively.
Eqn.~(4) of Lemma 4 of \cite{DKM06a} gives a bound for such matrices $\Omega, \matS$
constructed with randomized sampling with replacement and any set of probabilities
$p_1, p_2,...,p_n$ (over the columns of $\matE$):
$$ \Expect{ \FNormS{ \matE \matZ - \matE \Omega \matS \matS\transp \Omega\transp \matZ } } \le
\sum_{i=1}^{n} \frac{ \TNormS{\matE^{(i)}} \TNormS{ \matZ_{(i)} } }{r p_i} - \frac{1}{r} \FNormS{\matE \matZ}. $$
Notice that $\matE \matZ = \bm{0}_{m \times k}$ by construction. Now, 
replace the values of $p_i = \frac{ \TNormS{ \matZ_{(i)}  }}{k}$ (in Definition~\ref{def:sampling}) and rearrange:
$$ \Expect{ \FNormS{\matE \Omega \matS \matS\transp \Omega\transp \matZ } } \le
\frac{k}{r} \FNormS{\matE}.$$
Finally, apply Markov's inequality 
to the random variable $x = \FNormS{\matE \Omega \matS \matS\transp \Omega\transp \matZ }$ to wrap up. 
\paragraph{Statement 2.}
First, by Lemma \ref{lem:random} and our choice of $r$, w.p. $1 - \delta$:
$$1 -  \epsilon  \leq  \sigma_i^2(\matV\transp \Omega \matS)   \leq 1 + \epsilon.$$
Let $\matX = \matZ\transp \Omega \matS \in \R^{k \times r}$ with SVD:
$\matX = \matU_{\matX} \Sigma_{\matX} \matV_\matX\transp$.
Here,  $\matU_{\matX} \in \R^{k \times k}$, 
and $\Sigma_{\matX} \in \R^{k \times k}$, and $\matV_\matX \in \R^{r \times k}$.
By taking the
SVD of $\matX^+$ and $\matX\transp$:
$$
\TNorm{ (\matZ\transp \Omega \matS)^+ - (\matZ\transp \Omega \matS)\transp } = 
\TNorm{ \matV_{\matX} \Sigma_{\matX}^{-1} \matU_{\matX}\transp - \matV_{\matX}\Sigma_{\matX} \matU_{\matX}\transp } = 
\TNorm{ \matV_{\matX}(\Sigma_{\matX}^{-1} - \Sigma_{\matX}) \matU_{\matX}\transp } =$$ 
$$ = \TNorm{ \Sigma_{\matX}^{-1} - \Sigma_{\matX} },$$ 
since $\matV_{\matX}$ and
$\matU_{\matX}\transp $ can be dropped without changing the spectral norm. 
Let $\matY = \Sigma_{\matX}^{-1} - \Sigma_{\matX} \in \R^{k \times k}$ be diagonal;
Then, for all $i=1,...,k$:
$$ \matY_{ii}  = \frac{ 1 - \sigma_i^2(\matX)  }{ \sigma_{i}(\matX) }.$$
Since $\matY$ is a diagonal matrix:
$$ \TNorm{ \matY } = 
\max_{1 \leq i \leq k} \abs{ \matY_{ii} }= 
\max_{1 \leq i \leq k} \abs{\frac{ 1 - \sigma_i^2(\matX)}{ \sigma_{i}(\matX) } } =
\max_{1 \leq i \leq k} \frac{ \abs{1 - \sigma_i^2(\matX)}}{ \sigma_{i}(\matX) } \le
\frac{\epsilon}{\sqrt{1-\epsilon}} .
$$
The inequality follows by using the bounds for $\sigma_{i}^2(\matX)$ from above. 
The failure probability is $\delta$ because the bounds for $\sigma_{i}^2(\matX)$
fail with this probability. 

\paragraph{Statement 3.}

First, notice that from the first statement of the lemma $\sigma_k(\matZ\transp \Omega \matS) > 0$ w.p. $1-\delta$.
This implies that $\matZ\transp \Omega \matS$ is a rank $k$ matrix w.p $1 - \delta$, so, with the same probability:
$$(\matZ\transp \Omega \matS)(\matZ\transp \Omega \matS)^+ = \matI_{k}.$$
Now notice that:  
$$ \matB \matZ\transp - \matB \matZ\transp \Omega \matS (\matZ\transp \Omega \matS)^+\matZ\transp =
\matB \matZ\transp - \matB \matI_{k} \matZ\transp = \bm{0}_{m \times n}.$$
Next, we manipulate the term 
$\FNorm{ \matB \matZ\transp - \matA \Omega \matS (\matZ\transp \Omega \matS)^+\matZ\transp}$ as follows ($\matA = \matB\matZ\transp + \matE$):
\begin{eqnarray*}
\FNorm{ \matB \matZ\transp - \matA \Omega \matS (\matZ\transp \Omega \matS)^+\matZ\transp} 
&=& 
\FNorm{ 
\underbrace{ \matB \matZ\transp - \matB \matZ\transp \Omega \matS (\matZ\transp \Omega \matS)^+\matZ\transp}_{\bm{0}_{m \times n}} 
- \matE \Omega \matS (\matZ\transp \Omega \matS)^+\matZ\transp } \\
         &=&  \FNorm{\matE \Omega \matS (\matZ\transp \Omega \matS)^+\matZ\transp }.
\end{eqnarray*}
For notational convenience, let $\matX = (\matZ\transp \Omega \matS)^+ - (\matZ\transp \Omega \matS)\transp$.
Finally, we manipulate the term $\FNorm{\matE \Omega \matS (\matZ\transp \Omega \matS)^+\matZ\transp }$
as follows:
\begin{eqnarray*} 
\FNorm{\matE \Omega \matS (\matZ\transp \Omega \matS)^+\matZ\transp } 
&\le& \FNorm{\matE \Omega \matS (\matZ\transp \Omega \matS)^+ } \\ 
&\le& \FNorm{\matE \Omega \matS (\matZ\transp \Omega \matS)\transp} + \FNorm{\matE \Omega \matS} \TNorm{\matX} \\
&\le& \sqrt{\frac{k}{\delta r}} \FNorm{\matE} + \frac{1}{\sqrt{\delta}} \FNorm{\matE} \frac{\epsilon}{\sqrt{1-\epsilon}} \\
&\le& \left( \sqrt{ \frac{k}{\delta r} } + \frac{\epsilon}{ \sqrt{\delta} \sqrt{1-\epsilon}} \right) \FNorm{\matE}  \\   
&\le& \left( \frac{\epsilon}{2 \sqrt{\delta}} \frac{1}{\sqrt{\ln(2k/\delta)}} + \frac{\epsilon}{\sqrt{\delta}\sqrt{1-\epsilon}} \right)\FNorm{\matE} \\
&\le& \left( \frac{\epsilon}{2 \ln(4) \sqrt{\delta}}  + \frac{\epsilon}{\sqrt{\delta}\sqrt{1-\epsilon}} \right)\FNorm{\matE} \\
&\le& \frac{ 1.6 \epsilon}{\sqrt{\delta}}  \FNorm{\matE}.
\end{eqnarray*}
The first inequality follows by spectral submultiplicativity and the fact that $\TNorm{\matZ\transp}=1$.
The second inequality follows by the triangle inequality for matrix norms. 
In the third inequality, 
the bound for the term $\FNorm{\matE \Omega \matS (\matZ\transp \Omega \matS)\transp}$ follows by the first
statement of the lemma w.p. $1-\delta$, $\FNorm{\matE \Omega \matS}$ is bounded using Lemma~\ref{lem:fnorm} w.p. $1-\delta$, while we bound $\TNorm{\matX}$ by the second statement of the lemma w.p. $1-\delta$. So, by the union bound, the failure probability so far is $3\delta$.
The rest of the argument follows by our choice of $r$, assuming $k > 1$, $\epsilon < \frac{1}{3}$, 
and simple algebraic manipulations.

\section{A Lower Bound for Spectral Column-based Matrix Approximation}
\label{sec:lower}
\begin{theorem}
\label{theorem:lower1}
For any \math{\alpha>0}, any $k \geq 1$, and any $r \geq k$, there exists a matrix \math{\matA \in \mathbb{R}^{m \times n}} for which
$$\frac{\TNorm{\matA-\matC\matC^+\matA}^2}{\TNorm{\matA-\matA_k}^2} \ge\frac{n+\alpha^2}{r+\alpha^2}.$$
Here $\matC$ is any matrix that consists of $r$ columns of $\matA$. As $\alpha \rightarrow 0$, this implies a lower bound of $n/r$ for the approximation ratio of the spectral norm column-based matrix reconstruction problem.
\end{theorem}
\begin{proof}
We extend the lower bound in \cite{DR10} to arbitrary $r > k$. Consider the matrix
$$\matA = [\e_1+\alpha\e_2, \e_1+\alpha\e_3,\ldots, \e_1+\alpha\e_{n+1}] \in\R^{(n+1)\times n},$$
where \math{\e_i\in\R^{n+1}} are the standard basis vectors. Then,
$$\matA\transp\matA=\bm{1}_n\bm{1}_n\transp+\alpha^2\matI_{n}, \qquad
\sigma_1^2(\matA)=n+\alpha^2, \qquad \mbox{and} \qquad \sigma_i^2(\matA)=\alpha^2 \mbox{\ \ for\ \ } i>1.$$
Thus, for all $k\ge1$, $\TNorm{\matA-\matA_k}^2=\alpha^2$. Intuitively, as \math{\alpha\rightarrow0}, \math{\matA} is a rank-one matrix. Consider any \math{r} columns of $A$ and note that, up to row permutations, all sets of $r$ columns of $A$ are equivalent. So, without loss of generality, let \math{\matC} consist of the first \math{r} columns of \math{\matA}. We now compute the optimal reconstruction of $\matA$ from $\matC$ as follows: let $\a_j$ be the $j$-th column of $A$. In order to reconstruct \math{\a_j}, we minimize \math{\TNormS{\a_j-\matC\x}} over all vectors \math{\x\in\R^{r}}. Note that if \math{j\le r} then the reconstruction error is zero. For $j > r$, \math{\a_j=\e_1+\alpha\e_{j+1}} and \math{\matC\x=\e_1\sum_{i=1}^r x_i+\alpha\sum_{i=1}^rx_i\e_{i+1}}. Then,
$$\TNorm{\a_j-\matC\x}^2=
\TNorm{\e_1\left(\sum_{i=1}^rx_i-1\right)+
\alpha\sum_{i=1}^rx_i\e_{i+1}-e_{j+1}}^2
=
\left(\sum_{i=1}^rx_i-1\right)^2+\alpha^2\sum_{i=1}^rx_i^2+1.
$$
The above quadratic form in \math{\x} is minimized when \math{x_i=\left(r+\alpha^2\right)^{-1}} for all $i=1,\ldots,r$. Let $\hat\matA = A - CC^+A$ and let the $j$-th column of $\hat\matA$ be $\hat\a_j$. Then, for \math{j\le r}, \math{\hat\a_j} is an all-zeros vector; for \math{j>r}, \math{\hat\a_j=\alpha\e_{j+1}-\frac{\alpha}{r+\alpha^2}\sum_{i=1}^r\e_{i+1}}.
Thus,
$$\hat\matA\transp\hat\matA=
\left[
\begin{matrix}
\bm0_{r\times r}&\bm0_{r\times (n-r)}\\
\bm0_{(n-r)\times r}&\matZ
\end{matrix}
\right],$$
where \math{\matZ=\frac{\alpha^2}{r+\alpha^2}\bm1_{n-r}\bm1_{n-r}\transp +\alpha^2\matI_{n-r}}. This immediately implies that
$$\TNorm{\matA-\matC\matC^+\matA}^2 = \TNorm{\hat\matA}^2 = \TNorm{\hat\matA^T \hat\matA} = \TNorm{\matZ}^2 =  \frac{(n-r)\alpha^2}{r+\alpha^2}+\alpha^2 = \frac{n+\alpha^2}{r+\alpha^2}\alpha^2=$$
$$ = \frac{n+\alpha^2}{r+\alpha^2}\TNormS{A-A_k}.$$
This concludes our proof.
\end{proof}

\chapter{VERY RECENT WORK ON MATRIX SAMPLING ALGORITHMS}
While writing the results of this thesis, we learned the
existence of two independent concurrent works that study low-rank
column-based matrix approximation in the Frobenius norm
and coreset construction for constrained least-squares regression, respectively. 
The recent work on Column-based Matrix Reconstruction appeared
as a technical report in~\cite{GK11}; while, the recent
work on coreset construction for Linear Regression will appear
soon in the Proceedings of the 43rd ACM Symposium on Theory of 
Computing (STOC). We comment on these results below. 

\section{Column-based Reconstruction in the Frobenius Norm}
Two exciting results for column-based matrix reconstruction appeared recently in~\cite{GK11}. More specifically, 
Theorem 1 in~\cite{GK11} describes an algorithm that
on input a matrix $\matA \in \R^{m \times n}$, a target rank $k$, and an
oversampling parameter $r \ge k$, constructs $\matC \in \R^{m \times r}$ deterministically such that:
$$ \FNormS{ \matA - \matC \matC^+\matA } \le \frac{r+1}{r+1-k} \FNormS{\matA - \matA_k}. $$
Moreover, this algorithm runs in $O(r n m^3 \log(m))$. 
Notice that, for any $\epsilon > 0$, it suffices to select only $r = k + \frac{k}{\epsilon}-1$
columns to get a relative error approximation. Theorem 2 in~\cite{GK11} presents a faster
algorithm that runs in $O(rnm^2)$ (assuming $m \le n$) and constructs $\matC \in \R^{m \times r}$ randomly such that:
$$ \Expect{\FNormS{ \matA - \matC \matC^+\matA }} \le \frac{r+1}{r+1-k} \FNormS{\matA - \matA_k}. $$
Both these results are based on volume sampling approaches (see Section~\ref{chap313}). The
work of~\cite{GK11} was motivated by applications of column selection to 
approximation algorithms for Quadratic Integer Programming~\cite{GK11b}.
Moreover,~\cite{GK11} proved a new lower bound: 
$$ \frac{ \FNormS{ \matA - \matC \matC^+\matA } } { \FNormS{\matA - \matA_k} } \ge  1 + \frac{k}{r} - o(1),$$ 
which implies that the bounds from~\cite{GK11} are asymptotically optimal. Finally, we should note that
the bounds from~\cite{GK11} hold for the residual error $\FNorm{ \matA - \matC \matC^+\matA }$ and
it is not obvious if they can be extended to the rank $k$ matrix $\Pi_{\matC,k}^{F}(\matA)$. 
Recall that for arbitrary $r > k$:
$$ \FNorm{ \matA - \matC \matC^+\matA } \le \FNorm{ \matA - \Pi_{\matC,k}^{F}(\matA) }.$$
So, an upper bound for the term $\FNorm{ \matA - \matC \matC^+\matA }$ does
not imply an upper bound for the term $\FNorm{ \matA - \Pi_{\matC,k}^{F}(\matA) }$.

\section{Coreset Construction for Linear Regression}

Feldman and Langberg in~\cite{FL11} proved that a ($1+\epsilon$)-coreset with
$r = O(n / \epsilon^2)$ rows can be found w.p. $0.5$ in 
$O( mn^2 + n^2 \log(m)/\epsilon^2)$ time. Although not explicitly stated, it is easy to see
that the coreset of~\cite{FL11} applies to arbitrary constrained regression. 
The upcoming full version of~\cite{FL11} will give all the details
of this result. The novel technical ingredient in~\cite{FL11} is
the randomized fast analog of 
Corollary~\ref{cor:1set} of Section~\ref{chap316}. \cite{FL11}
gives a coreset which is as good (in terms of coreset size) as 
ours in Theorem~\ref{lem:regression}. The running time of the
method of~\cite{FL11} is better than ours; this trades the failure
probability introduced in~\cite{FL11}. Recall that our theorem
constructs the coreset deterministically.

\end{document}